\documentclass[final,3p, times, 12pt]{elsarticle}

\usepackage{amssymb}
\usepackage[english]{babel}
\usepackage{graphicx,epsfig,color}
\usepackage{float}
\usepackage{wrapfig}
\usepackage{subfig}
\usepackage{amsfonts,amsmath,latexsym}
\usepackage{amsmath}
\usepackage{amsthm}
\usepackage[makeroom]{cancel}
\usepackage{tabularx}
\usepackage{bbm, dsfont}
\usepackage{upgreek}
\usepackage{multirow}
\usepackage{soul}

\addcontentsline{toc}{section}{References}
\definecolor{light}{gray}{.85}
 \newcommand{\norm}[1]{\left\lVert#1\right\rVert}

\newtheorem{cor}{Corollary}[section]
\newtheorem{prn}{Proposition}[section]
\newtheorem{asm}{Assumption}

\numberwithin{equation}{section} \numberwithin{prn}{section}
\numberwithin{cor}{section} \numberwithin{thm}{section}
\numberwithin{lea}{section}

\makeatletter
\def\hlinewd#1{%
\noalign{\ifnum0=`}\fi\hrule \@height #1 %
\futurelet\reserved@a\@xhline}
\makeatother

\usepackage{lipsum}
\makeatletter
\def\ps@pprintTitle{%
 \let\@oddhead\@empty
 \let\@evenhead\@empty
 \def\@oddfoot{}%
 \let\@evenfoot\@oddfoot}
\makeatother
\begin{document}
\begin{frontmatter}



\title{\bf
Estimation of multiple
change points under a generalised Ornstein-Uhlenbeck 
framework}


\author[*]
{Fuqi Chen}~~~~~~~~~~~~~~~~~
\author[*,AM]
{Rogemar Mamon\corref{cor1}}
\author[*,AM]{Matt Davison}
\address[*]
{Department of Statistical and Actuarial Sciences,
University of Western Ontario, London, Ontario, Canada}
\address[AM]{Department of Applied Mathematics,
University of Western Ontario,
London, Ontario, Canada}
\cortext[cor1]{Corresponding author. Department of Statistical
and Actuarial Sciences, University of Western Ontario,
1151 Richmond Street, London, Ontario, Canada, N6A 5B7.
Email: \textit{rmamon@stats.uwo.ca}}

\begin{abstract}
The use of an Ornstein-Uhlenbeck (OU) process is ubiquitous
in business, economics and finance to capture
various price processes and evolution of economic indicators
exhibiting mean-reverting properties. When structural changes
happen, econo
mic dynamics drastically change and the
times at which these occur are of particular interest
to policy makers, investors and financial product providers.
This paper addresses
the change-point problem under a generalised OU model and
investigates the associated statistical inference. We propose two estimation methods to locate multiple change
points and show the
asymptotic properties of the estimators.
An informational approach is employed in
detecting the change points, and the
consistency of our methods is also theoretically
demonstrated. Estimation is considered under the setting where
both the number and location of change points are unknown.
Three computing algorithms are further developed for implementation.
The practical applicability of our methods is illustrated using simulated
and observed financial market data.
\end{abstract}

\end{frontmatter}

\section{Introduction}
\noindent We examine the change-point detection problem on the drift parameters of a generalised version of the Ornstein-Uhlenbeck (OU) process introduced in Dehling, et al. (2010); see also Dehling, et al. (2014) and Zhang~(2015). Such a process $X_t$ is a solution to the stochastic differential equation (SDE)
\begin{eqnarray}\label{ou1}
dX_t&=&S(\boldsymbol{\mathrm \uptheta},t,X_t)dt+\sigma dW_t,\quad 0<t\leq T,
\end{eqnarray}
where $S(\boldsymbol{\mathrm \uptheta},t,X_t)=L(t)-a X_t=\sum_{i=1}^p\mu_i\varphi_k(t)-a X_t$, $i=1,~\ldots,~p$,
$\boldsymbol{\mathrm \uptheta}=(\mu_1,...,\mu_p,-a)^{\top}$ and $\top$ denotes the transpose of a matrix. Here, $W_t$ is a one-dimensional standard Brownian motion defined on some probability space $(\Omega, {\cal F}, P).$ In particular, if $L(t)=\mu$ then (\ref{ou1}) is the SDE of the classical OU process, which is
commonly used to model the stochastic dynamics of various financial variables. \\
\ \\
\noindent Many economic indicators, prices in the financial market as well as processes in the natural and physical sciences and engineering are captured sufficiently by the OU model. The classical work of Vasicek
(1977) employs an OU model for bond valuation. The importance of
this stochastic process is also demonstrated by its ubiquity in many fields.
For instance, the OU process is {\color{black}used in mathematical models of the electricity market} (e.g., Erlwein, et al. (2010)), commodity futures market
(e.g., Date, et al. (2010)), weather derivatives (e.g., Elias,
et al. (2014)), central-bank {\color{black}rate setting} policy (e.g., Elliott and Wilson (2007)), {\color{black} spreads between pairs of securities (e.g., Elliott (2005)),}
stochastic control-driven insurance problems (e.g., Liang, et al. (2011)),
spot freight rates in the shipping industry (e.g., Benth, et al. (2015)), risk management (e.g.,
Date and Bustreo (2015), and power generation (e.g., Howell, et al. (2011)).
{\color{black} In the OU-modelling context, Tenyakov, et al. (2016) {\color{black}proposed} a
signal processing-based approach to determine presence of market liquidity regimes.}
Various applications of the OU process are also highlighted in biology (e.g., Rohlfs, et al. (2010)),
neurology (e.g., Shinomoto, et al. (1999)), survival analysis (e.g., Aalen and Gjessing (2004)),
physics (e.g., Lansk\'{y} {\color{black}and} Sacerdote (2001)), and chemistry (e.g., Lu (2003) and (2004)).
\\
\ \\
\noindent We note that the mean-reverting level of an OU process is constant, which {\color{black} can be} a notable weakness {\color{black} for many financial datasets}. This
{\color{black}may be} rectified by introducing a generalised OU process where a time-dependent function
describes its level of mean reversion. Such a generalised version
incorporates time-inhomogeneity and seasonality of mean reversion simultaneously.
Dehling, et al. (2014) developed the framework to study a change-point phenomenon under
the generalised OU process. This allows the model to capture drastic changes {\color{black}at} certain time points (e.g., {\color{black}drastic-moving} interest rates
due to the outbreaks of financial crisis or war).
In practice, many data series are characterised by some potential changes
in their evolution structure, i.e., a sudden change in mean or variance and other model parameters.
It is then of interest to determine the (i) existence
and (ii) location of the change point. This implies segregating the data series
into different segments and analysing them in a {\color{black}less efficient but more accurate} way.
{\color{black} Thus, our research contributions
support and {\color{black}complement} the objective of papers employing regime-switching OU-process
as we provide a methodology to verify the switching phenomenon in the data.
We go further by precisely estimating where the switch occurred
and how many switches are possible given a data set. An instance of this
support and {complementarity} are depicted in Subsection 3.1 of Tenyakov, et al. (2016),
where a simple statistical testing of regime-switching in the data was performed.}
\\
\ \\
\noindent Pioneering contributions to this field of change-point detection
were spearheaded by
Page~(1954) and Shiryaev~(1963). Advances in recent years
have tackled the (i) estimation of change points and coefficients
of linear regression models with multiple change points (Bai and Perron~(1998);
Perron and Qu~(2006); Lu and Lund~(2007), Gombay~(2010), and Chen and Nkurunziza~(2015)); (ii) change-point
testing for the drift parameters of a periodic mean-reverting process (cf. Dehling,
et al.~(2014)); (iii) applications in finance
(cf. Spokoiny~(2009)); (iv) detection of malware within software
(Yan, et al.~(2008)); (v) climatology (Reeves, et al. (2007), Robbins, et al.~(2011), Gallagher et al.~(2012));
and epidemiology (Yu, et al.~(2013)).
The analysis of change points could be described more generally as a
hypothesis-testing problem for the existence of change points in various
locations.  This could be viewed, from another perspective, as a model selection problem
where the change points are additional unknown parameters to be estimated. \\
\ \\
\noindent The change-point problems
are typically examined depending on two {\color{black}alternatives: (i) the} number of change points is known but their exact
locations are unknown (Perron and Qu (2006) and Chen and Nkurunziza~(2015)) and
(ii) both the number and the exact locations
of the change points are unknown. The estimation methods under the first scenario only require
the identification of the exact locations of the change points. {\color{black}It is easier for the first alternative than for the second}. Closed-form solutions
{\color{black}for} the direct calculation of the change point are usually not available.
Current change-point estimation approaches are
normally constructed to perform a search at every possible location of
unknown candidate change points via some efficient computational algorithms
subject to some constraint or criteria.
Examples of well-known algorithms for
change point detection include: (i) {\color{black}the} binary segmentation type algorithm
(Scott and Knott, (1974); Sen and Shrivastava, (1975)),
(ii) {\color{black}the} segment-neighbourhood type algorithm (Auger and Lawrence, (1989);
Bai and Perron, (1998)) with adaption to the restricted regression
model (Perron and Qu, (2006)); and (iii) {\color{black}the} optimal partitioning type algorithm (Jackson et al., (2005)) and its pruned version, PELT method by Killick, et al. (2012). Further details of these algorithms
can be found in Killick, et al. (2012) and Maidstone, et al., (2014).\\
\ \\
\noindent The intents of our work are motivated by
two major research results. The first motivation is from Dehling, et al. (2010) that derives
a maximum likelihood estimator (MLE) for the drift parameters of the diffusion
process and establishes its asymptotic properties. This was extended in
Dehling, et al. (2014), where
there is one unknown change point and a likelihood-ratio test statistic was
constructed to
determine such change point. The second motivation is from
Zhang~(2015) that establishes the asymptotic properties of both the unrestricted and restricted MLE for
the drift parameters of the generalised OU process with a single change point.
A James-Stein-type shrinkage estimator
for the drift parameters is proposed in Zhang~(2015) as an improvement
and it is also shown that the previously established asymptotic properties also hold for any consistent
estimator for the rate of the change point. \\
\ \\
\noindent {\color{black}Neither Dehling, et al.~(2014) nor Zhang~~(2015) offer} a specific methodology to identify the change point.
This led us to the three main contributions {\color{black}of} this paper.
First, we extend the single-change point framework to the multiple-change point setting and  present two consistent methods to estimate the unknown locations of change points.
Second, we prove the asymptotic normality of the drift parameters' MLE.
Third, we employ information-based statistics to resolve the issue
of estimating the unknown number of change points and then created three
algorithms to implement the calculations. We validate
the performance of our estimation techniques using simulated and
real market data.
\\
\ \\
\noindent This paper is structured as follows.
In Section 2, we present the formulation of the multiple change-point problem.
Section 3 summarises the
results of Dehling, et al. (2014) and Zhang~(2015)
on MLE and the related asymptotic properties, which provide an
impetus on the asymptotic performance of our proposed methods.
Two estimation methods are put forward to determine
the unknown locations of change
points in section 4 along with the discussion of the asymptotic properties of the estimators;
we find that
the asymptotic properties obtained in
in Zhang~(2015) also hold in our proposed techniques.
Section 5 deals with
the problem of both the existence issue and location of the change points
using an information approach. We develop
computing algorithms in section 6 in order to implement the
proposed methods. In section 7, we assess the
applicability of our methods through numerical examples
on simulated and observed financial
market data. Finally, section 8 provides some concluding remarks.

\section{Problem description in determining change points}
\noindent We study the
generalised version of the OU process with SDE representation given in (\ref{ou1}).
It is assumed that there {\color{black}exist} $m$ ($m\geq 1$) unknown change points $\tau_j=s_j T$, where $j=1,~\ldots,~m$ and $0<s_1<~\ldots~<s_m<1$. To simplify the notation, we let $\tau_0=0$ and $\tau_{m+1}=T$. {\color{black}In our setup} $\boldsymbol{\mathrm \uptheta}=(\boldsymbol{\mathrm \uptheta}^{(1)\top}, \ldots,~\boldsymbol{\mathrm \uptheta}^{(m+1)\top})^{\top}$ with $\boldsymbol{\mathrm \uptheta}^{(j)}=(\mu_1^{(j)},...,\mu_p^{(j)}, -a^{(j)})^{\top}$ for $\tau_{j-1}<t\leq \tau_j$ and
\begin{equation}\label{eq-S}
S(\boldsymbol{\mathrm \uptheta},t,X_t)=\sum_{j=1}^{m+1}\left(\sum_{k=1}^p\mu_k^{(j)}\varphi_k(t)-a^{(j)} X_t\right) \mathbbm{1}_{( \tau_{j-1}<t<\tau_j)}
\end{equation}
with $\mathbbm{1}_{(.)}$ as the indicator function. {\color{black}Note that} $\boldsymbol{\uptheta}^{(j)}$ may also be a vector.\\
\ \\
\noindent We start by assuming that the number of change points $m$ is known, but the exact value of each change point denoted by $\tau_1^0,~\ldots~,\tau_m^0$ (and correspondingly the exact rates $s_j^0$, $j=1,~\ldots,~m$) are unknown. Furthermore, considering that we have multiple change points in the model, we posit that these change points are asymptotically distinct.  We further impose the following assumptions.\\
\begin{asm}\label{asm1}$\tau_j^0=Ts_j^0$, $0< s_1^0<~\dots~<s_m^0<1$.
We call $\displaystyle s_j=\frac{\tau_j^0}{T}$ the {\color{black} change points' arrival rate}, and if we have $\widehat s_j$,
the value $\widehat \tau_j^0$ is immediate.
\end{asm}
\noindent Assumption~\ref{asm1} implies that the length of each regime $[\tau_{j-1}^0,\tau_j^0]$ is proportional to $T$. The structure of the model in each regime is similar to that {\color{black}of} the no-change point process studied in Dehling~et. al (2010); see also Zhang~(2015) for the case of a single change point. This means that the results established in the existing literature could also be adapted to the case of multiple change points. \\

\noindent MLEs
for the drift parameters and their asymptotic properties were shown
in Dehling, et al. (2010) for the case of no change point and
in Zhang (2015) for the case of one change point. Certainly, Zhang~(2015)
is a special case of our study with
with $m=1$. The next section reviews previous results and extends them
to the multiple change points problem.\\

\section{Prior MLE-based results and our extension}\label{section-mle}
\noindent \noindent The asymptotic normality for the MLE estimator of the drift parameters
in Zhang (2015) assumes that the estimator is already consistent. In our case,
we shall prove ({\color{black}rather than simply assume}) that such an estimator of the change point is consistent.\\
\ \\
\noindent In the {\color{black}subsequent} discussion, we write ``$\xrightarrow[T\rightarrow \infty]{p}$", $``\xrightarrow[T\rightarrow \infty]{D}"$, and $``\xrightarrow[T\rightarrow \infty]{a.s.}"$ to mean convergence in probability, convergence in distribution, and convergence almost surely, respectively. The notation $||.||$ denotes the Frobenius norm for matrices. We use bold, unitalicised English or Greek letters in lowercase for vectors; and bold, unitalicised English or Greek letters in upper case for matrices. \\
\ \\
The ``$O(\cdot)$'' denotes the Landau symbol, also known as the ``Big O'' notation, which is used to describe the asymptotic behaviour of functions. So, for a set of random variables $U_n$ and a corresponding set of constants
$a_n$, $U_n=O_p(a_n)$ means $U_n/a_n$ is stochastically bounded. Formally, this means $\forall \epsilon > 0
,~~\exists~M >0
,~~\ni~P(\left|U_n/a_n \right| > M) < \epsilon,~~\forall n.$
On the other hand, the symbol involving ``small o'', i.e., $U_n=o_p(a_n)$ means $U_n/a_n$ converges in probability to zero as $n$ approaches an appropriate limit. So, since $U_n=o_p(a_n)$ is equivalent to $U_n/a_n=0_p(1)$,
convergence in probability {\color{black}is here} defined as $\displaystyle \lim_{n \rightarrow \infty}\left(
P(|U_n/a_n)| \geq \epsilon \right)=0.$

\subsection{Log likelihood function}\label{section21}
\noindent The following assumption from Dehling, et. al.~(2010) is also {\color{black}retained} here.
\begin{asm}\label{asm2}$\textrm{P}\left(\int_{0}^{T}S^2(\boldsymbol{\mathrm \uptheta},t,X_t)<\infty\right)=1$, for all $0<T<\infty$ and elements $\boldsymbol{\mathrm \uptheta}^{(j)}$ of $\boldsymbol{\mathrm \uptheta}$ involved in $S(\boldsymbol{\mathrm \uptheta},t,X_t)$ given by equation (\ref{eq-S}).\end{asm}
\noindent Under Assumptions \ref{asm1}--\ref{asm2} and Theorem 7.6 of Lipster and Shiryayev~(2001), the corresponding likelihood function in our modelling framework is
\begin{eqnarray*}
\ell^*(\boldsymbol{\mathrm \uptheta},X_t)&=&\exp\left(\frac{1}{\sigma^2}\int_{0}^T S\left(\boldsymbol{\mathrm \uptheta},t,X_t\right)dX_t-\frac{1}{2\sigma^2}\int_{0}^TS^2\left(\boldsymbol{\mathrm \uptheta},t,X_t\right)dt\right).
\end{eqnarray*}
The log-likelihood function is therefore
\begin{eqnarray*}
\log \ell(\boldsymbol{\mathrm \uptheta},X_t)&=&\frac{1}{\sigma^2}\int_{0}^T S(\boldsymbol{\mathrm \uptheta},t,X_t)dX_t-\frac{1}{2\sigma^2}\int_{0}^TS^2(\boldsymbol{\mathrm \uptheta},t,X_t)dt
=\frac{1}{\sigma^2}\sum_{j=1}^{m+1}\boldsymbol{\mathrm \uptheta}^{(j)\top}\boldsymbol{\mathrm{\tilde{r}}}_{(\tau_{j-1}^0,\tau_j^0)}\\&&-\frac{1}
{2\sigma^2}\sum_{j=1}^{m+1}\boldsymbol{\mathrm \uptheta}^{(j)\top}\boldsymbol{\mathrm Q}_{(\tau_{j-1}^0,\tau_j^0)}\boldsymbol{\mathrm \uptheta}^{(j)},
\end{eqnarray*}
where
$$\boldsymbol{\mathrm Q}_{(\tau_{j-1}^0,\tau_j^0)}=
\left[\begin{array}{cccc}
    \int_{\tau_{j-1}^0}^{\tau_j^0}\varphi_1^2(t)dt &\dots& \int_{\tau_{j-1}^0}^{\tau_j^0}\varphi_1(t)\varphi_p(t)dt
    & -\int_{\tau_{j-1}^0}^{\tau_j^0}\varphi_1(t)X_tdt \\
    \dots&&& \\
    -\int_{\tau_{j-1}^0}^{\tau_j^0}\varphi_p(t)X_tdt &\dots & -\int_{\tau_{j-1}^0}^{\tau_j^0}\varphi_p(t)X_tdt & \int_{\tau_{j-1}^0}^{\tau_j^0}X_t^2dt
\end{array}\right]$$
and $\boldsymbol{\mathrm{\tilde{r}}}_{(\tau_{j-1}^0,\tau_j^0)}
=\left( \int_{\tau_{j-1}^0}^{\tau_j^0}\varphi_1(t)dX_t,~\ldots,~\int_{\tau_{j-1}^0}^{\tau_j^0}\varphi_p(t)dX_t,
-\int_{\tau_{j-1}^0}^{\tau_j^0}X_tdX_t \right)^{\top}.$

\subsection{Maximum likelihood estimators for the draft parameters}\label{section22}
\noindent By setting the first partial derivatives with respect to each of the parameters of $\ell(\boldsymbol{\mathrm \uptheta},X_t)$ to 0, we obtain the MLE of the drift parameters, provided $\boldsymbol{\mathrm Q}_{(\tau_{j-1}^0,\tau_j^0)}$ is invertible for each $j=1,~\ldots,~m+1$.
When $\boldsymbol{\mathrm Q}_{(\tau_{j-1}^0,\tau_{j}^0)}^{-1}$ exists, Remark 3
of Dehling, et al. (2010) shows that $(s_{j}^0-s_{j-1}^0)T\boldsymbol{\mathrm Q}_{(\tau_{j-1}^0,\tau_{j}^0)}^{-1}$
must exist almost
surely if $T$ is large enough. Moreover, Proposition 2.1.1. of
Zhang~(2015) is also adapted to give the positive definiteness of $\frac{1}{(s_{j}^0-s_{j-1}^0)T}\boldsymbol{\mathrm Q}_{(\tau_{j-1}^0,\tau_{j}^0)}$
under the following assumption.
\begin{asm}\label{asm3}For any $T>0$, the base function $\{\varphi_k(t),k=1,~\ldots,~p\}$ is
Riemann-integrable on $[0,T]$ and satisfies two properties.
\begin{enumerate}
  \item Periodicity. That is, $\varphi_k(t+v)=\varphi_k(t)$ $\forall$ $i=1,~\ldots,~p$ and $v$ is the
  period observed in the data.
  \item Orthogonality. That is, $\forall$ $k_1,k_2=1,~\ldots,~p$, $\int_{0}^v\varphi_{k_1}(t)\varphi_{k_2}(t)dt$ is
  equal to $v$ if $k_1=k_2$ and 0 otherwise.
\end{enumerate}
\end{asm}
\noindent By Assumption~\ref{asm3}, $\varphi_k(t)$ is bounded on $\boldsymbol{\mathrm r}_{+}$ (i.e.  $\varphi_k(t)\leq K_{\varphi}$ for some $0< K_{\varphi}<\infty$ ) as for every $k$ the base function $\varphi_k(t)$ is bounded on $[0,T]$ and $v$-periodic. The following result is obtained by reducing the time period from $[0,T]$ to $[\tau_{j-1}^0,\tau_j^0]$, $j=1,~\ldots,~m+1$ and utilising the same arguments as in Proposition 2.1.1 of Zhang~(2015).

\begin{prn}\label{prn01}
Suppose Assumptions \ref{asm1}-\ref{asm3} hold for any $T>0$ and $t\in [\tau_{j-1}^0,\tau_j^0]$. The base
functions $\{ \varphi_k(t),k=1,...,p \}$ are incomplete if and
only if $\frac{1}{(s_{j}^0-s_{j-1}^0)T}\boldsymbol{\mathrm Q}_{(\tau_{j-1}^0,\tau_{j}^0)}$ is a positive definite matrix.
\end{prn}
\noindent For the rest of this paper, we assume that the sample
size $T$ is an integral multiple of the period length $v$,
i.e., $T=Nv$ for some integer $N$. Without loss of generality,
we let $v=1$, {\color{black}implies} that $\varphi_k(t+1)=\varphi_k(t)$.\\
\ \\
\noindent Using the results in Dehling, et. al (2010) and Zhang~(2015), the MLE of the drift parameters based on the log likelihood function provided above are {\color{black}given} by $\boldsymbol{\mathrm{\hat{\uptheta}}}=(\boldsymbol{\mathrm{\hat{\uptheta}}}^{(1)\top},...,\boldsymbol{\mathrm{\hat{\uptheta}}}^{(m+1)\top})^{\top}$ with
\begin{equation}\label{mlem1}\boldsymbol{\mathrm{\hat{\uptheta}}}^{(j)}
=\boldsymbol{\mathrm Q}_{(\tau_j^0,\tau_{j-1}^0)}^{-1}\boldsymbol{\mathrm{\tilde{r}}}_{(\tau_j^0,\tau_{j-1}^0)},\quad j=1,\ldots,m+1.\end{equation}
Substituting (\ref{ou1}) into (\ref{mlem1}) and going through some algebraic computations will lead to
\begin{equation}\label{mlem2}\boldsymbol{\mathrm{\hat{\uptheta}}}^{(j)}=\boldsymbol{\mathrm \uptheta}^{(j)}+\sigma T \boldsymbol{\mathrm Q}_{(\tau_j^0,\tau_{j-1}^0)}^{-1}\frac{1}{T}\boldsymbol{\mathrm r}_{(\tau_j^0,\tau_{j-1}^0)},\quad j=1,~\ldots,~m+1,\end{equation}
where  $\boldsymbol{\mathrm r}_{(a,b)}=\left(\int_a^b\varphi_1(t)dW_t,...,\int_a^b\varphi_p(t)dW_t,-\int_a^b X_tdW_t\right)^{\top}$ for $0\leq a<b\leq T$.
\subsection{Asymptotic properties of the MLE}\label{section23}
\noindent To study the asymptotic proprieties of the MLE in the next section, equation (\ref{mlem1}) to be precise, we review the established asymptotic results in Dehling, et al.~(2010) for the case where there is no change point ($m=0$) and also the results in Zhang~(2015) when there exists one change point ($m=1$).\\
\ \\
\noindent If there is no change point ($m=0$ and $\boldsymbol{\mathrm \uptheta}=\boldsymbol{\mathrm \uptheta}^{(1)}$ for $0<t\leq T$), the solution of the SDE (\ref{ou1}) has the explicit representation
\begin{equation}\label{sol1}X_t=e^{-a^{(1)}t}X_{0}+h^{(1)}(t)+\boldsymbol{\mathrm z}_t^{(1)},\quad 0<t\leq T,\end{equation}
where
$\displaystyle{h^{(j)}(t)=e^{-a^{(j)}t}\sum_{i=1}^p\mu_i^{(j)}\int_{0}^te^{a^{(j)}s}\varphi_k(s)ds}$
and
$\displaystyle{\boldsymbol{\mathrm z}_t^{(j)}=\sigma e^{-a^{(j)}t}\int_{0}^te^{a^{(j)}s}dW_s}.$\\
Note that {\color{black}as} the process $\{X_t,t\geq 0\}$ is not stationary in the ordinary sense, it is impossible to apply the ergodic theorem directly. To circumvent this, a stationary solution, for $t\in \mathbb{R}$ instead of $t\geq 0$, was introduced in Dehling, et al.~(2010). Consider
\begin{equation}\label{cx0}\tilde{X}_t=\tilde{h}^{(1)}(t)+\tilde{Z}_t^{(1)},\quad 0<t\leq T,\end{equation}
where
$\displaystyle{\tilde{h}^{(j)}(t)=e^{-a^{(j)}t}\sum_{i=1}^p\mu_i^{(j)}\int_{-\infty}^te^{a^{(j)}s}\varphi_k(s)ds}$,
$\displaystyle{\tilde{Z}_t=\sigma e^{-a^{(j)}t}\int_{0}^te^{a^{(j)}s}d\tilde{B}_s},$
and $(\tilde{B}_s)_{s\in  \mathbb{R}}$ denotes a bilateral Brownian motion, i.e.,
$$\tilde{B}_s={B}_s \mathbbm{1}_{ \boldsymbol{\mathrm r}_{+}}(s)+\bar{B}_{-s} \mathbbm{1}_{ \boldsymbol{\mathrm r}_{-}}(s),$$
where $({B}_s)_{s\geq 0}$ and  $(\bar{B}_s)_{s\geq 0}$ are two independent standard Brownian motions. Then, from Lemma~{4.3} in Dehling, et al. (2010), the sequence of $\mathcal{C}[0,1]$-valued random variables $W_k(s)=\tilde{X}_{k-1+s}$, $0\leq s\leq 1$, $k\in \mathbb{N}$ is stationary and ergodic. In this case, by Proposition 4.5 of Dehling, et al. (2010),
\begin{equation}\label{cx2}\frac{1}{T}\int_0^T\tilde{X}_t\varphi_k(t)dt
\xrightarrow[T\rightarrow \infty]{a.s.}\int_0^1(\tilde{h}^{(1)}(t))\varphi_k(t)dt \quad
\mbox{and}\quad \frac{1}{T}\int_0^T\tilde{X}_t^2dt\xrightarrow[T\rightarrow \infty]
{a.s.}\int_0^1(\tilde{h}^{(1)}(t))^2dt+\frac{\sigma^2}{2a^{(1)}}.\end{equation}
\noindent Moreover, under Assumptions \ref{asm1}-\ref{asm3}, Lemma 4.4 in Dehling, et al. (2010), \begin{equation}\label{cx1}|\tilde{X}_t-X_t|\xrightarrow[t \rightarrow \infty]{a.s.} 0.\end{equation}
Using (\ref{cx1}), we have the following properties:
$$\frac{1}{T}\int_0^T\tilde{X}_t\varphi_k(t)dt-\frac{1}{T}\int_0^T{X}_t\varphi_k(t)dt
\xrightarrow[T\rightarrow \infty]{a.s.} 0\quad\mbox{and}~~\frac{1}{T}\int_0^T\tilde{X}_t^2dt-\frac{1}{T}\int_0^T{X}_t^2dt
\xrightarrow[T\rightarrow \infty]{a.s.} 0.$$
It follows from (\ref{cx2}) that
$$\frac{1}{T}\int_0^T{X}_t\varphi_k(t)dt\xrightarrow[T\rightarrow
\infty]{a.s.}\int_0^1\tilde{h}^{(1)}(t)\varphi_k(t)dt\quad \mbox{and}~~\frac{1}{T}\int_0^T{X}_t^2dt\xrightarrow[T\rightarrow \infty]{a.s.}
\int_0^1(\tilde{h}^{(1)}(t))^2dt+\frac{\sigma^2}{2a^{(1)}}.$$
Hence,
\begin{equation}\label{conQ-0}T\boldsymbol{\mathrm Q}_{(0,T)}^{-1}\xrightarrow[T\rightarrow
\infty]{a.s.}\boldsymbol{\mathrm \Sigma}_0^{-1},\end{equation}
where
$$\boldsymbol{\mathrm \Sigma}_0=\begin{bmatrix}
    \boldsymbol{\mathrm I}_p  &\boldsymbol{\mathrm \Lambda}_{(0,T)}\\
    \boldsymbol{\mathrm \Lambda}_{(0,T)}^{\top}   & w \\
\end{bmatrix}$$
with $\boldsymbol{\mathrm \Lambda}_{(0,T)}=\left( \int_0^1\tilde{h}^{(1)}(t)\varphi_1(t)dt,...,\int_0^1
\tilde{h}^{(1)}(t)\varphi_p(t)dt \right)^{\top}$ and $w=\int_0^1(\tilde{h}^{(1)})^2(t)dt+\frac{\sigma^2}{2a^{(1)}}$. \\
\ \\
\noindent Furthermore, under Assumptions \ref{asm1}--\ref{asm3}, the following properties for $\boldsymbol{\mathrm r}_{(0,T)}$ hold.
\begin{enumerate}
\item $\{\boldsymbol{\mathrm r}_{(0,T)}, T>0\}$ is a  martingale.
\item $\frac{1}{T}\boldsymbol{\mathrm r}_{(0,T)}\xrightarrow[T\rightarrow \infty]{a.s.}0$.
\item $\frac{1}{\sqrt{T}}\boldsymbol{\mathrm r}_{(0,T)}\xrightarrow[T\rightarrow
\infty]{D}\boldsymbol{\mathrm r}\sim \mathcal{N}_{p+2}(0,\boldsymbol{\mathrm \Sigma}_0)$.
\end{enumerate}
Detailed proofs of the above are elaborated in Zhang~(2015). The above properties, together with Slutsky's Theorem, yields
$$\frac{1}{\sqrt{T}}(\boldsymbol{\mathrm{\hat{\uptheta}}}-\boldsymbol{\mathrm \uptheta})\xrightarrow[T\rightarrow \infty]
{D}\boldsymbol{\mathrm r}\sim \mathcal{N}_{p+1}\left( 0,\boldsymbol{\mathrm \Sigma}_0^{-1} \right).$$
\noindent Zhang~(2015) extended the above asymptotic properties to the case
of a single change point. Using similar arguments, we extend these results in the context of  multiple change points. We first present a result
covering the coefficients of SDE (\ref{ou1}).
\begin{prn}\label{prncoef}
Under Assumptions \ref{asm1} -\ref{asm3}, the coefficients in SDE (\ref{ou1}) for $m\geq 1$ satisfy both the space-variable Lipschitz and the spatial growth conditions.\end{prn}

\noindent \textit{Proof:} See \ref{appendixsection2}. \\
\ \\
\noindent Using Proposition~\ref{prncoef} and similar methods employed for the proof of Proposition 2.2.1 in Zhang~(2015), it may be verified that SDE (\ref{ou1}) admits a strong and unique solution that is uniformly bounded in $L^2$ and
\begin{equation}\label{sol2}\sup_{t\geq 0}\mathrm{E}\left(X_t^2\right)\leq K_1
\end{equation}
for some $0<K_1<\infty$.\\
\ \\
\noindent Employing (\ref{sol2}), we (\ref{sol1}) and (\ref{cx0}) to their representations in context of multiple change points ($m\geq 1$) given by
 \begin{equation}\label{solm1}X_t=e^{-a^{(j)}t}X_{0}^{(j)}+h^{(j)}(t)+\boldsymbol{\mathrm z}_t^{(j)},\quad \tau_{j-1}^0<t\leq \tau_{j}^0,
\quad j=1,~\ldots,~m+1,\end{equation}
where $X_{0}^{(j)}=X_{\tau_{j-1}}$ and
\begin{equation}\label{cxm0}\tilde{X}_t=\tilde{h}^{(j)}(t)+\tilde{Z}_t^{(j)},\quad \tau_{j-1}^0<t\leq \tau_{j}^0,
\quad j=1,~\ldots,~m+1.\end{equation}
Then
\begin{equation}\label{cxm2}\frac{1}{T}\int_{\tau_{j-1}^0}^{\tau_j^0}\tilde{X}_t\varphi_k(t)dt\xrightarrow[T\rightarrow \infty]{a.s.}(s_{j}-s_{j-1}^0)\int_0^1\tilde{h}^{(j)}(t)\varphi_k(t)dt\end{equation}
and
\begin{equation}\label{cxm3}\frac{1}{T}\int_{\tau_{j-1}^0}^{\tau_j^0}\tilde{X}_t^2dt\xrightarrow[T\rightarrow \infty]{a.s.}(s_{j}^0-s_{j-1}^0)\left(\int_0^1(\tilde{h}^{(j)}(t))^2dt+\frac{\sigma^2}{2a^{(j)}}\right).\end{equation}
Using (\ref{cx1}) and (\ref{sol2}) and similar arguments in the proof of Propositions 2.2.2 and 2.2.3 in Zhang~(2015), the following properties hold:
$$\frac{1}{T}\int_{\tau_{j-1}^0}^{\tau_j^0}\tilde{X}_t\varphi_k(t)dt-\frac{1}{T}\int_{\tau_{j-1}^0}^{\tau_j^0}{X}_t\varphi_k(t)dt
\xrightarrow[T\rightarrow \infty]{a.s.} 0,\quad\mbox{and}~~\frac{1}{T}\int_{\tau_{j-1}^0}^{\tau_j^0}\tilde{X}_t^2dt-\frac{1}{T}\int_{\tau_{j-1}^0}^{\tau_j^0}{X}_t^2dt
\xrightarrow[T\rightarrow \infty]{a.s.} 0.$$
Thus,
$$\frac{1}{T}\int_{\tau_{j-1}^0}^{\tau_j^0}{X}_t\varphi_k(t)dt\xrightarrow[T\rightarrow
\infty]{a.s.}\int_0^1\tilde{h}^{(j)}(t)\varphi_k(t)dt\quad \mbox{and}~~\frac{1}{T}\int_{\tau_{j-1}^0}^{\tau_j^0}{X}_t^2dt\xrightarrow[T\rightarrow \infty]{a.s.}
\int_0^1(\tilde{h}^{(j)}(t))^2dt+\frac{\sigma^2}{2a^{(j)}}.$$
So,
\begin{eqnarray}
\frac{1}{T}\boldsymbol{\mathrm Q}_{(\tau_{j-1}^0,\tau_{j}^0)}&\xrightarrow[T\rightarrow \infty]{a.s.}&(s_j^0-s_{j-1}^0)\boldsymbol{\mathrm \Sigma}_j,\label{convQ-m1}\end{eqnarray}
where
\begin{equation}\label{sigmaj}\boldsymbol{\mathrm \Sigma}_j=\begin{bmatrix}
    \boldsymbol{\mathrm I}_p  & \boldsymbol{\mathrm \Lambda}_j \\
    \boldsymbol{\mathrm \Lambda}_j'   & w_j \\
\end{bmatrix}\end{equation}
with $\boldsymbol{\mathrm \Lambda}_{j}=\left(\int_0^1\tilde{h}^{(j)}(t)\varphi_1(t)dt,...,\int_0^1\tilde{h}^{(j)}(t)\varphi_p(t)dt\right)^{\top}$ and $w_j=\int_0^1(\tilde{h}^{(j)})^2(t)dt+\frac{\sigma^2}{2a^{(j)}}$, $j=1,...,m+1$. Further, by the Continuous Mapping Theorem,
\begin{equation}T\boldsymbol{\mathrm Q}_{(\tau_{j-1}^0,\tau_{j}^0)}^{-1}\xrightarrow[T\rightarrow \infty]{a.s.}\frac{1}{(s_j^0-s_{j-1}^0)}\boldsymbol{\mathrm \Sigma}_j^{-1}.\label{convQ-m2}\end{equation}

\noindent So long as Assumptions \ref{asm1}--\ref{asm3} hold and with the aid of similar argument used in the proof of Proposition 2.2.6 in Zhang~(2015), it may be shown that $\boldsymbol{\mathrm \Sigma}_j$ is positive definite.\\
\ \\
\noindent Note that (\ref{convQ-m1}) and  (\ref{convQ-m2}) are key elements in analysing the asymptotic properties of $\boldsymbol{\mathrm Q}_{({\hat{\tau}}_{j-1},{\hat{\tau}}_j)}$ and its inverse, where ${\hat{\tau}}_j$ is the estimator of $\tau_j^0$.\\
\ \\
\noindent Invoking the boundedness property of $\varphi_k(t)$, we have
\begin{equation}
\mathrm{E}\left(\int_{0}^T\left(\frac{1}{\sqrt{T}}\varphi_k(t)\mathbbm{1}(\tau_{j-1}^0< t\leq \tau_{j}^0)\right)^2dt\right)
=\int_{0}^T\frac{1}{T}\varphi_k^2(t)\mathbbm{1}(\tau_{j-1}^0< t\leq \tau_{j}^0)dt \leq K_{\varphi}^2(s_j^0-s_{j-1}^0)<\infty
\end{equation}
for $k=1,~\ldots,~p$ and $j=1,~\ldots,~m$. Similarly, using (\ref{sol2}),
\begin{equation}
\mathrm{E}\left(\int_{0}^T\left(\frac{1}{\sqrt{T}}X_t\mathbbm{1}(\tau_{j-1}^0< t\leq \tau_{j}^0)\right)^2dt\right)
=\int_{0}^T\frac{1}{T}\mathrm{E}(X_t^2)\mathbbm{1}(\tau_{j-1}^0< t\leq \tau_{j}^0)dt \leq K_{1}^2(s_j^0-s_{j-1}^0)<\infty.
\end{equation}
From Proposition~1.21 in Kutoyants~(2004) (see also Proposition 2.2.10 in Zhang,~2015), we get
\begin{equation}\label{am-R}
\left( \frac{1}{\sqrt{T}}\boldsymbol{\mathrm r}_{(0, \tau_1^0)},~\ldots,\frac{1}{\sqrt{T}}\boldsymbol{\mathrm r}_{(\tau_{m}^0, T)}\right)\xrightarrow[T\rightarrow \infty]{D}\mathcal{N}_{(m+1)(p+1)}\left(\boldsymbol{\mathrm 0},\boldsymbol{\tilde{\mathrm \Sigma}}\right),
\end{equation}
where $\boldsymbol{\mathrm 0}$ is a vector of zeros and $\boldsymbol{\tilde{\mathrm \Sigma}}=\mathrm{diag}({s_1^0}\boldsymbol{\mathrm \Sigma}_1,(s_2^0-s_1^0)\boldsymbol{\mathrm \Sigma}_2,~\ldots,(1-s_m^0)\boldsymbol{\mathrm \Sigma}_{m+1})$.
Combining
(\ref{convQ-m1}), (\ref{convQ-m2}) and (\ref{am-R}), along with Slutsky's Theorem and some algebraic computations,
\begin{equation}\label{am-theta0}
\frac{1}{\sqrt{T}}(\boldsymbol{\mathrm{\hat{\uptheta}}}-\boldsymbol{\mathrm \uptheta}^0)\xrightarrow[T\rightarrow \infty]{D}\mathcal{N}_{(m+1)(p+1)}(\boldsymbol{\mathrm 0}, \sigma^2\boldsymbol{\tilde{\mathrm \Sigma}}^{-1}),
\end{equation}
where $\boldsymbol{\tilde{\mathrm \Sigma}^{-1}}=\mathrm{diag}\left( \frac{1}{s_1^0}\boldsymbol{\mathrm \Sigma}_1^{-1},\frac{1}{s_2^0-s_1^0}\boldsymbol{\mathrm \Sigma}_2^{-1},~\ldots,\frac{1}{1-s_m^0}\boldsymbol{\mathrm \Sigma}_{m+1}^{-1}
\right)$.\\
\ \\
The above asymptotic properties are established based on the exact values of the locations of the change points $\tau_j^0$, $j=1,~\ldots,~m$. However, in practice, $\tau_j^0$ are often unknown. Hence, we shall devise
methods to estimate the unknown $\tau_j^0$ and investigate whether the above asymptotic normality still hold for the estimated change points.

\section{Estimation of change points and pertinent asymptotic properties}\label{mcpe}
\noindent In Subsections 4.1 and 4.2, we develop two techniques to estimate the unknown locations of change points. The asymptotic normality of $\widehat{\boldsymbol{\mathrm \uptheta}}$ based on the estimated change points is discussed in Subsection 4.3.

\subsection{Least sum squared error method}\label{LSSE}
\noindent We introduce the  least sum of squared errors (LSSE) method then investigate the consistency
of our proposed estimator. Consider a partition $0= t_0<~\dots~<t_n=T$ on a time period $[0,T]$ with constant increment $\Delta_t=t_{i+1}-t_{i}$. Also, let $Y_i=X_{t_{i+1}}-X_{t_i}$ and $\boldsymbol{\mathrm z}_{i}=(\varphi_1(t_i),...,\varphi_p(t_i),-X_{t_i})\Delta_t$. \\
\ \\
\noindent Due to the uncertain locations of estimated change points, the exact value of the drift parameters $\boldsymbol{\mathrm \uptheta}$ and the MLE may have different indices. For example, if ${\hat{\tau}}_j>{\tau}_j^0$, then for all $t_i\in(\tau_j^0,{\hat{\tau}}_j]$, the associated exact value of the drift parameters is $\boldsymbol{\mathrm \uptheta}^{(j+1)}$ but the MLE is $\boldsymbol{\mathrm{\hat{\uptheta}}}^{(j)}$. So here, $\boldsymbol{\mathrm \uptheta}_{i}$ and $\boldsymbol{\mathrm{\hat{\uptheta}}}_{i}$ refer to the exact value and MLE of the drift parameters at time point $t_i$, respectively. In this case, $\boldsymbol{\mathrm \uptheta}_i=\sum_{j=1}^{m+1}\boldsymbol{\mathrm \uptheta}^{(j)}\mathbbm{1}(\tau_{j-1}^0\leq t_i\leq \tau_j^0)$  and $\boldsymbol{\mathrm{\hat{\uptheta}}}_i=\sum_{j=1}^{m+1}\boldsymbol{\mathrm{\hat{\uptheta}}}^{(j)}\mathbbm{1}({\hat{\tau}}_{j-1}\leq t_i\leq {\hat{\tau}}_j)$ with
$\boldsymbol{\mathrm{\hat{\uptheta}}}^{(j)}=\boldsymbol{\mathrm Q}_{({\hat{\tau}}_{j-1},{\hat{\tau}}_{j})}^{-1}\boldsymbol{\mathrm{\tilde{r}}}_{({\hat{\tau}}_{j-1},{\hat{\tau}}_{j})}$ for $j=1,~\ldots,~m+1$. Then by the Euler-Maruyama discretisation method,
\begin{eqnarray}\label{oum2}
Y_{i}&=&\boldsymbol{\mathrm z}_{i}\boldsymbol{\mathrm \uptheta}_i+\epsilon_i,\quad i=1,~\ldots~,T,
\end{eqnarray}
where $\epsilon_i$ is the error term $\sigma\sqrt{\Delta_t}\omega$, and $\omega$ is
the standard normal term. Therefore, we could now use the LSSE method to
estimate the change points. \\
\ \\
\noindent From (\ref{oum2}), the estimates for the multiple change points $\tau^0=(\tau_1^0,...,\tau_m^0)$ are given by
\begin{eqnarray}\label{cpm1}
\boldsymbol{\hat{\mathrm \tau}}=\arg\min_{\tau} SSE([0,T],\boldsymbol{\mathrm \tau},\boldsymbol{\mathrm{\hat{\uptheta}}}(\boldsymbol{\mathrm \tau})),
\end{eqnarray}
where
\begin{eqnarray}\label{ssr1}
SSE([0,T],\boldsymbol{\mathrm \tau},\boldsymbol{\mathrm{\hat{\uptheta}}}(\boldsymbol{\mathrm \tau}))&=&\sum_{t_i\in[0,T]}(Y_i-\boldsymbol{\mathrm z}_i\boldsymbol{\mathrm{\hat{\uptheta}}}_i)^{\top}(Y_i-\boldsymbol{\mathrm z}_i\boldsymbol{\mathrm{\hat{\uptheta}}}_i).
\end{eqnarray}

\noindent {\bf Consistency of the proposed estimator}\\
Under Assumptions \ref{asm1}--\ref{asm3}, $\sum_{t_i\in (\tau_{j-1}^0,\tau_j^0]}\boldsymbol{\mathrm z}_i^{\top}\boldsymbol{\mathrm z}_i$ for $j=1,~\ldots~,m$, the discretised versions of $\boldsymbol{\mathrm Q}_{(\tau_{j-1}^0,{\tau}_j^0)}$ are both positive definite with probability 1 provided that the base
functions $\{ \varphi_k(t), i=1,...,p \}$ are incomplete. Moreover, using (\ref{convQ-m1}) (see also Proposition 2.2.6 of Zhang~(2015)), one can show that $\frac{1}{(s_{j}^0-s_{j-1}^0)T}\boldsymbol{\mathrm Q}_{(\tau_{j-1}^0,{\tau}_j^0)}$ converges in probability to some positive definite matrices for large $T$, {\color{black}as do} their respective discretised versions. Hence, for large $T$, it is reasonable to impose a useful assumption in proving the consistency of the estimators of the change points.

\begin{asm}\label{asm4}For every $j=1,~\ldots~,m$, there exists an
$L_0>0$ such that for all $L>L_0$ the minimum eigenvalues of
$\frac{1}\ell\sum_{t_i\in(\tau_j^0,\tau_j^0+L]}\boldsymbol{\mathrm z}_i'\boldsymbol{\mathrm z}_i$ and of
$\frac{1}\ell\sum_{t_i\in(\tau_j^0-L,\tau_j^0]}\boldsymbol{\mathrm z}_i'\boldsymbol{\mathrm z}_i$, as well as their respective continuous-time
versions $\frac{1}\ell\boldsymbol{\mathrm Q}_{(\tau_j^0,\tau_j^0+L)}$ and $\frac{1}\ell\boldsymbol{\mathrm Q}_{(\tau_j^0-L,\tau_j^0]}$,
are all bounded away from 0.\end{asm}
\ \\
\noindent For the motivation of the above assumption, see Perron and Qu (2006) and Chen and Nkurunziza~(2015). The next two propositions provide results characterising consistency.

\begin{prn}\label{prnm13}
Suppose that $\boldsymbol{\mathrm \uptheta}^{(1)}-\boldsymbol{\mathrm \uptheta}^{(2)}$, the shift in the drift parameters, is of fixed non-zero magnitude independent of $T$. Then, under Assumption \ref{asm1}--\ref{asm4}, and $\hat{s}_j-s_j^0\xrightarrow[T\rightarrow \infty]{P}0$, $j=1,~\ldots,~m$.
\end{prn}
\noindent \textit{Proof:} See \ref{appendixb}.

\begin{prn}\label{prnm14}
Suppose the conditions in Proposition~\ref{prnm13} hold. Then, for every $\epsilon>0$, there exists a $C>0$ such that for large $T$, $P\left(T|\hat{s}_j-s_j|>C\right)<\epsilon$, for every $j=1,...,m$.
\end{prn}
\noindent \textit{Proof:} See \ref{appendixb}.

\noindent Proposition~\ref{prnm13} shows that the estimated rate $\hat{s}_j=\frac{{\hat{\tau}}_j}{T}$ is consistent for $s_j^0$, for $j=1,~\ldots,~m$. Proposition~\ref{prnm14} gives the convergence rate $T$ of ${\hat{\tau}}_j$, $j=1,~\ldots,~m$. In reality we may encounter the case where the shift is time-dependent, and, in particular as $T$ tends to infinity, the shift may shrink towards 0 at rate $v_T$, i.e., $\boldsymbol{\mathrm \uptheta}^{(j)}-\boldsymbol{\mathrm \uptheta}^{(j-1)}=\mathbf{M}v_T$, where $\mathbf{M}$ is independent of $T$ and $v_T\xrightarrow[T\rightarrow \infty]{}0$. In this case, the validity of Propositions~\ref{prnm13} and \ref{prnm14} depends on the speed $v_T$. In fact, using similar arguments as in the proofs of these two propositions (see \ref{appendixb}), we have the following corollary.

\begin{cor}~\label{corm-1}
Suppose that $\boldsymbol{\mathrm \uptheta}^{(j)}-\boldsymbol{\mathrm \uptheta}^{(j-1)}=\mathbf{M}v_T$, where $j=1,\ldots,m+1$, $\mathbf{M}$ is independent of $T$ and $v_T\xrightarrow[T\rightarrow \infty]{}0$ but $T^{1/2-r^*}v_T\xrightarrow[T\rightarrow \infty]{}\infty$ for some $0<r^*<1/2$. Then under Assumptions 1--\ref{asm4}, we have (i) $\hat{s}-s^0\xrightarrow[T\rightarrow \infty]{P}0$ and (ii) for every $\epsilon>0$, there exists a $C>0$ such that for
large $T$, $P\left(Tv_T^2|\hat{s}-s|>C\right)<\epsilon$.
\end{cor}
\noindent \textit{Proof:} See \ref{appendixb}.

\subsection{Maximum log-likelihood method}
\noindent We introduce an alternative method to estimate the location of the change points based on the maximum of log-likelihood function. Recall that the log-likelihood function for (\ref{ou1}) with the exact change points $\tau_1^0,~\ldots~\tau_m^0$ is given by
\begin{eqnarray}
\log \ell(\tau_1^0,~\ldots~,\tau_m^0,\boldsymbol{\mathrm \uptheta})&=&
\frac{1}{\sigma^2}\int_{0}^TS(\boldsymbol{\mathrm \uptheta},t,X_t)dX_t
-\frac{1}{2\sigma^2}\int_{0}^TS^2(\boldsymbol{\mathrm \uptheta},t,X_t)dt \nonumber \\
&=&\frac{1}{\sigma^2}\sum_{j=1}^{m+1}\int_{\tau_{j-1}^0}^{\tau_{j}^0}S(\boldsymbol{\mathrm \uptheta}^{(j)},t,X_t)dX_t -\frac{1}{2\sigma^2}\sum_{j=1}^{m+1}\int_{\tau_{j-1}^0}^{\tau_{j}^0}S^2(\boldsymbol{\mathrm \uptheta}^{(j)},t,X_t)dt.\label{loglm}
\end{eqnarray}
\noindent From \eqref{loglm}, the estimator of the location of the change points is given
\begin{equation}\label{mlem1}\boldsymbol{\hat{\mathrm \tau}}=\arg\max_{\tau} \log \ell(\boldsymbol{\mathrm \tau},\hat{\boldsymbol{\mathrm \uptheta}}(\boldsymbol{\mathrm \tau})),\end{equation}
where $\hat{\boldsymbol{\mathrm \uptheta}}(\boldsymbol{\mathrm \tau})$ is the MLE of $\boldsymbol{\mathrm \uptheta}$ based on the given change points $\boldsymbol{\mathrm \tau}=(\tau_1,~\ldots,~\tau_m)$.\\
\ \\
\noindent In practice, the calculation of $\log \ell(\boldsymbol{\mathrm \tau},\hat{\boldsymbol{\mathrm \uptheta}}(\boldsymbol{\mathrm \tau}))$ in (\ref{mlem1}) relies on numerical approximation methods (see Auger and Lawrence (1989)) to compute the integrals inside $\log \ell(\boldsymbol{\mathrm \tau},\boldsymbol{\mathrm{\hat{\uptheta}}}(\boldsymbol{\mathrm \tau}))$. For example, by approximating the Riemann sum based on a partition $0=t_0^*<...<t_n^*=T$ with $\Delta_t^*=t_{i+1}^*-t_i^*$, (\ref{mlem1}) is calculated as
\begin{equation}\label{rmloglm} \log \ell^*([0,T],\boldsymbol{\mathrm \tau},\boldsymbol{\mathrm{\hat{\uptheta}}}(\boldsymbol{\mathrm \tau}))
=\frac{1}{\sigma^2}\sum_{j=1}^{m+1}\sum_{t_i^*\in(\tau_{j-1},\tau_{j}]}\boldsymbol{\mathrm{\hat{\uptheta}}}^{(j)\prime}V(t)'(X_{t_{i+1}^*}-X_{t_i^*})-\frac{1}{2\sigma^2}\sum_{j=1}^{m+1}\sum_{t_i^*\in(\tau_{j-1},\tau_{j}]}\left(\boldsymbol{\mathrm{\hat{\uptheta}}}^{(j)\prime}V(t)'\right)^2\Delta_t^*.\end{equation}
The approximated version of (\ref{mlem1}) is then
\begin{equation}\label{mlem2}\boldsymbol{\hat{\mathrm \tau}}=\arg\max_{\boldsymbol{\mathrm \tau}=(\tau_1,...,\tau_m)} \log \ell^*([0,T],\boldsymbol{\mathrm \tau},\hat{\boldsymbol{\mathrm \uptheta}}(\boldsymbol{\mathrm \tau})).\end{equation}

\noindent {\bf Consistency of the proposed estimator}\\
\noindent We now link results (\ref{mlem2}) and (\ref{cpm1}), which are
the respective results from the LSSE- and MLE-based methods.
\begin{prn}\label{prn-mlem1}
Consider the observed process $X_t$, $t\in [0,T]$. If the increment $\Delta_t^*$ is equal to $\Delta_t$ defined in Section~\ref{LSSE} then under Assumptions \ref{asm1}--\ref{asm4}, the asymptotic results given in Propositions~\ref{prnm13} and \ref{prnm14} as well as in Corollary~\ref{corm-1} also hold for (\ref{mlem2}).
\end{prn}

\noindent \textit{Proof:} See \ref{appendixb}. \\
\ \\
\noindent Using the {\color{black}consistency} properties, we can establish the asymptotic normality for the MLE of drift parameters based on the estimated change points.

\subsection{Asymptotic normality of $\hat{\uptheta}$ based on the estimated change points}\label{sectionanormality}
\noindent Previously, we established the T-rate consistency of the estimated change points. Based on these asymptotic consistency results, we extend the asymptotic normality results in Zhang~(2015) to the case of multiple change points.

\begin{prn}\label{anm1}
 Let $\boldsymbol{\hat{\mathrm \tau}}=\{{\hat{\tau}}_1,~\ldots~,{\hat{\tau}}_m\}$ be the estimated change points using (\ref{cpm1}) or (\ref{mlem1}).  Then, under Assumptions \ref{asm1}--\ref{asm4}, we have that for $i=1,~\ldots,~p$,
\begin{equation}\label{anm1-1}\frac{1}{T}\int_{{\hat{\tau}}_{j-1}}^{{\hat{\tau}}_j}{X}_t\varphi_k(t)dt\xrightarrow[T\rightarrow \infty]{p}(s_{j}-s_{j-1})\int_0^1\tilde{h}(t)\varphi_k(t)dt.\end{equation}
Similarly,
\begin{equation}\label{anm1-1}\frac{1}{T}\int_{{\hat{\tau}}_{j-1}}^{{\hat{\tau}}_j}{X}_t^2dt\xrightarrow[T\rightarrow \infty]{p}(s_{j}-s_{j-1})\left(\int_0^1\tilde{h}^2(t)dt+\frac{\sigma^2}{2 a^{(j)}}\right).\end{equation}
\end{prn}
\noindent \textit{Proof:} See \ref{appendixanormality}.

\begin{prn}\label{anm3}
 Under the same conditions as in Proposition~\ref{anm1}, we have that for $i=1,~\ldots,~p$,
\begin{equation}P\left(\int_{0}^{T}(\frac{1}{\sqrt{T}}\varphi_k(t)\mathbbm{1}_{{\hat{\tau}}_{i-1}<\leq {\hat{\tau}}_i})^2<\infty\right)=1,
\mbox{ and }
P\left(\int_{0}^{T}(\frac{-1}{\sqrt{T}}X_t\mathbbm{1}_{{\hat{\tau}}_{i-1}<\leq {\hat{\tau}}_i})^2<\infty\right)=1.\end{equation}
\end{prn}
\noindent \textit{Proof:} This follows from (\ref{sol2}) and the similar {\color{black}arguments} in the proof of Proposition~2.3.5 in Zhang~(2015). \\
\ \\
\noindent
Employing Propositions~\ref{anm1} and \ref{anm3}, we have the following results.
\begin{prn}\label{anm4}
\begin{equation}T\boldsymbol{\mathrm Q}_{({\hat{\tau}}_{j-1}, {\hat{\tau}}_{j})}^{-1} \xrightarrow[T\rightarrow \infty]{P} \frac{1}{s_j^0-s_{j-1}^0}\boldsymbol{\mathrm \Sigma}_j^{-1}, j=1,~\ldots,~m,\end{equation}
where $\boldsymbol{\mathrm \Sigma}_j$ is defined in (\ref{sigmaj}).
\end{prn}
\noindent \textit{Proof:} See \ref{appendixanormality}.

\begin{prn}\label{anm5}
Let $\boldsymbol{\hat{\mathrm \tau}}=\{{\hat{\tau}}_1,~\ldots,~{\hat{\tau}}_m\}$ be the estimated change points from (\ref{cpm1}), and suppose that Assumptions \ref{asm1}--\ref{asm4} hold. Then,
\begin{equation}\frac{1}{\sqrt{T}}\boldsymbol{\mathrm r}_{({\hat{\tau}}_{j-1}, {\hat{\tau}}_{j})}-\frac{1}{\sqrt{T}}\boldsymbol{\mathrm r}_{({\tau}_{j-1}^0, {\tau}_{j}^0)} \xrightarrow[T\rightarrow \infty]{P} \boldsymbol{\mathrm 0},\end{equation}
where
$$\boldsymbol{\mathrm r}_{({\hat{\tau}}_{j-1}, {\hat{\tau}}_{j})}=\left(\int_{{\hat{\tau}}_{j-1}}^{{\hat{\tau}}_j}\varphi_1(t)dW_t,\ldots,\int_{{\hat{\tau}}_{j-1}}^{{\hat{\tau}}_j}\varphi_p(t)dW_t, -\int_{{\hat{\tau}}_{j-1}}^{{\hat{\tau}}_j}X_tdW_t\right)^{\top}.$$
\end{prn}
\noindent \textit{Proof:} See \ref{appendixanormality}. \\
\ \\
\noindent The next result following from Propositions \ref{anm4} and \ref{anm5} plays an essential role in proving the asymptotic {\color{black}normality} of MLE of the drift paramter based on the proposed estimated change points.
\begin{prn}\label{anm6}
 Suppose $\boldsymbol{\hat{\mathrm \tau}}=\{{\hat{\tau}}_1,~\ldots~,{\hat{\tau}}_m\}$ is the estimated set of change points from (\ref{cpm1}) and Assumptions \ref{asm1}-\ref{asm4} hold. Then
\begin{equation}\left( \frac{1}{\sqrt{T}}\boldsymbol{\mathrm r}_{(0, \tau_1^0)},~\ldots~,\frac{1}{\sqrt{T}}\boldsymbol{\mathrm r}_{(\tau_{m}^0, T)} \right)\xrightarrow[T\rightarrow \infty]{D}  \boldsymbol{\mathrm r}\sim N_{(m+1)(p+1)}\left(0,\boldsymbol{\tilde{\mathrm \Sigma}}\right),\end{equation}
where $\boldsymbol{\tilde{\mathrm \Sigma}}=\mathrm{diag}\left({s_1^0}\boldsymbol{\mathrm \Sigma}_1,(s_2^0-s_1^0)\boldsymbol{\mathrm \Sigma}_2,~\ldots~,(1-s_m^0)\boldsymbol{\mathrm \Sigma}_{m+1}\right)$.
\end{prn}
\noindent \textit{Proof:} See \ref{appendixanormality}. \\
\ \\
\noindent From Proposition \ref{anm6} together with Slutsky's Theorem, the following corollary establishes the asymptotic normality for $\hat{\boldsymbol{\mathrm \uptheta}}=\left(\boldsymbol{\mathrm{\hat{\uptheta}}}^{(1)},~\ldots~,\boldsymbol{\mathrm{\hat{\uptheta}}}^{(m)}\right)$.

\begin{cor}\label{anc1}
 Let $\boldsymbol{\hat{\mathrm \tau}}=\{{\hat{\tau}}_1,~\ldots~,{\hat{\tau}}_m\}$ be the estimated change points from (\ref{cpm1}), and supposes that Assumption \ref{asm1}-\ref{asm4} hold.
\begin{equation}\sqrt{T}(\boldsymbol{\mathrm{\hat{\uptheta}}}-\boldsymbol{\mathrm \uptheta})\xrightarrow[T\rightarrow \infty]{D} \rho \sim N_{((m+1)(p+1))}\left(0,\sigma^2\boldsymbol{\tilde{\mathrm \Sigma}}^{-1}\right),\end{equation}
where $\boldsymbol{\tilde{\mathrm \Sigma}}^{-1}=\mathrm{diag} \left( \frac{1}{s_1^0}\boldsymbol{\mathrm \Sigma}_1^{-1}, \frac{1}{s_2^0-s_1^0}\boldsymbol{\mathrm \Sigma}_2^{-1},~\ldots~,\frac{1}{1-s_m^0}\boldsymbol{\mathrm \Sigma}_{m+1}^{-1}
\right)$.
\end{cor}
\noindent \textit{Proof:} See \ref{appendixanormality}.

\section{Estimating the number of change points}
\noindent In the last section, we developed two consistent estimation methods for the case when
the number of change points is known. {\color{black}In this section,} we extend our examination of the change-point
problem when the number of change points is also unknown. Hence,
we are interested in knowing the number of change points as well as their exact locations. \\
\ \\
\noindent One popular methodology {\color{black}for} detecting the
unknown number of change points is {\color{black} to treat} this issue as a model-selection problem.
For instance, adding one change point into (\ref{ou1}) brings $p+1$ extra drift parameters into the model. Thus, detecting the number of change points can be considered as selecting the most suitable statistical model from a series of candidate models with different number of change points, and this can be solved using an informational approach. Such approach deems the most appropriate model as the one {\color{black}which} minimises the log-likelihood-based information criterion
\begin{equation}\label{ic0} \mathcal{IC}(m)=-2\log\ell(\boldsymbol{\mathrm \tau}, \hat{\uptheta})
+(m+1)h(p)\phi(T).\end{equation}
In \eqref{ic0}, $\log\ell(\boldsymbol{\mathrm \tau}, \hat{\uptheta})$ is defined in (\ref{loglm}); $\boldsymbol{\hat{\mathrm \tau}}$
is obtained via (\ref{mlem1}) corresponding to each $m$; $h(p)=p+1$ if there is no
change in $\sigma$ (or $p+2$ if there is a change in $\sigma$); $\phi(T)$ is a non-decreasing
function of $T$, the length of the data set; and $m$ is the potential number of change points to
be determined. \\
\ \\
\noindent  Based on the asymptotic results for the Riemann sum approximation of the log-likelihood function $\log\ell(\boldsymbol{\mathrm \tau}, \hat{\uptheta})$, we use the criterion
\begin{equation}\label{ic} \mathcal{IC}(m)=-2\log\ell^*([0,T],\boldsymbol{\mathrm \tau},\hat{\boldsymbol{\mathrm \uptheta}}(\boldsymbol{\mathrm \tau}))
+(m+1)(p+1)\phi(T),\end{equation}
where $\log\ell^*([0,T],\boldsymbol{\mathrm \tau},\hat{\boldsymbol{\mathrm \uptheta}}(\boldsymbol{\mathrm \tau}))$ is given in (\ref{rmloglm}).\\
\ \\
\noindent  Note that, if the number of change points is known, the term
$(m+1)h(p)\phi(T)$ is fixed and the approach covering (\ref{ic}) is equivalent to the maximum
log-likelihood method introduced in the previous section. The efficiency of
information criterion depends on the choice of the penalty criterion $\phi(T)$. For example, if $\phi(T)=2$,
then (\ref{ic}) reduces to the well-known Akaike information criterion (AIC)~\cite{Ak73}. However, in practice, a model selected by minimising the AIC may not
be asymptotically consistent in terms of the model order; see Schwarz~(1978).
Modified information criteria were, thus, proposed to overcome this problem. One example is the Schwarz information criterion (SIC) \cite{Sc78}, which
sets $\phi(T)$ as the {\color{black}logarithm} of the sample size.
{\color{black}The} SIC has been
successfully applied to the change-point analysis in the
literature. As it gives an asymptotically consistent estimate
of the order of the true model, we also adapt {\color{black}the} SIC for our theoretical development.  \\
\ \\
{\color{black}Note that the penalty term $(m+1)(p+1)\phi(T)$ in SIC increases as the sample size increases. Hence, for large sample size SIC tends to ignore the relatively small changes in the process. This feature makes it useful for those who are mainly interested in studying only the major changes within certain time period.} Further, based on the SIC, we have the following asymptotic results for (\ref{ic}).
\begin{prn}\label{asic1}
Under Assumptions 1--4, we have that for large $T$,
(i) $\mathcal{IC}(m=m^0)< \mathcal{IC}(m<m^0)$
with probability 1 and (ii) $\mathcal{IC}(m=m^0)< \mathcal{IC}(m>m^0)$ with probability 1.
\end{prn}
\noindent \textit{Proof:} See \ref{appendixsic}.\\
\ \\
\noindent  Proposition~\ref{asic1} tells us that, for large T, $\mathcal{IC}(m)$ reaches its minimum value when $m=m^0$ and this allows us to detect the exact value of $m^0$.

\section{Computing algorithms}
\noindent In this section, we put forward algorithms for computing (\ref{cpm1}) and (\ref{mlem2}) when $m^0$ is known. {\color{black}We also provide algorithms for computing} (\ref{ic}) when $m^0$ is unknown. {\color{black}Based on these algorithms, a simulation study to examine the efficiency of the proposed methods for different time periods $T$ is presented in Section~\ref{simulation}.} Our numerical results will show that, {\color{black}with our sample parameter set,} the proposed methods perform well {\color{black}for values of $T$ as small as $T=5$}.

\subsection{Algorithm for (\ref{cpm1}) and (\ref{mlem2}) (with known $m^0$) }\label{alg-1}
\noindent In estimating the unknown locations of change points, a standard searching method is to compute the criteria, i.e., least squared errors for (\ref{cpm1}) or maximum log likelihoods for (\ref{mlem2}), through all possible locations of change points and search for the one that returns the optimal value. However, for $m$ change points the associated costs for the above searching procedure are of order $O((T/\Delta_t)^m)$. Thus, for large $T$ and small $\Delta_t$, the computations can be time consuming. To overcome this problem, we adopt to our two proposed LSSE and MLE methods
a dynamic programming algorithm due to Bai and Perron~(2003), see also Perron and Qu~(2006),
to reduce the computational cost to $O(m(T/\Delta_t)^2)$ for $m$ change points. {\color{black}This algorithm} is very efficient when $m\geq 2$. \\
\ \\
\noindent {\bf Algorithm 1}\\
\noindent Let $H_1(r, T_r)$ be either $H_1(r, T_r)=\min_{\tau}SSE([0,T_r],\boldsymbol{\mathrm \tau},\boldsymbol{\mathrm{\hat{\uptheta}}}(\boldsymbol{\mathrm \tau}))$, the least sum squared error for (\ref{cpm1}) or $H_1(r, T_r)=\max_{\tau}\log\ell^*([0,T_r],\boldsymbol{\mathrm \tau},\boldsymbol{\mathrm{\hat{\uptheta}}}(\boldsymbol{\mathrm \tau}))$, the maximum Riemann sum approximation of log likelihood for (\ref{mlem2})) computed based on the optimal partition of time interval $[0,T_r]$ that contains $r$ change points. Also, let $H_2(a, b)$ be the SSE for (\ref{cpm1}) or Riemann sum approximation of log likelihood for (\ref{mlem2})) computed based on a time regime $(a,b]$. Further, we assume that Assumption~\ref{asm1} holds and let $h=\epsilon T$ be the minimal permissible length of a time regime. Then (\ref{cpm1}) or (\ref{mlem2}) with known $m=m^0$ can be computed as follows.  \\
\begin{enumerate}[Step 1:]
\item Compute and save $H_2(a,b)$ for all time periods $(a, b]$ that satisfy $b-a\geq h$.
\item Compute and save $H_1(1,T_1)$ for all $T_1 \in[2h, T-(m-1)h]$ by solving the optimisation problem
\begin{equation*}H_1(1,T_1)=
\begin{cases}
\min_{a\in[h,T_1-h]}[H_2(0,a)+H_2(a,T_1)]& \mbox{for (\ref{cpm1})}\\
\max_{a\in[h,T_1-h]}[H_2(0,a)+H_2(a,T_1)]& \mbox{for (\ref{mlem2})}.
\end{cases}
\end{equation*}
\item Sequentially compute and save
\begin{equation*}H_1(r,T_r)=
\begin{cases}
\min_{a\in[rh,T_r-h]}[H_1(r-1,a)+H_2(a,T_r)]& \mbox{for (\ref{cpm1})}\\
\max_{a\in[rh,T_r-h]}[H_1(r-1,a)+H_2(a,T_r)]& \mbox{for (\ref{mlem2})}
\end{cases}
\end{equation*}
for $r=2,~\ldots~m-1$, and $T_r \in [(r+1)h,T-(m-r)h]$.
\item Finally, the estimated change points are obtained by solving
\begin{equation*}H_1(m,T)=
\begin{cases}
\min_{a\in[mh,T-h]}[H_1(m-1,a)+H_2(a,T)]& \mbox{for (\ref{cpm1})}\\
\max_{a\in[mh,T-h]}[H_1(m-1,a)+H_2(a,T)]& \mbox{for (\ref{mlem2})}
\end{cases}
\end{equation*}
and $H_1(m-1,a)=H_2(0,a)$ if $m=1$.
\end{enumerate}
\noindent The steps in Algorithm 1 can be viewed as a combination of two components. Step 1 computes all possible choices of $H_2$, and the computations in this step are at most of order $O((T/\Delta_t)^2)$ as there are at most $(T/\Delta_t)^2$ different time periods $(a,b]$ in the dataset. This step is useful, {\color{black}since} in the succeeding steps some pairs of $(a,b]$ will be visited more than once during the optimisation process, {\color{black}so} using previously saved results {\color{black}for} $H_2(a,b)$ will be helpful to reduce computations. Steps 2-4 can be treated as an application of {\color{black}the} Segment Neighbourhood Search (SNS) method introduced by Auger and Lawrence (1989). The goal of Steps 2--4 is to search for the global optimal locations of the $m$ change points and the total computation costs in these steps are also of $O(m(T/\Delta_t)^2)$. Note that when $m=m^0=1$ ({\color{black}a} single change point), only the last step is needed to search for the optimal location of a change point, and the related computations costs are of $O(T/\Delta_t)$.   \\
\ \\
\noindent When $m^0>1$ and $T$ is large, Algorithm 1 can be extremely time-consuming because of the $O(m(T/\Delta_t)^2)$ computations, we aim to decrease the computational costs in this case. Apparently, some computations in Step 1 may be redundant. For example, in Step 2, the domain for the optimisation problem is $a\in[h,T_1-h]$ for each
 $T_1 \in[2h, T-(m^0-1)h]$; so the calculations of $H_2(0,b^*)$ for all $b^*>T-m^0h$ in Step 1 become unnecessary as these results will not be used. Thus, the computations in Step 1 could be moved into Steps 2- 4 so that only the necessary $H_2$ are computed and stored. This means one could begin the algorithm from Step 2 by computing and storing $H_2(0,b)$ for $b\in [h,T-mh]$ and $H_2(a,b)$ for $a\in [h,T-mh]$, $b\in [2h,T-(m-1)h]$, then solve for $H_1(1,T_1)$. In Step 3, for $r=2,~\ldots~,m-1$, one only needs to compute and save $H_2(a,b)$ for ($a\in [T-(m-r+2)h,T-(m-r+1)h], b\in [(r+1)h,T-(m-r)h])\bigcup (a\in [rh,T-(m-r+2)h], b\in [T-(m-r+1)h,T-(m-r)h])$ before solving for $H_1(1,T_1)$. Finally, in Step 4, we compute and store $H_2(a,T)$ for $a\in [mh,T-h]$ before solving for $H_1(m,T)$.

\subsection{Algorithm for (\ref{ic}) (with unknown $m^0$) }\label{alg-2}
\noindent When $m^0$ is unknown, one may compute and compare the $m$ values of (\ref{ic}) and $m$ varies from $0$ up to $m_{\max}$ for some $0\leq m_{\max}\leq
\left \lceil [T/h] \right \rceil.$  The upper bound $m_{\max}$ can also be predetermined from the descriptive analysis of the observed processes. For each $m$, one can first apply Algorithm 1 to obtain the estimated change points and compute (\ref{ic}) accordingly. After the $m_{\max}$ computations, the desired $\hat{m}$ is the one that returns the minimum value of (\ref{ic}). By Proposition~\ref{asic1}, $\hat{m}$ is consistent when $T$ is large, provided $m^0\in[0,m_{\max}]$. \\
\ \\
\noindent If we directly apply Algorithm 1 to $m=1,~\ldots~,m_{\max}$, the total computations will be of order $O((1+2+...+m_{\max})(T/\Delta_t)^2)$. To further simplify the computations, we study the behaviour of Step 2--3 in Algorithm 1 when $m$ increases from $m^*$ to $m^*+1$. In this case, the ranges of $T_r$ reduces from $T_r \in[(r+1)h, T-(m^*-1)h]$ to $T_r \in[(r+1)h, T-(m^*)h]$ for each $r=1,~\ldots,~m-1$. This implies that the stored optimisation results of Steps 2 and 3 in Algorithm 1 at the previous step ($m=m^*$) can also be used in the current step ($m=m^*+1$). Therefore, with the previously stored results, the only step that needs to be updated for each $m$ is when $r=m-1$ and $r=m$, and the associated computations are of order $O((T/\Delta_t)^2)$ at $r=m-1$ and $O(T/\Delta_t)$ at $r=m$. Based on these considerations, we tailor the SNS algorithm for (\ref{ic}) and the related computations for $m=0~\ldots,~m_{\max}$ are of order $O(m_{\max}(T/\Delta_t)^2)$.\\
\ \\
\newpage
 \noindent {\bf Algorithm 2 (SNS method)}
\begin{enumerate}[Step 1:]
\item Follows all steps in Algorithm 1 to search for the optimal locations of the $m$ estimated change points then store the computed value of $(\ref{ic})$ for $m=0,~1,~2$. Note that the results of $H_2 (a,b)$ for all $(a,b]$ such that $a-b\geq h$ as well as the optimisation results of $H_1(r,T_r)$ for all $r=1,~\ldots,~m$ and $T_r \in [(r+1)h,T-(m-r)h]$ need to be stored for future use.
\item For $m=3,~\ldots,~m_{\max}$, first let $r=m-1$ and $T_r \in [(r+1)h,T-(m-r)h]$ then compute and store $H_1(r,T_r)$. Next let $r=m$ and the estimated change points are obtained by solving $H_1(m,T)$, where $H_1(r,T_r)$ and $H_1(m,T)$ are defined in Algorithm 1. Finally, based on the estimated $m$ change points, compute and store $\mathcal{IC}(m)$.
\item $\hat{m}$ is obtained from $m=1,~\ldots~,m_{\max}$ that returns the smallest value of (\ref{ic}).
\end{enumerate}
\noindent The advantage of SNS method is that it returns the optimal locations of change points for every $m=1,~\ldots,~m_{\max}$. Hence, it is useful if one interested in investigating the relationships between the locations of change points and $m$. However, for large $T$ and $m_{\max}$, {\color{black}the $O(m_{\max}(T/\Delta_t)^2)$ computational costs in the SNS method may be high}.\\
\ \\
\noindent In addition to SNS method, another dynamic programming algorithm for finding the unknown number of change points is called the Optimal Partitioning (OP) algorithm introduced by Jackson, et al. (2005). The related computational costs are of order $O(T/\Delta_t)^2)$ for any $m$; hence, it is more efficient than SNS when $m_{\max}$ is large. Based on the OP algorithm, Killick, et al.~(2012) introduced the Pruned Exact Linear Time (PELT) method. Although the maximum computational costs for the PELT algorithm is still up to $O(n^2)$ for a data set with size $n$, the computations in the PELT method, which involved pruning of the solution space {\color{black}under some conditions} can be much less than those {\color{black}required} in the OP. \\
\ \\
\noindent {\color{black}However}, the PELT method introduced in Killick, et al.~(2012) may not satisfy Assumption~\ref{asm1}, which is essential for most of the theoretical properties developed in this paper. In fact, under Assumption~\ref{asm1}: (i) there is no change point in time period $[0,h)$, where $h$ is defined in Algorithm 1; (ii) if there is a change point $\tau^*\in [h,2h)$, then there is no non-zero change point prior to $\tau^*$; and (iii) for any potential change point $\tau^*\in [2h,T]$, the minimal distance between it and the most recent change point prior to this change point is at least $h$. Based on these considerations, we use the following modified version of the PELT algorithm. \\
\ \\
\noindent {\bf Algorithm 3 ({\color{black}Modified} PELT method)}\\
Let $n=T/\Delta_t$ be the length of the data set based on the partition $0=t_0<t_1<\ldots<t_n=T$ of time period $[0,T]$ with increment $\Delta_t$, and let $t_h=h/\Delta_t$. Set $SS_{2t_h}=\{0\}$ and $F(t_i)=0$ for $i=1,~\ldots,~t_{h}-1$. Then, for $i=t_h,~\ldots,~n$, compute and store the values obtained from {\color{black}the} following steps.
\begin{enumerate}[Step 1:]
\item For $i=t_h,~\ldots,~2t_h-1$, compute and store $F(t_i)=-2*\log\ell^*([0,t_i],\boldsymbol{\mathrm{\hat{\uptheta}}})+(p+1)\log(n)$, where $\log\ell^*([a,b],\hat{\boldsymbol{\mathrm \uptheta}})$ is the Riemann sum approximation {\color{black}(midpoints)} of the log likelihood (no change point) evaluated at time period $[a,b]$ with $\boldsymbol{\mathrm \uptheta}=\hat{\boldsymbol{\mathrm \uptheta}}$.
\item For $i=2t_h,~\ldots,~n$, compute and store: (i) $F(t_i)=\min_{t\in SS_{t_i}}F(t)-2*\log\ell^*((t,t_i],\hat{\boldsymbol{\mathrm \uptheta}})+(p+1)\log(n)$; (2). $\tau_{t_i}^*=\arg\min_{t\in SS_{t_i}}F(t)-2*\log\ell^*((t,t_i],\hat{\uptheta})+(p+1)\log(n)$; (3). $SS_{t_{i+1}}=\{0\} \cup \{t\in SS_{t_i}\cup \{t_i-t_h+1\}: F(t)-2*\log\ell^*((t,t_i],\hat{\boldsymbol{\mathrm \uptheta}})\leq F(t_i)\}$.
\item Denote $cp(0)=\emptyset$. Then, the optimal change points can be obtained by solving $cp(t^*)=(cp(t^*=\tau_{t^*}^*),\tau_{t}^*)$ with $t^*$ starts from $T$ and iterates recursively until $\tau_{t^*}^*=0$.
\end{enumerate}

\section{Numerical demonstrations}\label{nd}
\noindent The Monte-Carlo simulation technique will be used in Subsection~\ref{simulation}
(i) {\color{black} to evaluate the comparative} performance
of the two estimation methods, viz. LSSE in (\ref{cpm1}) and MLL in (\ref{mlem2}) to determine the unknown
location of $m^0$ change points assumed to already exist; and (ii)  {\color{black}to test the} method in (\ref{ic}) for detecting the unknown number of change points. In Subsection \ref{realdata}, we illustrate the various implementation details of our proposed methods on some observed financial market data.
\ \\
\subsection{Monte-Carlo simulation study}\label{simulation}
\noindent Our simulation considers two different scenarios (or cases). In the first case, we study
the performance of the proposed methods under the classical OU process defined by
\begin{equation}\label{sim1-1}
dX_t=
(\mu^{(j)}-\alpha^{(j)}X_t)dt+\sigma dW_t, \quad \mbox{if }s_{j-1}^0 T<t<s_j^0 T, \quad j=1,~\ldots,~m.
\end{equation}
In the second case, the performance
evaluation of the proposed methods is tested assuming a periodic mean-reverting OU process, with 2-dimensional periodic incomplete set of functions
$\left\{1, \sqrt{2}\cos\left( \frac{\pi t}{2\Delta_t} \right)\right\}$ {\color{black}(which are orthogonal on [0,$T$] with weight fixed to 1}),
given by
\begin{equation}\label{sim2-1}
dX_t=
\left[ \mu_1^{(j)}+\mu_2^{(j)} \sqrt{2}\cos\left( \frac{\pi t}{2\Delta_t}\right)-\alpha^{(j)}X_t \right]dt+\sigma dW_t,\quad \mbox{if }s_{j-1}^0 T<t<s_j^0 T, \quad j=1,~\ldots,~m.
\end{equation}
where $\Delta_t=t_{i+1}-t_i$ is the increment for $[0,T]$. \\
\ \\
\noindent Each case consists of 500 iterations. {\color{black}Although an exact solution is available, we choose to use the Euler-Maruyama discretisation scheme to be consistent with {\color{black}the results of} Zhang~(2015)}. In each iteration,
we first generate a desired simulated process based on a given period $T$  with pre-assigned ``true" parameters such as the number and location of change points and the model coefficients. To evaluate the performance of (\ref{cpm1}) and (\ref{mlem2}), we {\color{black}specify} the number of change points to be known but the rate is unknown. Then, we estimate and record the {\color{black}change points' arrival rates} by applying (\ref{cpm1}) and (\ref{mlem2}) on the simulated process. The detailed simulation setup and results are reported in Subsection~\ref{simulation-1}.  \\
\ \\
\noindent In Subsection~\ref{simulation-2}, we also use the Monte-Carlo simulation method to  investigate the performance of (\ref{ic}). That is, we assume that $m^0$ is unknown and apply (\ref{ic}) with $m$ ranging from $0$ to $m^0+3$. Then, the $m$ that returns the minimum value of (\ref{ic}) is chosen as the estimated value for the number of change points. After 500 iterations, we analyse the performance of the proposed methods based on the recorded results.

\subsubsection{Estimating the rate $s_j$ of change points}\label{simulation-1}
\noindent We first study the performance of (\ref{cpm1}) and (\ref{mlem2}) in estimating the rates of the change points with known $m^0$. For the simulation setup, we consider the case where $m^0=2$ and $3$ with different time periods $T=5,~10,~20$, respectively.
{\color{black}The pre-assigned values of the coefficients are provided in Table~\ref{cp-setup}.}

\ \\
%

\begin{table}[!htbp]
\small \caption{Pre-assigned coefficients (with known number of change points)} \centering
\scriptsize
\begin{tabular}{|c|c|c|c|c|c|c|c|c|}
\hline
\multicolumn{2}{|c|}{} &\multicolumn{3}{c|}{2 change points ($0.35T$, $0.7T$)}& \multicolumn{4}{c|}{ 3 change points ($0.25T$, $0.5T$, $0.75T$)}\\
\hline
\small Case& Coefficient & $j=1$ & $j=2$ & $j=3$ & $j=1$ &  $j=2$  &$j=3$ &$j=4$  \\
\hline
\multirow{2}{*}{1} & $\mu^{(j)}$ & 0.08 & 2.50 & 0.08 & 0.08 & 2.50 & 0.08 & 2.50\\
\cline{2-9}
 & $\alpha^{(j)}$ & 0.10 & 1.00 & 0.50 &0.10 & 1.00 & 0.50 & 1.00\\
 \hline
\multirow{3}{*}{2} & $\mu_1^{(j)}$ & 0.08 & 2.50 & 0.08 & 0.08 & 2.50 & 0.08 & 2.50\\
\cline{2-9}
 & $\mu_2^{(j)}$ & 0.02 & 1.20 & 0.02 & 0.02 & 1.20 & 0.02 & 1.20\\
\cline{2-9}
 & $\alpha^{(j)}$ & 0.10 & 1.00 & 0.50 &0.10 & 1.00 & 0.50 & 1.00\\
 \hline
\end{tabular}
\label{cp-setup}
\end{table}

\noindent  For each case, after 500 iterations we record the mean of the estimates based on (\ref{cpm1}) and (\ref{mlem2}), {\color{black} together with the $95\%$ empirical confidence interval (i.e., locating the $2.5$ and $97.5$ percentiles) and also the mean-squared error.} The results are reported in Tables~\ref{m2-case1}--\ref{m3-case2}. 
{\color{black}
In this section, we also report the following histograms: Figure~\ref{fig:figure1} presents the histograms of the estimated rates based on MLL method for Case 1 with $m^0=3$ when $T$ increases from 5 to 20, which shows the behavior of the estimated rates as $T$ increases. Figure~\ref{fig:figure2} shows the histograms of the estimated rates for the two cases under the same conditions (i.e. $m^0=3$, $T=10$ and the rates are estimated by LSSE method). Moreover, for the convenience of the reader, all histograms of the estimated rates for $m^0=3$ are provided in \ref{appendixhistogram} (Figures~\ref{figure4}--\ref{figure6}) for reference.} \\
\ \\
\noindent From Tables~\ref{m2-case1}--\ref{m3-case2} (along with {\color{black}Figures~\ref{fig:figure1}, \ref{fig:figure2} and \ref{figure4}--\ref{figure6} in \ref{appendixhistogram}), we see that in Cases 1 and 2, both proposed methods (\ref{cpm1}) and (\ref{mlem2}) estimate very accurately the exact rates of change points. In particular,
the sample means of the estimated change points' arrival rates are close to the exact values, and the results obtained by the 2 proposed methods are very close, which confirms Proposition~\ref{prn-mlem1}. We also clearly observe that as $T$ increases from 5 to 20, the lengths of the $95\%$ empirical confidence intervals and MSEs of the two estimators all decrease. This is well substantiated for example by the pertinent histograms in Figure~\ref{fig:figure1}, which shows that when MLL method is employed to estimate the change points' arrival rates in Case 1 with $m^0=3$, the central tendencies of the estimated rates are all close to their exact values, and the sample variances decrease as $T$ becomes larger. Similar evidences are shown for other choices of scenarios (different combinations of cases, methods and time period $T$) as illustrated in Figures~\ref{figure4}-- \ref{figure6}. Although not shown in this paper, the histograms for the case $m^0=2$ exhibit similar features}. These outcomes confirm the theoretical findings regarding the asymptotic consistency of our two proposed methods. \\
\ \\
\noindent Also, from Tables~\ref{m2-case1}--\ref{m3-case2} and Figure~\ref{fig:figure1} (see \ref{appendixhistogram} as well), the lengths of the $95\%$ confidence intervals (C.I.'s), with corresponding MSEs, of the estimated rates for the first and the last unknown change points ($\hat{s}_1$ and $\hat{s}_3$) are more accurate than those of the middle change point ($\hat{s}_2$). Moreover, the improved accuracy of $\hat{s}_1$ and $\hat{s}_3$ is more sensitive to the increase of $T$ as compared to that of $\hat{s}_2$.
This is because the unknown change point's arrival rate $s_j$ satisfies
$s_{j-1}<s_j<s_{j+1}$ for $j=1,\ldots,m$.
For $s_1$ and $s_m$, one of their boundaries is known (which are $s_0=0$ and $s_{m+1}=1$, respectively),
whilst for the intermediate rates both the upper and lower bounds $s_{j-1}$ and $s_{j+1}$ are unknown.
Therefore, under the same condition, the uncertainties of the first and the last change points' arrival rates would be lower than that in the middle. \\
\ \\
\noindent For the comparisons between Case 1 and 2, we can see from the selected scenario ($m^0=3$, $T=10$, LSSE method) shown in Figure~\ref{fig:figure2}, along with the results in Tables~\ref{m2-case1}--\ref{m3-case2} that under the same conditions, the lengths of the 95\% empirical C.I. and MSEs in Case 2 are all smaller than those in Case 1. However, since the simulated processes in the 2 cases are generated from different SDEs, it may be inappropriate to make conclusion based only on the results shown in the provided tables and figures. In fact, note that the main difference between (\ref{sim1-1}) and (\ref{sim2-1}) is the number of coefficients in the models. Therefore, the comparison between the 2 models when fitting them to the same process can be considered as a model selection problem that has been well studied in the literature.

\begin{table}[!htbp]
\small \caption{Simulation results of $\hat{s}_j$, $j=1,~\ldots,~m^0$ for case 1,~(\ref{sim1-1}) with $m^0=2$, {\color{black}true values of parameters as in Table~\ref{cp-setup}}} \centering
\scriptsize
\begin{tabular}{|c|c|c|c|c|c|c|c|c|c|c|}
\hline
 &\multicolumn{3}{c|}{ $T$=5}& \multicolumn{3}{c|}{ $T$=10}&\multicolumn{3}{c|}{ $T$=20}\\
\hline
$\hat{s}_j$& Mean & $95\%$ C.I. & MSE & Mean & $95\%$ C.I.& MSE & Mean & $95\%$ C.I. & MSE \\
\hline
$\hat{s}_1$, (LSSE)  & 0.348 & (0.313, 0.371) & 1.75$\times 10^{-4}$ &0.349 & (0.333, 0.363)& 9.35$\times 10^{-5}$ &0.350 &(0.341, 0.356) &3.47$\times 10^{-5}$ \\
\hline
$\hat{s}_1$, (MLL) &0.348 &  (0.313, 0.371) & 1.76$\times 10^{-4}$ &0.349 & (0.333, 0.363)& 9.35$\times 10^{-5}$ &0.350&(0.341, 0.356) &3.75$\times 10^{-5}$\\
\hline
$\hat{s}_2$, (LSSE) & 0.701 &(0.638, 0.742)& 4.61$\times 10^{-4}$ &0.702 & (0.676, 0.736) & 1.76$\times 10^{-4}$ &0.700 &(0.682, 0.716)&5.62$\times 10^{-5}$\\
\hline
$\hat{s}_2$, (MLL)& 0.701 & (0.638,0.742) & 4.61$\times 10^{-4}$  &0.702 &(0.676, 0.736)& 1.76$\times 10^{-4}$ &0.700 &(0.682, 0.716)  &5.19$\times 10^{-5}$\\
\hline
\end{tabular}
\label{m2-case1}
\end{table}
\begin{table}[!htbp]
\small \caption{Simulation results of $\hat{s}_j$, $j=1,~\ldots,~m^0$ for case 2,~(\ref{sim2-1}) with $m^0=2$, {\color{black}true values of parameters as in Table~\ref{cp-setup}}} \centering
\scriptsize
\begin{tabular}{|c|c|c|c|c|c|c|c|c|c|c|}
\hline
 &\multicolumn{3}{c|}{ $T$=5}& \multicolumn{3}{c|}{ $T$=10}&\multicolumn{3}{c|}{ $T$=20}\\
\hline
$\hat{s}_j$& Mean & $95\%$ C.I. & MSE & Mean & $95\%$ C.I.& MSE & Mean & $95\%$ C.I. & MSE \\
\hline
$\hat{s}_1$, (LSSE) & 0.349 &(0.333, 0.362)& 6.97$\times 10^{-5}$ &0.350 &(0.341, 0.358) & 1.51$\times 10^{-5}$ &0.350 &(0.344, 0.355)&7.63$\times 10^{-6}$\\
\hline
$\hat{s}_1$, (MLL) &0.349 &(0.332, 0.362)& 0.70$\times 10^{-5}$ &0.350 &(0.342, 0.358) & 1.53$\times 10^{-5}$ &0.350 &(0.344, 0.355)&7.60$\times 10^{-6}$\\
\hline
$\hat{s}_2$, (LSSE) & 0.701 &(0.667, 0.734) & 1.97$\times 10^{-4}$ &0.700 &(0.684, 0.718)& 6.38$\times 10^{-5}$ &0.70 &(0.690, 0.707)&1.76$\times 10^{-5}$\\
\hline
$\hat{s}_2$, (MLL)& 0.701&(0.667, 0.733) & 2.01$\times 10^{-4}$ &0.700 &(0.684,0.718)& 6.40$\times 10^{-5}$ &0.70 &(0.690, 0.707) &1.80$\times 10^{-5}$\\
\hline
\end{tabular}
\label{m2-case2}
\end{table}
\begin{table}[!htbp]
\small \caption{Simulation results of $\hat{s}_j$, $j=1,~\ldots,~m^0$ for case 1,~(\ref{sim1-1}) with $m^0=3$, {\color{black}true values of parameters as in Table~\ref{cp-setup}}} \centering
\scriptsize
\begin{tabular}{|c|c|c|c|c|c|c|c|c|c|}
\hline
 &\multicolumn{3}{c|}{ $T$=5}& \multicolumn{3}{c|}{ $T$=10}&\multicolumn{3}{c|}{ $T$=20}\\
\hline
$\hat{s}_j$& Mean & $95\%$ C.I. & MSE & Mean &$95\%$ C.I. & MSE &Mean &$95\%$ C.I.  & MSE \\
\hline
$\hat{s}_1$, (LSSE)  & 0.248 & (0.217, 0.263) & 1.23 $\times 10^{-4}$ &0.249 &(0.231, 0.258) & 3.95$\times 10^{-5}$ &0.250 & (0.242, 0.256) &2.66 $\times 10^{-5}$ \\
\hline
$\hat{s}_1$, (MLL)  &0.248 &(0.217, 0.264) & 1.32 $\times 10^{-4}$ &0.249 &(0.231, 0.258) & 4.14$\times 10^{-5}$ &0.250 &(0.242, 0.256) &2.18$\times 10^{-5}$ \\
\hline
$\hat{s}_2$, (LSSE) & 0.502 &  (0.468, 0.549) & 3.38$\times 10^{-4}$ &0.502 & (0.477, 0.532) & 1.58$\times 10^{-4}$ &0.501 &(0.485, 0.522) &6.44$\times 10^{-5}$ \\
\hline
$\hat{s}_2$, (MLL)&0.502 & (0.468, 0.549)  & 3.39$\times 10^{-4}$ &0.502&  (0.477, 0.532)& 1.58$\times 10^{-4}$  &0.501 & (0.486, 0.522)& 6.23$\times 10^{-5}$ \\
\hline
$\hat{s}_3$, (LSSE)& 0.752 & (0.725, 0.785) &1.74 $\times 10^{-4}$&0.751 & (0.740, 0.770)& 4.78$\times 10^{-5}$  &0.750 & (0.746, 0.755) & 5.65$\times 10^{-6}$ \\
\hline
$\hat{s}_3$, (MLL) &0.752 &  (0.725, 0.785)  & 1.73$\times 10^{-4}$ &0.751 & (0.740, 0.770) & 4.78$\times 10^{-5}$  &0.750 &(0.746, 0.755)  &5.65$\times 10^{-6}$ \\
\hline
\end{tabular}
\label{m3-case1}
\end{table}
\begin{table}[!htbp]
\small \caption{Simulation results of $\hat{s}_j$, $j=1,~\ldots,~m^0$ for case 2,~(\ref{sim2-1}) with $m^0=3$, {\color{black}true values of parameters as in Table~\ref{cp-setup}}} \centering
\scriptsize
\begin{tabular}{|c|c|c|c|c|c|c|c|c|c|c|}
\hline
 &\multicolumn{3}{c|}{ $T$=5}& \multicolumn{3}{c|}{ $T$=10}&\multicolumn{3}{c|}{ $T$=20}\\
\hline
$\hat{s}_j$ & Mean & $95\%$ C.I. & MSE & Mean & $95\%$ C.I. & MSE &Mean& $95\%$ C.I. & MSE \\
\hline
$\hat{s}_1$, (LSSE) & 0.248 & (0.231, 0.262) & 6.48 $\times 10^{-5}$ &0.250& (0.240, 0.259) &1.93$\times 10^{-5}$ &0.250&(0.244, 0.253)  &5.29$\times 10^{-6}$ \\
\hline
$\hat{s}_1$, (MLL) &0.248 & (0.233, 0.263) & 6.37$\times 10^{-5}$ &0.249 & (0.242, 0.259)  &1.83$\times 10^{-5}$ &0.250& (0.244, 0.254) &5.39$\times 10^{-5}$ \\
\hline
$\hat{s}_2$, (LSSE) & 0.502 &(0.475, 0.539) & 1.94$\times 10^{-4}$ &0.500 &(0.483, 0.515) &5.09$\times 10^{-5}$ &0.499 &(0.489, 0.506)&1.72$\times 10^{-5}$ \\
\hline
$\hat{s}_2$, (MLL)&0.502&(0.475, 0.539) &1.94$\times 10^{-4}$ &0.500&(0.484, 0.517) &5.03$\times 10^{-5}$  &0.500 &(0.489, 0.507) &1.73$\times 10^{-5}$ \\
\hline
$\hat{s}_3$, (LSSE) & 0.750 &(0.727, 0.774)&1.20 $\times 10^{-4}$&0.750 & (0.743, 0.760) & 1.39$\times 10^{-5}$  &0.750 &(0.746, 0.753) &3.20$\times 10^{-6}$ \\
\hline
$\hat{s}_3$, (MLL) &0.751& (0.727, 0.774)&1.17$\times 10^{-4}$ &0.750 & (0.743, 0.760) &1.42$\times 10^{-5}$  &0.750 &(0.746, 0.754) &3.07$\times 10^{-6}$ \\
\hline
\end{tabular}
\label{m3-case2}
\end{table}

\begin{figure}[htbp]
\includegraphics[height=1.6in,width=2.04in]{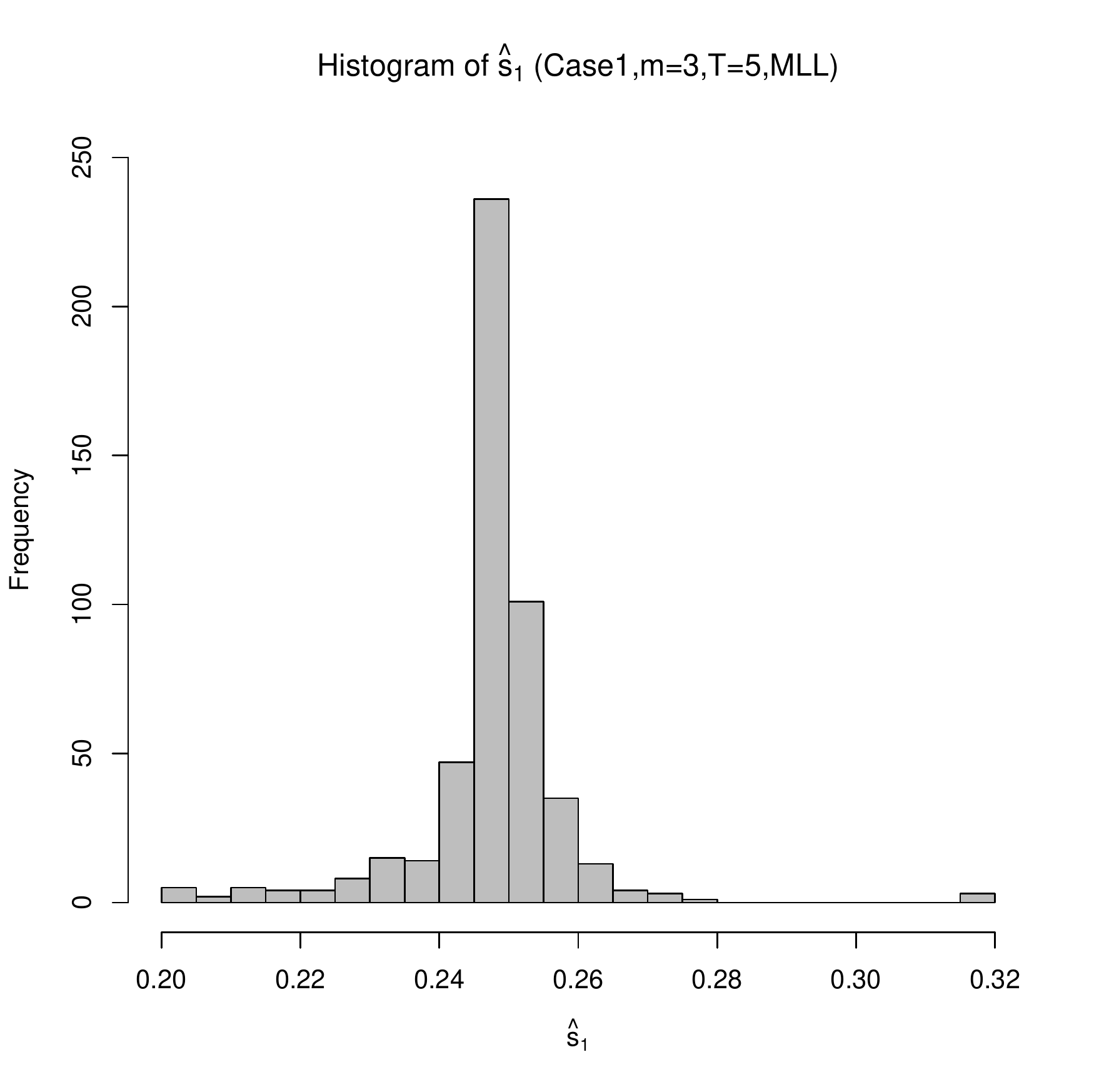}
\includegraphics[height=1.6in,width=2.04in]{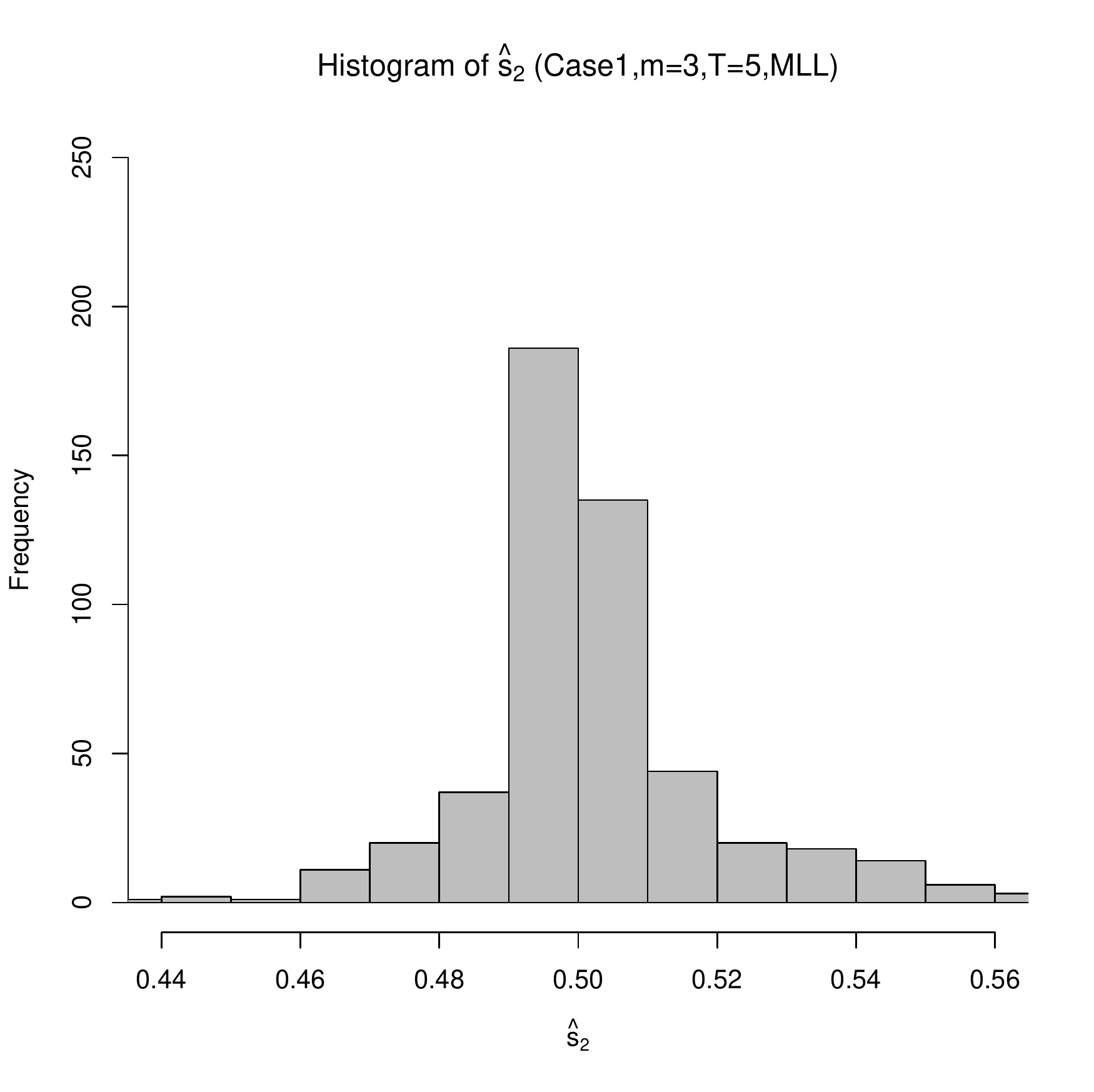}
\includegraphics[height=1.6in,width=2.04in]{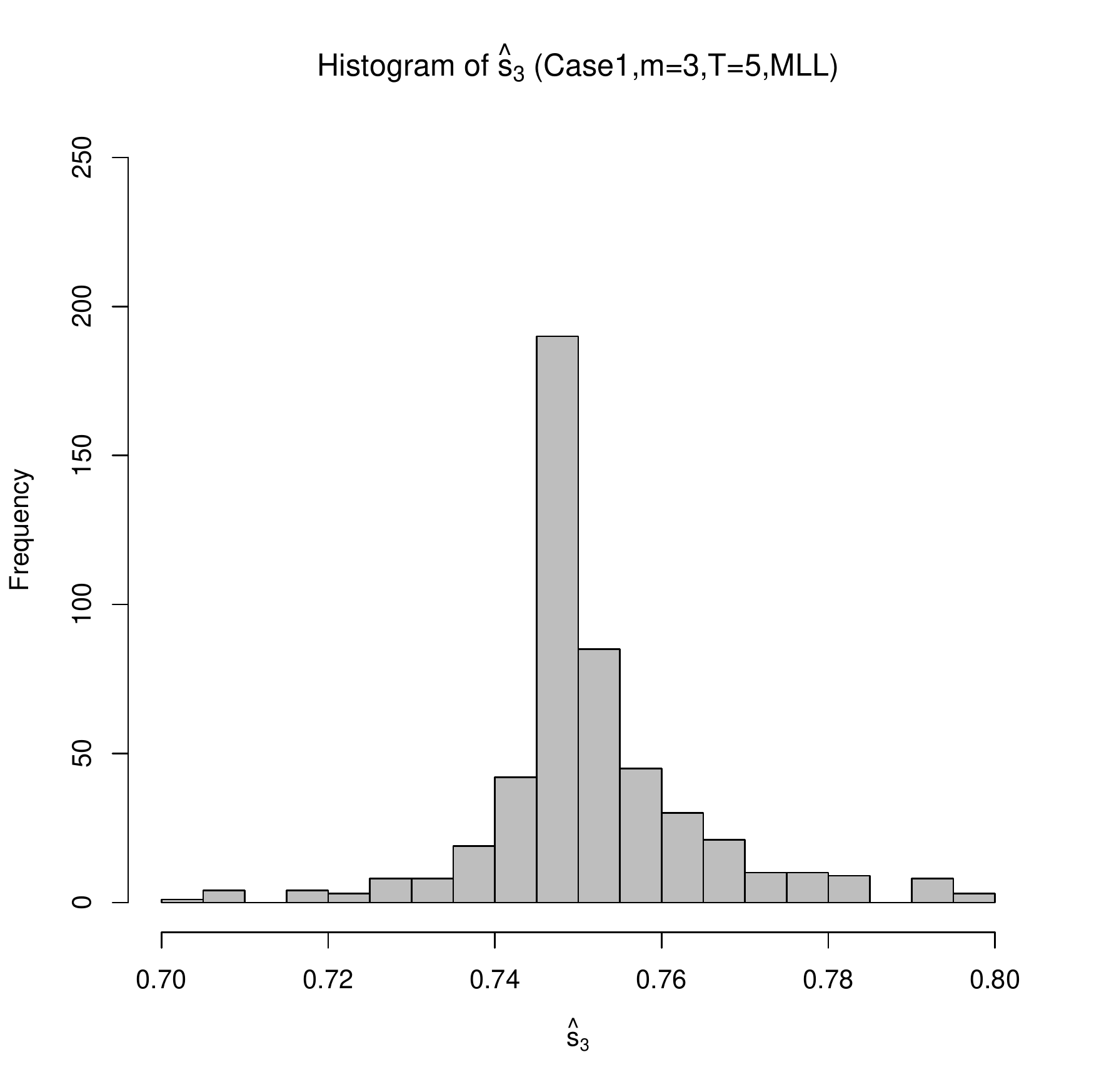}\\
\includegraphics[height=1.6in,width=2.04in]{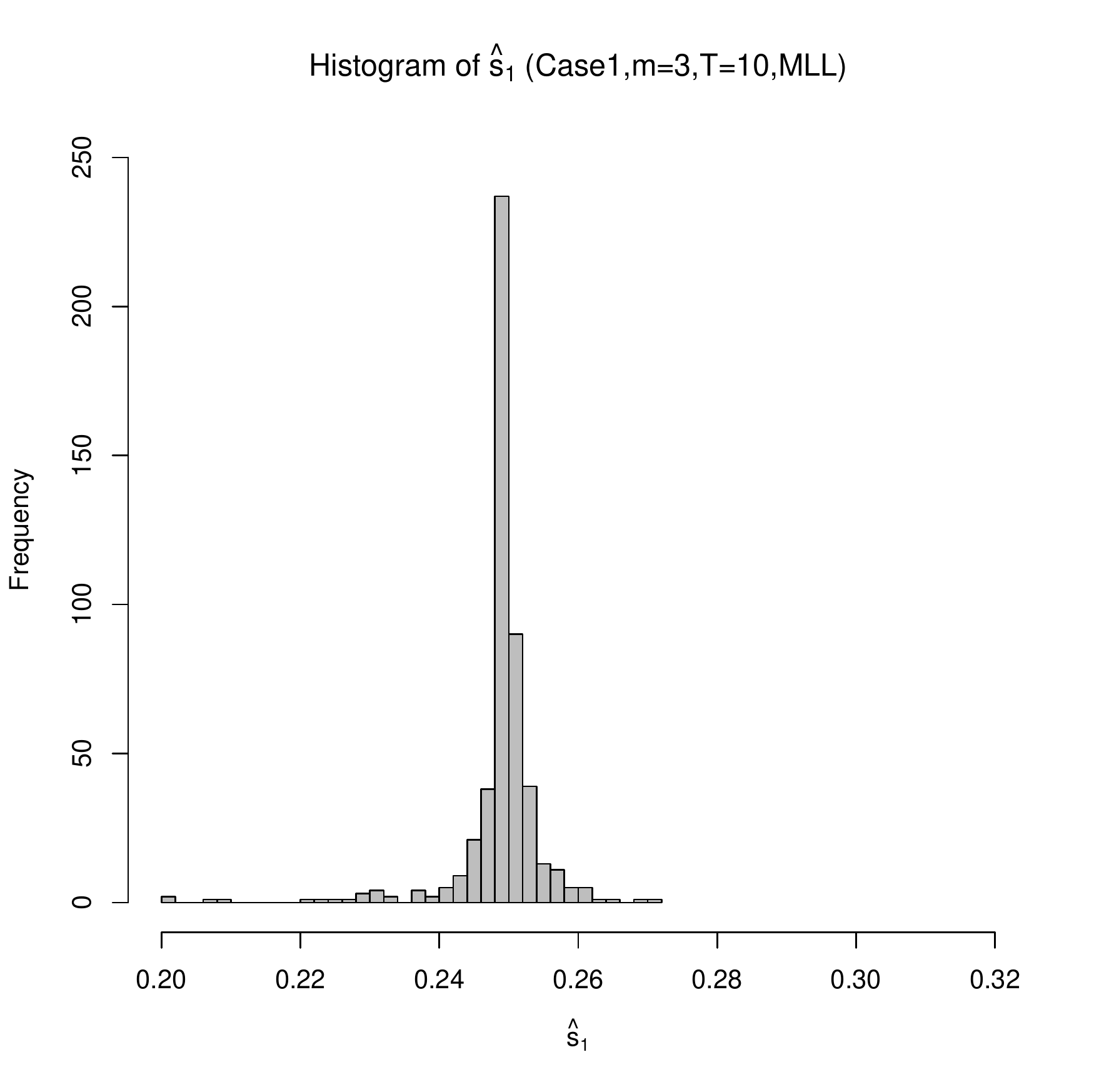}
\includegraphics[height=1.6in,width=2.04in]{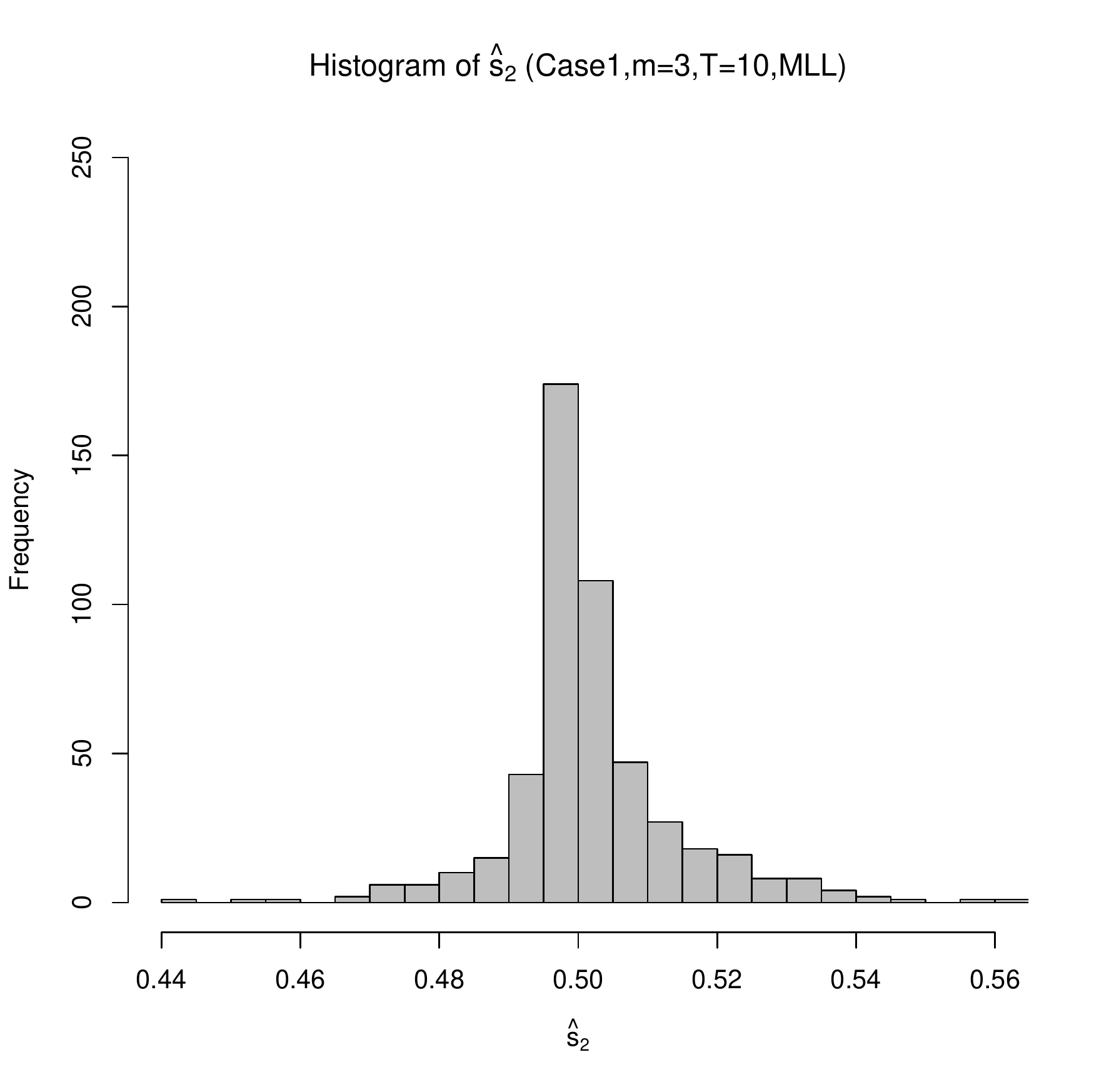}
\includegraphics[height=1.6in,width=2.04in]{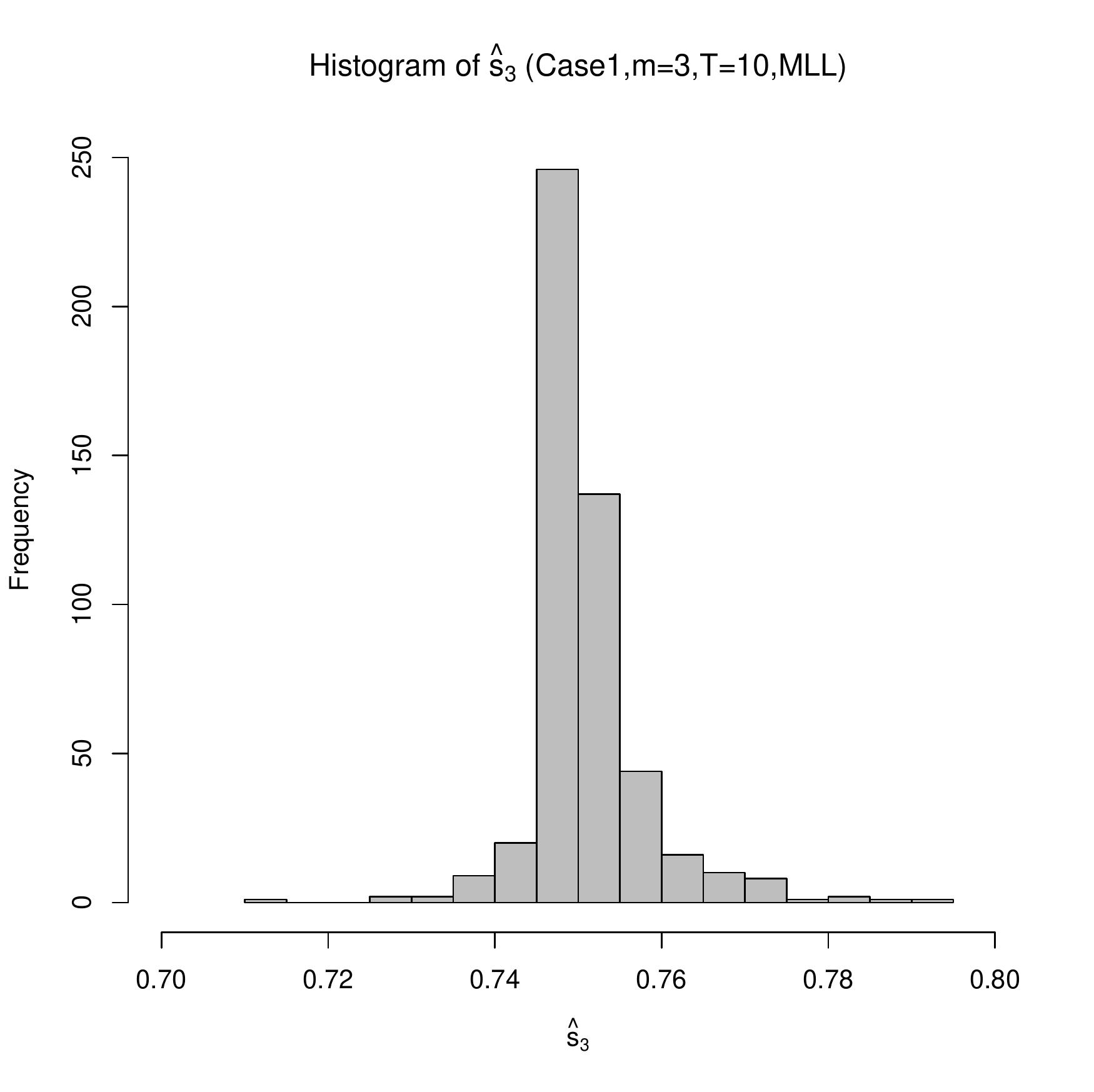}\\
\includegraphics[height=1.6in,width=2.04in]{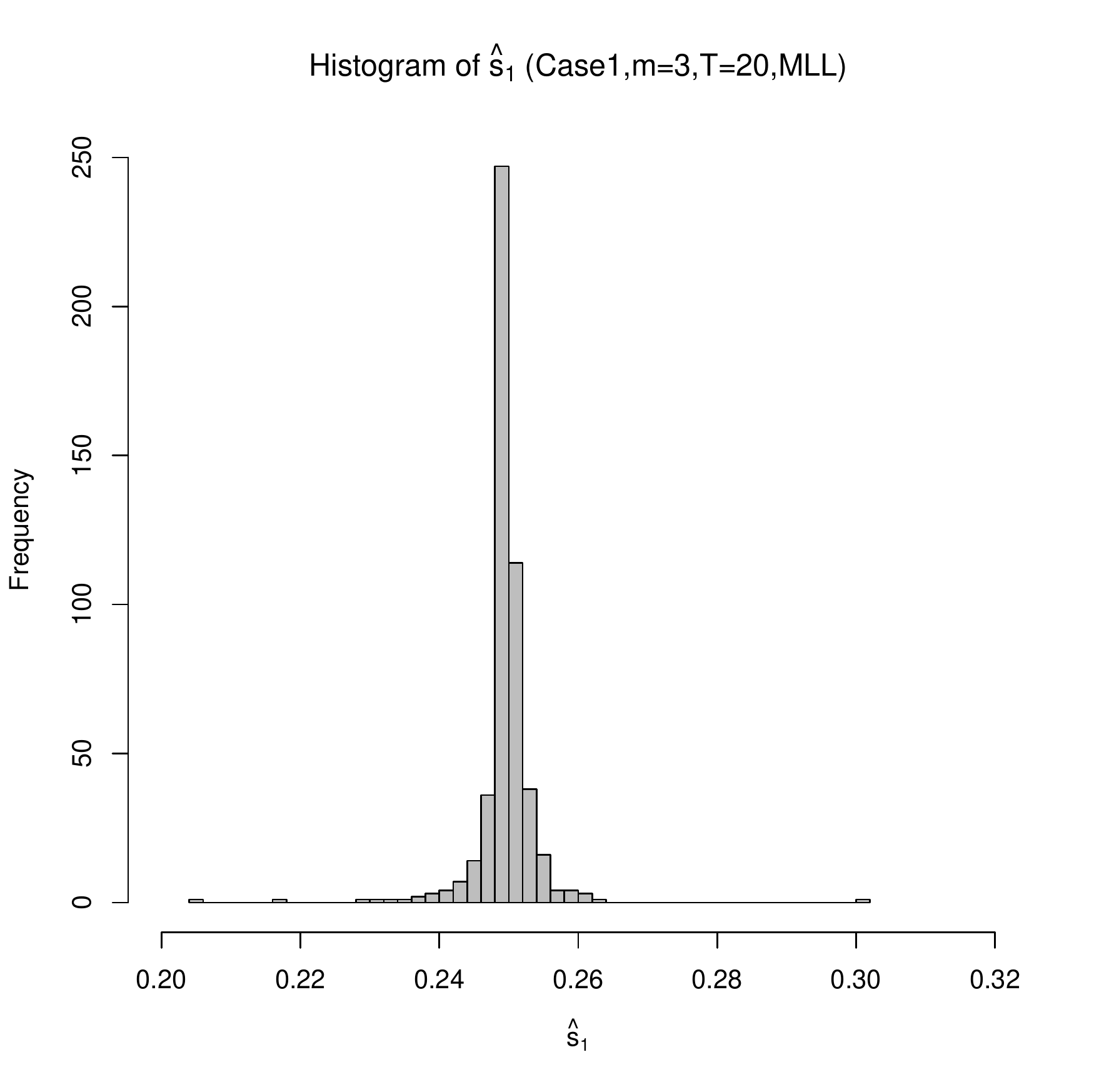}
\includegraphics[height=1.6in,width=2.04in]{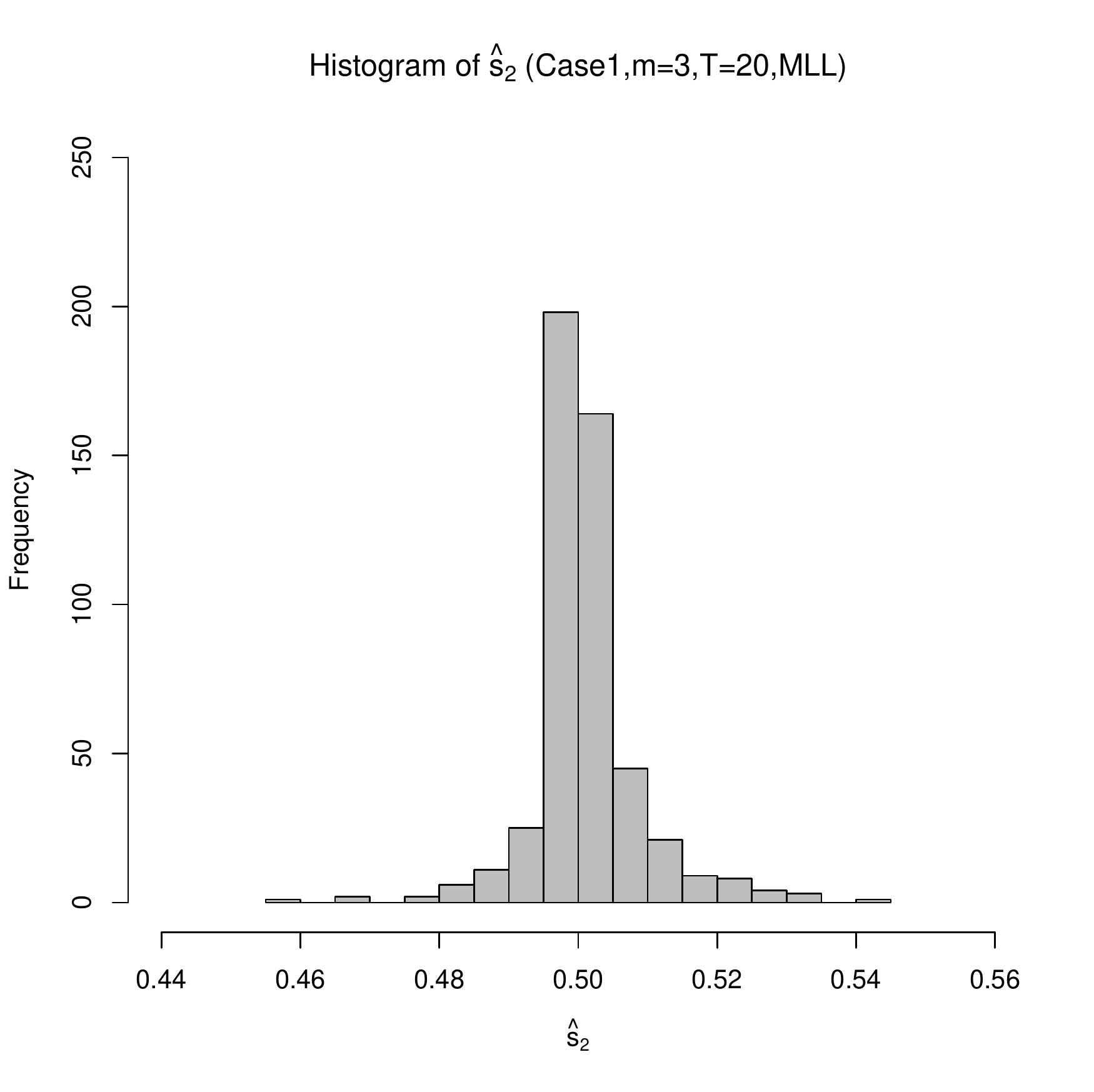}
\includegraphics[height=1.6in,width=2.04in]{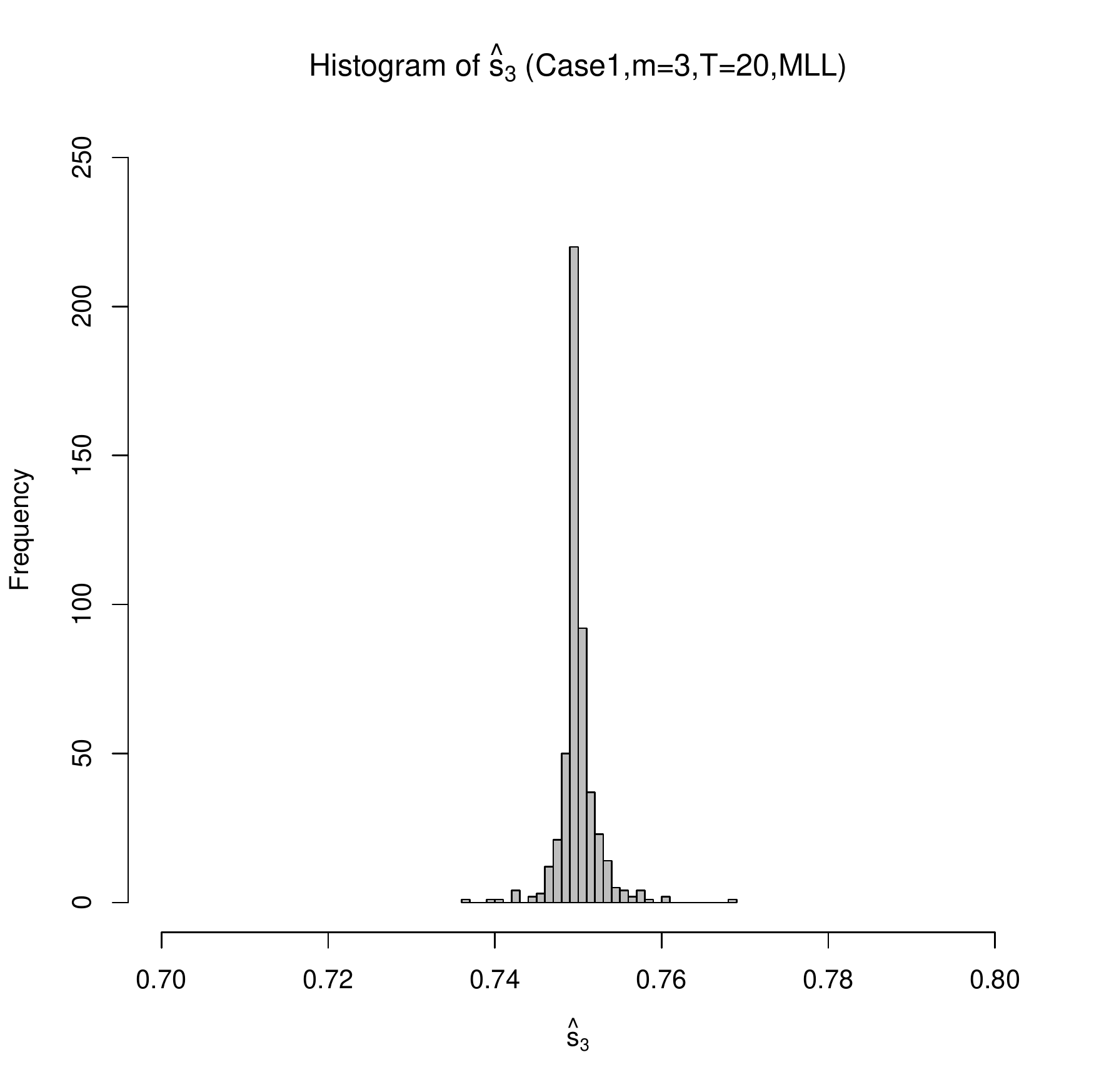}\\

\caption{\small Histogram of $\hat{s}$ based on MLL method for Case 1 with $m^0=3$, $T=(5, 10, 20)$ and exact value $s^0=(0.25, 0.50, 0.75)$}
\label{fig:figure1}
\end{figure}

\begin{figure}[htbp]
\includegraphics[height=1.6in,width=2.04in]{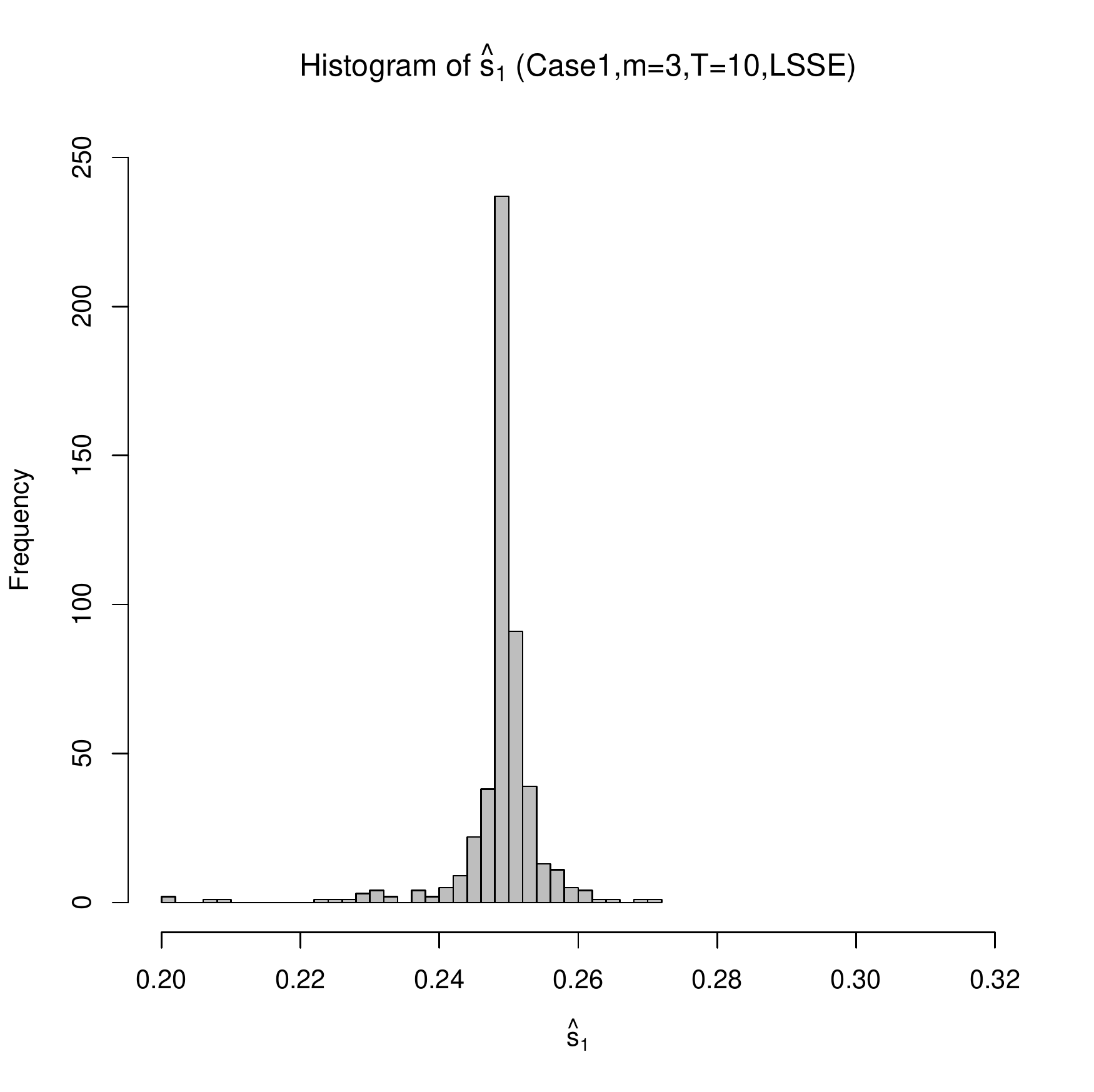}
\includegraphics[height=1.6in,width=2.04in]{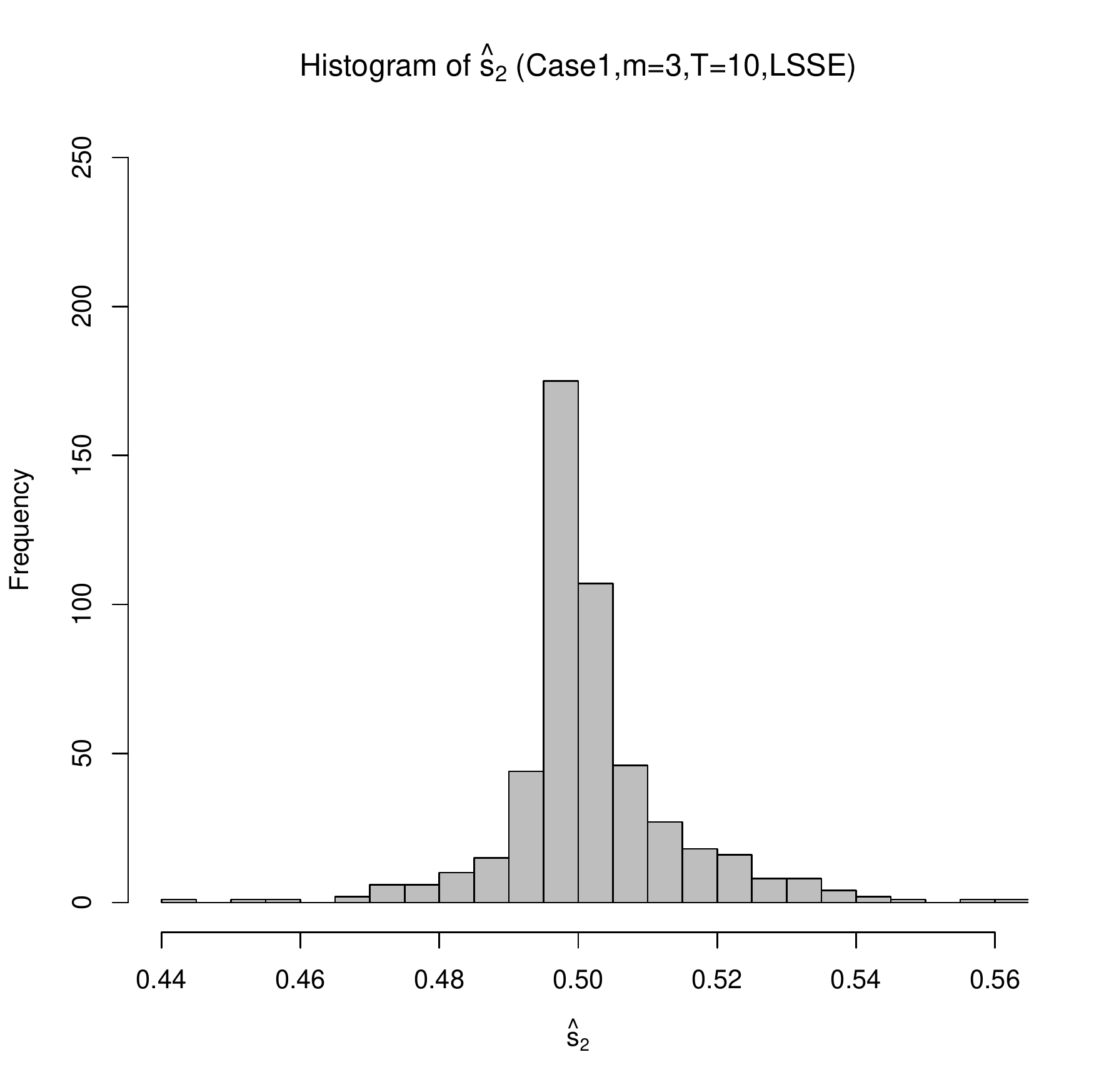}
\includegraphics[height=1.6in,width=2.04in]{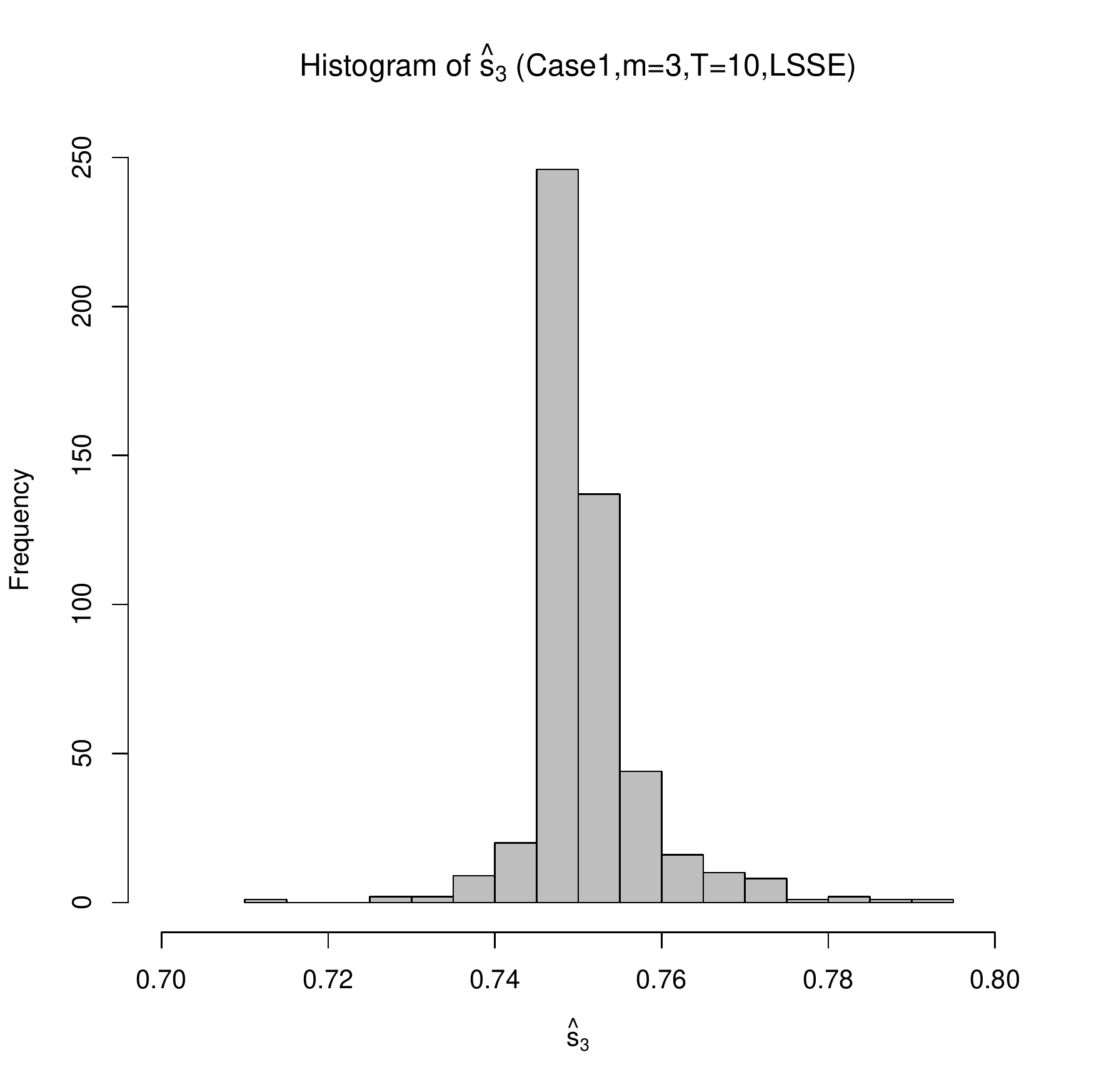}\\
\includegraphics[height=1.6in,width=2.04in]{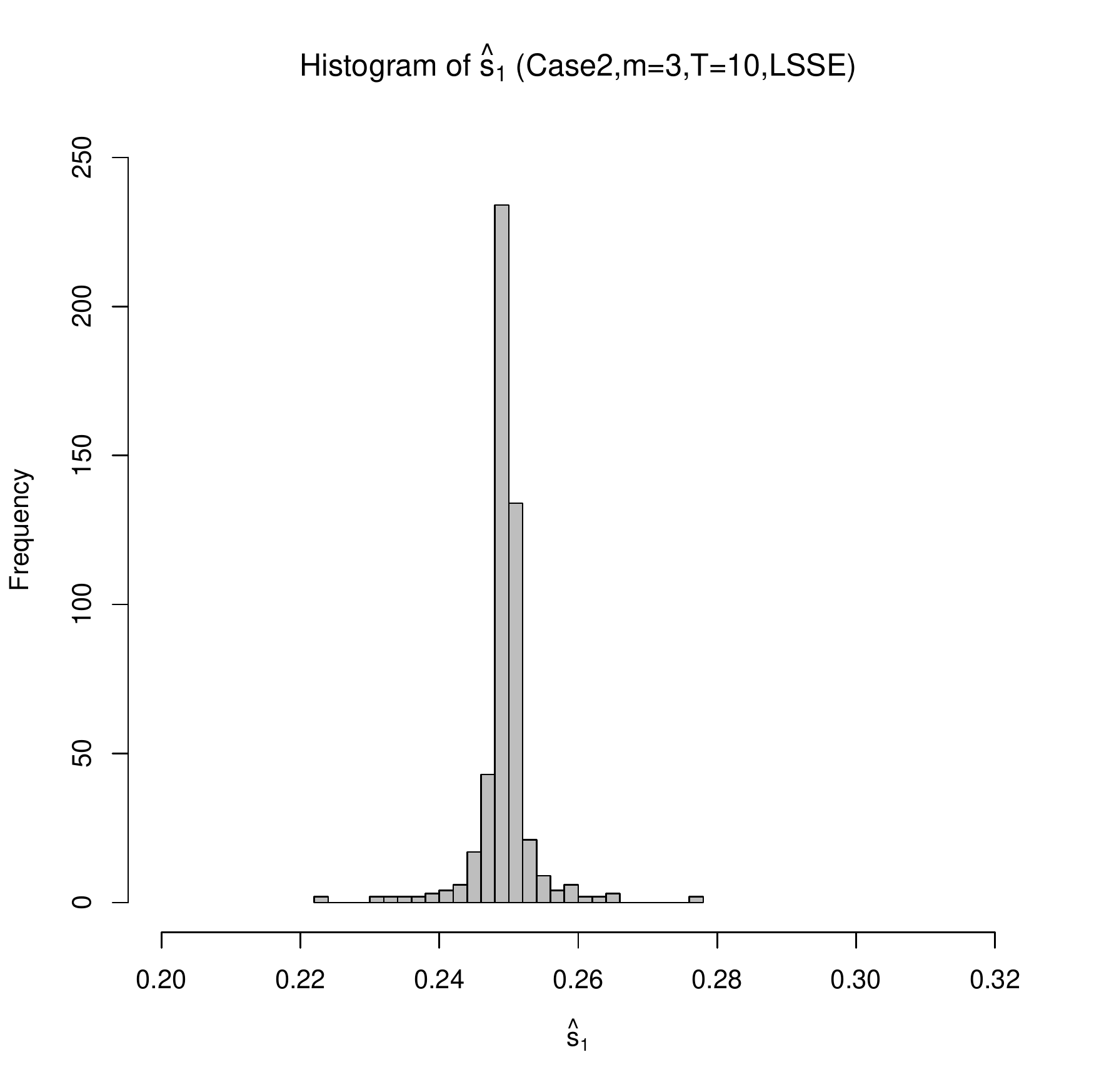}
\includegraphics[height=1.6in,width=2.04in]{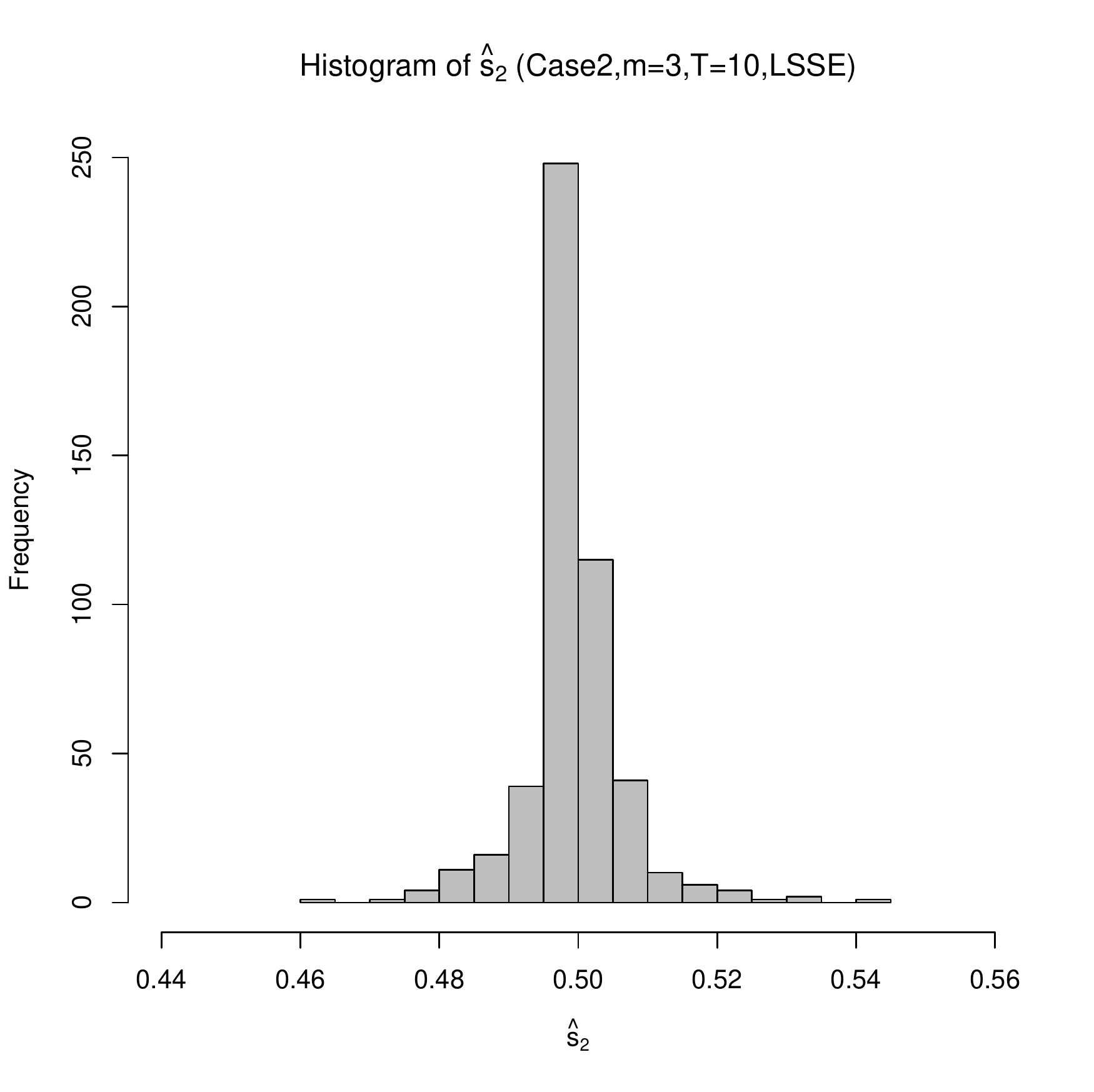}
\includegraphics[height=1.6in,width=2.04in]{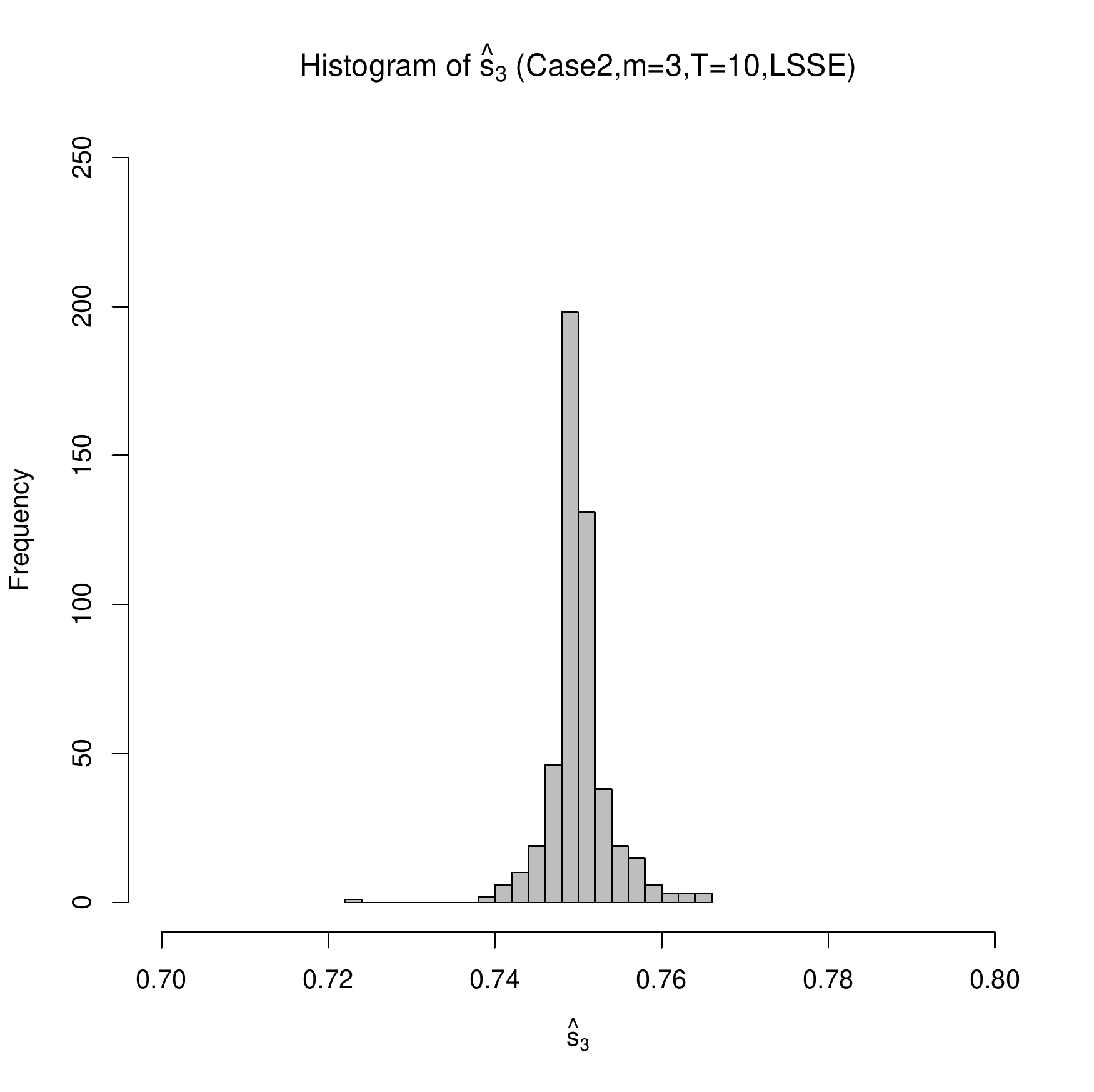}\\
\caption{\small Histogram of $\hat{s}$ based on LSSE method for Case 1 and 2, when $m^0=3$, $T=10$ and exact value $s^0=(0.25, 0.50, 0.75)$}
\label{fig:figure2}
\end{figure}

\subsubsection{Estimating the number of change points}\label{simulation-2}
\noindent In this subsection, we study the performance of (\ref{ic}) in estimating the unknown number of the change points based on Algorithms 2 (SNS) and 3 (Modified PELT). For the simulation {\color{black}setup}, we assume the exact value of  $m^0=2$, with different time periods $T=5, 10, 15, 20$. {\color{black}The pre-assigned coefficients are provided in Table~\ref{mcp-setup}}.\\

\noindent Based on the simulated process, we apply Algorithms 2 and 3 with $m$ ranging from $0$ to $5$ to estimate the unknown number of change points. In Tables~\ref{m-3}, we count and report the cumulative frequency (CF) of 500 iterations that return the correct estimates $\left( \sum_{i=1}^{500} \mathbbm{1}(\hat{m}_i=m^0)\right)$ and the relative frequency (RF) $\left(\frac{1}{500}\sum_{i=1}^{500} \mathbbm{1}(\hat{m}_i=m^0)\times 100 \% \right)$.\\
\begin{table}[!htbp]
\small \caption{Pre-assigned coefficients (with unknown number of change points)} \centering
\scriptsize
\begin{tabular}{|c|c|c|c|c|c|c|c|c|}
\hline
\small Case& Coefficient & $j=1$ & $j=2$ & $j=3$   \\
\hline
\multirow{2}{*}{1} & $\mu^{(j)}$ & 0.08 & 2.50 & 0.08 \\
\cline{2-5}
 & $\alpha^{(j)}$ & 0.10 & 1.00 & 0.50 \\
 \hline
\multirow{3}{*}{2} & $\mu_1^{(j)}$ & 0.08 & 2.50 & 0.08\\
\cline{2-5}
 & $\mu_2^{(j)}$ & 0.02 & 1.20 & 0.02 \\
\cline{2-5}
 & $\alpha^{(j)}$ & 0.10 & 1.00 & 0.50 \\
 \hline
\end{tabular}
\label{mcp-setup}
\end{table}

\begin{table}[!htbp]
\small \caption{Cumulative frequency and relative frequency in 500 iterations that return the correct estimates} \centering
\scriptsize
\begin{tabular}{|c|c|c|c|c|c|c|c|c|c|}
\hline
& &\multicolumn{2}{c|}{ $T$=5 }&\multicolumn{2}{c|}{$T$=10 } &\multicolumn{2}{c|}{$T$=15 } &\multicolumn{2}{c|}{ $T$=20 }\\
\hline
Case& Algorithm  &  CF& RF &CF&RF &CF&RF & CF&RF\\
\hline
1&2 (SNS) & 492 &$98.4 \%$ &498 & $99.7\%$ &500 &$100.0 \%$ &500 & $100.0 \%$\\
\hline
2&2 (SNS) &500 & $100.0 \%$ &500 &$100.0 \%$ &500 &$100.0 \%$ &500 & $100.0 \%$\\
\hline
1&3 (PELT)&494 &98.8\% &499 &99.8\% &500 & 100.0\%&500 & 100.0\%\\
\hline
2&3 (PELT)&497 &99.4\% &500 &100.0\% &500 &100.0\% &500 &100.0\% \\
\hline
\end{tabular}
\label{m-3}
\end{table}

\noindent For the estimated number $\hat{m}$ of change points,
one could see from Table~\ref{m-3} that, when $m^0=2$, the proposed methods perform very well in both cases
with different time periods. Furthermore, the accuracy of the estimating results in different cases all increase as $T$ increases. These results suggest that our proposed method is asymptotically consistent, which confirms the theoretical finding in Proposition~\ref{asic1}.

\subsection{Implementation on observed financial market data with discussion}\label{realdata}
\noindent We apply the estimation methods to the Brent oil one-month futures settlement daily price data for the period 18 March 1993 to 25 September 2015. The data set is available
at www.quandl.com. \\
\ \\
\noindent The empirical studies of Schwartz (1997) and Chen (2010) showed that mean-reversion features hold for prices of several commodities including oil. Hence, we use OU-types processes to model such price behaviour. We first fit the classical OU process without {\color{black}any} change point, i.e., using the dynamics $dX_t=
(\mu-\alpha X_t)dt+\sigma dW_t$), to the log-transformed data series with $\sigma$'s estimate as the data's realised volatility $\left(\hat{\sigma}=\sqrt{\sum_{t_i\in[0,T]}(X_{t_{i+1}}-X_{t_i})^2/T}\right)$. The MLE of the drift parameters are given by $\hat{\mu}=0.48$ and $\hat{\alpha}=0.12$. Based on these MLE values, the log likelihood (via the Riemann-sum approximation) is about 1.02 and $\mathcal{IC}(m=0)=15.27$.  \\
\ \\
\noindent However, from the plot of the price series in Figure~\ref{fig:figure7}, there are several changes in the shapes of the price evolution. This observed feature suggests that it may be more appropriate to use (\ref{sim1-1}) with unknown ($m>0$) change points.\\
\begin{figure}[htbp]
\includegraphics[height=2.1in,width=6.5in]{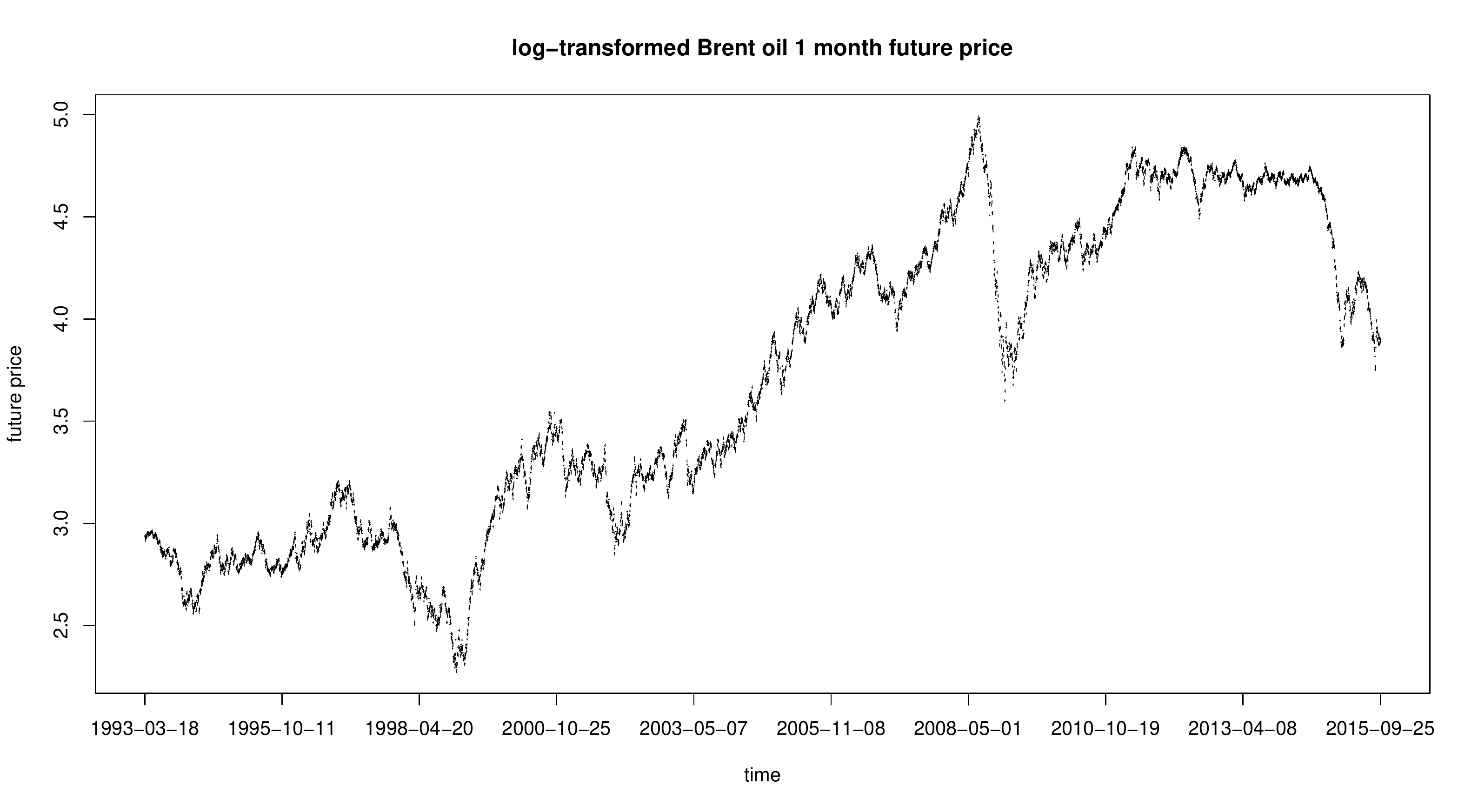}
\caption{\small Log-transformed Brent oil 1 month futures settlement prices: 18 March 1993 -- 25 September 2015}
\label{fig:figure7}
\end{figure}
\ \\
\noindent {\color{black} The data set covers approximately 22.5 years giving a sample size of 5735 trading days. The recorded yearly number of trading days varies from year to year; so, for convenience, we let $\Delta_t=22.5/5735$. Since the size of the data set is large, we apply Algorithm 3 using a minimum permissible regime-time length $h$ of $0.25/\Delta_t\approx 63$ trading days (quarterly)}. The algorithm detects $\hat{m}=2$ change points, which occurs on 24 September 2008 and 23 December 2008, with a corresponding log likelihood increase (via Riemann-sum approximation) of 26.81 and $\mathcal{IC}(m=2)=-1.69$ lower than $\mathcal{IC}(m=0)$. To confirm the results, we also apply Algorithm 2 with LSSE and MLL methods respectively and $m_{\max}$ set to be 10. The results are the same as that obtained from Algorithm 3. \\
\ \\
\noindent The plot of the price series, with change points indicated, is depicted in Figure~\ref{fig:figure8}. It shows that, from September 2008 to March 2009, there is a decreasing trend in the log-transformed futures prices and then the trend is slightly increasing after this period. Further, based on the estimated change points, the MLE of the drift parameters, {\color{black}and the associated statistics such as (long-term) means (i.e., $\hat{\mu}^{(j)}/\hat{\alpha}^{(j)}$) and variances (i.e., $\hat{\sigma}^2/(2\hat{\alpha}^{(j)})$)} are given in Table~\ref{m-4}. From Table~\ref{m-4}, there are huge changes in the MLEs of the drift parameters under different $T$'s. Based on the MLEs, we also plot two simulated series based on the OU process with and without change points; see Figures~\ref{fig:figure9} and
\ref{fig:figure10}, respectively. As most notably expected, the simulated series based on the OU process with two change points is closer to the original series, especially during the period spanning 25 September--23 December 2008, than the simulated series based on the OU process without change point. Judging from these observed characteristics and taking into account the above-mentioned substantial improvements in the log likelihood and SIC values, we conclude that the OU-process with two change points occurring on 25 September 2008 and 23 December 2008 is the appropriate model for the data set that we analysed.
\begin{figure}[htbp]
\includegraphics[height=2.1in,width=6.5in]{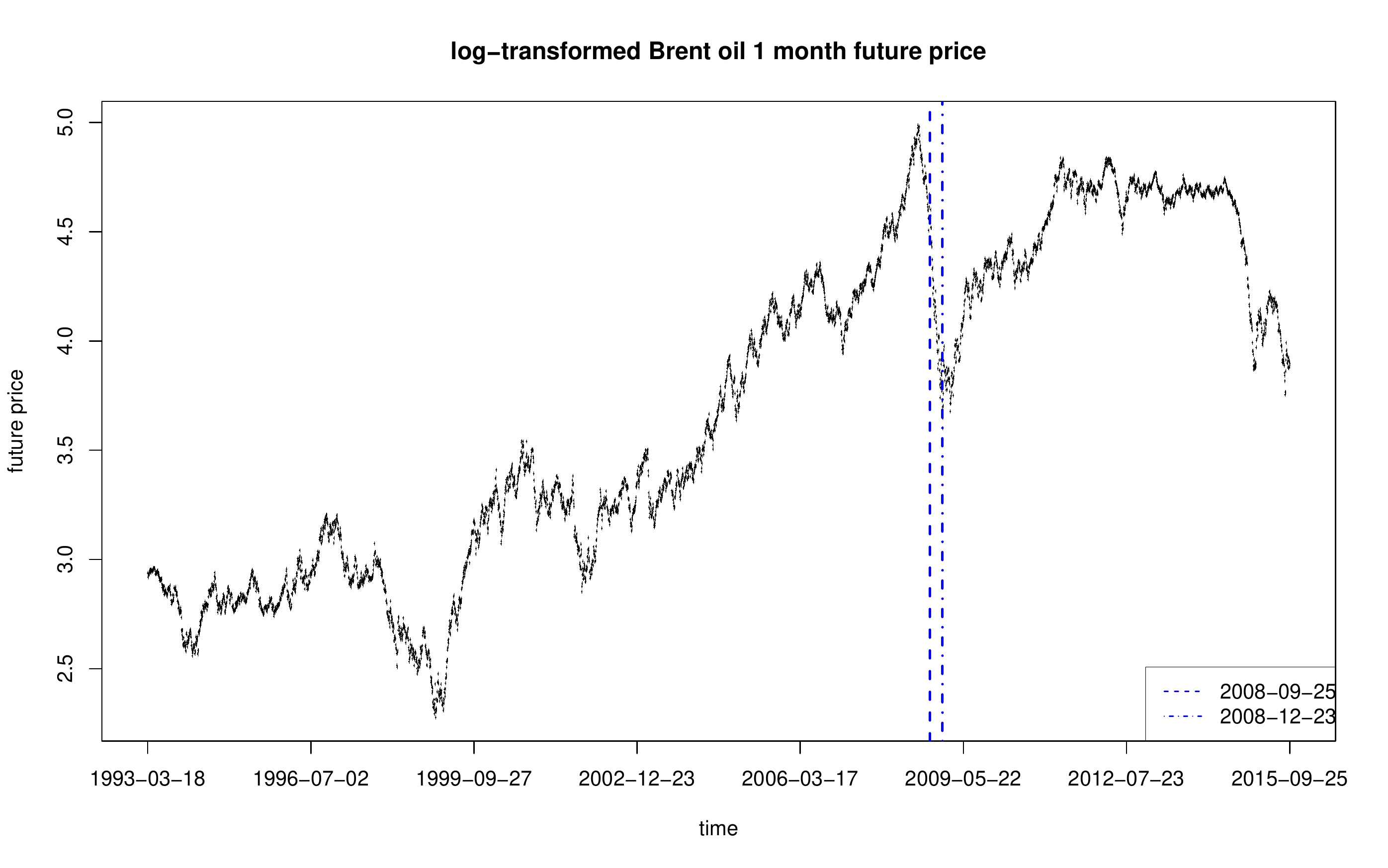}
\caption{\small Highlighting the two detected change points in the log-transformed one-month futures settlement prices on Brent oil: 18 March 1993--25 September 2015}
\label{fig:figure8}
\end{figure}
\begin{figure}[htbp]
\includegraphics[height=2.2in,width=6.5in]{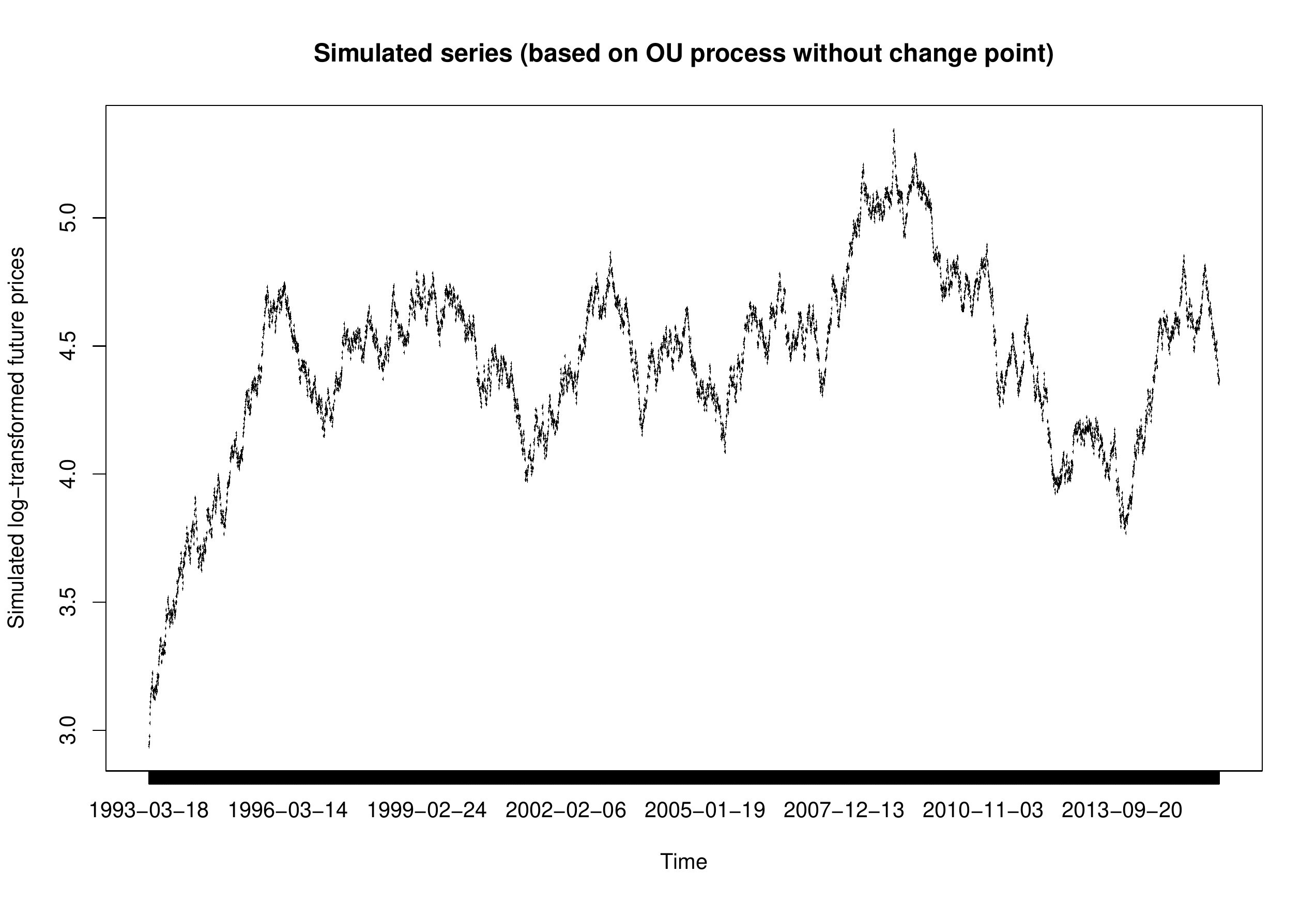}
\caption{\small Simulated log-transformed one-month futures settlement prices on Brent oil based on OU process without change point: 18 March 1993--25 September 2015}
\label{fig:figure9}
\end{figure}
\begin{figure}[htbp]
\includegraphics[height=2.2in,width=6.5in]{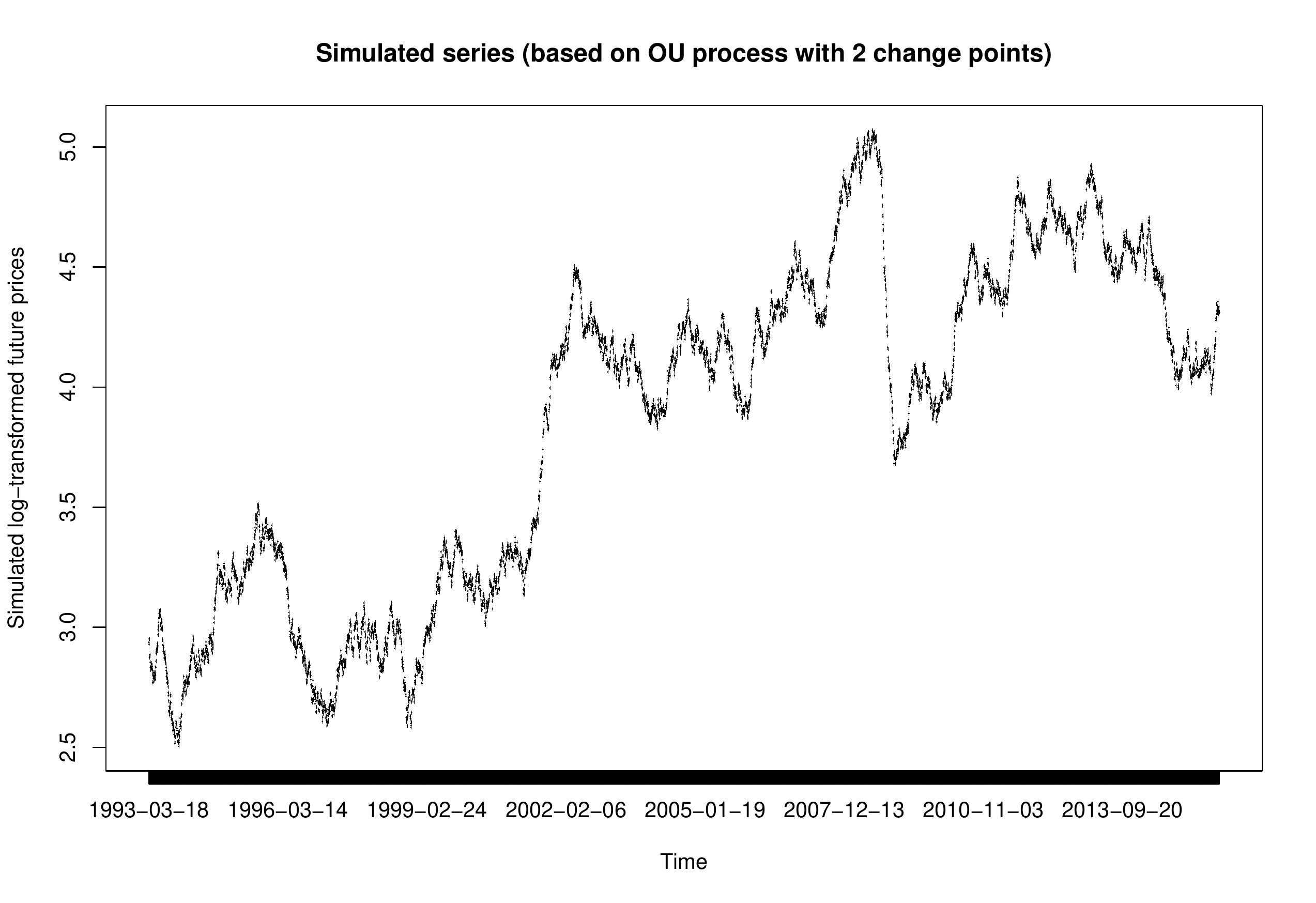}
\caption{\small Simulated log-transformed one-month futures settlement prices on Brent oil based on OU process with two change points: 18 March 1993--25 September 2015}
\label{fig:figure10}
\end{figure}
\begin{table}[!htbp]
\small \caption{MLEs of the drift parameters, (long term) mean, variances and approximate log likelihoods for the OU process with two change points in modelling the Brent oil's one-month futures settlement prices from 18 March 1993 to 25 September 2015} \centering
\begin{tabular}{|c|c|c|c|c|c|c|}
\hline
Time period & $\hat{\mu}$ & $\hat{\alpha}$ &  $\hat{\mu}/\hat{\alpha}$& $\hat{\sigma}$ & $\hat{\sigma}^2/(2\hat{\alpha})$ & $\log\ell$ \\
\hline
18 March 1993 to 25 September 2008 & 0.128 & 0.005 & 25.794& 0.328 &  10.889 & 0.884\\
\hline
26 September 2008 to 23 December 2008 & 5.501 & 2.418 & 2.275& 0.328 & 0.022 & 23.367 \\
\hline
24 December 2008 to 25 September 2015 & 3.977 & 0.879 &4.524 &0.328 & 0.061 & 2.557\\
\hline
\end{tabular}
\label{m-4}
\end{table}

{\color{black}
\noindent Moreover, we see from Figure~\ref{fig:figure8} that from 18 March 1993 to 25 September 2008 (the first estimated change point), there are several noticeable changes in the series. For example, from 1996 to 1998 there was a fall in the futures price. However, based on the SIC, these changes are not significant enough to {\color{black}warrant the inference of a regime change} and they are therefore ignored by the proposed methods. We take a closer look concentrating only on the 18-Mar-1993-to-25-Sep-2008 data set to see if there is any change point at all. This analysis is equivalent to reducing the sample size and the penalty term of the SIC accordingly. Algorithm 3 is re-applied to the reduced data set with a sample size 3928 and maintaining the same $\Delta_t$ and $\sigma$ as those in our other experiments. The resulting SIC indicates still no change point during the shortened period. Furthermore, we employ the estimated parameters and sample size in our reduced data set to run a simulation similar to that in Subsection~\ref{simulation-2} in assessing the performance of the proposed methods, and we obtained an RF of $86.2\%$ in producing the correct estimates.  \\
\ \\
\noindent The above result tells us that an OU process without a change point would be appropriate to model the data series from 18 March 1993 to 25 September 2008. However, from Table~\ref{m-4} the respective long-term mean and variance $\displaystyle \frac{\hat \mu}{\hat \alpha}$ and $\displaystyle \frac{\sigma^2}{2 \hat \alpha}$ are
$25.794$ and $10.889$, which are both higher than those in the two other time periods. Such high statistics may be less preferable in practice, although they reasonably explain the increasing trend in the investigated time period. On the other hand, we know that imposing more change points, which is equivalent to increasing the number of coefficients in the model, can reduce the variance.
Hence, we examine the potential reduction in the variance by imposing a change point into this period. To this end, we fit an OU process with one change point to the series and use Algorithm 1 to estimate the location in the OU process. The estimated change point is at 15 February 1999, which is near the bottom of the series; see  Figure~\ref{fig:figure11}.
Based on Table~\ref{m-5}, by imposing a change point at 15 February 1999, the means $\displaystyle \frac{\hat \mu}{\hat \alpha}$
before and after the change point both strikingly decrease to 2.646 and 4.407, respectively, in comparison to the previous result of $25.794$. Additionally, the variances $\displaystyle \frac{\sigma^2}{2 \hat \alpha}$ for the time periods before and after the change point also markedly go down to respective values of 0.085 and 0.163.  The approximate log likelihood increases correspondingly to 4.199 for the period 16 February 1999--25 September 2008.
This implies that imposing a change point (15 February 1999) into the model does improve the accuracy. Nevertheless, recall that our SIC-based method shows no change point at this time period. Therefore, we demonstrated
a strong potential for an over-fitting problem to arise when a change-point assumption is unnecessarily introduced into the model. This also reminds us of a trade-off between accuracy improvement and issue of over fitting that must be
avoided whenever possible.}   \ \\
\ \\
\begin{table}[!htbp]
\small \caption{MLEs of the drift parameters, (long-term) mean, variances and approximate log likelihoods for for the OU process with one change point in modeling Brent oil's one-month futures settlement prices from 18 March, 1993 to 25 September, 2008} \centering
\begin{tabular}{|c|c|c|c|c|c|c|}
\hline
Time period & $\hat{\mu}$ & $\hat{\alpha}$&  $\hat{\mu}/\hat{\alpha}$ & $\hat{\sigma}$ & $\hat{\sigma}^2/(2\hat{\alpha})$ & $\log\ell$ \\
\hline
18 March 1993 to 15 February 1999 & 1.688 & 0.638 & 2.646 & 0.328 &  0.085 & 0.762\\
\hline
16 February 1999 to 25 September 2008 & 1.450 & 0.329 & 4.407 & 0.328 & 0.163  & 4.199\\
\hline
\end{tabular}
\label{m-5}
\end{table}
\begin{figure}[htbp]
\includegraphics[height=2.1in,width=6.5in]{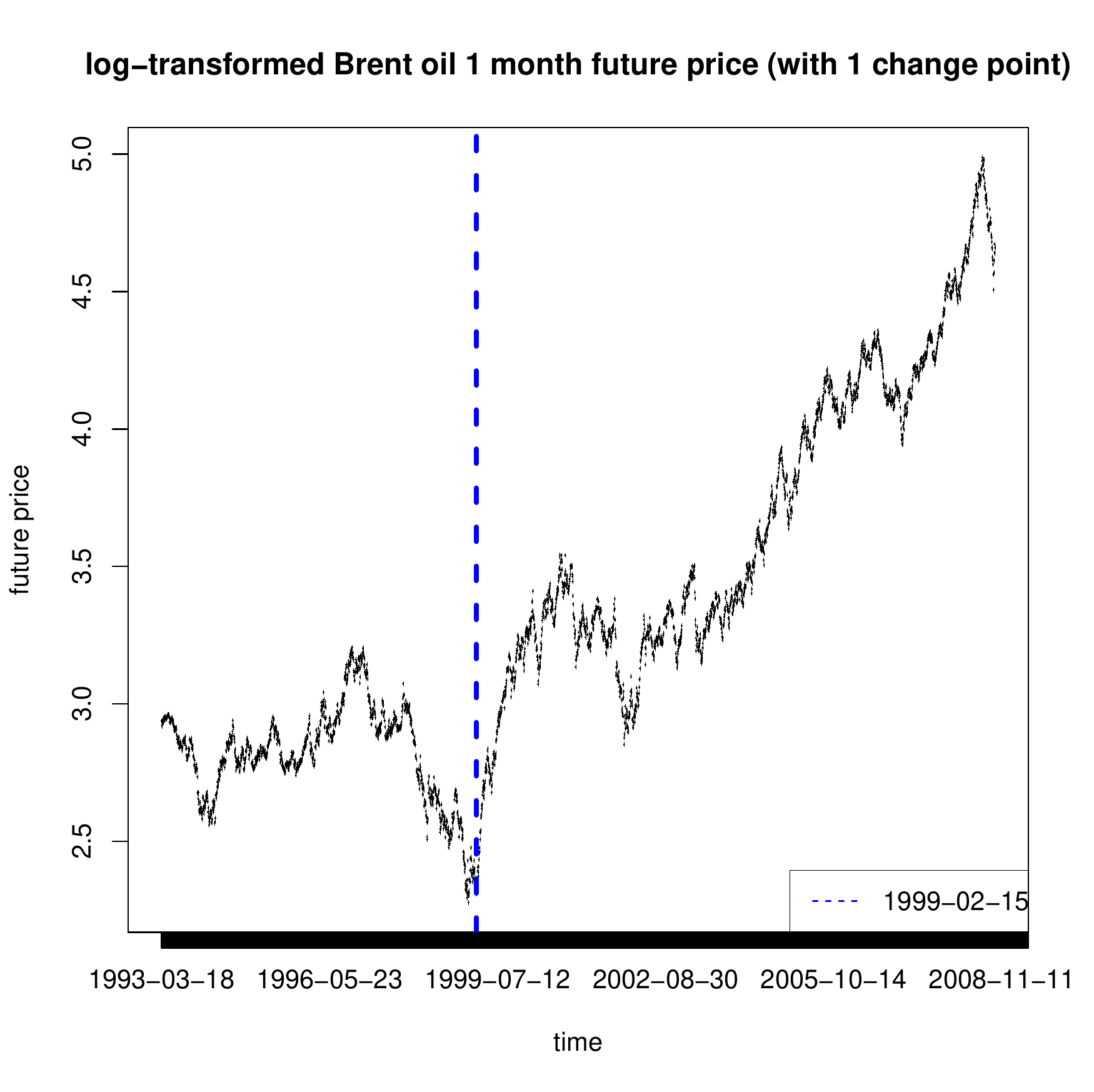}
\caption{\small log-transformed one-month futures settlement prices on Brent oil based on OU process with one change point: 18 March 1993--25 September 2008}
\label{fig:figure11}
\end{figure}

\section{Conclusion}
\noindent The main contribution of this paper is
the development of MLL- and LSSE-based methods
in detecting the unknown number of multiple change points along with the
identification of their locations in a generalised univariate
OU process. Additionally, we showed that our proposed
estimators for the change points' locations satisfy the
asymptotically consistent and normality properties under
certain suitable conditions that we painstakingly imposed. These results
guided the design of three computing algorithms customised for the efficient
implementation of our proposed methods. The numerical applications we showcased
covering both simulated and observed data illustrated the excellent performance
and accuracy of the estimation approaches that we created to handle change points detection.
The usefulness of our results have relevance to regulatory authorities'
policy-making, trading strategy's construction by investors and provider's of financial
products and services, and other scientific endeavours in the natural and social
sciences impacted by sudden and significant changes (e.g., break, jumps, shifts, etc)
in the time-series data.
This work provides impetus for the investigation and
development of methodology suited in tackling further the
multiple-change point problem for a multivariate OU process
and other closely related modelling challenges in the research literature and practice
that entail the statistical inference of stochastic processes.
\ \\
\appendix
\renewcommand{\thesection}{Appendix \Alph{section}}
\section{Proof of Proposittion~\ref{prncoef}}\label{appendixsection2}
\begin{proof}[Proof of Proposition~\ref{prncoef}]
We first need to prove that the coefficients of (\ref{ou1}) satisfy the space-variable Lipschitz condition. That is,
$$|\mu(t,x)-\mu(t,y)|^2+|\sigma(t,x)-\sigma(t,y)|^2\leq K_a |x-y|^2, \quad \mbox{for some}\quad K_a> 0.$$
Given that $\sigma(t,x)$ is constant in our modelling framework, the second term above is 0. Hence,
\begin{eqnarray*}
&&|\mu(t,x)-\mu(t,y)|^2+|\sigma(t,x)-\sigma(t,y)|^2=\left(\sum_{j=1}^{m+1}\left(S\left(\boldsymbol{\mathrm \uptheta}^{(j)},t,x\right)-S\left(\boldsymbol{\mathrm \uptheta}^{(j)},t,y\right)\right)\mathbbm{1}_{\{\tau_{j-1}^0< t\leq \tau_j^0\}}\right)^2\\
&&=\sum_{j=1}^{m+1}\left(S\left(\boldsymbol{\mathrm \uptheta}^{(j)},t,x\right)-S\left(\boldsymbol{\mathrm \uptheta}^{(j)},t,y\right)\right)^2\mathbbm{1}_{\{\tau_{j-1}^0< t\leq \tau_j^0\}}
=\sum_{j=1}^{m+1}(L^{(j)}(t)-a^{(j)}x-L^{(j)}(t)+a^{(j)}y)^2\mathbbm{1}_{\{\tau_{j-1}^0< t\leq \tau_j^0\}}\\
&&=\sum_{j=1}^{m+1}(a^{(j)}(x-y))^2\mathbbm{1}_{\{\tau_{j-1}^0< t\leq \tau_j^0\}}
\leq\sum_{j=1}^{m+1}(a^{(j)}(x-y))^2=\left(\sum_{j=1}^{m+1}(a^{(j)})^2\right)(x-y)^2.
\end{eqnarray*}
Since $a^{(j)}<\infty$ for $j=1,~\ldots,~p$, there exists a $K_a>0$ such that $\left(\sum_{j=1}^{m+1}(a^{(j)})^2\right)\leq K_a$. Then,
$$|\mu(t,x)-\mu(t,y)|^2+|\sigma(t,x)-\sigma(t,y)|^2\leq K_a |x-y|^2.$$
\ \\
\noindent Next, we prove the spatial growth condition. That is,
$$|\mu(t,x)|^2+|\sigma(t,x)|^2\leq K_b (1+x^2), \quad \mbox{for some}\quad K_b> 0.$$
Note that
\begin{eqnarray*}
&&|\mu(t,x)|^2+|\sigma(t,x)|^2
=\left[\sum_{j=1}^{m+1}S\left(\boldsymbol{\mathrm \uptheta}^{(j)},t,x\right)\mathbbm{1}_{\{\tau_{j-1}^0< t\leq \tau_j^0\}}\right]^2+\sigma^2 \nonumber \\
&&=\sum_{j=1}^{m+1}\left(L^{(j)}(t)-a^{(j)}x\right)^2\mathbbm{1}_{\{\tau_{j-1}^0< t\leq \tau_j^0\}}+\sigma^2
\leq \sum_{j=1}^{m+1}\left(L^{(j)}(t)-a^{(j)}x\right)^2+\sigma^2.
\end{eqnarray*}
Using the identity $(a+b)^2\leq 2a^2+2b^2$, we have
$$|\mu(t,x)|^2+|\sigma(t,x)|^2\leq 2\sum_{j=1}^{m+1}\left(L^{(j)}(t)\right)^2+\sigma^2+2\sum_{j=1}^{m+1}\left(a^{(j)}\right)^2x^2.$$
Since $\{\varphi_k(t),~i=1,~\ldots,~p\}$ are bounded, we can find a constant $K_b>0$ such that
$$\max\left(2\sum_{j=1}^{m+1}\left(L^{(j)}(t)\right)^2+\sigma^2,2\sum_{j=1}^{m+1}\left(a^{(j)}\right)^2\right) \leq K_b.$$
This gives
$$|\mu(t,x)|^2+|\sigma(t,x)|^2\leq K_b (1+x^2).$$
Moreover, let $K_c=\max(K_a,K_b)$. Then
$$|\mu(t,x)-\mu(t,y)|^2+|\sigma(t,x)-\sigma(t,y)|^2\leq K_c |x-y|^2$$
and
$$|\mu(t,x)|^2+|\sigma(t,x)|^2\leq K_c (1+x^2).$$
\end{proof}

\section{Proof of the propositions in section~\ref{mcpe}} \label{appendixb}
\begin{proof}[Proof of Proposition~\ref{prnm13}]
Let $\hat{u}_i$ be the residual of the $i$th element based on the estimated change points $\{{\hat{\tau}}_j\}$, $j=1,~\ldots,~m+1$, i.e., $\hat{u}_i=Y_i-\boldsymbol{\mathrm z}_i\boldsymbol{\mathrm{\hat{\uptheta}}}_i=\boldsymbol{\mathrm z}_i{\uptheta}_i-\boldsymbol{\mathrm z}_i\boldsymbol{\mathrm{\hat{\uptheta}}}_i+u_i$, for $t_i\in [0,T]$, where $\boldsymbol{\mathrm \uptheta}_i$ and $\boldsymbol{\mathrm{\hat{\uptheta}}}_i$ are defined in Section~\ref{mcpe} as the true value and MLE of the parameters
based on the assigned estimated change points of the coefficients associated with the $i$th element. Also, let $\hat{u}_i^0$ be the residual of the $i$th element based on the exact change points $\{\tau_j^0, j=1,~\ldots~,m+1\}$ and
$\boldsymbol{\mathrm{\hat{\uptheta}}}_i^{0}=\sum_{j=1}^{m+1}\boldsymbol{\mathrm{\hat{\uptheta}}}^{(j,0)}\mathbbm{1}(\tau_{j-1}^0\leq t_i\leq \tau_j^0)$ with $\boldsymbol{\mathrm{\hat{\uptheta}}}^{(j,0)}=\boldsymbol{\mathrm Q}_{(\tau_j^0,\tau_{j-1}^0)}^{-1}\boldsymbol{\mathrm{\tilde{r}}}_{(\tau_j^0,\tau_{j-1}^0)}$.\\
\ \\
The proof relies on investigating the behaviour of
\begin{equation}\label{prnm1-eq0}
\frac{1}{\phi}\left(\sum_{t_i\in[0,T]}\hat{u}_i'\hat{u}_i-\sum_{t_i\in [0,T]}\hat{u}_i^{0\prime}\hat{u}_i^0\right),
\end{equation}
where $\phi=T$. By (\ref{cpm1}), $(\ref{prnm1-eq0})\leq 0$ with probability 1. Hence, it remains to show that if one of the change points is not consistently estimated, $(\ref{prnm1-eq0})> 0$ with positive probability yielding a contradiction. \\
\ \\
\noindent Using the quadratic expansion $(a+b)^2=a^2+2ab+b^2$ and some algebraic computations, (\ref{prnm1-eq0}) can be expressed as
\begin{eqnarray}
(\ref{prnm1-eq0})&=&\frac{1}{T}\sum_{t_i\in[0,T]}\left(\boldsymbol{\mathrm z}_i\left(\boldsymbol{\mathrm \uptheta}_i-\boldsymbol{\mathrm{\hat{\uptheta}}}_i\right)
\right)^2-\frac{1}{T}\sum_{t_i\in[0,T]}\left(\boldsymbol{\mathrm z}_i\left(\boldsymbol{\mathrm \uptheta}_i-\boldsymbol{\mathrm{\hat{\uptheta}}}_i^0\right)\right)^2 \nonumber \\
&&~~+\frac{2}{T}\sum_{t_i\in[0,T]}\left(u_i'\boldsymbol{\mathrm z}_i\left(\boldsymbol{\mathrm \uptheta}_i-\boldsymbol{\mathrm{\hat{\uptheta}}}_i^0\right)\right)
-\frac{2}{T}\sum_{t_i\in[0,T]}\left(u_i'\boldsymbol{\mathrm z}_i\left(\boldsymbol{\mathrm \uptheta}_i-\boldsymbol{\mathrm{\hat{\uptheta}}}_i\right)\right).\label{prnm1-eq01}\end{eqnarray}
\ \\
\noindent We aim to show that if there exists a change point $\tau_j^0$, $j=1~\ldots,~m$, and {\color{black}it is not} consistently estimated, the first term $\frac{1}{T}\sum_{t_i\in[0,T]}\left(\boldsymbol{\mathrm z}_i\left(\boldsymbol{\mathrm \uptheta}_i-\boldsymbol{\mathrm{\hat{\uptheta}}}_i\right)\right)^2$ is larger than a positive constant with positive probability, whilst the rest of terms is of $o_p(1)$ and hence, $(\ref{prnm1-eq0})>0$ with positive probability. To this end, we first provide a lemma, which will be useful in deriving the asymptotic consistency of the estimated change points.\\
\ \\
{\bf Lemma C.1} If at least one of the change points, say  $\tau_{j}^0$, can not be consistently estimated, then
for large $T$,
$$\frac{1}{T}\sum_{t_i\in[0,T]}\left(\boldsymbol{\mathrm z}_i\left(\boldsymbol{\mathrm \uptheta}_i-\boldsymbol{\mathrm{\hat{\uptheta}}}_i\right)\right)^2\geq C_0\norm{\boldsymbol{\mathrm \uptheta}^{(j)}-\boldsymbol{\mathrm \uptheta}^{(j+1)}}^2
~~\mbox{with positive probability}.$$
\begin{proof}
If the change point $\tau_{j}^0$ {\color{black}is not} consistently estimated, then with some positive probability there exists an $\eta>0$ such that there is no estimated change point in $[\tau_{j}^0-\eta T,\tau_{j}^0+\eta T]$ for some $\eta>0$. Without loss of generality, let ${\hat{\tau}}_{k-1}\leq \tau_{j}^0-\eta T\leq\tau_{j}^0+\eta T\leq {\hat{\tau}}_k$. Then, since $(\boldsymbol{\mathrm z}_i({\uptheta}_i-\boldsymbol{\mathrm{\hat{\uptheta}}}_i
 ))^2\geq 0$ for each $i$,
\begin{eqnarray}
\frac{1}{T}\sum_{t_i\in[0,T]}\left(\boldsymbol{\mathrm z}_i\left(\boldsymbol{\mathrm \uptheta}_i-\boldsymbol{\mathrm{\hat{\uptheta}}}_i\right)\right)^2&\geq&\eta \left(\left(\boldsymbol{\mathrm \uptheta}^{(j)}-\boldsymbol{\mathrm{\hat{\uptheta}}}^{(k)}\right)^{\top}\frac{1}{\eta T}\sum_{t_i\in(\tau_{(j)}^0-\eta T,\tau_{j}^0]} \boldsymbol{\mathrm z}_i^{\top}\boldsymbol{\mathrm z}_i\left(\boldsymbol{\mathrm \uptheta}^{(j)}-\boldsymbol{\mathrm{\hat{\uptheta}}}^{(k)}\right)\right)+\eta \left(\left(\boldsymbol{\mathrm \uptheta}^{(j+1)}-\boldsymbol{\mathrm{\hat{\uptheta}}}^{(k)}\right)^{\top} \right. \nonumber\\&&\left. \times \frac{1}{\eta T}\sum_{t_i\in(\tau_{j}^0,\tau_{j}^0+\eta T]}\boldsymbol{\mathrm z}_i'\boldsymbol{\mathrm z}_i\left(\boldsymbol{\mathrm \uptheta}^{(j+1)}-\boldsymbol{\mathrm{\hat{\uptheta}}}^{(k)}\right)\right).\label{prnm1-eq02}\end{eqnarray}
\noindent Let $\gamma_1$ and $\gamma_2$ be the smallest eigenvalues of $\frac{1}{\eta T}\sum_{t_i\in(\tau_{j}^0-\eta T,\tau_{j}^0]}\boldsymbol{\mathrm z}_i^{\top}\boldsymbol{\mathrm z}_i$ and $\frac{1}{\eta T}\sum_{t_i\in(\tau_{j}^0,\tau_{j}^0+\eta T]} \boldsymbol{\mathrm z}_i^{\top}\boldsymbol{\mathrm z}_i$, respectively. Then,
\begin{eqnarray*}
(\ref{prnm1-eq02})&\geq&\eta \gamma_1\norm{\boldsymbol{\mathrm \uptheta}^{(j)}-\boldsymbol{\mathrm{\hat{\uptheta}}}^{(k)}}^2+\eta \gamma_2\norm{\boldsymbol{\mathrm \uptheta}^{(j+1)}-\boldsymbol{\mathrm{\hat{\uptheta}}}^{(k)}}^2
\geq\eta \min (\gamma_1, \gamma_2)(\norm{\boldsymbol{\mathrm \uptheta}^{(j)}-\boldsymbol{\mathrm{\hat{\uptheta}}}^{(k)}}^2+\norm{\boldsymbol{\mathrm \uptheta}^{(j+1)}-\boldsymbol{\mathrm{\hat{\uptheta}}}^{(k)}}^2).\end{eqnarray*}
Using the convexity of a quadratic function, we have
$$\norm{\boldsymbol{\mathrm \uptheta}^{(j)}-\boldsymbol{\mathrm{\hat{\uptheta}}}^{(k)}}^2+\norm{\boldsymbol{\mathrm \uptheta}^{(j+1)}-\boldsymbol{\mathrm{\hat{\uptheta}}}^{(k)}}^2\geq\frac{1}{2}\norm{\boldsymbol{\mathrm \uptheta}^{(j)}-\boldsymbol{\mathrm \uptheta}^{(j+1)}}^2.$$
Hence,
$$\frac{1}{T}\sum_{t_i\in[0,T]}\left(\boldsymbol{\mathrm z}_i\left(\boldsymbol{\mathrm \uptheta}_i-\boldsymbol{\mathrm{\hat{\uptheta}}}_i\right)\right)^2\geq\displaystyle \eta \frac{\min (\gamma_1, \gamma_2)}{2}\norm{\boldsymbol{\mathrm \uptheta}^{(j)}-\boldsymbol{\mathrm \uptheta}^{(j+1)}}^2.$$
Under Assumption~\ref{asm4}, $\gamma_1$ and $\gamma_2$ are both bounded away from 0 and $\min (\gamma_1, \gamma_2)$ is also bounded away from 0. Therefore, the right-hand side of the above inequality is positive. Then, the proof is complete by letting $\displaystyle C_0=\eta \frac{\min (\gamma_1, \gamma_2)}{2}$.
\end{proof}
\noindent {\bf Lemma C.2} Under Assumptions \ref{asm1}--\ref{asm3}, $\frac{1}{T}\boldsymbol{\mathrm Q}_{(\hat{s}_{j-1}T,\hat{s}_j T)}
\xrightarrow[T\rightarrow \infty]{a.s.}(\hat{s}_j-\hat{s}_{j-1})\boldsymbol{\mathrm \Sigma}_j$ for $s_{j-1}^0\leq \hat{s}_{j-1}<\hat{s}_j\leq s_{j}^0$. \\

\noindent Lemma C.2. directly follows from Proposition 2.2.6 in Zhang~(2015) with $s_{j-1}^0$ and $s_j^0$ replaced by $\hat{s}_{j-1}$ and $\hat{s}_j$, respectively. To emphasise again, both lemmas~C.1 and C.2 are key in proving the asymptotic properties of the proposed estimators. \\
\\
\noindent Next, note that for $t_i\in (\tau_{j-1}^0,\tau_j^0]$, $j=1,~\ldots,~m$, we have
$\boldsymbol{\mathrm \uptheta}_i=\boldsymbol{\mathrm \uptheta}^{(j)}$ and $\boldsymbol{\mathrm{\hat{\uptheta}}}_i^{0}=\boldsymbol{\mathrm \uptheta}^{(j)}+\boldsymbol{\mathrm Q}_{({\tau}_{j-1}^0,\tau_j^0)}^{-1}\sigma \boldsymbol{\mathrm r}_{({\tau}_{j-1}^0,\tau_j^0)}$. Substituting these expressions into $\frac{1}{T}\sum_{t_i\in[0,T]}\left(\boldsymbol{\mathrm z}_i\left(\boldsymbol{\mathrm \uptheta}_i-\boldsymbol{\mathrm{\hat{\uptheta}}}_i^{0}\right)\right)^2$, we get
\begin{eqnarray}
\sum_{j=1}^{m}\frac{1}{T}\boldsymbol{\mathrm r}_{(\tau_{j-1}^0,\tau_j^0)}^{\top}
\boldsymbol{\mathrm Q}_{(\tau_{j-1}^0,\tau_j^0)}^{-1}\sum_{t_i\in(\tau_{j-1}^0,\tau_j^0]}\boldsymbol{\mathrm z}_i^{\top}\boldsymbol{\mathrm z}_i\boldsymbol{\mathrm Q}_{(\tau_{j-1}^0,\tau_j^0)}^{-1} \boldsymbol{\mathrm r}_{(\tau_{j-1}^0,\tau_j^0)}.\label{prnm1-eq03}
\end{eqnarray}
\ \\
To proceed further, we first prove the following inequality.
Suppose $0<\tau_1^*<\tau_2^*\leq T$. By the Markov inequality, It\^o's isometry and (\ref{sol2}), we have
\begin{eqnarray}\label{convm-R1}
\textrm{P}\left(\frac{|\int_{\tau_1^*}^{{\tau_2}^*}X_td W_t|}{\sqrt{\tau_2^*-\tau_1^*}}>K^*\right)
\leq \frac{\mathrm{E}\left(|\int_{\tau_1^*}^{\tau_2^*}X_td W_t|^2\right)}{(\tau_2^*-\tau_1^*)(K^*)^2}
=\frac{\int_{\tau_1^*}^{{\tau}_2^*}E(X_t^2)dt}{(\tau_2^*-\tau_1^*)(K^*)^2}
\leq \frac{K_1(\tau_2^*-\tau_1^*)}{(\tau_2^*-\tau_1^*)(K^*)^2}=\frac{K_1}{(K^*)^2}.
\end{eqnarray}
\noindent Therefore, by letting $K^*=(\log T)^{a^*}$, for some $0<a^*<1/2$, the above probability tends to 0 as $T$ tends to infinity. This implies that for some $0<a^*<1/2$,
\begin{equation}\label{convm_R2}\frac{1}{\sqrt{\tau_2^*-\tau_1^*}}\norm{\boldsymbol{\mathrm r}_{(\tau_1^*,\tau_2^*)}}
=O_p(\log^{a^*} T)\quad \textrm{for any}
\quad 0<\tau_1^*<\tau_2^*\leq T.\end{equation}
\ \\
Now, continuing the proof of Proposition 4.1, we note that under Assumption \ref{asm3},
$\displaystyle{\sup_{0\leq t\leq T}|\varphi_k(t)|\leq K_{\varphi_k}<\infty}$ for $i=1,~\ldots,~p$.
By similar argument used to obtain (\ref{convm_R2}), we have $\frac{1}{\sqrt{\tau_2^*-\tau_1^*}}\int_{{\hat{\tau}}_{j-1}}^{{\hat{\tau}}_j}\varphi_k(t)d W_t
=O_p\left( (\log T)^{a^*}\right)$ for $i=1,~\ldots~,p$ so that $\norm{\frac{1}{\sqrt{\tau_2^*-\tau_1^*}}\boldsymbol{\mathrm r}_{({\hat{\tau}}_{j-1},{\hat{\tau}}_{j})}}=O_p\left((\log T)^{a^*}\right)$.
Moreover, since $\sum_{t_i\in({\tau}_{1}^*,{\tau}_2^*]}u_i^{\top}\boldsymbol{\mathrm z}_i$  is the discretised versions of $\boldsymbol{\mathrm r}_{({\tau}_{1}^*,{\tau}_2^*)}$, it similarly follows that $\norm{\frac{1}{\sqrt{\tau_2^*-\tau_1^*}}\sum_{t_i\in({\tau}_{1}^*,{\tau}_2^*]}u_i^{\top}\boldsymbol{\mathrm z}_i}=O_p\left(
(\log T)^{a^*} \right)$. Also, $\sum_{t_i\in(\tau_{j-1}^0,\tau_j^0]}\boldsymbol{\mathrm z}_i^{\top}\boldsymbol{\mathrm z}_i$ is the discretised version of $\boldsymbol{\mathrm Q}_{(\tau_{j-1}^0,\tau_j^0)}$, thus the asymptotic results in (\ref{convQ-m1}) and  (\ref{convQ-m2}) also hold for $\sum_{t_i\in(\tau_{j-1}^0,\tau_j^0]}\boldsymbol{\mathrm z}_i^{\top}\boldsymbol{\mathrm z}_i$. Therefore, by the Cauchy-Schwarz inequality, (\ref{convQ-m1}) and (\ref{convQ-m2}) we have, after some algebraic manipulations,
\begin{eqnarray}
(\ref{prnm1-eq03})\leq \sum_{j=1}^m\frac{1}{T} \norm{\frac{1}{\sqrt{(s_j^0-s_{j-1}^0)T}}\boldsymbol{\mathrm r}_{(\tau_{j-1}^0,\tau_j^0)}}^2
\norm{(s_{j}^0-s_{j-1}^0)T\boldsymbol{\mathrm Q}_{(\tau_{j-1}^0,\tau_j^0)}^{-1}}^2
\norm{\frac{1}{(s_{j}^0-s_{j-1}^0)T}\sum_{t_i\in(\tau_{j-1}^0,\tau_j^0]}\boldsymbol{\mathrm z}_i^{\top}\boldsymbol{\mathrm z}_i}=o_p(1).\quad \label{prnm1-eq3-1}
\end{eqnarray}
It remains to investigate the quantity
\begin{equation}\frac{2}{T}\sum_{t_i\in[0,T]}(u_i^{\top}\boldsymbol{\mathrm z}_i({\uptheta}_i-\boldsymbol{\mathrm{\hat{\uptheta}}}_i^0))
-\frac{2}{T}\sum_{t_i\in[0,T]}(u_i^{\top}\boldsymbol{\mathrm z}_i({\uptheta}_i-\boldsymbol{\mathrm{\hat{\uptheta}}}_i))
=\frac{2}{T}\sum_{t_i\in[0,T]}(u_i^{\top}\boldsymbol{\mathrm z}_i(\boldsymbol{\mathrm{\hat{\uptheta}}}_i-\boldsymbol{\mathrm{\hat{\uptheta}}}_i^0))\label{prnm1-eq04}.\end{equation}
\ \\
Note that the structure of $\boldsymbol{\mathrm{\hat{\uptheta}}}_i$ is affected by the location of the estimated change points. It is, therefore, difficult to substitute the expressions for $\boldsymbol{\mathrm{\hat{\uptheta}}}_i$ into (\ref{prnm1-eq04}) directly. Without loss of generality, we consider $m=2$ (but the procedure can be extended to the general case $m>0$), and assume that $0=\tau_0<{\hat{\tau}}_1<{\tau}_1^0<{\tau}_2^0<{\hat{\tau}}_2<\tau_3^0=T$. Other cases can be analysed in a similar manner. With $m=2$, (\ref{prnm1-eq04}) reduces to
\begin{eqnarray}
(\ref{prnm1-eq04})&=&\frac{2}{T}\sum_{t_i\in (0,{\hat{\tau}}_1]}\left(u_i\boldsymbol{\mathrm z}_i(\boldsymbol{\mathrm{\hat{\uptheta}}}^{(1,0)}-\boldsymbol{\mathrm{\hat{\uptheta}}}^{(1)})\right)
+\frac{2}{T}\sum_{t_i\in ({\hat{\tau}}_1,\tau_1^0]}\left(u_i\boldsymbol{\mathrm z}_i\left(\boldsymbol{\mathrm{\hat{\uptheta}}}^{(1,0)}-\boldsymbol{\mathrm{\hat{\uptheta}}}^{(2)}\right)\right)
+\frac{2}{T}\sum_{t_i\in (\tau_1^0,\tau_2^0]}\left(u_i\boldsymbol{\mathrm z}_i\left(\boldsymbol{\mathrm{\hat{\uptheta}}}^{(2,0)}-\boldsymbol{\mathrm{\hat{\uptheta}}}^{(2)}\right)\right)\nonumber\\
&&~~+\frac{2}{T}\sum_{t_i\in ({\tau}_2^0,{\hat{\tau}}_2]}\left(u_i\boldsymbol{\mathrm z}_i\left(\boldsymbol{\mathrm{\hat{\uptheta}}}^{(3,0)}-\boldsymbol{\mathrm{\hat{\uptheta}}}^{(2)}\right)\right)+\frac{2}{T}\sum_{t_i\in ({\hat{\tau}}_2,T]}\left( u_i\boldsymbol{\mathrm z}_i\left(\boldsymbol{\mathrm{\hat{\uptheta}}}^{(3,0)}-\boldsymbol{\mathrm{\hat{\uptheta}}}^{(3)}\right)\right),
\end{eqnarray}
where $\boldsymbol{\mathrm{\hat{\uptheta}}}^{(1)}=\boldsymbol{\mathrm \uptheta}^{(1)}+\sigma \boldsymbol{\mathrm Q}_{(0,{\hat{\tau}}_1)}^{-1}  \boldsymbol{\mathrm r}_{(0,{\hat{\tau}}_1)}$,
$\boldsymbol{\mathrm{\hat{\uptheta}}}^{(2)}=\boldsymbol{\mathrm Q}_{({\hat{\tau}}_1,{\hat{\tau}}_2)}^{-1}(\boldsymbol{\mathrm Q}_{({\hat{\tau}}_1,{\tau}_1^0)}\boldsymbol{\mathrm \uptheta}^{(1)}+\boldsymbol{\mathrm Q}_{({\tau}_1^0,{\tau}_2^0)}\boldsymbol{\mathrm \uptheta}^{(2)}+\boldsymbol{\mathrm Q}_{({\tau}_2^0,{\hat{\tau}}_2)}\boldsymbol{\mathrm \uptheta}^{(3)}+\sigma \boldsymbol{\mathrm r}_{({\hat{\tau}}_1,{\hat{\tau}}_2)})$,
$\boldsymbol{\mathrm{\hat{\uptheta}}}^{(3)}=\boldsymbol{\mathrm Q}_{({\hat{\tau}}_2,T)}^{-1}(\boldsymbol{\mathrm Q}_{({\hat{\tau}}_2,T)}\boldsymbol{\mathrm \uptheta}^{(3)}+\sigma \boldsymbol{\mathrm r}_{({\hat{\tau}}_2,T)})$,  $\boldsymbol{\mathrm{\hat{\uptheta}}}^{(1,0)}=\boldsymbol{\mathrm Q}_{(0,{\tau}_1^0)}^{-1}(\boldsymbol{\mathrm Q}_{(0,{\tau_1^0})}\boldsymbol{\mathrm \uptheta}^{(1)}+\sigma \boldsymbol{\mathrm r}_{(0,{\tau}_1^0)})$, $\boldsymbol{\mathrm{\hat{\uptheta}}}^{(2,0)}=\boldsymbol{\mathrm Q}_{({\tau}_1^0,{\tau}_2^0)}^{-1}(\boldsymbol{\mathrm Q}_{({\tau}_1^0,\tau_2^0)}\boldsymbol{\mathrm \uptheta}^{(2)}+\sigma \boldsymbol{\mathrm r}_{({\tau}_1^0,{\tau}_2^0)})$, and $\boldsymbol{\mathrm{\hat{\uptheta}}}^{(3,0)}=\boldsymbol{\mathrm Q}_{({\tau}_2^0,T)}^{-1}(\boldsymbol{\mathrm Q}_{({\tau}_2^0,T)}\boldsymbol{\mathrm \uptheta}^{(2)}+\sigma \boldsymbol{\mathrm r}_{({\tau}_2^0,T)})$.\\
\ \\
Using the Cauchy-Schwarz inequality, (\ref{convm_R2}), Lemma~C.2 along with the Continuous Mapping Theorem, the first term in (\ref{prnm1-eq04}) is bounded above, i.e.,
\begin{eqnarray*}
\frac{2}{T}\sum_{t_i\in [0,{\hat{\tau}}_1]}(u_i'\boldsymbol{\mathrm z}_i(\boldsymbol{\mathrm{\hat{\uptheta}}}^{(1,0)}-\boldsymbol{\mathrm{\hat{\uptheta}}}^{(1)}))
\leq\frac{2\sqrt{\hat{s}_1}}{\sqrt{s_1^0}T}\norm{\frac{1}{\sqrt{\hat{s}_1T}}\sum_{t_i\in [0,{\hat{\tau}}_1]}u_i^{\top}\boldsymbol{\mathrm z}_i} \norm{s_1^0T\boldsymbol{\mathrm Q}_{(0,\tau_1^0)}^{-1}}
\norm{\frac{1}{\sqrt{s_1^0T}}\boldsymbol{\mathrm r}_{(0,\tau_1^0)}}
\\-\frac{2}{T}
\norm{\frac{1}{\sqrt{\hat{s}_1T}}\sum_{t_i\in [0,{\hat{\tau}}_1]}u_i^{\top}\boldsymbol{\mathrm z}_i||\hat{s}_1T\boldsymbol{\mathrm Q}_{(0,{\hat{\tau}}_1)}^{-1}}
\norm{\frac{1}{\sqrt{\hat{s}_1T}}\boldsymbol{\mathrm r}_{(0,{\hat{\tau}}_1)}}=o_p(1).\end{eqnarray*}
Similarly, one can show that the rest of the terms in (\ref{prnm1-eq04}) are all of $o_p(1)$. Following these arguments, it can be shown that in {\color{black}the} general case when $m>0$, the terms in (\ref{prnm1-eq04}) are all of order $o_p(1)$. So, if one {\color{black}of the change points' arrival rates}, say $s_j^0$, is not consistently estimated, (\ref{prnm1-eq0}) is dominated by the first term, which is larger than 0 with positive probability. This gives a contradiction. Therefore, $\hat{s}_j-s_j^0\xrightarrow[T\rightarrow \infty]{p}0$ for every $j=1,~\ldots,~m$.
\end{proof}

\begin{proof}[Proof of Proposition~\ref{prnm14}]
\noindent Without loss of generality, assume that $m=3$. We provide an explicit proof dealing with the rate $T$-consistency for ${\hat{\tau}}_2$ only. The consistency analysis for ${\hat{\tau}}_1$ and ${\hat{\tau}}_3$ can be similarly completed. By Proposition~\ref{prnm13}, ${\hat{\tau}}_j \in \{\tau: |\tau_j-\tau_j^0|\leq \eta T, 1\leq j\leq m\}$ for each $\eta>0$. For $C>0$, define the set $V_{\eta}(C):=\{{\tau}: |{\tau_j}-{\tau_j^0}|\leq \eta T, 1\leq j\leq m;~ |\tau_2-\tau_2^0|>C \}$. \\
\ \\
Let $SSE_1=SSE(\tau_1,\tau_2,\tau_3)=\sum_{t_i\in[0,T]}(Y_i-\boldsymbol{\mathrm z}_i\boldsymbol{\mathrm{\hat{\uptheta}}}_i)^2$, where $\boldsymbol{\mathrm{\hat{\uptheta}}}_i$ is the MLE of $\uptheta$ associated with the $i$th element under the change point $(\tau_1,\tau_2,\tau_3)$, and let $SSE_2=SSE(\tau_1,\tau_2^0,\tau_3)$. So, we have $\min_{(\tau_1,\tau_2,\tau_3)}(SSE_1-SSE_2)\leq 0$ with probability 1. If we can show that  for each $\epsilon>0$, there exists a $C>0$, $0<\eta<1$  such that for large $T$ and any ${\tau}\in V_{\eta}(C)$, $P(\min_{{\tau}\in V_\eta(C)}(SSE_1-SSE_2)>0)<\epsilon$, this would imply that for some $C>0$, the global optimisation can not be achieved on the set $V_{\eta}(C)$. Thus with large probability, $|\boldsymbol{\hat{\mathrm \tau}}-\tau^0|\leq C$. \\
\ \\
Let $\boldsymbol{\hat{\mathrm \tau}}=\arg \min_{\tau\in V_{\eta}(C)}(SSE_1-SSE_2)$.
Assume, without loss of generality, that ${\hat{\tau}}_1\leq\tau_1^0<\tau_2^0<{\hat{\tau}}_2<{\hat{\tau}}_3\leq \tau_3^0$. We now focus on the behaviour of
\begin{equation}\label{prnm2-eq0}
\left(SSE_1-SSE_2\right)/({\hat{\tau}}_2-\tau_2^0).
\end{equation}
Applying the identity $(a+b)^2=a^2+2ab+b^2$ into (\ref{prnm2-eq0}) and then breaking the time period $[0,T]$ into different intervals, we have
\begin{eqnarray}
&&SSE_1-SSE_2\nonumber\\&=&-\sum_{t_i\in({\hat{\tau}}_1,{\tau}_2^0]}\left(\boldsymbol{\mathrm z}_i\left(\boldsymbol{\mathrm \uptheta}^{(1)}-\boldsymbol{\mathrm{\hat{\uptheta}}}^{(2,0)}\right)\right)^2
-\sum_{t_i\in (\tau_1^0,\tau_2^0]}\left(\boldsymbol{\mathrm z}_i\left(\boldsymbol{\mathrm \uptheta}^{(2)}-\boldsymbol{\mathrm{\hat{\uptheta}}}^{(2,0)}\right)\right)^2
 -\sum_{t_i\in ({\tau}_2^0,\hat{\tau}_3]}\left(\boldsymbol{\mathrm z}_i\left(\boldsymbol{\mathrm \uptheta}^{(3)}-\boldsymbol{\mathrm{\hat{\uptheta}}}^{(3,0)}\right)\right)^2\label{prnm14-1-2}\\
&&+\sum_{t_i\in (\hat{\tau}_1,\tau_2^0]}2u_i^{\top}\boldsymbol{\mathrm z}_i\left(\boldsymbol{\mathrm{\hat{\uptheta}}}^{(2,0)}-\boldsymbol{\mathrm{\hat{\uptheta}}}^{(2)}\right)
+\sum_{t_i\in({\tau}_2^0,{\hat{\tau}}_2]}2u_i^{\top}\boldsymbol{\mathrm z}_i\left(\boldsymbol{\mathrm{\hat{\uptheta}}}^{(3,0)}-\boldsymbol{\mathrm{\hat{\uptheta}}}^{(2)}\right)
+\sum_{t_i\in ({\hat{\tau}}_2,{\hat{\tau}}_3]}2u_i^{\top}\boldsymbol{\mathrm z}_i\left(\boldsymbol{\mathrm{\hat{\uptheta}}}^{(3,0)}-\boldsymbol{\mathrm{\hat{\uptheta}}}^{(3)}\right)\label{prnm14-1-3}\\
&&+\sum_{t_i\in (\hat{\tau}_1,\tau_1^0]}\left(\boldsymbol{\mathrm z}_i\left(\boldsymbol{\mathrm \uptheta}^{(1)}-\boldsymbol{\mathrm{\hat{\uptheta}}}^{(2)}\right)\right)^2
+\sum_{t_i\in(\tau_1^0,\tau_2^0]}\left(\boldsymbol{\mathrm z}_i\left(\boldsymbol{\mathrm \uptheta}^{(2)}-\boldsymbol{\mathrm{\hat{\uptheta}}}^{(2)}\right)\right)^2+\sum_{t_i\in(\tau_2^0,{\hat{\tau}}_2]}\left(\boldsymbol{\mathrm z}_i\left(\boldsymbol{\mathrm \uptheta}^{(3)}-\boldsymbol{\mathrm{\hat{\uptheta}}}^{(2)}\right)\right)^2\nonumber\\&&+\sum_{t_i\in({\hat{\tau}}_2,{\hat{\tau}}_3]}\left(\boldsymbol{\mathrm z}_i\left(\boldsymbol{\mathrm \uptheta}^{(3)}-\boldsymbol{\mathrm{\hat{\uptheta}}}^{(3)}\right)\right)^2\qquad\quad \label{prnm14-1-1}
\end{eqnarray}
where
$\boldsymbol{\mathrm{\hat{\uptheta}}}^{(2)}=\boldsymbol{\mathrm Q}_{({\hat{\tau}}_1,{\hat{\tau}}_2)}^{-1}\left(\boldsymbol{\mathrm Q}_{({\hat{\tau}}_1,\tau_1^0)}\boldsymbol{\mathrm \uptheta}^{(1)}+\boldsymbol{\mathrm Q}_{(\tau_1^0,{\tau}_2^0)}\boldsymbol{\mathrm \uptheta}^{(2)}+\boldsymbol{\mathrm Q}_{(\tau_2^0,{\hat{\tau}}_2)}\boldsymbol{\mathrm \uptheta}^{(3)}+\sigma \boldsymbol{\mathrm r}_{({\hat{\tau}}_1,{\hat{\tau}}_2)}\right)$,
$\boldsymbol{\mathrm{\hat{\uptheta}}}^{(3)}=\boldsymbol{\mathrm Q}_{({\hat{\tau}}_2,{\hat{\tau}}_3)}^{-1}\left(\boldsymbol{\mathrm Q}_{({\hat{\tau}}_2,\hat{\tau_3})}\boldsymbol{\mathrm \uptheta}^{(3)}+\sigma \boldsymbol{\mathrm r}_{({\hat{\tau}}_2,{\hat{\tau}}_3)}\right)$,
$\boldsymbol{\mathrm{\hat{\uptheta}}}^{(2,0)}=\boldsymbol{\mathrm Q}_{({\hat{\tau}}_1,{\tau}_2^0)}^{-1}\left(\boldsymbol{\mathrm Q}_{({\hat{\tau}}_1,{\tau_1^0})}\boldsymbol{\mathrm \uptheta}^{(1)}+\boldsymbol{\mathrm Q}_{({\tau}_1^0,\tau_2^0)}\boldsymbol{\mathrm \uptheta}^{(2)}+\sigma \boldsymbol{\mathrm r}_{({\hat{\tau}}_1,{\tau}_2^0)}\right),$
and
$\boldsymbol{\mathrm{\hat{\uptheta}}}^{(3,0)}=\boldsymbol{\mathrm Q}_{(\tau_2^0, {\hat{\tau}}_3)}^{-1}\left(\boldsymbol{\mathrm Q}_{({\tau}_2^0,{\hat{\tau}}_3)}\boldsymbol{\mathrm \uptheta}^{(3)}+\sigma \boldsymbol{\mathrm r}_{({\tau}_2^0, {\hat{\tau}}_3)}\right).$\\
\ \\
It remains to show that for large $T$,  $\displaystyle \frac{(\ref{prnm14-1-1})+(\ref{prnm14-1-2})+(\ref{prnm14-1-3})}{{\hat{\tau}}_2-\tau_2^0}$ is positive with probability 1. Note that this depends on the choice of ${\hat{\tau}}_2\in V_\eta(C)$; $(\ref{prnm14-1-2})/({\hat{\tau}}_2-\tau_2^0)$ can be $O_p(1)$ instead of $o_p(1)$, and thus the arguments we used in the proof of Proposition~\ref{prnm13} cannot be applied here directly. To overcome this problem, we need to expand each term in (\ref{prnm14-1-1}). We note that
\begin{eqnarray*}
\boldsymbol{\mathrm \uptheta}^{(1)}-\boldsymbol{\mathrm{\hat{\uptheta}}}^{(2)}&=&\boldsymbol{\mathrm Q}_{({\hat{\tau}}_1,{\hat{\tau}}_2)}^{-1}
\left(\boldsymbol{\mathrm Q}_{(\tau_1^0,{\tau}_2^0)}\left(\boldsymbol{\mathrm \uptheta}^{(1)}-\boldsymbol{\mathrm \uptheta}^{(2)}\right)+\boldsymbol{\mathrm Q}_{(\tau_2^0,{\hat{\tau}}_2)}\left(\boldsymbol{\mathrm \uptheta}^{(1)}-\boldsymbol{\mathrm \uptheta}^{(3)}\right)\right)+\sigma \boldsymbol{\mathrm Q}_{({\hat{\tau}}_1,{\hat{\tau}}_2)}^{-1} \boldsymbol{\mathrm r}_{({\hat{\tau}}_1,{\hat{\tau}}_2)}.
\end{eqnarray*}
So,
\begin{eqnarray*}
\sum_{t_i\in({\hat{\tau}}_1,{\tau}_2^0]}\left( \boldsymbol{\mathrm z}_i(\boldsymbol{\mathrm \uptheta}^{(1)}-\boldsymbol{\mathrm{\hat{\uptheta}}}^{(2)})\right)^2
=\left( \boldsymbol{\mathrm \uptheta}^{(1)}-\boldsymbol{\mathrm \uptheta}^{(2)} \right)^{\top}\boldsymbol{\mathrm Q}_{(\tau_1^0,{\tau}_2^0)}\boldsymbol{\mathrm Q}_{({\hat{\tau}}_1,{\hat{\tau}}_2)}^{-1}
\sum_{t_i\in({\hat{\tau}}_1,{\tau}_2^0]}\boldsymbol{\mathrm z}_i^{\top}\boldsymbol{\mathrm z}_i\boldsymbol{\mathrm Q}_{({\hat{\tau}}_1,{\hat{\tau}}_2)}^{-1}\boldsymbol{\mathrm Q}_{(\tau_1^0,{\tau}_2^0)}
\left(\boldsymbol{\mathrm \uptheta}^{(1)}-\boldsymbol{\mathrm \uptheta}^{(2)}\right)
\\+2\left(\boldsymbol{\mathrm \uptheta}^{(1)}-\boldsymbol{\mathrm \uptheta}^{(2)}\right)^{\top}\boldsymbol{\mathrm Q}_{(\tau_1^0,{\tau}_2^0)}\boldsymbol{\mathrm Q}_{({\hat{\tau}}_1,{\hat{\tau}}_2)}^{-1}\sum_{t_i\in({\hat{\tau}}_1,{\tau}_2^0]}
\boldsymbol{\mathrm z}_i^{\top}\boldsymbol{\mathrm z}_i\boldsymbol{\mathrm Q}_{({\hat{\tau}}_1,{\hat{\tau}}_2)}^{-1}\boldsymbol{\mathrm Q}_{(\tau_2^0,{\hat{\tau}}_2)}\left(\boldsymbol{\mathrm \uptheta}^{(1)}-\boldsymbol{\mathrm \uptheta}^{(3)}\right)
+2\left(\boldsymbol{\mathrm \uptheta}^{(1)}-\boldsymbol{\mathrm \uptheta}^{(3)}\right)^{\top}\boldsymbol{\mathrm Q}_{(\tau_2^0,{\hat{\tau}}_2)}\boldsymbol{\mathrm Q}_{({\hat{\tau}}_1,{\hat{\tau}}_2)}^{-1}\\ \times
\sum_{t_i\in({\hat{\tau}}_1,{\tau}_2^0]}\boldsymbol{\mathrm z}_i^{\top}\boldsymbol{\mathrm z}_i\boldsymbol{\mathrm Q}_{({\hat{\tau}}_1,{\hat{\tau}}_2)}^{-1}
\boldsymbol{\mathrm Q}_{(\tau_2^0,{\hat{\tau}}_2)}
\left(\boldsymbol{\mathrm \uptheta}^{(1)}-\boldsymbol{\mathrm \uptheta}^{(3)}\right)
+2\sigma \left(\boldsymbol{\mathrm \uptheta}^{(1)}-\boldsymbol{\mathrm \uptheta}^{(2)}\right)^{\top}\boldsymbol{\mathrm Q}_{(\tau_1^0,{\tau}_2^0)}\boldsymbol{\mathrm Q}_{({\hat{\tau}}_1,{\hat{\tau}}_2)}^{-1}
\sum_{t_i\in({\hat{\tau}}_1,{\tau}_2^0]}\boldsymbol{\mathrm z}_i^{\top}\boldsymbol{\mathrm z}_i\boldsymbol{\mathrm Q}_{({\hat{\tau}}_1,{\hat{\tau}}_2)}^{-1}\boldsymbol{\mathrm r}_{({\hat{\tau}}_1,{\hat{\tau}}_2)}
\\+2\sigma \left(\boldsymbol{\mathrm \uptheta}^{(1)}-\boldsymbol{\mathrm \uptheta}^{(3)}\right)\boldsymbol{\mathrm Q}_{(\tau_2^0,{\hat{\tau}}_2)}\boldsymbol{\mathrm Q}_{({\hat{\tau}}_1,{\hat{\tau}}_2)}^{-1}
 \sum_{t_i\in({\hat{\tau}}_1,{\tau}_2^0]}\boldsymbol{\mathrm z}_i^{\top}\boldsymbol{\mathrm z}_i \boldsymbol{\mathrm Q}_{({\hat{\tau}}_1,{\hat{\tau}}_2)}^{-1}\boldsymbol{\mathrm r}_{({\hat{\tau}}_1,{\hat{\tau}}_2)}
+\boldsymbol{\mathrm r}_{({\hat{\tau}}_1,{\hat{\tau}}_2)}^{\top}\boldsymbol{\mathrm Q}_{({\hat{\tau}}_1,{\hat{\tau}}_2)}^{-1}\sum_{t_i\in({\hat{\tau}}_1,{\tau}_2^0]}
\boldsymbol{\mathrm z}_i^{\top}\boldsymbol{\mathrm z}_i\boldsymbol{\mathrm Q}_{({\hat{\tau}}_1,{\hat{\tau}}_2)}^{-1}\boldsymbol{\mathrm r}_{({\hat{\tau}}_1,{\hat{\tau}}_2)},
\end{eqnarray*}
and
\begin{eqnarray*}
\sum_{t_i\in({\hat{\tau}}_1,{\tau}_2^0]}\left(\boldsymbol{\mathrm z}_i\left(\boldsymbol{\mathrm \uptheta}^{(1)}-\boldsymbol{\mathrm z}_i\boldsymbol{\mathrm{\hat{\uptheta}}}^{(2,0)}\right)\right)^2
=\left(\boldsymbol{\mathrm \uptheta}^{(1)}-\boldsymbol{\mathrm \uptheta}^{(2)}\right)^{\top}\boldsymbol{\mathrm Q}_{(\tau_1^0,{\tau}_2^0)}\boldsymbol{\mathrm Q}_{({\hat{\tau}}_1,\tau_2^0)}^{-1}
\sum_{t_i\in ({\hat{\tau}}_1,{\tau}_2^0]}\boldsymbol{\mathrm z}_i^{\top}\boldsymbol{\mathrm z}_i\boldsymbol{\mathrm Q}_{({\hat{\tau}}_1,\tau_2^0)}^{-1}
\boldsymbol{\mathrm Q}_{(\tau_1^0,{\tau}_2^0)}\left(\boldsymbol{\mathrm \uptheta}^{(1)}-\boldsymbol{\mathrm \uptheta}^{(2)}\right)\\
+2\sigma \left(\boldsymbol{\mathrm \uptheta}^{(1)}-\boldsymbol{\mathrm \uptheta}^{(2)}\right)^{\top} \boldsymbol{\mathrm Q}_{(\tau_1^0,{\tau}_2^0)}\boldsymbol{\mathrm Q}_{({\hat{\tau}}_1,\tau_2^0)}^{-1}\sum_{t_i\in({\hat{\tau}}_1,{\tau}_2^0]}
\boldsymbol{\mathrm z}_i^{\top}\boldsymbol{\mathrm z}_i\boldsymbol{\mathrm Q}_{({\hat{\tau}}_1,\tau_2^0)}^{-1}\boldsymbol{\mathrm r}_{({\hat{\tau}}_1,\tau_2^0)}
+\boldsymbol{\mathrm r}_{({\hat{\tau}}_1,\tau_2^0)}^{\top}\boldsymbol{\mathrm Q}_{({\hat{\tau}}_1,\tau_2^0)}^{-1}\sum_{t_i\in({\hat{\tau}}_1,{\tau}_2^0]}
\boldsymbol{\mathrm z}_i^{\top}\boldsymbol{\mathrm z}_i\boldsymbol{\mathrm Q}_{({\hat{\tau}}_1,\tau_2^0)}^{-1}\boldsymbol{\mathrm r}_{({\hat{\tau}}_1,\tau_2^0)}.
\end{eqnarray*}
By Tobing and McGlichrist (1992), we have \begin{equation}\label{TM-1}\boldsymbol{\mathrm Q}_{({\hat{\tau}}_1,{\hat{\tau}}_2)}^{-1}=\boldsymbol{\mathrm Q}_{({\hat{\tau}}_1,{\tau_2^0})}^{-1}+O_p\left(\frac{{\hat{\tau}}_2-\tau_2^0}{T^2}\right).\end{equation} Hence,
\begin{eqnarray}
&&\left(\boldsymbol{\mathrm \uptheta}^{(1)}-\boldsymbol{\mathrm \uptheta}^{(2)}\right)^{\top}\boldsymbol{\mathrm Q}_{(\tau_1^0,{\tau}_2^0)}
\boldsymbol{\mathrm Q}_{({\hat{\tau}}_1,{\hat{\tau}}_2)}^{-1}\sum_{t_i\in({\hat{\tau}}_1,{\tau}_2^0]}
\boldsymbol{\mathrm z}_i^{\top}\boldsymbol{\mathrm z}_i\boldsymbol{\mathrm Q}_{({\hat{\tau}}_1,{\hat{\tau}}_2)}^{-1}\boldsymbol{\mathrm Q}_{(\tau_1^0,{\tau}_2^0)}\left(\boldsymbol{\mathrm \uptheta}^{(1)}-\boldsymbol{\mathrm \uptheta}^{(2)}\right)
-\left(\boldsymbol{\mathrm \uptheta}^{(1)}-\boldsymbol{\mathrm \uptheta}^{(2)}\right)^{\top}\boldsymbol{\mathrm Q}_{(\tau_1^0,{\tau}_2^0)}\boldsymbol{\mathrm Q}_{({\hat{\tau}}_1,\tau_2^0)}^{-1}\nonumber\\ &&\times \sum_{t_i\in({\hat{\tau}}_1,{\tau}_2^0]}\boldsymbol{\mathrm z}_i'\boldsymbol{\mathrm z}_i\boldsymbol{\mathrm Q}_{({\hat{\tau}}_1,\tau_2^0)}^{-1}\boldsymbol{\mathrm Q}_{(\tau_1^0,{\tau}_2^0)}\left(\boldsymbol{\mathrm \uptheta}^{(1)}-\boldsymbol{\mathrm \uptheta}^{(2)}\right)
=2O_p\left(\frac{{\hat{\tau}}_2-\tau_2^0}{T^2}\right)\left(\boldsymbol{\mathrm \uptheta}^{(1)}-\boldsymbol{\mathrm \uptheta}^{(2)}\right)^{\top}
\boldsymbol{\mathrm Q}_{(\tau_1^0,{\tau}_2^0)}\boldsymbol{\mathrm Q}_{({\hat{\tau}}_1,\tau_2^0)}^{-1}
\sum_{t_i\in({\hat{\tau}}_1,{\tau}_2^0]}\boldsymbol{\mathrm z}_i^{\top}\boldsymbol{\mathrm z}_i\nonumber\\
&&\times\boldsymbol{\mathrm Q}_{(\tau_1^0,{\tau}_2^0)}\left(\boldsymbol{\mathrm \uptheta}^{(1)}-\boldsymbol{\mathrm \uptheta}^{(2)}\right)+O_p\left(\frac{({\hat{\tau}}_2-\tau_2^0)^2}{T^4}\right)\left(\boldsymbol{\mathrm \uptheta}^{(1)}-\boldsymbol{\mathrm \uptheta}^{(2)}\right)^{\top}
\boldsymbol{\mathrm Q}_{(\tau_1^0,{\tau}_2^0)}\sum_{t_i\in ({\hat{\tau}}_1,{\tau}_2^0]}\boldsymbol{\mathrm z}_i^{\top}\boldsymbol{\mathrm z}_i\boldsymbol{\mathrm Q}_{(\tau_1^0,{\tau}_2^0)}\left(\boldsymbol{\mathrm \uptheta}^{(1)}-\boldsymbol{\mathrm \uptheta}^{(2)}\right).\quad  \label{prnm14-2-1}
\end{eqnarray}
Using the asymptotic results in Proposition~\ref{prnm13}, we have that $\frac{\tau_1^0-{\hat{\tau}}_1}{T}\xrightarrow[T\rightarrow \infty]{p}0$. Therefore,
\begin{eqnarray*}
&&\frac{2}{{\hat{\tau}}_2-\tau_2^0}O_p\left(\frac{{\hat{\tau}}_2-\tau_2^0}{T^2}\right)\left(\boldsymbol{\mathrm \uptheta}^{(1)}-\boldsymbol{\mathrm \uptheta}^{(2)}\right)^{\top}
\boldsymbol{\mathrm Q}_{(\tau_1^0,{\tau}_2^0)}\boldsymbol{\mathrm Q}_{({\hat{\tau}}_1,\tau_2^0)}^{-1}\sum_{t_i\in({\hat{\tau}}_1,{\tau}_2^0]}\boldsymbol{\mathrm z}_i'\boldsymbol{\mathrm z}_i\boldsymbol{\mathrm Q}_{(\tau_1^0,{\tau}_2^0)}\left(\boldsymbol{\mathrm \uptheta}^{(1)}-\boldsymbol{\mathrm \uptheta}^{(2)}\right)
\leq 2O_p\left(\frac{{\hat{\tau}}_2-\tau_2^0}{T^2}\right)\\ &&\times\frac{(\tau_2^0-{\tau}_1^0)^2(\tau_1^0-{\hat{\tau}}_1)}{\left(\tau_2^0-{\hat{\tau}}_1\right)} \norm{\boldsymbol{\mathrm \uptheta}^{(1)}-\boldsymbol{\mathrm \uptheta}^{(2)}}^2
\norm{\frac{1}{(\tau_2^0-{\tau}_1^0)}\boldsymbol{\mathrm Q}_{(\tau_1^0,{\tau}_2^0)}}^2
\norm{{(\tau_2^0-{\hat{\tau}}_1)}\boldsymbol{\mathrm Q}_{({\hat{\tau}}_1,\tau_2^0)}^{-1}} \frac{1}{(\tau_1^0-{\hat{\tau}}_1)}\norm{\sum_{t_i\in({\hat{\tau}}_1,{\tau}_2^0]}\boldsymbol{\mathrm z}_i'\boldsymbol{\mathrm z}_i}=o_p(1).
\end{eqnarray*}
\ \\
\noindent Similarly,
\begin{eqnarray*}
&&\frac{1}{{\hat{\tau}}_2-\tau_2^0}O_p\left(\frac{({\hat{\tau}}_2-\tau_2^0)^2}{T^4}\right)
\left(\boldsymbol{\mathrm \uptheta}^{(1)}-\boldsymbol{\mathrm \uptheta}^{(2)}\right)^{\top}\boldsymbol{\mathrm Q}_{(\tau_1^0,{\tau}_2^0)}
\sum_{t_i\in({\hat{\tau}}_1,{\tau}_2^0]}\boldsymbol{\mathrm z}_i^{\top}\boldsymbol{\mathrm z}_i\boldsymbol{\mathrm Q}_{(\tau_1^0,{\tau}_2^0)}\left(\boldsymbol{\mathrm \uptheta}^{(1)}-\boldsymbol{\mathrm \uptheta}^{(2)}\right)=o_p(1).
\end{eqnarray*}
Thus, $(\ref{prnm14-2-1})/({\hat{\tau}}_2-\tau_2^0)=o_p(1)$. Furthermore, using the same argument above together with (\ref{convm_R2}), we have
\begin{eqnarray*}
&&\frac{2\sigma}{{\hat{\tau}}_2-\tau_2^0}\left(\boldsymbol{\mathrm \uptheta}^{(1)}-\boldsymbol{\mathrm \uptheta}^{(2)}\right)^{\top}\boldsymbol{\mathrm Q}_{(\tau_1^0,{\tau}_2^0)}\boldsymbol{\mathrm Q}_{({\hat{\tau}}_1,{\hat{\tau}}_2)}^{-1}
\sum_{t_i\in({\hat{\tau}}_1,{\tau}_2^0]}\boldsymbol{\mathrm z}_i^{\top}\boldsymbol{\mathrm z}_i\boldsymbol{\mathrm Q}_{({\hat{\tau}}_1,{\hat{\tau}}_2)}^{-1}
\boldsymbol{\mathrm r}_{({\hat{\tau}}_1,{\hat{\tau}}_2)}- \frac{2\sigma}{{\hat{\tau}}_2-\tau_2^0} \left(\boldsymbol{\mathrm \uptheta}^{(1)}-\boldsymbol{\mathrm \uptheta}^{(2)}\right)^{\top}\boldsymbol{\mathrm Q}_{(\tau_1^0,{\tau}_2^0)}\boldsymbol{\mathrm Q}_{({\hat{\tau}}_1,\tau_2^0)}^{-1}\\ &&\times \sum_{t_i\in({\hat{\tau}}_1,{\tau}_2^0]}\boldsymbol{\mathrm z}_i^{\top}\boldsymbol{\mathrm z}_i  \boldsymbol{\mathrm Q}_{({\hat{\tau}}_1,\tau_2^0)}^{-1}\boldsymbol{\mathrm r}_{({\hat{\tau}}_1,\tau_2^0)}
=O_p\left(\frac{1}{T^2}\right) 4\sigma \left(\boldsymbol{\mathrm \uptheta}^{(1)}-\boldsymbol{\mathrm \uptheta}^{(2)}\right)^{\top}\boldsymbol{\mathrm Q}_{(\tau_1^0,{\tau}_2^0)}\sum_{t_i\in({\hat{\tau}}_1,{\tau}_2^0]}\boldsymbol{\mathrm z}_i^{\top}\boldsymbol{\mathrm z}_i \boldsymbol{\mathrm Q}_{({\hat{\tau}}_1,\tau_2^0)}^{-1}\boldsymbol{\mathrm r}_{({\hat{\tau}}_1,\tau_2^0)}\\ &&
+O_p\left(\frac{({\hat{\tau}}_2-\tau_2^0)}{T^4}\right) 2\sigma \left(\boldsymbol{\mathrm \uptheta}^{(1)}-\boldsymbol{\mathrm \uptheta}^{(2)}\right)^{\top} \boldsymbol{\mathrm Q}_{(\tau_1^0,{\tau}_2^0)}\sum_{t_i\in({\hat{\tau}}_1,{\tau}_2^0]}\boldsymbol{\mathrm z}_i'\boldsymbol{\mathrm z}_i \boldsymbol{\mathrm r}_{({\hat{\tau}}_1,\tau_2^0)}=o_p(1).
\end{eqnarray*}
\ \\
\noindent With the aid of the identity $\boldsymbol{\mathrm r}_{({\hat{\tau}}_1,{\hat{\tau}}_2)}=\boldsymbol{\mathrm r}_{({\hat{\tau}}_1,{\tau}_2^0)}+\boldsymbol{\mathrm r}_{({\tau}_2^0,{\hat{\tau}}_2)}$,  (\ref{TM-1}) and  the Cauchy-Schwarz inequality, we get
\begin{eqnarray*}
\frac{\sigma^2}{{\hat{\tau}}_2-\tau_2^0}\boldsymbol{\mathrm r}_{({\hat{\tau}}_1,{\hat{\tau}}_2)}^{\top}
\boldsymbol{\mathrm Q}_{({\hat{\tau}}_1,{\hat{\tau}}_2)}^{-1}\sum_{t_i\in({\hat{\tau}}_1,{\tau}_2^0]}
\boldsymbol{\mathrm z}_i^{\top}\boldsymbol{\mathrm z}_i\boldsymbol{\mathrm Q}_{({\hat{\tau}}_1,{\hat{\tau}}_2)}^{-1}\boldsymbol{\mathrm r}_{({\hat{\tau}}_1,{\hat{\tau}}_2)}
-\frac{\sigma^2}{{\hat{\tau}}_2-\tau_2^0} \boldsymbol{\mathrm r}_{({\hat{\tau}}_1,\tau_2^0)}^{\top}\boldsymbol{\mathrm Q}_{({\hat{\tau}}_1,\tau_2^0)}^{-1}
\sum_{t_i\in({\hat{\tau}}_1,{\tau}_2^0]}\boldsymbol{\mathrm z}_i^{\top}\boldsymbol{\mathrm z}_i\boldsymbol{\mathrm Q}_{({\hat{\tau}}_1,\tau_2^0)}^{-1}
\boldsymbol{\mathrm r}_{({\hat{\tau}}_1,\tau_2^0)}=o_p(1)\end{eqnarray*}
and
\begin{eqnarray*}
&&\frac{1}{{\hat{\tau}}_2-\tau_2^0}2\left(\boldsymbol{\mathrm \uptheta}^{(1)}-\boldsymbol{\mathrm \uptheta}^{(3)}\right)^{\top}
\boldsymbol{\mathrm Q}_{(\tau_2^0,{\hat{\tau}}_2)}\boldsymbol{\mathrm Q}_{({\hat{\tau}}_1,{\hat{\tau}}_2)}^{-1}
\sum_{t_i\in({\hat{\tau}}_1,{\tau}_2^0]}\boldsymbol{\mathrm z}_i^{\top}\boldsymbol{\mathrm z}_i\boldsymbol{\mathrm Q}_{({\hat{\tau}}_1,{\hat{\tau}}_2)}^{-1}
\boldsymbol{\mathrm Q}_{(\tau_2^0,{\hat{\tau}}_2)}\left(\boldsymbol{\mathrm \uptheta}^{(1)}-\boldsymbol{\mathrm \uptheta}^{(3)}\right)=o_p(1)\\
&&\frac{2\sigma}{{\hat{\tau}}_2-\tau_2^0}\left(\boldsymbol{\mathrm \uptheta}^{(1)}-\boldsymbol{\mathrm \uptheta}^{(3)}\right)
\boldsymbol{\mathrm Q}_{(\tau_2^0,{\hat{\tau}}_2)}\boldsymbol{\mathrm Q}_{({\hat{\tau}}_1,{\hat{\tau}}_2)}^{-1}
\sum_{t_i\in({\hat{\tau}}_1,{\tau}_2^0]}\boldsymbol{\mathrm z}_i^{\top}\boldsymbol{\mathrm z}_i\boldsymbol{\mathrm Q}_{({\hat{\tau}}_1,{\hat{\tau}}_2)}^{-1}
\boldsymbol{\mathrm r}_{({\hat{\tau}}_1,{\hat{\tau}}_2)}=o_p(1),
\end{eqnarray*}
Combining the above results, we have
$$\frac{\sum_{t_i\in({\hat{\tau}}_1,{\tau}_2^0]}\left(\boldsymbol{\mathrm z}_i\left(\boldsymbol{\mathrm \uptheta}^{(1)}-\boldsymbol{\mathrm{\hat{\uptheta}}}^{(2)}\right)\right)^2
-\sum_{t_i\in({\hat{\tau}}_1,{\tau}_2^0]}\left(\boldsymbol{\mathrm z}_i\left(\boldsymbol{\mathrm \uptheta}^{(1)}-\boldsymbol{\mathrm{\hat{\uptheta}}}^{(2,0)}\right)\right)^2}{{\hat{\tau}}_2
-\tau_2^0}=o_p(1).$$
Similarly,
\begin{eqnarray}
&&\sum_{t_i\in ({\hat{\tau}}_2,{\hat{\tau}}_3]}(\boldsymbol{\mathrm z}_i(\boldsymbol{\mathrm \uptheta}^{(3)}-\boldsymbol{\mathrm{\hat{\uptheta}}}^{(3)}))^2
-\sum_{t_i\in({\tau}_2^0,{\hat{\tau}}_3]}(\boldsymbol{\mathrm z}_i(\boldsymbol{\mathrm \uptheta}^{(3)}-\boldsymbol{\mathrm{\hat{\uptheta}}}^{(3,0)}))^2
=\sigma^2 \boldsymbol{\mathrm r}_{({\hat{\tau}}_2,{\hat{\tau}}_3)}^{\top}\boldsymbol{\mathrm Q}_{({\hat{\tau}}_2,{\hat{\tau}}_3)}^{-1}\sum_{t_i\in ({\hat{\tau}}_2,{\hat{\tau}}_3]}\boldsymbol{\mathrm z}_i^{\top}\boldsymbol{\mathrm z}_i\boldsymbol{\mathrm Q}_{({\hat{\tau}}_2,{\hat{\tau}}_3)}^{-1} \boldsymbol{\mathrm r}_{({\hat{\tau}}_2,{\hat{\tau}}_3)}\nonumber\\ &&
-\sigma^2 \boldsymbol{\mathrm r}_{({\tau}_2^0,{\hat{\tau}}_3)}^{\top} \boldsymbol{\mathrm Q}_{(\tau_2^0,{\hat{\tau}}_3)}^{-1}\sum_{t_i\in ({\tau}_2^0,{\hat{\tau}}_3]}\boldsymbol{\mathrm z}_i^{\top}\boldsymbol{\mathrm z}_i\boldsymbol{\mathrm Q}_{(\tau_2^0,{\hat{\tau}}_3)}^{-1} \boldsymbol{\mathrm r}_{(\tau_2^0,{\hat{\tau}}_3)}.\label{prnm14-3-1}
\end{eqnarray}
\ \\
\noindent Note that $\boldsymbol{\mathrm r}_{({\tau}_2^0,{\hat{\tau}}_3)}=\boldsymbol{\mathrm r}_{({\tau}_2^0,{\hat{\tau}}_2)}+\boldsymbol{\mathrm r}_{({\hat{\tau}}_2,{\hat{\tau}}_3)}.$ So, from (\ref{convm-R1}),  $\boldsymbol{\mathrm r}_{({\tau}_2^0,{\hat{\tau}}_2)}/\sqrt{{\hat{\tau}}_2-\tau_2^0}=O_p\left( ({\hat{\tau}}_2-\tau_2^0)^{a^*} \right)$ for some $0<a^*<1/2$. Moreover, $\sum_{t_i\in({\tau}_2^0,{\hat{\tau}}_3]}\boldsymbol{\mathrm z}_i^{\top}\boldsymbol{\mathrm z}_i=\sum_{t_i\in(\tau_2^0,{\hat{\tau}}_2]}\boldsymbol{\mathrm z}_i^{\top}\boldsymbol{\mathrm z}_i
+\sum_{t_i\in({\hat{\tau}}_2,{\hat{\tau}}_3]}\boldsymbol{\mathrm z}_i^{\top}\boldsymbol{\mathrm z}_i$, and by Tobing and McGlichrist (1992), $\boldsymbol{\mathrm Q}_{(\tau_2^0,{\hat{\tau}}_3)}^{-1}=\boldsymbol{\mathrm Q}_{({\hat{\tau}}_2,{\hat{\tau}}_3)}^{-1}
+O_p\left(\frac{{\hat{\tau}}_2-\tau_2^0}{T^2}\right)$. Again, when these results are combined,
one may verify that, with a suitable choice of $C$, for large $T$ the order of $\displaystyle \frac{(\ref{prnm14-3-1})}{({\hat{\tau}}_2-\tau_2^0)}$ is $o_p(1)$. \\
\ \\
\noindent The above procedure can also be applied to investigate the behaviour of (\ref{prnm14-1-3}). For example,
\begin{eqnarray*}
&&\sum_{t_i\in({\hat{\tau}}_1,{\tau}_2^0]}2u_i^{\top}\boldsymbol{\mathrm z}_i\left(\boldsymbol{\mathrm \uptheta}^{(2,0)}-\boldsymbol{\mathrm z}_i\boldsymbol{\mathrm{\hat{\uptheta}}}^{(2)}\right)
=\sum_{t_i\in({\hat{\tau}}_1,{\tau}_2^0]}2u_i^{\top}\boldsymbol{\mathrm z}_i\boldsymbol{\mathrm Q}_{({\hat{\tau}}_1,{\hat{\tau}}_2)}^{-1}
\boldsymbol{\mathrm Q}_{(\tau_1^0,{\tau}_2^0)}\left(\boldsymbol{\mathrm \uptheta}^{(1)}-\boldsymbol{\mathrm \uptheta}^{(2)}\right)\\
&&-\sum_{t_i\in({\hat{\tau}}_1,{\tau}_2^0]}2u_i^{\top}\boldsymbol{\mathrm z}_i\boldsymbol{\mathrm Q}_{({\hat{\tau}}_1,{\tau}_2^0)}^{-1}
\boldsymbol{\mathrm Q}_{(\tau_1^0,{\tau}_2^0)}\left(\boldsymbol{\mathrm \uptheta}^{(1)}-\boldsymbol{\mathrm \uptheta}^{(2)}\right)
+\sum_{t_i\in({\hat{\tau}}_1,{\tau}_2^0]}2u_i^{\top}\boldsymbol{\mathrm z}_i\boldsymbol{\mathrm Q}_{({\hat{\tau}}_1,{\hat{\tau}}_2)}^{-1}
\boldsymbol{\mathrm Q}_{(\tau_2^0,{\hat{\tau}}_2)}\left(\boldsymbol{\mathrm \uptheta}^{(1)}-\boldsymbol{\mathrm \uptheta}^{(3)}\right)\\&&
+2\sigma\sum_{t_i\in({\hat{\tau}}_1,{\tau}_2^0]}u_i^{\top}\boldsymbol{\mathrm z}_i \boldsymbol{\mathrm Q}_{({\hat{\tau}}_1,{\hat{\tau}}_2)}^{-1} \boldsymbol{\mathrm r}_{({\hat{\tau}}_1,{\hat{\tau}}_2)}
-2\sigma\sum_{t_i\in({\hat{\tau}}_1,{\tau}_2^0]}u_i^{\top}\boldsymbol{\mathrm z}_i \boldsymbol{\mathrm Q}_{({\hat{\tau}}_1,{\tau}_2^0)}^{-1} \boldsymbol{\mathrm r}_{({\hat{\tau}}_1,{\tau}_2^0)},
\end{eqnarray*}
with
\begin{eqnarray*}
\frac{2}{{\hat{\tau}}_2-\tau_2^0}\left(\sum_{t_i\in({\hat{\tau}}_1,{\tau}_2^0]}u_i^{\top}
\boldsymbol{\mathrm z}_i\boldsymbol{\mathrm Q}_{({\hat{\tau}}_1,{\hat{\tau}}_2)}^{-1}\boldsymbol{\mathrm Q}_{(\tau_1^0,{\tau}_2^0)}
\left(\boldsymbol{\mathrm \uptheta}^{(1)}-\boldsymbol{\mathrm \uptheta}^{(2)}\right)-\sum_{t_i\in({\hat{\tau}}_1,{\tau}_2^0]}u_i^{\top}\boldsymbol{\mathrm z}_i
\boldsymbol{\mathrm Q}_{({\hat{\tau}}_1,{\tau}_2^0)}^{-1}\boldsymbol{\mathrm Q}_{(\tau_1^0,{\tau}_2^0)}\left(\boldsymbol{\mathrm \uptheta}^{(1)}-\boldsymbol{\mathrm \uptheta}^{(2)}\right)\right)\\=
O_p\left( \frac{1}{T^2}\right) \sum_{t_i\in({\hat{\tau}}_1,{\tau}_2^0]}u_i'\boldsymbol{\mathrm z}_i \boldsymbol{\mathrm Q}_{(\tau_1^0,{\tau}_2^0)}\left(\boldsymbol{\mathrm \uptheta}^{(1)}-\boldsymbol{\mathrm \uptheta}^{(2)}\right)
=o_p(1),\end{eqnarray*}
\begin{eqnarray*}
\frac{2}{{\hat{\tau}}_2-\tau_2^0}\sum_{t_i\in({\hat{\tau}}_1,{\tau}_2^0]}
u_i^{\top}\boldsymbol{\mathrm z}_i\boldsymbol{\mathrm Q}_{({\hat{\tau}}_1,{\hat{\tau}}_2)}^{-1}\boldsymbol{\mathrm Q}_{(\tau_2^0,{\hat{\tau}}_2)}\left(\boldsymbol{\mathrm \uptheta}^{(1)}-\boldsymbol{\mathrm \uptheta}^{(3)}\right)\leq
\norm{\frac{1}{T}\sum_{t_i\in({\hat{\tau}}_1,{\tau}_2^0]}u_i^{\top}\boldsymbol{\mathrm z}_i }
\norm{T\boldsymbol{\mathrm Q}_{({\hat{\tau}}_1,{\hat{\tau}}_2)}^{-1} }
\norm{\frac{2}{{\hat{\tau}}_2-\tau_2^0}\boldsymbol{\mathrm Q}_{(\tau_2^0,{\hat{\tau}}_2)} }\\ \times
\norm{ \left(\boldsymbol{\mathrm \uptheta}^{(1)}-\boldsymbol{\mathrm \uptheta}^{(3)}\right) }=o_p(1)\end{eqnarray*}
\noindent and
\begin{eqnarray*}
\frac{2\sigma}{{\hat{\tau}}_2-\tau_2^0}\left(\sum_{t_i\in({\hat{\tau}}_1,{\tau}_2^0]}u_i^{\top}\boldsymbol{\mathrm z}_i \boldsymbol{\mathrm Q}_{({\hat{\tau}}_1,{\hat{\tau}}_2)}^{-1} \boldsymbol{\mathrm r}_{({\hat{\tau}}_1,{\hat{\tau}}_2)}
-\sum_{t_i\in({\hat{\tau}}_1,{\tau}_2^0]}u_i^{\top}\boldsymbol{\mathrm z}_i \boldsymbol{\mathrm Q}_{({\hat{\tau}}_1,{\tau}_2^0)}^{-1} \boldsymbol{\mathrm r}_{({\hat{\tau}}_1,{\tau}_2^0)}\right)
=\frac{2\sigma}{{\hat{\tau}}_2-\tau_2^0}\sum_{t_i\in({\hat{\tau}}_1,{\tau}_2^0]}u_i^{\top}\boldsymbol{\mathrm z}_i \boldsymbol{\mathrm Q}_{({\hat{\tau}}_1,{\tau}_2^0)}^{-1} \boldsymbol{\mathrm r}_{({\tau}_2^0,{\hat{\tau}}_2)}\\+
2\sigma O_p(\frac{1}{T^2})\sum_{t_i\in({\hat{\tau}}_1,{\tau}_2^0]}u_i^{\top}\boldsymbol{\mathrm z}_i  \boldsymbol{\mathrm r}_{({\tau}_2^0,{\hat{\tau}}_2)}=o_p(1).
\ \\
\noindent \end{eqnarray*}
\ \\
\noindent Hence, for large $T$,
\begin{eqnarray*}
&&\sum_{t_i\in({\hat{\tau}}_1,{\tau}_2^0]}2u_i^{\top}
\boldsymbol{\mathrm z}_i\left(\boldsymbol{\mathrm{\hat{\uptheta}}}^{(2,0)}-\boldsymbol{\mathrm z}_i\boldsymbol{\mathrm{\hat{\uptheta}}}^{(2)}\right)/({{\hat{\tau}}_2-\tau_2^0})=o_p(1).
\end{eqnarray*}
\ \\
\noindent
In an analogous manner, one can show that, with a suitable choice of $C$, for large $T$,
\begin{eqnarray*}
&&\frac{\sum_{t_i\in({\tau}_2^0,{\hat{\tau}}_2]}2u_i^{\top}\boldsymbol{\mathrm z}_i\left(\boldsymbol{\mathrm{\hat{\uptheta}}}^{(3,0)}-\boldsymbol{\mathrm z}_i\boldsymbol{\mathrm{\hat{\uptheta}}}^{(2)}\right)
{{\hat{\tau}}_2-\tau_2^0}
+\sum_{t_i\in({\hat{\tau}}_2,{\hat{\tau}}_3]}2u_i^{\top}\boldsymbol{\mathrm z}_i\left(\boldsymbol{\mathrm{\hat{\uptheta}}}^{(3,0)}-\boldsymbol{\mathrm{\hat{\uptheta}}}^{(3)}\right)}{
{{\hat{\tau}}_2-\tau_2^0}}=o_p(1).\end{eqnarray*}
\ \\
\noindent Finally, we investigate the behaviour of the remaining term.

\begin{eqnarray*}
\frac{\sum_{t_i\in(\tau_2^0,{\hat{\tau}}_2]}\left(\boldsymbol{\mathrm z}_i\left(\boldsymbol{\mathrm \uptheta}^{(3)}-\boldsymbol{\mathrm{\hat{\uptheta}}}^{(2)}\right)\right)^2}{{\hat{\tau}}_2-\tau_2^0}
=\left(\boldsymbol{\mathrm \uptheta}^{(2)}-\boldsymbol{\mathrm \uptheta}^{(3)}\right)^{\top}\boldsymbol{\mathrm Q}_{({\tau}_1^0,\tau_2^0)}
\boldsymbol{\mathrm Q}_{({\hat{\tau}}_1,{\hat{\tau}}_2)}^{-1}\frac{1}{{\hat{\tau}}_2
-\tau_2^0}\sum_{t_i\in({\tau}_2^0,{\hat{\tau}}_2]}\boldsymbol{\mathrm z}_i^{\top}\boldsymbol{\mathrm z}_i\boldsymbol{\mathrm Q}_{({\hat{\tau}}_1,{\hat{\tau}}_2)}^{-1}
\boldsymbol{\mathrm Q}_{({\tau}_1^0,{\tau}_2^0)}
\left(\boldsymbol{\mathrm \uptheta}^{(2)}-\boldsymbol{\mathrm \uptheta}^{(3)}\right)\\
+\left(\boldsymbol{\mathrm \uptheta}^{(1)}-\boldsymbol{\mathrm \uptheta}^{(3)}\right)^{\top}\boldsymbol{\mathrm Q}_{({\hat{\tau}}_1,{\tau}_1^0)}\boldsymbol{\mathrm Q}_{({\hat{\tau}}_1,{\hat{\tau}}_2)}^{-1}
\frac{1}{{\hat{\tau}}_2-\tau_2^0}\sum_{t_i\in({\tau}_2^0,{\hat{\tau}}_2]}\boldsymbol{\mathrm z}_i^{\top}\boldsymbol{\mathrm z}_i
\boldsymbol{\mathrm Q}_{({\hat{\tau}}_1,{\hat{\tau}}_2)}^{-1}\boldsymbol{\mathrm Q}_{({\hat{\tau}}_1,{\tau}_1^0)}\left(\boldsymbol{\mathrm \uptheta}^{(1)}-\boldsymbol{\mathrm \uptheta}^{(3)}\right)
+2\left(\boldsymbol{\mathrm \uptheta}^{(1)}-\boldsymbol{\mathrm \uptheta}^{(3)}\right)^{\top}\boldsymbol{\mathrm Q}_{({\hat{\tau}}_1,{\tau}_1^0)}\\
\times\boldsymbol{\mathrm Q}_{({\hat{\tau}}_1,{\hat{\tau}}_2)}^{-1}\frac{1}{{\hat{\tau}}_2-\tau_2^0}\sum_{t_i\in(\tau_2^0,{\hat{\tau}}_2]}\boldsymbol{\mathrm z}_i^{\top}\boldsymbol{\mathrm z}_i
\boldsymbol{\mathrm Q}_{({\hat{\tau}}_1,{\hat{\tau}}_2)}^{-1}\boldsymbol{\mathrm Q}_{({\tau}_1^0,{\tau}_2^0)}\left(\boldsymbol{\mathrm \uptheta}^{(2)}-\boldsymbol{\mathrm \uptheta}^{(3)}\right)
+2\left(\boldsymbol{\mathrm \uptheta}^{(1)}-\boldsymbol{\mathrm \uptheta}^{(3)}\right)^{\top}\boldsymbol{\mathrm Q}_{({\hat{\tau}}_1,{\tau}_1^0)}
\boldsymbol{\mathrm Q}_{({\hat{\tau}}_1,{\hat{\tau}}_2)}^{-1}\frac{1}{{\hat{\tau}}_2-\tau_2^0}
\sum_{t_i\in(\tau_2^0,{\hat{\tau}}_2]}\boldsymbol{\mathrm z}_i^{\top}\boldsymbol{\mathrm z}_i\\ \times\boldsymbol{\mathrm Q}_{({\hat{\tau}}_1,{\hat{\tau}}_2)}^{-1}
\boldsymbol{\mathrm r}_{({\hat{\tau}}_1,{\hat{\tau}}_2)}
+2\left(\boldsymbol{\mathrm \uptheta}^{(2)}-\boldsymbol{\mathrm \uptheta}^{(3)}\right)^{\top}\boldsymbol{\mathrm Q}_{({\tau}_1^0,{\tau}_2^0)}
\boldsymbol{\mathrm Q}_{({\hat{\tau}}_1,{\hat{\tau}}_2)}^{-1}\frac{1}{{\hat{\tau}}_2-\tau_2^0}
\sum_{t_i\in(\tau_2^0,{\hat{\tau}}_2]}\boldsymbol{\mathrm z}_i^{\top}\boldsymbol{\mathrm z}_i\boldsymbol{\mathrm Q}_{({\hat{\tau}}_1,{\hat{\tau}}_2)}^{-1}\boldsymbol{\mathrm r}_{({\hat{\tau}}_1,{\hat{\tau}}_2)}
+\boldsymbol{\mathrm r}_{({\hat{\tau}}_1,{\hat{\tau}}_2)}^{\top}\boldsymbol{\mathrm Q}_{({\hat{\tau}}_1,{\hat{\tau}}_2)}^{-1}
\sum_{t_i\in(\tau_2^0,{\hat{\tau}}_2]}\boldsymbol{\mathrm z}_i'\boldsymbol{\mathrm z}_i\\ \times\boldsymbol{\mathrm Q}_{({\hat{\tau}}_1,{\hat{\tau}}_2)}^{-1}\boldsymbol{\mathrm r}_{({\hat{\tau}}_1,{\hat{\tau}}_2)}.
\end{eqnarray*}
\ \\
\noindent By Proposition~\ref{prnm13}, $s_1^0-\hat{s}_1\xrightarrow[T \rightarrow\infty]{p} 0$. Hence, applying again the Cauchy-Schwarz inequality and similar reasoning as before, we have
\begin{eqnarray*}
&&\left(\boldsymbol{\mathrm \uptheta}^{(1)}-\boldsymbol{\mathrm \uptheta}^{(3)}\right)^{\top}\boldsymbol{\mathrm Q}_{({\hat{\tau}}_1,{\tau}_1^0)}
\boldsymbol{\mathrm Q}_{({\hat{\tau}}_1,{\hat{\tau}}_2)}^{-1}\frac{1}{{\hat{\tau}}_2-\tau_2^0}
\sum_{t_i\in(\tau_2^0,{\hat{\tau}}_2]}\boldsymbol{\mathrm z}_i^{\top}\boldsymbol{\mathrm z}_i\boldsymbol{\mathrm Q}_{({\hat{\tau}}_1,{\hat{\tau}}_2)}^{-1}
\boldsymbol{\mathrm Q}_{({\hat{\tau}}_1,{\tau}_1^0)}\left(\boldsymbol{\mathrm \uptheta}^{(1)}-\boldsymbol{\mathrm \uptheta}^{(3)}\right)
\leq (s_1^0-\hat{s}_1)^2
\norm{ \boldsymbol{\mathrm \uptheta}^{(1)}-\boldsymbol{\mathrm \uptheta}^{(3)}} \\ &&\times
 \norm{T\boldsymbol{\mathrm Q}_{({\hat{\tau}}_1,{\hat{\tau}}_2)}^{-1}}^2
 \norm{ \frac{1}{(s_1^0-\hat{s}_1)T}\boldsymbol{\mathrm Q}_{({\hat{\tau}}_1,{\tau}_1^0)}}^2
 \norm{ \frac{1}{{\hat{\tau}}_2-\tau_2^0}\sum_{t_i\in(\tau_2^0,{\hat{\tau}}_2]}\boldsymbol{\mathrm z}_i^{\top}\boldsymbol{\mathrm z}_i } \norm{\boldsymbol{\mathrm \uptheta}^{(1)}-\boldsymbol{\mathrm \uptheta}^{(3)}}=o_p(1).
\end{eqnarray*}
\ \\
Also,
\begin{eqnarray*}
&&2\left(\boldsymbol{\mathrm \uptheta}^{(1)}-\boldsymbol{\mathrm \uptheta}^{(3)}\right)^{\top}\boldsymbol{\mathrm Q}_{({\hat{\tau}}_1,{\tau}_1^0)}
\boldsymbol{\mathrm Q}_{({\hat{\tau}}_1,{\hat{\tau}}_2)}^{-1}\frac{1}{{\hat{\tau}}_2-\tau_2^0}
\sum_{t_i\in(\tau_2^0,{\hat{\tau}}_2]}\boldsymbol{\mathrm z}_i^{\top}\boldsymbol{\mathrm z}_i\boldsymbol{\mathrm Q}_{({\hat{\tau}}_1,{\hat{\tau}}_2)}^{-1}
\boldsymbol{\mathrm Q}_{({\tau}_1^0,{\tau}_2^0)}\left(\boldsymbol{\mathrm \uptheta}^{(2)}-\boldsymbol{\mathrm \uptheta}^{(3)}\right)=o_p(1),\\
&&2\left(\boldsymbol{\mathrm \uptheta}^{(1)}-\boldsymbol{\mathrm \uptheta}^{(3)}\right)^{\top}\boldsymbol{\mathrm Q}_{({\hat{\tau}}_1,{\tau}_1^0)}
\boldsymbol{\mathrm Q}_{({\hat{\tau}}_1,{\hat{\tau}}_2)}^{-1}\frac{1}{{\hat{\tau}}_2-\tau_2^0}
\sum_{t_i\in(\tau_2^0,{\hat{\tau}}_2]}\boldsymbol{\mathrm z}_i^{\top}\boldsymbol{\mathrm z}_i\boldsymbol{\mathrm Q}_{({\hat{\tau}}_1,{\hat{\tau}}_2)}^{-1}
\boldsymbol{\mathrm r}_{({\hat{\tau}}_1,{\hat{\tau}}_2)}=o_p(1),\\
&&2\left(\boldsymbol{\mathrm \uptheta}^{(2)}-\boldsymbol{\mathrm \uptheta}^{(3)}\right)^{\top}\boldsymbol{\mathrm Q}_{({\tau}_1^0,{\tau}_2^0)}
\boldsymbol{\mathrm Q}_{({\hat{\tau}}_1,{\hat{\tau}}_2)}^{-1}\frac{1}{{\hat{\tau}}_2-\tau_2^0}
\sum_{t_i\in(\tau_2^0,{\hat{\tau}}_2]}\boldsymbol{\mathrm z}_i^{\top}\boldsymbol{\mathrm z}_i\boldsymbol{\mathrm Q}_{({\hat{\tau}}_1,{\hat{\tau}}_2)}^{-1}
\boldsymbol{\mathrm r}_{({\hat{\tau}}_1,{\hat{\tau}}_2)}=o_p(1),
\end{eqnarray*}
and
\begin{eqnarray*}
&&\boldsymbol{\mathrm r}_{({\hat{\tau}}_1,{\hat{\tau}}_2)}^{\top}\boldsymbol{\mathrm Q}_{({\hat{\tau}}_1,{\hat{\tau}}_2)}^{-1}
\frac{1}{{\hat{\tau}}_2-\tau_2^0}\sum_{t_i\in(\tau_2^0,{\hat{\tau}}_2]}
\boldsymbol{\mathrm z}_i^{\top}\boldsymbol{\mathrm z}_i\boldsymbol{\mathrm Q}_{({\hat{\tau}}_1,{\hat{\tau}}_2)}^{-1}\boldsymbol{\mathrm r}_{({\hat{\tau}}_1,{\hat{\tau}}_2)}=o_p(1).
\end{eqnarray*}
\ \\
\noindent Moreover, it follows from the identity $\boldsymbol{\mathrm Q}_{({\tau}_1^0,{\tau}_2^0)}=\boldsymbol{\mathrm Q}_{({\hat{\tau}}_1,{\hat{\tau}}_2)}-\boldsymbol{\mathrm Q}_{({\hat{\tau}}_1,{\tau}_1^0)}-\boldsymbol{\mathrm Q}_{({\tau}_2^0,{\hat{\tau}}_2)}$  that
\begin{eqnarray*}
\left(\boldsymbol{\mathrm \uptheta}^{(2)}-\boldsymbol{\mathrm \uptheta}^{(3)}\right)^{\top}\boldsymbol{\mathrm Q}_{({\tau}_1^0,{\tau}_2^0)}
\boldsymbol{\mathrm Q}_{({\hat{\tau}}_1,{\hat{\tau}}_2)}^{-1}\frac{1}{{\hat{\tau}}_2-\tau_2^0}
\sum_{t_i\in(\tau_2^0,{\hat{\tau}}_2]}\boldsymbol{\mathrm z}_i^{\top}\boldsymbol{\mathrm z}_i\boldsymbol{\mathrm Q}_{({\hat{\tau}}_1,{\hat{\tau}}_2)}^{-1}
\boldsymbol{\mathrm Q}_{({\tau}_1^0,{\tau}_2^0)}\left(\boldsymbol{\mathrm \uptheta}^{(2)}-\boldsymbol{\mathrm \uptheta}^{(3)}\right)
=\left(\boldsymbol{\mathrm \uptheta}^{(2)}-\boldsymbol{\mathrm \uptheta}^{(3)}\right)^{\top} \frac{1}{{\hat{\tau}}_2-\tau_2^0}
\\
 \times\sum_{t_i\in(\tau_2^0,{\hat{\tau}}_2]}\boldsymbol{\mathrm z}_i^{\top}\boldsymbol{\mathrm z}_i\left(\boldsymbol{\mathrm \uptheta}^{(2)}-\boldsymbol{\mathrm \uptheta}^{(3)}\right)+\left(\boldsymbol{\mathrm \uptheta}^{(2)}-\boldsymbol{\mathrm \uptheta}^{(3)}\right)^{\top}
 \boldsymbol{\mathrm Q}_{({\hat{\tau}}_1,{\tau}_1^0)}\boldsymbol{\mathrm Q}_{({\hat{\tau}}_1,{\hat{\tau}}_2)}^{-1}
 \frac{1}{{\hat{\tau}}_2-\tau_2^0}\sum_{t_i\in(\tau_2^0,{\hat{\tau}}_2]}
 \boldsymbol{\mathrm z}_i^{\top}\boldsymbol{\mathrm z}_i\boldsymbol{\mathrm Q}_{({\hat{\tau}}_1,{\hat{\tau}}_2)}^{-1}\boldsymbol{\mathrm Q}_{({\hat{\tau}}_1,{\tau}_1^0)}\left(\boldsymbol{\mathrm \uptheta}^{(2)}-\boldsymbol{\mathrm \uptheta}^{(3)}\right)\\
+\left(\boldsymbol{\mathrm \uptheta}^{(2)}-\boldsymbol{\mathrm \uptheta}^{(3)}\right)^{\top}
\boldsymbol{\mathrm Q}_{({\tau}_2^0,{\hat{\tau}}_2)} \boldsymbol{\mathrm Q}_{({\hat{\tau}}_1,{\hat{\tau}}_2)}^{-1} \frac{1}{{\hat{\tau}}_2-\tau_2^0}\sum_{t_i\in(\tau_2^0,{\hat{\tau}}_2]}
\boldsymbol{\mathrm z}_i^{\top}\boldsymbol{\mathrm z}_i\boldsymbol{\mathrm Q}_{({\hat{\tau}}_1,{\hat{\tau}}_2)}^{-1}\boldsymbol{\mathrm Q}_{({\tau}_2^0,{\hat{\tau}}_2)}
\left(\boldsymbol{\mathrm \uptheta}^{(2)}-\boldsymbol{\mathrm \uptheta}^{(3)}\right)
+2\left(\boldsymbol{\mathrm \uptheta}^{(2)}-\boldsymbol{\mathrm \uptheta}^{(3)}\right)^{\top}\boldsymbol{\mathrm Q}_{({\hat{\tau}}_1,{\tau}_1^0)}
\end{eqnarray*}
\begin{eqnarray*}
 \times\boldsymbol{\mathrm Q}_{({\hat{\tau}}_1,{\hat{\tau}}_2)}^{-1} \frac{1}{{\hat{\tau}}_2-\tau_2^0}\sum_{t_i\in(\tau_2^0,{\hat{\tau}}_2]}
\boldsymbol{\mathrm z}_i^{\top}\boldsymbol{\mathrm z}_i\boldsymbol{\mathrm Q}_{({\hat{\tau}}_1,{\hat{\tau}}_2)}^{-1}\boldsymbol{\mathrm Q}_{({\tau}_2^0,{\hat{\tau}}_2)}\left(\boldsymbol{\mathrm \uptheta}^{(2)}-\boldsymbol{\mathrm \uptheta}^{(3)}\right)
-2\left(\boldsymbol{\mathrm \uptheta}^{(2)}-\boldsymbol{\mathrm \uptheta}^{(3)}\right)^{\top}\boldsymbol{\mathrm Q}_{({\tau}_2^0,{\hat{\tau}}_2)}
\boldsymbol{\mathrm Q}_{({\hat{\tau}}_1,{\hat{\tau}}_2)}^{-1}\frac{1}{{\hat{\tau}}_2-\tau_2^0}
\sum_{t_i\in(\tau_2^0,{\hat{\tau}}_2]}\boldsymbol{\mathrm z}_i^{\top}\boldsymbol{\mathrm z}_i\\ \times \left(\boldsymbol{\mathrm \uptheta}^{(2)}-\boldsymbol{\mathrm \uptheta}^{(3)}\right)
-2\left(\boldsymbol{\mathrm \uptheta}^{(2)}-\boldsymbol{\mathrm \uptheta}^{(3)}\right)^{\top}\boldsymbol{\mathrm Q}_{({\hat{\tau}}_1,{\tau}_1^0)}
\boldsymbol{\mathrm Q}_{({\hat{\tau}}_1,{\hat{\tau}}_2)}^{-1}\frac{1}{{\hat{\tau}}_2-\tau_2^0}
\sum_{t_i\in(\tau_2^0,{\hat{\tau}}_2]}\boldsymbol{\mathrm z}_i^{\top}\boldsymbol{\mathrm z}_i\left(\boldsymbol{\mathrm \uptheta}^{(2)}-\boldsymbol{\mathrm \uptheta}^{(3)}\right),
\end{eqnarray*}
with
\begin{eqnarray*}
&&\left(\boldsymbol{\mathrm \uptheta}^{(2)}-\boldsymbol{\mathrm \uptheta}^{(3)}\right)^{\top}\boldsymbol{\mathrm Q}_{({\hat{\tau}}_1,{\tau}_1^0)}
\boldsymbol{\mathrm Q}_{({\hat{\tau}}_1,{\hat{\tau}}_2)}^{-1}\frac{1}{{\hat{\tau}}_2
-\tau_2^0}\sum_{t_i\in(\tau_2^0,{\hat{\tau}}_2]}
\boldsymbol{\mathrm z}_i^{\top}\boldsymbol{\mathrm z}_i\boldsymbol{\mathrm Q}_{({\hat{\tau}}_1,{\hat{\tau}}_2)}^{-1}
\boldsymbol{\mathrm Q}_{({\hat{\tau}}_1,{\tau}_1^0)}\left(\boldsymbol{\mathrm \uptheta}^{(2)}
-\boldsymbol{\mathrm \uptheta}^{(3)}\right)=o_p(1),\\
&&\left(\boldsymbol{\mathrm \uptheta}^{(2)}-\boldsymbol{\mathrm \uptheta}^{(3)}\right)^{\top}\boldsymbol{\mathrm Q}_{({\tau}_2^0,{\hat{\tau}}_2)}
\boldsymbol{\mathrm Q}_{({\hat{\tau}}_1,{\hat{\tau}}_2)}^{-1}\frac{1}{{\hat{\tau}}_2
-\tau_2^0}\sum_{t_i\in(\tau_2^0,{\hat{\tau}}_2]}\boldsymbol{\mathrm z}_i^{\top}\boldsymbol{\mathrm z}_i
\boldsymbol{\mathrm Q}_{({\hat{\tau}}_1,{\hat{\tau}}_2)}^{-1}\boldsymbol{\mathrm Q}_{({\tau}_2^0,{\hat{\tau}}_2)}
\left(\boldsymbol{\mathrm \uptheta}^{(2)}-\boldsymbol{\mathrm \uptheta}^{(3)}\right)=o_p(1),\\
&&2\left(\boldsymbol{\mathrm \uptheta}^{(2)}-\boldsymbol{\mathrm \uptheta}^{(3)}\right)^{\top}\boldsymbol{\mathrm Q}_{({\hat{\tau}}_1,{\tau}_1^0)}
\boldsymbol{\mathrm Q}_{({\hat{\tau}}_1,{\hat{\tau}}_2)}^{-1}\frac{1}{{\hat{\tau}}_2-\tau_2^0}
\sum_{t_i\in(\tau_2^0,{\hat{\tau}}_2]}\boldsymbol{\mathrm z}_i^{\top}\boldsymbol{\mathrm z}_i\boldsymbol{\mathrm Q}_{({\hat{\tau}}_1,{\hat{\tau}}_2)}^{-1}
\boldsymbol{\mathrm Q}_{({\tau}_2^0,{\hat{\tau}}_2)}\left(\boldsymbol{\mathrm \uptheta}^{(2)}-\boldsymbol{\mathrm \uptheta}^{(3)}\right)=o_p(1),\\
&&2\left(\boldsymbol{\mathrm \uptheta}^{(2)}-\boldsymbol{\mathrm \uptheta}^{(3)}\right)^{\top}\boldsymbol{\mathrm Q}_{({\tau}_2^0,{\hat{\tau}}_2)}
\boldsymbol{\mathrm Q}_{({\hat{\tau}}_1,{\hat{\tau}}_2)}^{-1}\frac{1}{{\hat{\tau}}_2-\tau_2^0}
\sum_{t_i\in(\tau_2^0,{\hat{\tau}}_2]}\boldsymbol{\mathrm z}_i^{\top}\boldsymbol{\mathrm z}_i\left(\boldsymbol{\mathrm \uptheta}^{(2)}-\boldsymbol{\mathrm \uptheta}^{(3)}\right)=
o_p(1),\\
&&2\left(\boldsymbol{\mathrm \uptheta}^{(2)}-\boldsymbol{\mathrm \uptheta}^{(3)}\right)^{\top}\boldsymbol{\mathrm Q}_{({\hat{\tau}}_1,{\tau}_1^0)}
\boldsymbol{\mathrm Q}_{({\hat{\tau}}_1,{\hat{\tau}}_2)}^{-1}\frac{1}{{\hat{\tau}}_2-\tau_2^0}
\sum_{t_i\in(\tau_2^0,{\hat{\tau}}_2]}\boldsymbol{\mathrm z}_i^{\top}\boldsymbol{\mathrm z}_i\left(\boldsymbol{\mathrm \uptheta}^{(2)}-\boldsymbol{\mathrm \uptheta}^{(3)}\right)
=o_p(1).
\end{eqnarray*}
\ \\
\noindent Now, suppose $\gamma_4$ be the smallest eigenvalue of $\frac{1}{{\hat{\tau}}_2-\tau_2^0}\sum_{t_i\in(\tau_2^0,{\hat{\tau}}_2]}\boldsymbol{\mathrm z}_i^{\top}\boldsymbol{\mathrm z}_i$.
With a suitable $C$, $\gamma_4$ is bounded away from 0. Hence,
 \begin{eqnarray*}
\left(\boldsymbol{\mathrm \uptheta}^{(2)}-\boldsymbol{\mathrm \uptheta}^{(3)}\right)'\frac{1}{{\hat{\tau}}_2-\tau_2^0}
\sum_{t_i\in(\tau_2^0,{\hat{\tau}}_2]}\boldsymbol{\mathrm z}_i^{\top}\boldsymbol{\mathrm z}_i\left(\boldsymbol{\mathrm \uptheta}^{(2)}-\boldsymbol{\mathrm \uptheta}^{(3)}\right) &\geq& \gamma4 ||\boldsymbol{\mathrm \uptheta}^{(2)}-\boldsymbol{\mathrm \uptheta}^{(3)}||^2>0
\end{eqnarray*}
with probability 1 and this dominates the rest of the terms in (\ref{prnm2-eq0}) when $T$ is large. This implies that (\ref{prnm2-eq0}) is positive with probability 1, which gives a contradiction and it indicates that with large probability ${\hat{\tau}}_2$ cannot be in the set $V_{\eta}(C)$.
\end{proof}

\begin{proof}[Proof of Corollary~\ref{corm-1}]
Part (i) of Corollary~\ref{corm-1} follows from the same arguments utilised in the proof of Proposition~\ref{prnm13} with $\phi$ in (\ref{prnm1-eq0}) set to $\displaystyle \frac{1}{T^{2r^*}}$, together with the fact that $T\phi||\boldsymbol{\mathrm \uptheta}^{(1)}-\boldsymbol{\mathrm \uptheta}^{(2)}||^2=(T^{1-2r^*}v_T^2||\mathbf{M}||^2)$ $ \xrightarrow[T\rightarrow \infty]{}\infty$ and $\log T/T^{2r^*} \xrightarrow[T\rightarrow \infty]{}0$. On the other hand, part (ii) may be verified by employing similar arguments as in the proof of Proposition~\ref{prnm14} in investigating the set $V_{\eta}(C, v_T)=\{{\tau}: C/v_T^2<|{\tau}-{\tau^0}|<\eta T\}$ instead of $V_{\eta}(C)$.
\end{proof}
\ \\
\begin{proof}[Proof of Proposition~\ref{prn-mlem1}]
To examine the behaviour of (\ref{mlem2}), we first define $Y_i^*=X_{t_{i+1}^*}-X_{t_i^*}$ and $\boldsymbol{\mathrm z}_i^*=V(t=t_i^*)\Delta_t^*$. Then,
\begin{eqnarray*} \log \ell^*(\boldsymbol{\mathrm \tau},\boldsymbol{\mathrm{\hat{\uptheta}}}(\boldsymbol{\mathrm \tau}))&=&\frac{1}{\Delta_t^*\sigma^2}
\sum_{j=1}^{m+1}\sum_{t_i^*\in(\tau_{j-1},\tau_{j}]}\boldsymbol{\mathrm{\hat{\uptheta}}}^{(j){\top}}\boldsymbol{\mathrm z}_i^{*^{\top}}
Y_i^*-\frac{1}{2\Delta_t^*\sigma^2}\sum_{j=1}^{m+1}
\sum_{t_i^*\in(\tau_{j-1},\tau_{j}]}\left(\boldsymbol{\mathrm{\hat{\uptheta}}}^{(j)\prime}\boldsymbol{\mathrm z}_i^{*{\top}}\right)^2
-\frac{1}{2\Delta_t^*\sigma^2}\sum_{t_i^*\in[0,T]}Y_i^{*{\top}}Y_i^*
\\&&+\frac{1}{2\Delta_t^*\sigma^2}\sum_{t_i^*\in[0,T]}
Y_i^{*\top}Y_i^*=\frac{1}{2\Delta_t^*\sigma^2}\left(\sum_{t_i^*\in[0,T]}Y_i^{*\top}Y_i^*-
\sum_{j=1}^{m+1}\sum_{t_i^*\in(\tau_{j-1},\tau_{j}]}\left(Y_i^*-\boldsymbol{\mathrm z}_i^*\boldsymbol{\mathrm{\hat{\uptheta}}}^{(j)}\right)^2\right).\end{eqnarray*}
\ \\
\noindent The term $\sum_{j=1}^{m+1}\sum_{t_i^*\in(\tau_{j-1},\tau_{j}]}\left(Y_i^*-\boldsymbol{\mathrm z}_i^*\boldsymbol{\mathrm{\hat{\uptheta}}}^{(j)}\right)^2$ is non-negative. For an observed process $X_t=x_t$ with constant $\Delta_t^*$ and known $\sigma$, $\sum_{t_i^*\in[0,T]}Y_i^{*\top}Y_i^*$ is fixed and does not depend on the change points $\tau$. Hence, finding the change points $\tau=(\tau_1,~\ldots~,\tau_m)$ that maximise (\ref{mlem2}) is equivalent to
the minimisation of the term \begin{equation}\label{prnmlem1-eq1}\sum_{j=1}^{m+1}
\sum_{t_i^*\in(\tau_{j-1},\tau_{j}]}\left(Y_i^*-\boldsymbol{\mathrm z}_i^*\boldsymbol{\mathrm{\hat{\uptheta}}}_i\right)^2.\end{equation}
\ \\
If $\Delta_t^*=\Delta_t$, the structure of (\ref{prnmlem1-eq1}) is the same as $SSE(T,\boldsymbol{\mathrm \tau},\boldsymbol{\mathrm{\hat{\uptheta}}}(\boldsymbol{\mathrm \tau}))$ in (\ref{ssr1}). The rest of the proof for Proposition~\ref{prn-mlem1} follows directly via the same arguments
used in establishing Proposition~\ref{prnm13} and \ref{prnm14}.
\end{proof}
\ \\
\section{Proof of the properties in section~\ref{sectionanormality}}\label{appendixanormality}
\begin{proof}[Proof of Proposition~\ref{anm1}]
Note that by (\ref{cxm2}) and (\ref{cxm3}) ,
\begin{equation*}\frac{1}{T}\int_{{\tau}_{j-1}^0}^{{\tau}_j^0}\tilde{X}_t\varphi_k(t)dt\xrightarrow[T\rightarrow \infty]{a.s.}(s_{j}-s_{j-1})\int_0^1\tilde{h}(t)\varphi_k(t)dt \end{equation*}
and
\begin{equation}\label{cxm3}\frac{1}{T}\int_{\tau_{j-1}^0}^{\tau_j^0}\tilde{X}_t^2dt\xrightarrow[T\rightarrow \infty]{a.s.}(s_{j}^0-s_{j-1}^0)
\left(\int_0^1(\tilde{h}^{(j)}(t))^2dt+\frac{\sigma^2}{2a^{(j)}}\right).
\end{equation}
\ \\
In addition, invoking similar arguments found in the proofs of (B. 47) and (B. 67) in Zhang~(2015), we have
\begin{equation}\frac{1}{T}\int_{{\hat{\tau}}_{j-1}}^{{\hat{\tau}}_j}
X_t\varphi_k(t)dt-\frac{1}{T}\int_{{\hat{\tau}}_{j-1}}^{{\hat{\tau}}_j}
\tilde{X}_t\varphi_k(t)dt\xrightarrow[T\rightarrow \infty]{P}0\end{equation}
and
\begin{equation}\frac{1}{T}\int_{{\hat{\tau}}_{j-1}}^{{\hat{\tau}}_j}X_t^2dt-\frac{1}{T}\int_{{\hat{\tau}}_{j-1}}^{{\hat{\tau}}_j}\tilde{X}_t^2dt\xrightarrow[T\rightarrow \infty]{P}0.\end{equation}
\ \\
Therefore, it suffices to prove that
  \begin{equation}\label{anm1-2a}\frac{1}{T}\int_{{\hat{\tau}}_{j-1}}^{{\hat{\tau}}_j}\tilde{X}_t\varphi_k(t)dt- \frac{1}{T}\int_{{\tau}_{j-1}^0}^{{\tau}_j^0}\tilde{X}_t\varphi_k(t)dt\xrightarrow[T\rightarrow \infty]{p} 0.\end{equation}
\ \\
 Similarly,
 \begin{equation}\label{anm1-3a}\frac{1}{T}\int_{{\hat{\tau}}_{j-1}}^{{\hat{\tau}}_j}\tilde{X}_t^2dt- \frac{1}{T}\int_{{\tau}_{j-1}^0}^{{\tau}_j^0}\tilde{X}_t^2dt\xrightarrow[T\rightarrow \infty]{p}0.\end{equation}
\ \\
\noindent Let $0<\delta_j<\min((s_{j+1}^0-s_{j}^0),(s_j^0-s_{j-1}^0))/2$. Then, it follows form Proposition~\ref{prnm14} that $$\lim_{T\rightarrow \infty}P(|\hat{s}_j-s_j^0|>\delta_j)=0, \quad j=1,~\ldots,~m.$$
Therefore, we have
\begin{eqnarray*}
&&P\left(\frac{1}{T}|\int_{{\hat{\tau}}_{j-1}}^{{\hat{\tau}}_j}\tilde{X}_t\varphi_k(t)dt- \frac{1}{T}\int_{{\tau}_{j-1}^0}^{{\tau}_j^0}\tilde{X}_t\varphi_k(t)dt|> \epsilon \right)\\
&=&P\left( \frac{1}{T}|\int_{{\hat{\tau}}_{j-1}}^{{\hat{\tau}}_j}\tilde{X}_t\varphi_k(t)dt- \frac{1}{T}\int_{{\tau}_{j-1}^0}^{{\tau}_j^0}\tilde{X}_t\varphi_k(t)dt|> \epsilon, |\hat{s}_j-s_j^0|\leq\delta_j, |\hat{s}_{j-1}-s_{j-1}^0|\leq\delta_{j-1}\right).
\end{eqnarray*}

Since  $|\hat{s}_j-s_j^0|\leq\delta_j$ is equivalent to $s_j^0-\delta_j\leq \hat{s}_j\leq s_j^0+\delta_j$,
and it follows that, for every $j$,
\begin{eqnarray*}
\norm{ \int_{{\hat{\tau}}_{j-1}}^{{\hat{\tau}}_j}\tilde{X}_t\varphi_k(t)dt- \int_{{\tau}_{j-1}^0}^{{\tau}_j^0}\tilde{X}_t\varphi_k(t)dt}
&\leq& \int_{{\tau}_{j-1}^0-\delta_{j-1}T}^{\tau_{j-1}^0+\delta_{j-1}T}|\tilde{X}_t\varphi_k(t)|dt+\int_{{\tau}_{j}^0-\delta_{j}T}^{\tau_{j}^0+\delta_{j}T}|\tilde{X}_t\varphi_k(t)|dt.
\end{eqnarray*}
Hence,
\begin{eqnarray}
&&P\left( \frac{1}{T} \left| \int_{{\hat{\tau}}_{j-1}}^{{\hat{\tau}}_j}\tilde{X}_t\varphi_k(t)dt- \int_{{\tau}_{j-1}^0}^{{\tau}_j^0}\tilde{X}_t\varphi_k(t)dt \right|> \epsilon, |\hat{s}_j-s_j^0|\leq\delta_j, |\hat{s}_{j-1}-s_{j-1}^0|\leq\delta_{j-1} \right)\nonumber\\
&\leq&P\left( \frac{1}{T} \left( \int_{{\tau}_{j-1}^0-\delta_{j-1}T}^{\tau_{j-1}^0+\delta_{j-1}T}|\tilde{X}_t\varphi_k(t)|dt+\int_{{\tau}_{j}^0-\delta_{j}T}^{\tau_{j}^0+\delta_{j}T}|\tilde{X}_t\varphi_k(t)|dt\right)> \epsilon \right)\nonumber\\
&\leq&P\left( \frac{1}{T} \int_{{\tau}_{j-1}^0-\delta_{j-1}T}^{\tau_{j-1}^0+\delta_{j-1}T}|\tilde{X}_t\varphi_k(t)|dt>\epsilon/2\right)
+P\left( \frac{1}{T}\int_{{\tau}_{j}^0-\delta_{j}T}^{\tau_{j}^0+\delta_{j}T}|\tilde{X}_t\varphi_k(t)|dt> \epsilon/2 \right).\label{prn34-eq1}
\end{eqnarray}
\ \\
It then follows from the Markov's inequality and Jensen's inequality that
\begin{eqnarray*}
P\left(\frac{1}{T} \int_{{\tau}_{j-1}^0-\delta_{j-1}T}^{\tau_{j-1}^0+\delta_{j-1}T}|\tilde{X}_t\varphi_k(t)|dt>\epsilon/2\right)&\leq&\frac{4\mathrm{E}\left[\left( \int_{{\tau}_{j-1}^0-\delta_{j-1}T}^{\tau_{j-1}^0+\delta_{j-1}T}|\tilde{X}_t||\varphi_k(t)|dt
\right)^2\right]}{\epsilon^2T^2}
\leq \frac{4\mathrm{E}\left( \int_{{\tau}_{j-1}^0-\delta_{j-1}T}^{\tau_{j-1}^0+\delta_{j-1}T}|\tilde{X}_t|^2|\varphi_k(t)|^2dt\right)}{\epsilon^2T^2}.
\end{eqnarray*}
\noindent In the same vein,
\begin{eqnarray*}
P\left(\frac{1}{T}\int_{{\tau}_{j}^0-\delta_{j}T}^{\tau_{j}^0+\delta_{j}T}|\tilde{X}_t\varphi_k(t)|dt> \epsilon/2\right)&\leq&\frac{4\mathrm{E}\left[\left(
\int_{{\tau}_{j}^0-\delta_{j}T}^{\tau_{j}^0+\delta_{j}T}|\tilde{X}_t||\varphi_k(t)|dt\right)^2\right]}{\epsilon^2T^2}
\leq\frac{4\mathrm{E}\left( \int_{{\tau}_{j}^0-\delta_{j}T}^{\tau_{j}^0+\delta_{j}T}|\tilde{X}_t|^2|\varphi_k(t)|^2dt\right)}{\epsilon^2T^2}.
\end{eqnarray*}
\ \\
\noindent Consequently,
\begin{equation*}(\ref{prn34-eq1})\leq \frac{4\mathrm{E}\left( \int_{{\tau}_{j-1}^0-\delta_{j-1}T}^{\tau_{j-1}^0+\delta_{j-1}T}|\tilde{X}_t|^2|\varphi_k(t)|^2dt+ \int_{{\tau}_{j}^0-\delta_{j}T}^{\tau_{j}^0+\delta_{j}T}|\tilde{X}_t|^2|\varphi_k(t)|^2dt\right)}{\epsilon^2T^2}.\end{equation*}
\noindent Then, it follows from (\ref{sol2}) that $E(|\tilde{X}_t|^2)<K_1<\infty$, for all $t\geq 0$. Also, under Assumption~\ref{asm3}, $|\varphi_k(t)|\leq K_\varphi$. Therefore,
$$\frac{4\mathrm{E}\left( \int_{{\tau}_{j-1}^0-\delta_{j-1}T}^{\tau_{j-1}^0+\delta_{j-1}T}|\tilde{X}_t|^2|\varphi_k(t)|^2dt+\int_{{\tau}_{j}^0-\delta_{j}T}^{\tau_{j}^0+\delta_{j}T}|\tilde{X}_t|^2|\varphi_k(t)|^2dt\right)}{\epsilon^2T^2}
\leq \frac{8\left(\delta_{j-1}+\delta_j\right)T K_1^{2}K_\varphi^2 }{\epsilon^2 T^2}.$$
Hence,
\begin{eqnarray*}
&&P\left( \frac{1}{T}\left| \int_{{\hat{\tau}}_{j-1}}^{{\hat{\tau}}_j}\tilde{X}_t\varphi_k(t)dt- \frac{1}{T}\int_{{\tau}_{j-1}^0}^{{\tau}_j^0}\tilde{X}_t\varphi_k(t)dt\right|> \epsilon \right) \\
&\leq&\lim_{T\rightarrow \infty}P\left(\frac{1}{T}\left( \int_{{\tau}_{j-1}^0-\delta_{j-1}T}^{\tau_{j-1}^0+\delta_{j-1}T}|\tilde{X}_t\varphi_k(t)|dt+\int_{{\tau}_{j}^0-\delta_{j}T}^{\tau_{j}^0+\delta_{j}T}|\tilde{X}_t\varphi_k(t)|dt\right)> \epsilon/2\right)\\
&\leq& \lim_{T\rightarrow \infty} \frac{8(\delta_{j-1}+\delta_j)T K_1^2K_\varphi^2 }{\epsilon^2 T^2}=0.
\end{eqnarray*}
Similarly, we have
\begin{eqnarray}
&&P\left( \frac{1}{T}\left|\int_{{\hat{\tau}}_{j-1}}^{{\hat{\tau}}_j}\tilde{X}_t^2dt- \int_{{\tau}_{j-1}^0}^{{\tau}_j^0}\tilde{X}_t^2dt\right|> \epsilon, |\hat{s}_j-s_j^0|\leq\delta_j, |\hat{s}_{j-1}-s_{j-1}^0|\leq\delta_{j-1} \right)\nonumber\\
&\leq&P\left( \frac{1}{T} \int_{{\tau}_{j-1}^0-\delta_{j-1}T}^{\tau_{j-1}^0+\delta_{j-1}T}\tilde{X}_t^2dt>\epsilon/2\right)
+P\left( \frac{1}{T}\int_{{\tau}_{j}^0-\delta_{j}T}^{\tau_{j}^0+\delta_{j}T}\tilde{X}_t^2dt> \epsilon/2 \right).\label{prn34-eq2}
\end{eqnarray}
Again, using the Markov's inequality,
\begin{equation}(\ref{prn34-eq2})\leq \frac{4K_1(\delta_{j-1}+\delta_j)T}{\epsilon T}=\frac{4K_1(\delta_{j-1}+\delta_j)}{\epsilon}.\end{equation}
Using the consistency properties of the estimators $\hat{s}_j$ provided in Section~\ref{mcpe},
i.e., $\hat{s}_j-s_j^0\xrightarrow[T\rightarrow \infty]{P}0$,
we can choose arbitrarily small $\delta_{j-1}$ and $\delta_j$ such that
$$\lim_{T\rightarrow \infty}(\ref{prn34-eq2})\leq \lim_{T\rightarrow \infty}\frac{4K_1(\delta_{j-1}+\delta_j)}{\epsilon}=0.$$
   \end{proof}

\begin{proof}[Proof of Proposition~\ref{anm4}]
By Propositions \ref{anm1}--\ref{anm3}, it may be shown that
\begin{equation}\frac{1}{T}\boldsymbol{\mathrm Q}_{({\hat{\tau}}_{j-1}, {\hat{\tau}}_{j})} \xrightarrow[T\rightarrow \infty]{P} (s_j^0-s_{j-1}^0)\boldsymbol{\mathrm \Sigma}_j, j=1,~\ldots,~m. \end{equation}
Under Assumptions 2--4, $\displaystyle
\frac{1}{T}\boldsymbol{\mathrm Q}_{({\hat{\tau}}_{j-1}, {\hat{\tau}}_{j})}$ is positive definite provided that the base functions $\{\varphi_k(t):i=1~\ldots,~p\}$ are incomplete.
Then, by the Continuous Mapping Theorem,
\begin{equation}T\boldsymbol{\mathrm Q}_{({\hat{\tau}}_{j-1}, {\hat{\tau}}_{j})}^{-1}=g\left(
\frac{1}{T}\boldsymbol{\mathrm Q}_{({\hat{\tau}}_{j-1}, {\hat{\tau}}_{j})} \right) \xrightarrow[T\rightarrow \infty]{P} \frac{1}{s_j^0-s_{j-1}^0}\boldsymbol{\mathrm \Sigma}_j^{-1}=g\left(\left(s_j^0-s_{j-1}^0\right)\boldsymbol{\mathrm \mathrm{\Sigma}}_j\right), j=1,~\ldots,~m,\end{equation}
where $g(X)=X^{-1}$ for any positive definite matrix $X$.
\end{proof}

\begin{proof}[Proof of Proposition~\ref{anm5}]
Here we only prove that
\begin{equation}\frac{1}{\sqrt{T}}\int_{{\hat{\tau}}_{j-1}}^{{\hat{\tau}}_j}X_tdW_t
-\frac{1}{\sqrt{T}}\int_{{\tau}_{j-1}^0}^{{\tau}_j^0}X_tdW_t \xrightarrow[T\rightarrow \infty]{P} 0.\end{equation}
The convergence of the remaining components may be proved analogously by the same approach.\\
Since by  Proposition~\ref{prnm14}, $$\lim_{T\rightarrow \infty}P\left(|\hat{s}_j-s_j^0|>\delta_j\right)=0, \quad j=1,~\ldots,~m,$$
\begin{eqnarray*}
&&P\left( \frac{1}{T} \left| \int_{{\hat{\tau}}_{j-1}}^{{\hat{\tau}}_j}X_tdW_t- \frac{1}{T}\int_{{\tau}_{j-1}^0}^{{\tau}_j^0}X_tdW_t \right|> \epsilon \right)\\
&=&P\left( \frac{1}{T} \left| \int_{{\hat{\tau}}_{j-1}}^{{\hat{\tau}}_j}X_tdW_t- \int_{{\tau}_{j-1}^0}^{{\tau}_j^0}X_tdW_t \right|> \epsilon, |\hat{s}_j-s_j^0|\leq\delta_j, |\hat{s}_{j-1}-s_{j-1}^0|\leq\delta_{j-1}\right)
\end{eqnarray*}
Without loss of generality, we assume that ${\hat{\tau}}_{j-1}<\tau_{j-1}^0<\tau_j^0<{\hat{\tau}}_{j}$. Then,
\begin{eqnarray}
\left|\int_{{\hat{\tau}}_{j-1}}^{{\hat{\tau}}_j}X_tdW_t- \int_{{\tau}_{j-1}^0}^{{\tau}_j^0}X_tdW_t\right|
=\left| \int_{{\hat{\tau}}_{j-1}}^{\tau_{j-1}^0}X_tdW_t+ \int_{{\tau}_{j}^0}^{{\hat{\tau}}_{j}}X_tdW_t\right|
\leq \left|\int_{\tau_{j-1}^0-\delta_{j-1}T}^{{\hat{\tau}}_{j-1}}X_tdW_t\right|
+ \left|\int_{\tau_{j-1}^0-\delta_{j-1}T}^{{\tau}_{j-1}^0}X_tdW_t \right|\nonumber\\
+ \left|\int_{{\tau}_{j}^0}^{\tau_{j}^0+\delta_{j}T}X_tdW_t \right|
+ \left|\int_{{\hat{\tau}}_{j}}^{{\tau}_{j}+\delta_{j}T}X_tdW_t \right|.\qquad\quad
\end{eqnarray}
With the same arguments as in the proof of (\ref{prn34-eq1}), together with above inequality,
\begin{eqnarray}
&&P\left( \frac{1}{\sqrt{T}} \left| \int_{{\hat{\tau}}_{j-1}}^{{\hat{\tau}}_j}X_tdW_t- \frac{1}{T}\int_{{\tau}_{j-1}^0}^{{\tau}_j^0}X_tdW_t \right|> \epsilon, |\hat{s}_j-s_j^0|\leq\delta_j, |\hat{s}_{j-1}-s_{j-1}^0|\leq\delta_{j-1} \right)\\
&\leq&P\left( \frac{1}{\sqrt{T}} \left|\int_{\tau_{j-1}^0-\delta_{j-1}T}^{{\hat{\tau}}_{j-1}}X_tdW_t\right|> \epsilon/4\right)
+P\left(\frac{1}{\sqrt{T}} \left|\int_{\tau_{j-1}^0-\delta_{j-1}T}^{{\tau}_{j-1}^0}X_tdW_t\right|> \epsilon/4 \right)\\
&&+P\left( \frac{1}{\sqrt{T}} \left|\int_{{\hat{\tau}}_{j}}^{\tau_{j}^0+\delta_{j}T}X_tdW_t\right|> \epsilon/4\right)
+P\left(\frac{1}{\sqrt{T}} \left|\int_{{\tau}_{j}^0}^{\tau_{j}^0+\delta_{j}T}X_tdW_t\right|> \epsilon/4\right).\label{anm5-eq1}
\end{eqnarray}
\ \\
Then, by Markov inequality and It\^o's isometry,
\begin{eqnarray*}
P\left(\frac{1}{\sqrt{T}} \left|\int_{\tau_{j-1}^0-\delta_{j-1}T}^{{\tau}_{j-1}^0}X_tdW_t \right|> \epsilon/4 \right)
\leq\frac{16\mathrm{E}\left[ \left(\int_{\tau_{j-1}^0-\delta_{j-1}T}^{{\tau}_{j-1}^0}X_tdW_t \right)^2\right] }{\epsilon^2T}
=\frac{16\int_{\tau_{j-1}^0-\delta_{j-1}T}^{{\tau}_{j-1}^0}\mathrm{E}\left[X_t^2\right] dt}{T\epsilon^2}
\leq \frac{16K_1\delta_{j-1}T}{\epsilon^2T}\\=\frac{16K_1\delta_{j-1}}{\epsilon^2},
\end{eqnarray*}
which tends to 0 for an infinitesimal $\delta_{j-1}$. Similarly,
$\displaystyle P\left(\frac{1}{\sqrt{T}} \left|\int_{{\tau}_{j}^0}^{\tau_{j}^0+\delta_{j}T}X_tdW_t\right|> \epsilon/4 \right)$ also tends to 0
by choosing an infinitesimal $\delta_j$. Further, note that $\displaystyle \int_{\tau_{j-1}^0-\delta_{j-1}T}^{{\hat{\tau}}_{j-1}}X_tdW_t =\int_0^T X_t \mathbbm{1}(\tau_{j-1}^0-\delta_{j-1}T\leq t\leq {\hat{\tau}}_{j-1})dW_t$. So, again, by the Markov inequality and It\^o's isometry,
\begin{eqnarray*}
P\left( \frac{1}{\sqrt{T}} \left| \int_{\tau_{j-1}^0-\delta_{j-1}T}^{{\hat{\tau}}_{j-1}}X_tdW_t \right|> \epsilon/4 \right)
= P\left( \frac{1}{\sqrt{T}} \left|\int_0^T X_t \mathbbm{1}(\tau_{j-1}^0-\delta_{j-1}T\leq t\leq {\hat{\tau}}_{j-1})dW_t \right|> \epsilon/4\right)\\
\leq\frac{16\mathrm{E}\left[ \left( \int_0^T X_t \mathbbm{1}(\tau_{j-1}^0-\delta_{j-1}T\leq t\leq {\hat{\tau}}_{j-1})dW_t \right)^2 \right]}{\epsilon^2T}
=\frac{16\int_{\tau_{j-1}^0-\delta_{j-1}T}^{\tau_{j-1}^0+\delta_{j-1}T}\mathrm{E}\left[X_t^2\right]dt}{T\epsilon^2}
\leq \frac{32K_1\delta_{j-1}T}{\epsilon^2T}\\=\frac{32K_1\delta_{j-1}}{\epsilon^2},
\end{eqnarray*}
which goes to 0 for an infinitesimal $\delta_{j-1}$. Similarly, $P\left( \frac{1}{\sqrt{T}} \left| \int_{{\hat{\tau}}_{j}}^{\tau_{j}^0+\delta_{j}T}X_tdW_t \right|> \epsilon/4 \right)$ also approaches 0
by choosing an infinitesimal $\delta_j$. This implies that (\ref{anm5-eq1}) tends to 0 for an infinitesimal $\delta_{j-1}$ and $\delta_j$.
\end{proof}

\begin{proof}[Proof of Proposition~\ref{anm6}]
Replacing the interval $\left(0, T\right]$ by $\left(\tau_{j-1}^0,\tau_j^0\right]$ and using similar argument as in the proof of Proposition 2.1.6 in Zhang~(2015), we get
\begin{equation}\frac{1}{\sqrt{T}}\boldsymbol{\mathrm r}_{({\tau}_{j-1}^0, {\tau}_{j}^0)} \xrightarrow[T\rightarrow \infty]{D} \boldsymbol{\mathrm r}_j\sim {\cal N}_{(p+1)}\left(0,\boldsymbol{\mathrm \Sigma}_j\right).\end{equation}
By Proposition~\ref{anm5},
\begin{equation}\frac{1}{\sqrt{T}}\boldsymbol{\mathrm r}_{({\hat{\tau}}_{j-1}, {\hat{\tau}}_{j})}-\frac{1}{\sqrt{T}}\boldsymbol{\mathrm r}_{({\tau}_{j-1}^0, {\tau}_{j}^0)} \xrightarrow[T\rightarrow \infty]{P} 0.\end{equation}
Hence by Slutsky's Theorem,
\begin{equation}\frac{1}{\sqrt{T}}\boldsymbol{\mathrm r}_{({\hat{\tau}}_{j-1}, {\hat{\tau}}_{j})}=\frac{1}{\sqrt{T}}\boldsymbol{\mathrm r}_{({\hat{\tau}}_{j-1}, {\hat{\tau}}_{j})}-\frac{1}{\sqrt{T}}\boldsymbol{\mathrm r}_{({\tau}_{j-1}^0, {\tau}_{j}^0)}+\frac{1}{\sqrt{T}}\boldsymbol{\mathrm r}_{({\tau}_{j-1}^0, {\tau}_{j}^0)}
 \xrightarrow[T\rightarrow \infty]{D} \boldsymbol{\mathrm r}_j\sim {\cal N}_{(p+1)}\left(0,\boldsymbol{\mathrm \Sigma}_j\right).\end{equation}
\end{proof}
\ \\
\begin{proof}[Proof of Corollary~\ref{anc1}]
Let $\boldsymbol{\mathrm{\tilde{r}}}_T(\boldsymbol{\hat{\mathrm \tau}})=\left(\boldsymbol{\mathrm{\tilde{r}}}_{(0,{\hat{\tau}}_1)}^{\top},...,\boldsymbol{\mathrm{\tilde{r}}}_{({\hat{\tau}}_m,T)}^{\top}\right)^{\top}$,  $\boldsymbol{\mathrm r}_T(\boldsymbol{\hat{\mathrm \tau}})=\left(\boldsymbol{\mathrm r}_{(0,{\hat{\tau}}_1)}^{\top},~\ldots,~\boldsymbol{\mathrm r}_{({\hat{\tau}}_m,T)}^{\top}\right)^{\top}$ and let $\boldsymbol{\tilde{\mathrm Q}}(\boldsymbol{\hat{\mathrm \tau}})^{-1}=\mathrm{diag}\left(\boldsymbol{\mathrm Q}_{(0,{\hat{\tau}}_1)}^{-1},~\ldots~,\boldsymbol{\mathrm Q}_{({\hat{\tau}}_m,T)}^{-1}\right)$.
Then, we have $\boldsymbol{\mathrm{\hat{\uptheta}}}=\boldsymbol{\tilde{\mathrm Q}}(\boldsymbol{\hat{\mathrm \tau}})^{-1}\boldsymbol{\mathrm{\tilde{r}}}_T(\boldsymbol{\hat{\mathrm \tau}})$. By Propositions~\ref{anm4}-- \ref{anm6} and application of Slutsky's Theorem,
$$\sqrt{T}\boldsymbol{\tilde{\mathrm Q}}(\boldsymbol{\hat{\mathrm \tau}})^{-1}\boldsymbol{\mathrm{\tilde{r}}}_T(\boldsymbol{\hat{\mathrm \tau}})=T\boldsymbol{\tilde{\mathrm Q}}(\boldsymbol{\hat{\mathrm \tau}})^{-1}\frac{1}{\sqrt{T}}\boldsymbol{\mathrm r}_T(\boldsymbol{\hat{\mathrm \tau}})\xrightarrow[T\rightarrow \infty]{D} \rho \sim {\cal N}_{((m+1)(p+1))}\left(0, \boldsymbol{\tilde{\mathrm \Sigma}^{-1}}\right).$$
\ \\
\noindent Next, we investigate the asymptotic behaviour of $\boldsymbol{\mathrm{\hat{\uptheta}}}$ based on $\boldsymbol{\hat{\mathrm \tau}}$. Without loss of generality, we assume that for the $j$th block, $1\leq j \leq m$, we have ${\hat{\tau}}_{j-1}<\tau_{j-1}^0<{\hat{\tau}}_{j}<{\tau}_j^0$. In this case, we have
$$\boldsymbol{\mathrm{\hat{\uptheta}}}^{(j)}=\boldsymbol{\mathrm Q}_{({\hat{\tau}}_{j-1},{\hat{\tau}}_{j})}^{-1}\left(\boldsymbol{\mathrm Q}_{({\hat{\tau}}_{j-1},{\tau}_{j-1}^0)}\boldsymbol{\mathrm \uptheta}^{(j-1)}+\boldsymbol{\mathrm Q}_{({\tau}_{j-1}^0,{\hat{\tau}}_{j})}\boldsymbol{\mathrm \uptheta}^{(j)}+\sigma \boldsymbol{\mathrm r}_{({\hat{\tau}}_{j-1},{\hat{\tau}}_{j})}\right).$$
\ \\
\noindent Hence,
$$\sqrt{T}\left(\boldsymbol{\mathrm{\hat{\uptheta}}}^{(j)}-\boldsymbol{\mathrm \uptheta}^{(j)}\right)=T\boldsymbol{\mathrm Q}_{({\hat{\tau}}_{j-1},{\hat{\tau}}_{j})}^{-1}\left(\frac{1}{\sqrt{T}}\boldsymbol{\mathrm Q}_{({\hat{\tau}}_{j-1},{\tau}_{j-1}^0)}\left(\boldsymbol{\mathrm \uptheta}^{(j-1)}-\boldsymbol{\mathrm \uptheta}^{(j)}\right)+\sigma \frac{1}{\sqrt{T}} \boldsymbol{\mathrm r}_{({\hat{\tau}}_{j-1},{\hat{\tau}}_{j})}\right).$$
\ \\
By Proposition~\ref{prnm14}, $|{\hat{\tau}}_{j-1}-\tau_{j-1}^0|\leq C$ for some $C>0$ with probability 1. Invoking the Markov inequality,
$$P\left( \frac{1}{\sqrt{T}} \int_{{\tau}_{j-1}^0}^{{\hat{\tau}}_{j-1}}X_t^2dt> \epsilon \right)\leq\frac{2\mathrm{E}\left[\int_{{\tau}_{j-1}^0}^{{\hat{\tau}}_{j-1}}X_t^2dt\right]}{\epsilon \sqrt{T}}\leq \frac{2K_1C}{\epsilon \sqrt{T}}$$
and
$$P\left( \frac{1}{\sqrt{T}} \int_{{\tau}_{j-1}^0}^{{\hat{\tau}}_{j-1}}X_t\varphi_k(t)dt> \epsilon\right)\leq\frac{4\mathrm{E}\left[ \int_{{\tau}_{j-1}^0}^{{\hat{\tau}}_{j-1}}X_t\varphi_k(t)dt\right]^2}{\epsilon^2 {T}}\leq \frac{4K_1K_{\varphi}C}{\epsilon^2 {T}}.$$
Therefore, $\norm{\frac{1}{\sqrt{T}}\boldsymbol{\mathrm Q}_{({\hat{\tau}}_{j-1},{\tau}_{j-1}^0)}}
\xrightarrow[T\rightarrow \infty]{p} 0$, which means
$$\sqrt{T}(\boldsymbol{\mathrm{\hat{\uptheta}}}^{(j)}-\boldsymbol{\mathrm \uptheta}^{(j)})-\sigma T\boldsymbol{\mathrm Q}_{({\hat{\tau}}_{j-1},{\hat{\tau}}_{j})}^{-1} \frac{1}{\sqrt{T}}\boldsymbol{\mathrm r}_{({\hat{\tau}}_{j-1},{\hat{\tau}}_{j})}\xrightarrow[T\rightarrow \infty]{p}0.$$
So,
$$\sqrt{T}(\boldsymbol{\mathrm{\hat{\uptheta}}}-\boldsymbol{\mathrm \uptheta})-\sigma T\boldsymbol{\tilde{\mathrm Q}}(\boldsymbol{\hat{\mathrm \tau}})^{-1} \frac{1}{\sqrt{T}}\boldsymbol{\mathrm r}_T(\boldsymbol{\hat{\mathrm \tau}})\xrightarrow[T\rightarrow \infty]{p}0$$
and
\begin{eqnarray*}
\sqrt{T}(\boldsymbol{\mathrm{\hat{\uptheta}}}-\boldsymbol{\mathrm \uptheta})=\sqrt{T}(\boldsymbol{\mathrm{\hat{\uptheta}}}-\boldsymbol{\mathrm \uptheta})- \sigma T\tilde{Q}(\boldsymbol{\hat{\mathrm \tau}})^{-1} \frac{1}{\sqrt{T}}\boldsymbol{\mathrm r}_T(\boldsymbol{\hat{\mathrm \tau}})+\sigma T\tilde{Q}(\boldsymbol{\hat{\mathrm \tau}})^{-1} \frac{1}{\sqrt{T}}\boldsymbol{\mathrm r}_T(\boldsymbol{\hat{\mathrm \tau}}) \xrightarrow[T\rightarrow \infty]{D}\rho\\ \sim {\cal N}_{((m+1)(p+1))}\left(0,\sigma^2\boldsymbol{\tilde{\mathrm \Sigma}^{-1}}\right).\end{eqnarray*}
\end{proof}
\ \\
\section{Proof of Proposition ~\ref{asic1} }\label{appendixsic}
\noindent The proof of Proposition consists of two parts. In Part~(i), we prove that $\mathcal{IC}(m=m^0)<\mathcal{IC}(m<m^0)$, whilst in Part~(ii) we prove that $\mathcal{IC}(m=m^0)>\mathcal{IC}(m<m^0)$.
\begin{proof}[Part~(i): $\mathcal{IC}(m=m^0)<\mathcal{IC}(m<m^0)$ ]
From Proposition~\ref{prn-mlem1},
\begin{equation}\mathcal{IC}(m=m^0)=-2\frac{1}{2\Delta_t^*\sigma^2}\left(\sum_{t_i\in[0,T]}Y_i^{\prime}Y_i-\sum_{j=1}^{m^0+1}\sum_{t_i\in({\hat{\tau}}_{j-1},{\hat{\tau}}_{j}]}\left(Y_i-\boldsymbol{\mathrm z}_i\boldsymbol{\mathrm{\hat{\uptheta}}}^{(j)}\right)^2\right)+(m^0+1)(p+1)\log(T/\Delta_t),\end{equation}
where ${\hat{\tau}}_j$, $j=1~\ldots,~m^0$ are obtained via (\ref{mlem1}). Next, we define
\begin{equation}\mathcal{IC}^0(m=m^0)=-2\frac{1}{2\Delta_t\sigma^2}\left(\sum_{t_i\in[0,T]}Y_i^{\prime}Y_i-\sum_{j=1}^{m^0+1}\sum_{t_i\in({\tau}_{j-1}^0,{\tau}_{j}^0]}\left(Y_i-\boldsymbol{\mathrm z}_i\boldsymbol{\mathrm{\hat{\uptheta}}}^{(j,0)}\right)^2\right)+(m^0+1)(p+1)\log(T/\Delta_t),\end{equation}
where $\boldsymbol{\mathrm{\hat{\uptheta}}}^{(j,0)}$ was given in \ref{appendixb}.
Since ${\hat{\tau}}_j$, $j=1,~\ldots~,m^0$, are obtained by maximising $\log \ell^*(\boldsymbol{\mathrm \tau},\hat{\uptheta})$, or equivalently minimising $-2\log \ell^*(\boldsymbol{\mathrm \tau},\hat{\uptheta})$, we have that $\mathcal{IC}(m=m^0)\leq \mathcal{IC}^0(m=m^0)$ with probability 1. Hence, we must show that
\begin{equation}\mathcal{IC}(m<m^0)>\mathcal{IC}^0(m=m^0)\end{equation}
with probability 1.\\
\ \\
\noindent For any positive integer $m^*$ such that $0<m^*<m^0$, suppose that the estimated locations of these $m^*$ change points are ${\hat{\tau}}_1^*,~\ldots~,{\hat{\tau}}_{m^*}^*$, and the MLE of drift parameters associated with the $i$th observation is $\boldsymbol{\mathrm{\hat{\uptheta}}}_i^*=\sum_{j=1}^{m^*+1}\boldsymbol{\mathrm{\hat{\uptheta}}}^{(j,*)}\mathbbm{1}({\hat{\tau}}_{j-1}^*\leq t_i\leq {\hat{\tau}}_j^*)$, where $\boldsymbol{\mathrm{\hat{\uptheta}}}^{(j,*)}=\boldsymbol{\mathrm Q}_{({\hat{\tau}}_{j-1}^*,{\hat{\tau}}_{j}^*)}^{-1}\boldsymbol{\mathrm{\tilde{r}}}_{({\hat{\tau}}_{j-1}^*,{\hat{\tau}}_{j}^*)}$. Furthermore,
\begin{eqnarray}
\frac{1}{T}(\mathcal{IC}(m=m^*)- \mathcal{IC}^0(m=m^0))&=&\frac{1}{T \Delta_t \sigma^2}\left(\sum_{j=1}^{m^*+1}\sum_{t_i^*\in({\hat{\tau}}_{j-1},{\hat{\tau}}_{j}]}\left(Y_i-\boldsymbol{\mathrm z}_i\boldsymbol{\mathrm{\hat{\uptheta}}}_i^*\right)^2-\sum_{j=1}^{m^0+1}\sum_{t_i\in({\tau}_{j-1}^0,{\tau}_{j}^0]}\left(Y_i-\boldsymbol{\mathrm z}_i\boldsymbol{\mathrm{\hat{\uptheta}}}^{(j,0)}\right)^2\right)\nonumber\\&&-\frac{(m^0-m^*)(p+1)\log(T/\Delta_t^*) }{T},
\end{eqnarray}
Since $m^*<m^0$, there exists at least one change point that cannot be consistently estimated. Without loss of generality, let $\tau_j^0$ be the change point. With similar arguments utilised in the proof of Lemma~C.1, we get
\begin{eqnarray}
\frac{1}{T}(\mathcal{IC}(m<m^0)- \mathcal{IC}^0(m=m^0))&\geq& C^* ||\boldsymbol{\mathrm \uptheta}^{(j)}-\boldsymbol{\mathrm \uptheta}^{(j+1)}||^2 +o_p(1)
\end{eqnarray}
for some $C^{*}>0$ with probability 1. Therefore, $\mathcal{IC}(m<m^0)> \mathcal{IC}^0(m=m^0)$ with probability 1. This completes the proof of part (i).\end{proof}
\ \\
\begin{proof}[Part~(ii): $\mathcal{IC}(m=m^0)<\mathcal{IC}(m>m^0)$]
\noindent Since $\mathcal{IC}(m=m^0)\leq \mathcal{IC}^0(m=m^0)$, where $\mathcal{IC}^0(m=m^0)$ is defined in Part~(i), it remains to show that the difference $\mathcal{IC}(m=m^*>m^0)- \mathcal{IC}^0(m=m^0)$ is positive with probability 1.\\
\ \\
\noindent Note that for the case where $m=m^*>m^0$ and the estimated locations of the $m^*$ change points are given by ${\hat{\tau}}_1<{\hat{\tau}}_2<,...,<{\hat{\tau}}_{m^*}$, we have
\begin{eqnarray}\label{part2-eq1}
\mathcal{IC}(m=m^*)- \mathcal{IC}^0(m=m^0)&=&\frac{1}{ \Delta_t \sigma^2}\left(\sum_{t_i\in [0,T]}\left(Y_i-\boldsymbol{\mathrm z}_i\boldsymbol{\mathrm{\hat{\uptheta}}}_i^*\right)^2-\sum_{t_i\in [0,T]}\left(Y_i-\boldsymbol{\mathrm z}_i\boldsymbol{\mathrm{\hat{\uptheta}}}_i^0\right)^2\right)\nonumber\\&&+(m^*-m^0)(p+1)\log(T/\Delta_t^*),
\end{eqnarray}
where $\boldsymbol{\mathrm{\hat{\uptheta}}}_i^{*}=\sum_{j=1}^{m^*+1}\boldsymbol{\mathrm \uptheta}^{(j,*)}\mathbbm{1}(\tau_{j-1}^*\leq t_i\leq \tau_j^*)$ with $\boldsymbol{\mathrm{\hat{\uptheta}}}^{(j,*)}=\boldsymbol{\mathrm Q}_{({\hat{\tau}}_j^*,{\hat{\tau}}_{j-1}^*)}^{-1}\boldsymbol{\mathrm{\tilde{r}}}_{({\hat{\tau}}_j^*,{\hat{\tau}}_{j-1}^*)}$. We note that $m^*>m^0$, and from $m^*$ of these estimated change points, there are $m^*-m^0$ estimated change points that divide the time interval $[0,T]$ into $m^*-m^0+1$ regimes such that within each regime, the number of estimated change points is equal to the number of exact change points. For example, suppose that $m^0=2$ and $m^*=3$ with $0<{\hat{\tau}}_1^*<\tau_1^0<{\hat{\tau}}_2^*<\tau_2^0<{\hat{\tau}}_3^*<T$. Then, if we divide the given time interval into $[0, {\hat{\tau}}_2^*]$ and $({\hat{\tau}}_2^*, T]$, we can see that within these two intervals, the number of estimated change points is equal to the number of exact change points. \\
\ \\
\noindent Denote the particular $m^*-m^0$ estimated change points by $\{\tilde{\tau}_{j}^*, j=1,~\ldots~,m^*-m^0\}$. Also, let $\tilde{\tau}_0^*=0$ and
$\tilde{\tau}_{m^*-m^0+1}^*=T$. Then,
\begin{eqnarray*}
(\ref{part2-eq1})&=&\frac{1}{ \Delta_t \sigma^2}\sum_{j=1}^{m^*-m^0+1}\sum_{t_i\in(\tilde{\tau}_{j-1}^*,\tilde{\tau}_{j}^*]}\left[\left(Y_i-\boldsymbol{\mathrm z}_i\boldsymbol{\mathrm{\hat{\uptheta}}}_i^*\right)^2-\left(Y_i-\boldsymbol{\mathrm z}_i\boldsymbol{\mathrm{\hat{\uptheta}}}_i^0\right)^2+\frac{(m^*-m^0)(p+1)\log(T/\Delta_t^*)}{m^*-m^0+1}\right].
\end{eqnarray*}
Thus, it remains to show that in each regime $(\tilde{\tau}_{j-1}^*,\tilde{\tau}_{j}^*]$, $j=1,~\ldots,~m^*-m^0+1$,

\begin{eqnarray}
&&\frac{1}{ \Delta_t \sigma^2}\sum_{t_i\in(\tilde{\tau}_{j-1}^*,\tilde{\tau}_{j}^*]}\left[\left(Y_i-\boldsymbol{\mathrm z}_i\boldsymbol{\mathrm{\hat{\uptheta}}}_i^*\right)^2-\left(Y_i-\boldsymbol{\mathrm z}_i\boldsymbol{\mathrm{\hat{\uptheta}}}_i^0\right)^2+\frac{(m^*-m^0)(p+1)\log(T/\Delta_t^*)}{m^*-m^0+1}\right]\nonumber\\
&=&\frac{1}{ \Delta_t \sigma^2} \sum_{t_i\in(\tilde{\tau}_{j-1}^*,\tilde{\tau}_{j}^*]}\left[
\left(\boldsymbol{\mathrm z}_i\left({\uptheta}_i-\boldsymbol{\mathrm{\hat{\uptheta}}}_i^*\right)
\right)^2-\left(\boldsymbol{\mathrm z}_i\left(\boldsymbol{\mathrm \uptheta}_i-\boldsymbol{\mathrm{\hat{\uptheta}}}_i^0\right)\right)^2+2u_i'\boldsymbol{\mathrm z}_i\left(\boldsymbol{\mathrm{\hat{\uptheta}}}_i^*-\boldsymbol{\mathrm{\hat{\uptheta}}}_i^0\right)
+\frac{(m^*-m^0)(p+1)\log(\frac{T}{\Delta_t^*})}{m^*-m^0+1}\right]\qquad \label{part2-eq2}
\end{eqnarray}
is positive with probability 1.\\
\ \\
\noindent Since within $(\tilde{\tau}_{j-1}^*,\tilde{\tau}_{j}^*]$, the number of estimated change points and the number of exact change points are the same, we first consider the case where there is no any change points within $(\tilde{\tau}_{j-1}^*,\tilde{\tau}_{j}^*]$. In this case, we have $\tau_{k^*-1}^0<\tilde{\tau}_{j-1}^*<\tilde{\tau}_{j}^*<\tau_{k^*}^0$ for some $j$ and $k^*$. Then, $\boldsymbol{\mathrm{\hat{\uptheta}}}_i^*=\boldsymbol{\mathrm Q}_{(\tilde{\tau}_{j-1}^*,\tilde{\tau}_{j}^*)}^{-1}
\left(\boldsymbol{\mathrm Q}_{(\tilde{\tau}_{j-1}^*,\tilde{\tau}_{j}^*)}\boldsymbol{\mathrm \uptheta}_i+\sigma \boldsymbol{\mathrm r}_{(\tilde{\tau}_{j-1}^*,\tilde{\tau}_{j}^*)}\right)=\boldsymbol{\mathrm \uptheta}_i+\sigma \boldsymbol{\mathrm Q}_{(\tilde{\tau}_{j-1}^*,\tilde{\tau}_{j}^*)}^{-1}\boldsymbol{\mathrm r}_{(\tilde{\tau}_{j-1}^*,\tilde{\tau}_{j}^*)}$ and
$\boldsymbol{\mathrm{\hat{\uptheta}}}_i^0=\boldsymbol{\mathrm \uptheta}_i+\sigma \boldsymbol{\mathrm Q}_{({\tau}_{k^*-1}^0,{\tau}_{k^*}^0)}^{-1}
\boldsymbol{\mathrm r}_{({\tau}_{k^*-1}^0,{\tau}_{k^*}^0)}$.
Substituting the above expressions into (\ref{part2-eq2}), we have
\begin{eqnarray}
 (\ref{part2-eq2})&=&\frac{1}{ \Delta_t \sigma^2}\left[\sigma^2 \boldsymbol{\mathrm r}_{(\tilde{\tau}_{j-1}^*,\tilde{\tau}_{j}^*)}^{\top}\boldsymbol{\mathrm Q}_{(\tilde{\tau}_{j-1}^*,\tilde{\tau}_{j}^*)}^{-1}
 \sum_{t_i\in(\tilde{\tau}_{j-1}^*,\tilde{\tau}_{j}^*]}\boldsymbol{\mathrm z}_i^{\top}\boldsymbol{\mathrm z}_i\boldsymbol{\mathrm Q}_{(\tilde{\tau}_{j-1}^*,
 \tilde{\tau}_{j}^*)}^{-1}
 \boldsymbol{\mathrm r}_{(\tilde{\tau}_{j-1}^*,\tilde{\tau}_{j}^*)}-\sigma^2  \boldsymbol{\mathrm r}_{({\tau}_{k^*-1}^0,{\tau}_{k^*}^0)}^{\top}\boldsymbol{\mathrm Q}_{({\tau}_{k^*-1}^0,{\tau}_{k^*}^0)}^{-1}
 \right.
 \nonumber\\&&\times \sum_{t_i\in(\tilde{\tau}_{j-1}^*,\tilde{\tau}_{j}^*]}\boldsymbol{\mathrm z}_i^{\top}\boldsymbol{\mathrm z}_i\boldsymbol{\mathrm Q}_{({\tau}_{k^*-1}^0,{\tau}_{k^*}^0)}^{-1}
 \boldsymbol{\mathrm r}_{({\tau}_{k^*-1}^0,{\tau}_{k^*}^0)}+2\sum_{t_i\in(\tilde{\tau}_{j-1}^*,\tilde{\tau}_{j}^*]}u_i^{\top}\boldsymbol{\mathrm z}_i \sigma (\boldsymbol{\mathrm Q}_{(\tilde{\tau}_{j-1}^*,\tilde{\tau}_{j}^*)}^{-1}\boldsymbol{\mathrm r}_{(\tilde{\tau}_{j-1}^*,\tilde{\tau}_{j}^*)}
 -\boldsymbol{\mathrm Q}_{({\tau}_{k^*-1}^0,{\tau}_{k^*}^0)}^{-1}\boldsymbol{\mathrm r}_{({\tau}_{k^*-1}^0,{\tau}_{k^*}^0)})\nonumber\\
&&\left. +\frac{(m^*-m^0)(p+1)\log(T/\Delta_t^*)}{m^*-m^0+1}
\right]. \label{part2-eq3}
\end{eqnarray}
\ \\
\noindent From the approach used in the proof of Propositions~\ref{prnm13} and \ref{prnm14}, we have
\begin{eqnarray*} \norm{\boldsymbol{\mathrm r}_{(\tilde{\tau}_{j-1}^*,{\hat{\tau}}_{2}^*)}^{\top}
\boldsymbol{\mathrm Q}_{(\tilde{\tau}_{j-1}^*,{\hat{\tau}}_{2}^*)}^{-1}\sum_{t_i\in(\tilde{\tau}_{j-1}^*,{\hat{\tau}}_{2}^*]}
\boldsymbol{\mathrm z}_i^{\top}\boldsymbol{\mathrm z}_i\boldsymbol{\mathrm Q}_{(\tilde{\tau}_{j-1}^*,{\hat{\tau}}_{2}^*)}^{-1}\boldsymbol{\mathrm r}_{(\tilde{\tau}_{j-1}^*,{\hat{\tau}}_{2}^*)}}\leq
\norm{ \frac{1}{\sqrt{T}}\boldsymbol{\mathrm r}_{(\tilde{\tau}_{j-1}^*,{\hat{\tau}}_{2}^*)}}^2
\norm{T\boldsymbol{\mathrm Q}_{(\tilde{\tau}_{j-1}^*,{\hat{\tau}}_{2}^*)}^{-1}}^2
\norm{\frac{1}{T}\sum_{t_i\in(\tilde{\tau}_{j-1}^*,{\hat{\tau}}_{2}^*]}\boldsymbol{\mathrm z}_i^{\top}\boldsymbol{\mathrm z}_i}\\
=O_p\left(\log^{2a^*} (T)\right)\end{eqnarray*}
for some $0<a^*<1/2$. Similar results also hold for the second and the third terms of (\ref{part2-eq3}). \\
\ \\
Therefore, for large $T$, (\ref{part2-eq3}) is dominated by $\displaystyle \frac{(m^*-m^0)(p+1)\log(T/\Delta_t^*)}{m^*-m^0+1}$, which is positive. This implies that for large $T$, (\ref{part2-eq1}) is positive with probability 1.\\
\ \\
\noindent Now, consider the case where there exist $m_{j}$ ($0<m_j\leq m^0$) exact change points (so are the estimated change points) in $(\tilde{\tau}_{j-1}^*,\tilde{\tau}_{j}^*]$. We label these $m_{j}$ exact change points by
$\tilde{\tau}_{j-1}^*<{\tau}_{(m_j,1)}^0<,...,<{\tau}_{(m_j,m_j)}^0<\tilde{\tau}_{j}^*$ and similarly for the estimated change points, $\tilde{\tau}_{j-1}^*<{\hat{\tau}}_{(m_j,1)}^*<,~\cdots,~<{\hat{\tau}}_{(m_j,m_j)}^*<\tilde{\tau}_{j}^*$. \\
\ \\
\noindent By the quadratic structure,
$$\sum_{t_i\in(\tilde{\tau}_{j-1}^*,\tilde{\tau}_{j}^*]}(Y_i-\boldsymbol{\mathrm z}_i\boldsymbol{\mathrm{\hat{\uptheta}}}_i^0)^2\leq \sum_{t_i\in [0,T]}(Y_i-\boldsymbol{\mathrm z}_i\boldsymbol{\mathrm{\hat{\uptheta}}}_i^0)^2=\sum_{j=1}^{m}\boldsymbol{\mathrm r}_{(\tau_{j-1}^0,\tau_j^0)}^{\top}
\boldsymbol{\mathrm Q}_{(\tau_{j-1}^0,\tau_j^0)}^{-1}\sum_{t_i\in(\tau_{j-1}^0,\tau_j^0]}\boldsymbol{\mathrm z}_i'\boldsymbol{\mathrm z}_i\boldsymbol{\mathrm Q}_{(\tau_{j-1}^0,\tau_j^0)}^{-1} \boldsymbol{\mathrm r}_{(\tau_{j-1}^0,\tau_j^0)}.$$
\ \\
\noindent It follows from (\ref{prnm1-eq3-1}) that $\sum_{j=1}^{m}\boldsymbol{\mathrm r}_{(\tau_{j-1}^0,\tau_j^0)}^{\top}\boldsymbol{\mathrm Q}_{(\tau_{j-1}^0,\tau_j^0)}^{-1}
\sum_{t_i\in(\tau_{j-1}^0,\tau_j^0]}\boldsymbol{\mathrm z}_i^{\top}\boldsymbol{\mathrm z}_i\boldsymbol{\mathrm Q}_{(\tau_{j-1}^0,\tau_j^0)}^{-1} \boldsymbol{\mathrm r}_{(\tau_{j-1}^0,\tau_j^0)}=O_p\left( (\log T)^{2a^*}\right)$ for some $0<a^*<1/2$. By
 similar methods employed in the proof of Proposition~\ref{prnm13}, we have $\sum_{t_i\in(\tilde{\tau}_{j-1}^*,\tilde{\tau}_{j}^*]}2u_i^{\top}\boldsymbol{\mathrm z}_i(\boldsymbol{\mathrm{\hat{\uptheta}}}_i^*
 -\boldsymbol{\mathrm{\hat{\uptheta}}}_i^0)=O_p\left( (\log T)^{2a^*} \right)$. Since for large $T$, $(\log T)^{2a^*}<\displaystyle \frac{(m^*-m^0)(p+1)\log(T/\Delta_t^*)}{m^*-m^0+1}$, we have that for large $T$, $(\ref{part2-eq2})$ is dominated by either $\displaystyle \frac{(m^*-m^0)(p+1)\log(T/\Delta_t^*)}{(m^*-m^0+1)\Delta_t \sigma^2}$ or $\displaystyle \frac{1}{ \Delta_t \sigma^2} \sum_{t_i\in(\tilde{\tau}_{j-1}^*,\tilde{\tau}_{j}^*]}(\boldsymbol{\mathrm z}_i({\uptheta}_i-\boldsymbol{\mathrm{\hat{\uptheta}}}_i^*))^2$ and they are both positive. This implies that for large $T$, (\ref{part2-eq1}) is positive with probability 1.
\end{proof}
\section{Histograms of the estimated change points' arrival rates (Section~\ref{simulation-1})}\label{appendixhistogram}
\begin{figure}[htbp]
\includegraphics[height=1.85in,width=2.04in]{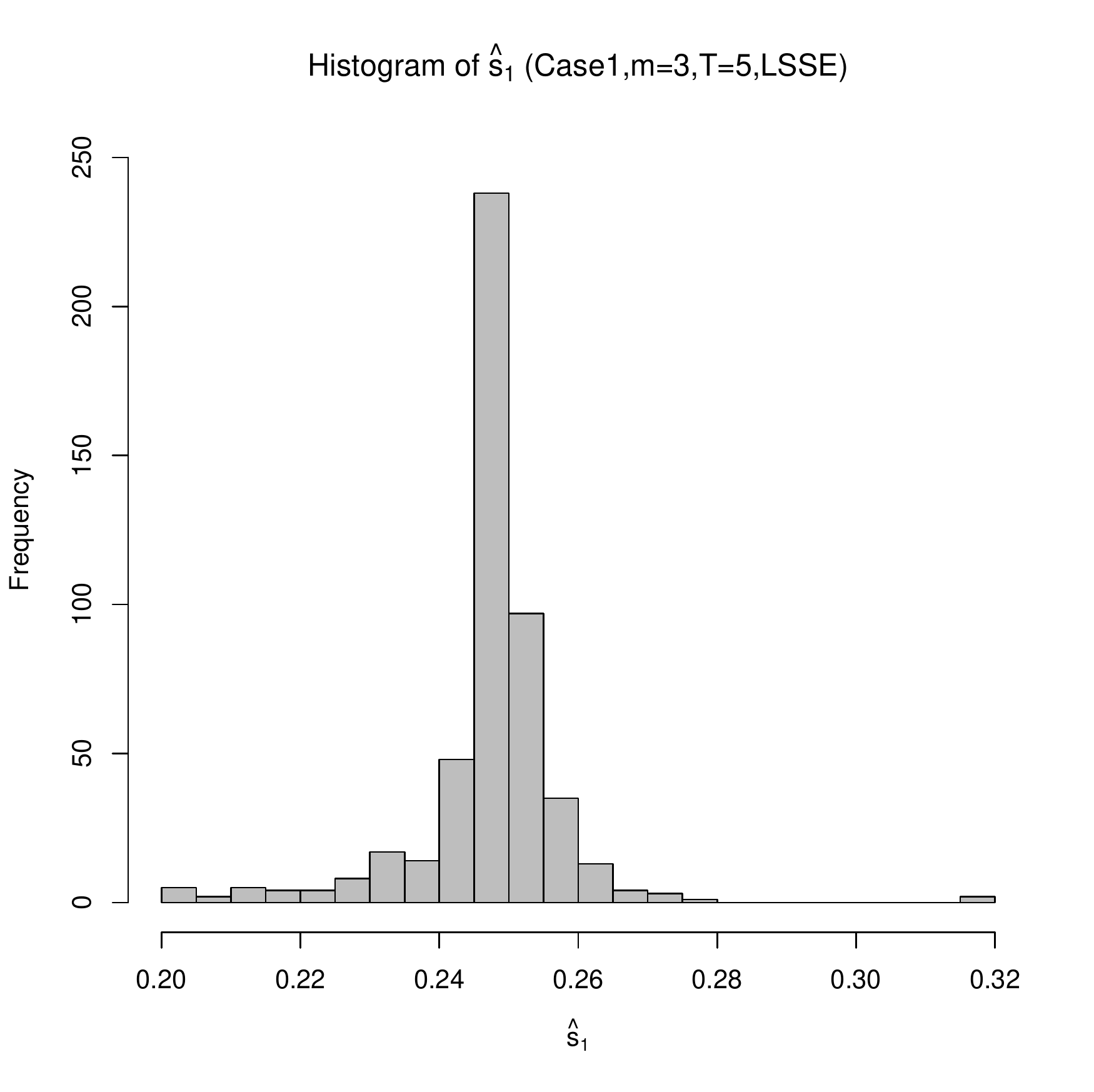}
\includegraphics[height=1.85in,width=2.04in]{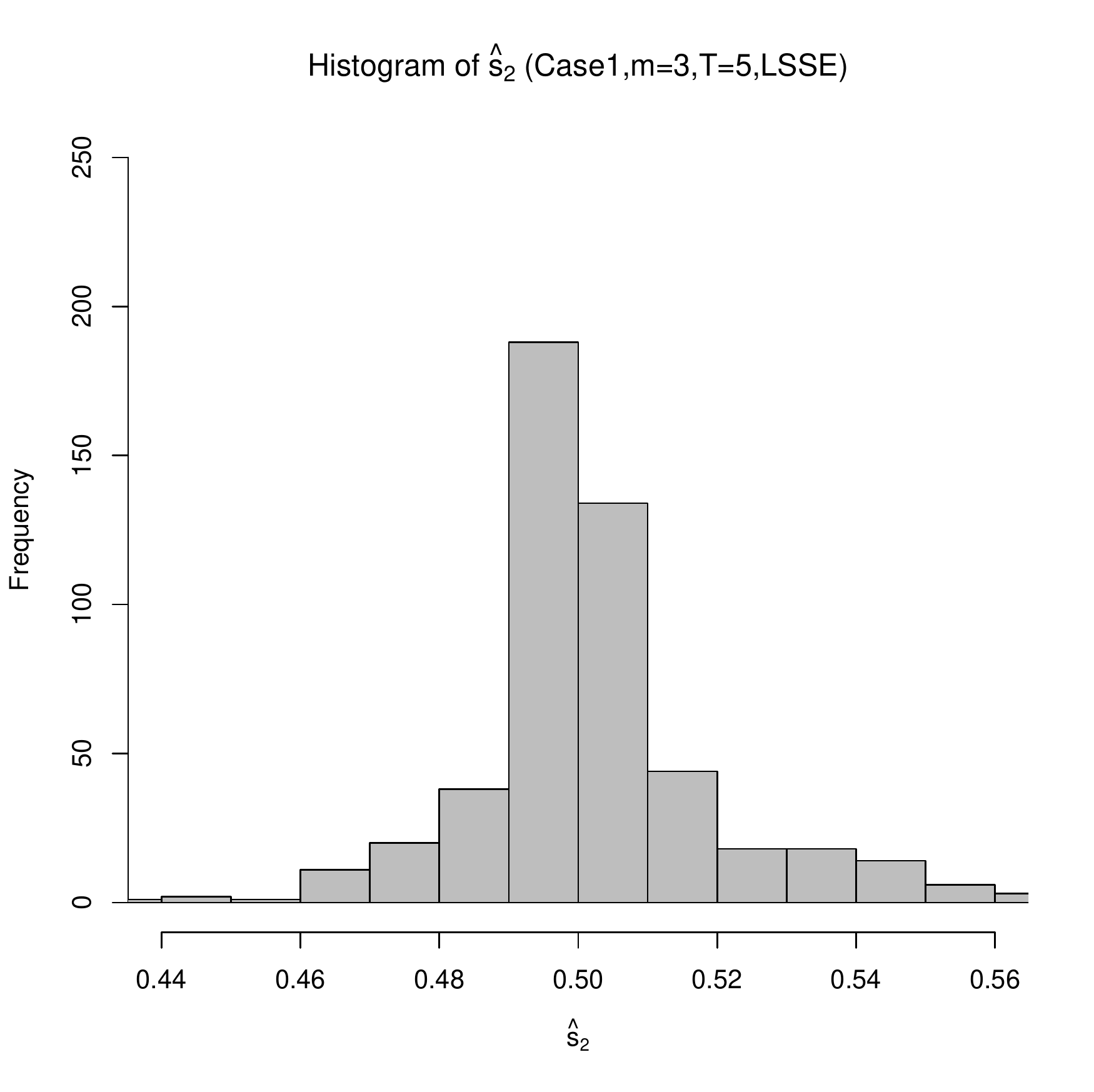}
\includegraphics[height=1.85in,width=2.04in]{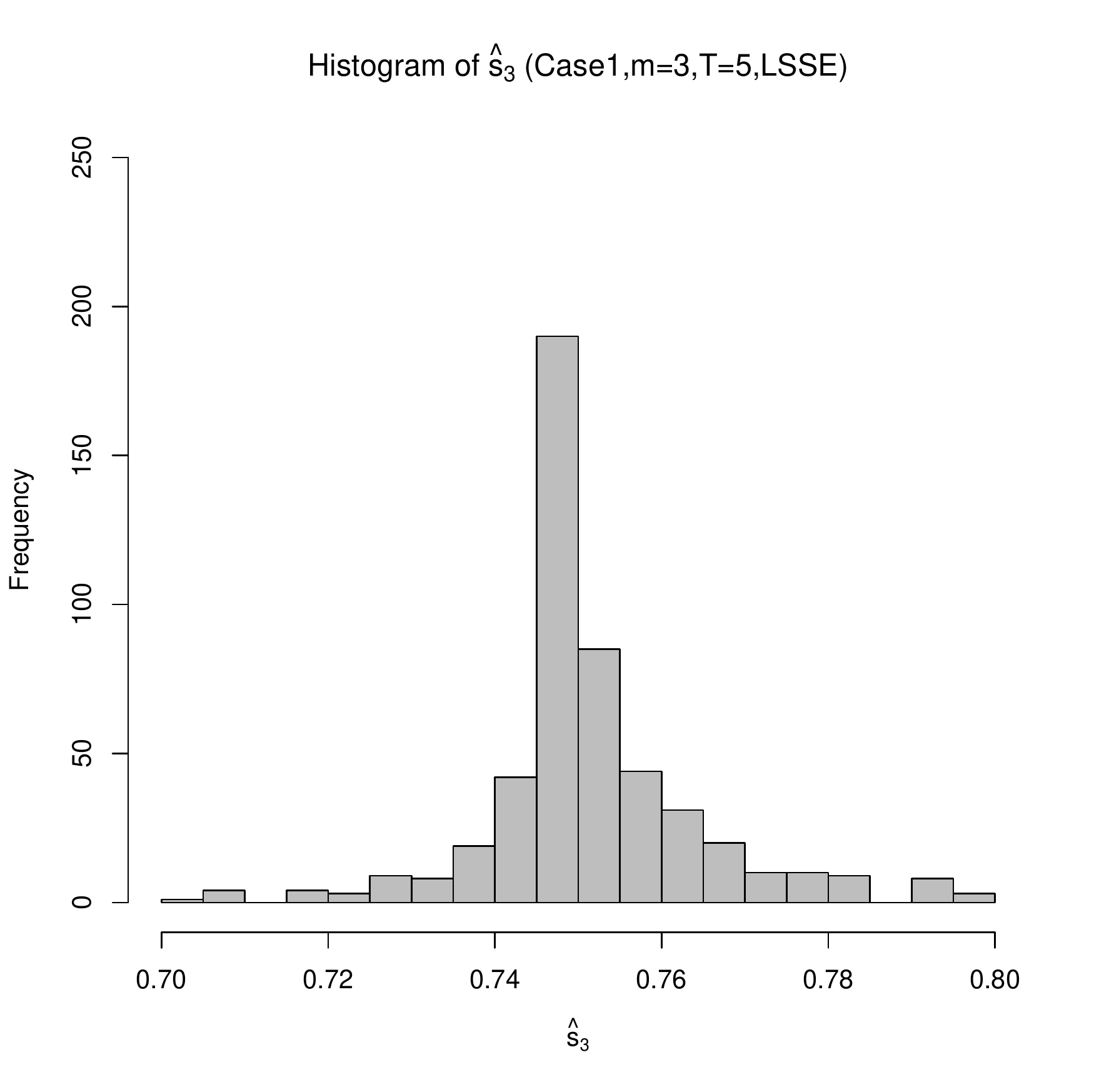}\\
\includegraphics[height=1.85in,width=2.04in]{histMLLCase1m3t5s1.pdf}
\includegraphics[height=1.85in,width=2.04in]{histMLLCase1m3t5s2.pdf}
\includegraphics[height=1.85in,width=2.04in]{histMLLCase1m3t5s3.pdf}\\
\includegraphics[height=1.85in,width=2.04in]{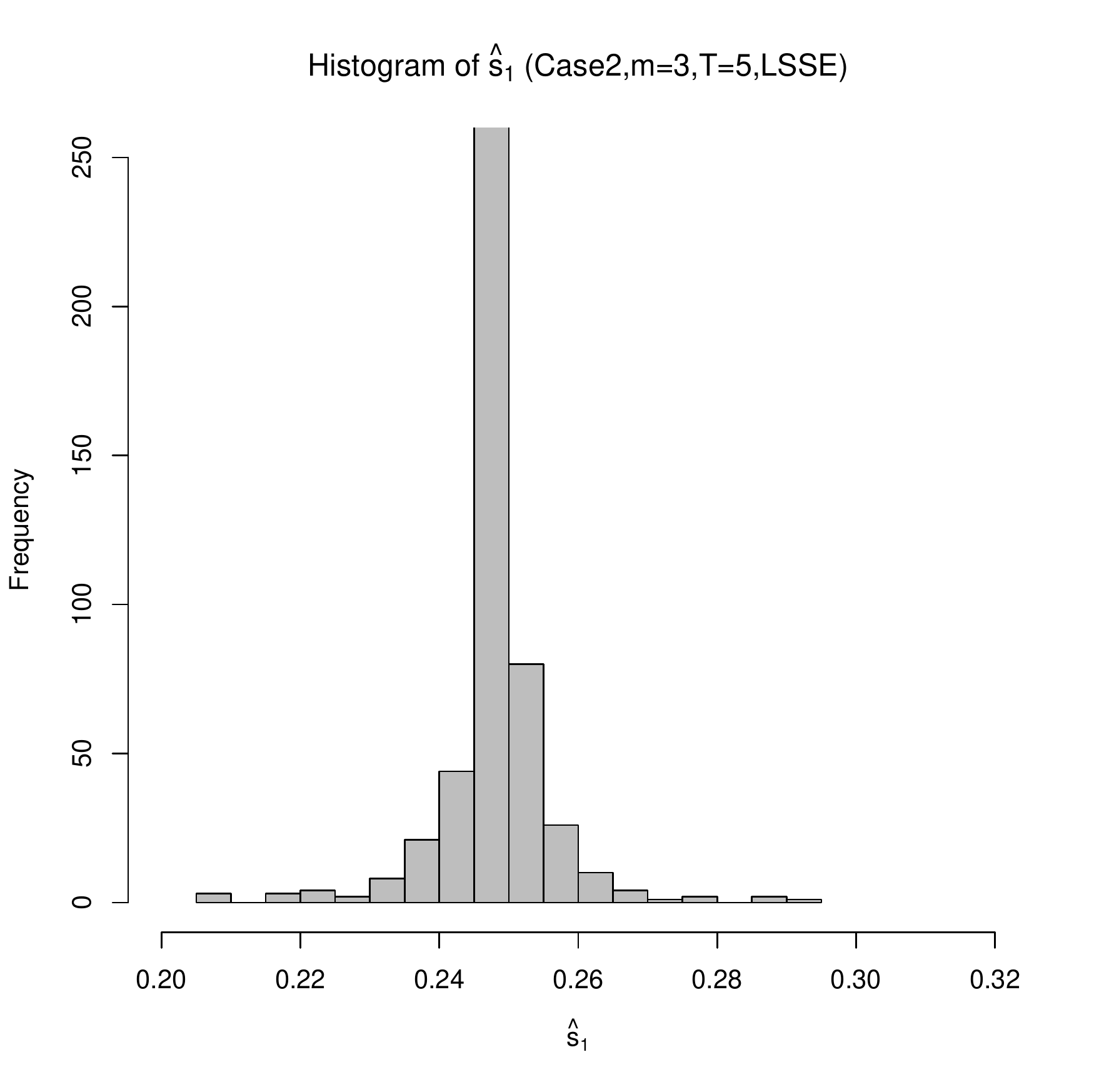}
\includegraphics[height=1.85in,width=2.04in]{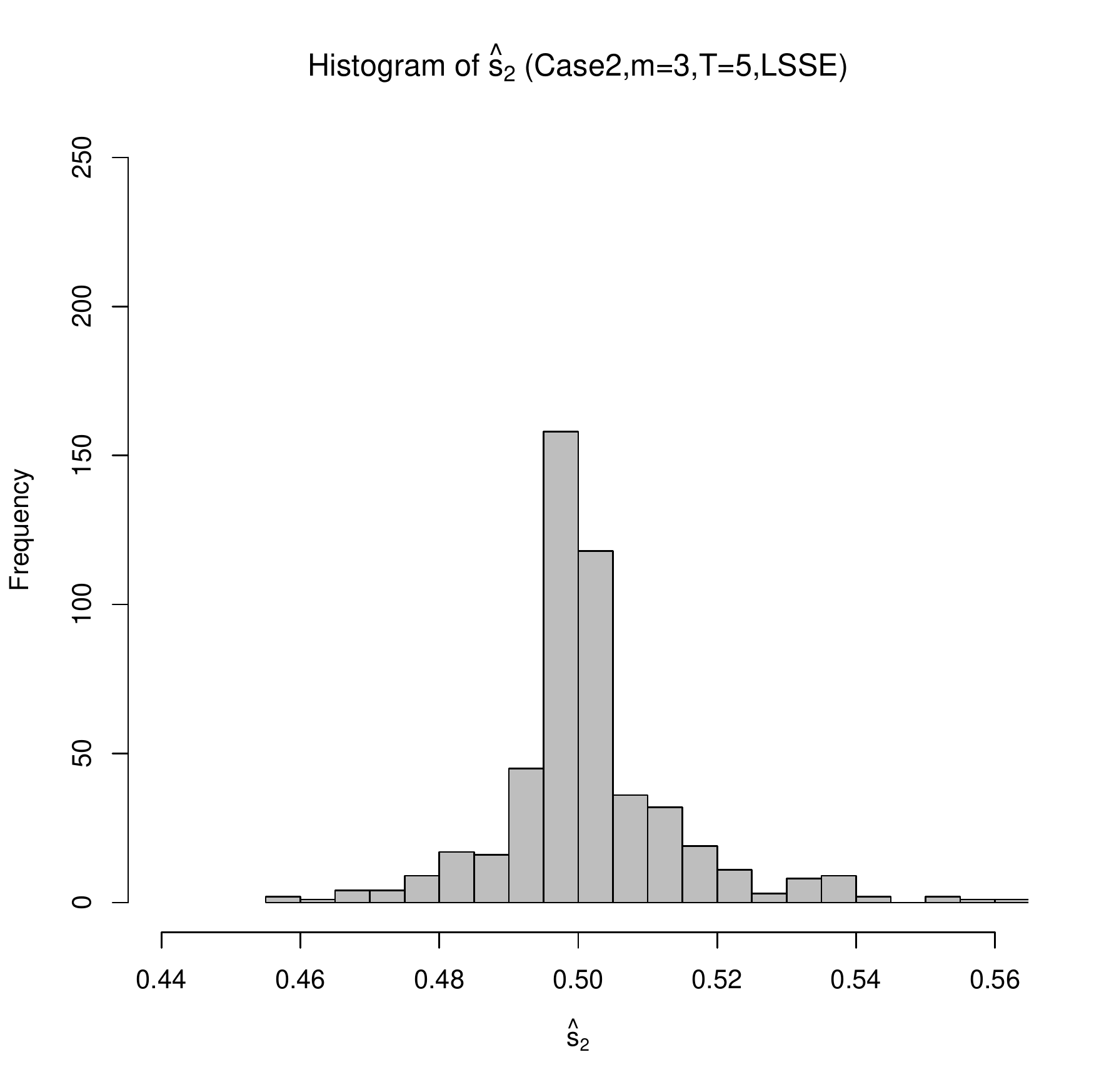}
\includegraphics[height=1.85in,width=2.04in]{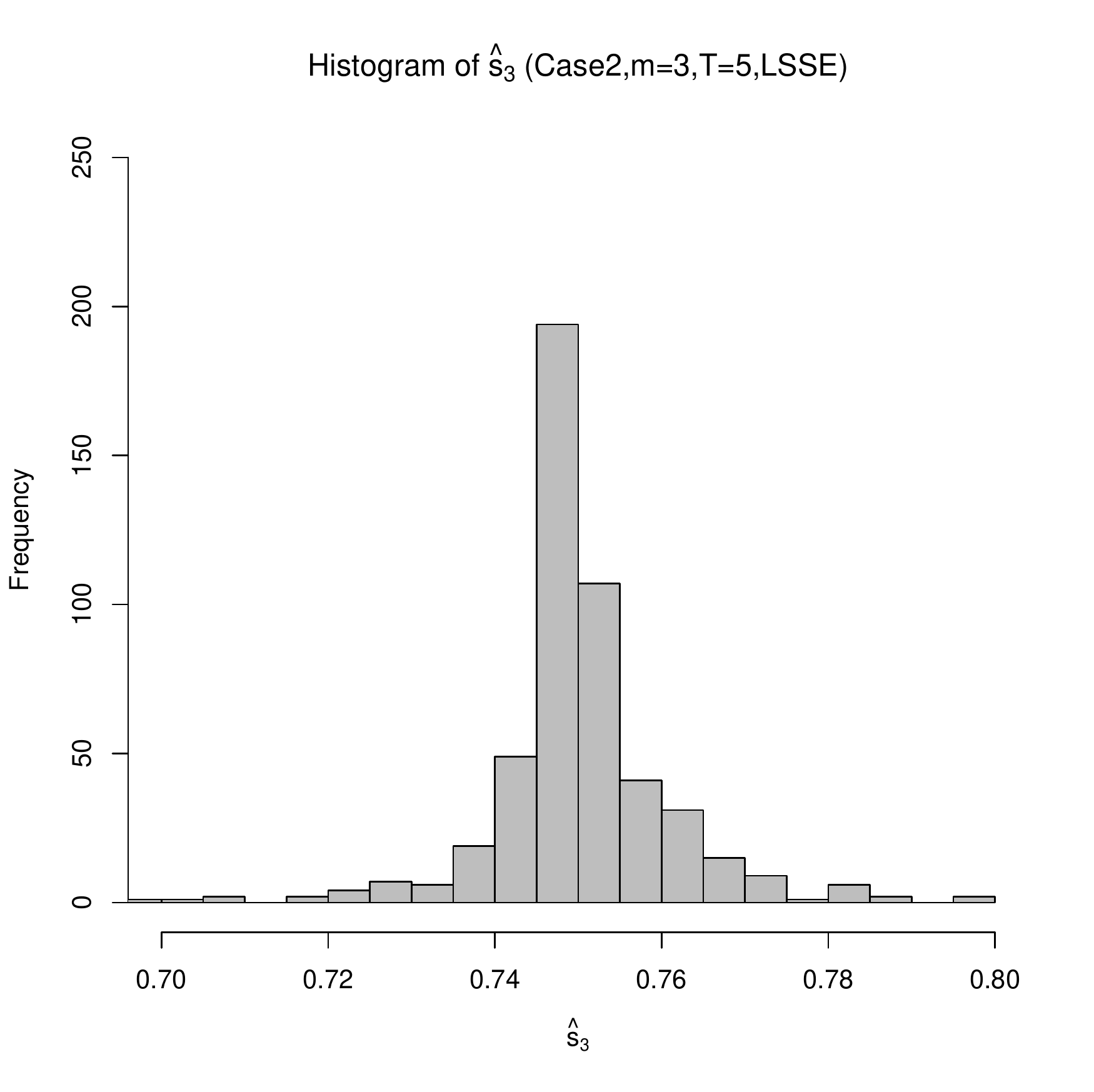}\\
\includegraphics[height=1.85in,width=2.04in]{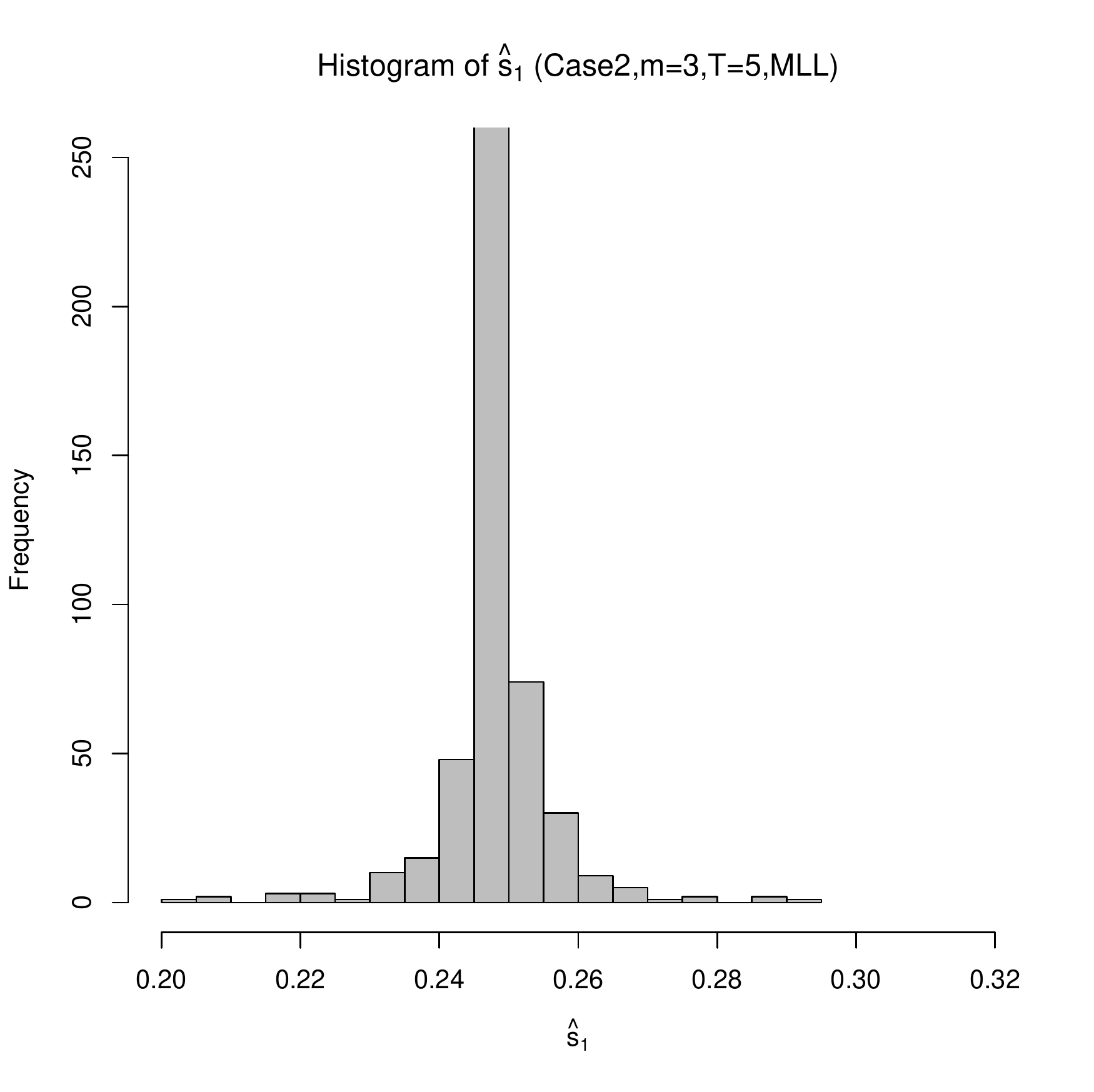}
\includegraphics[height=1.85in,width=2.04in]{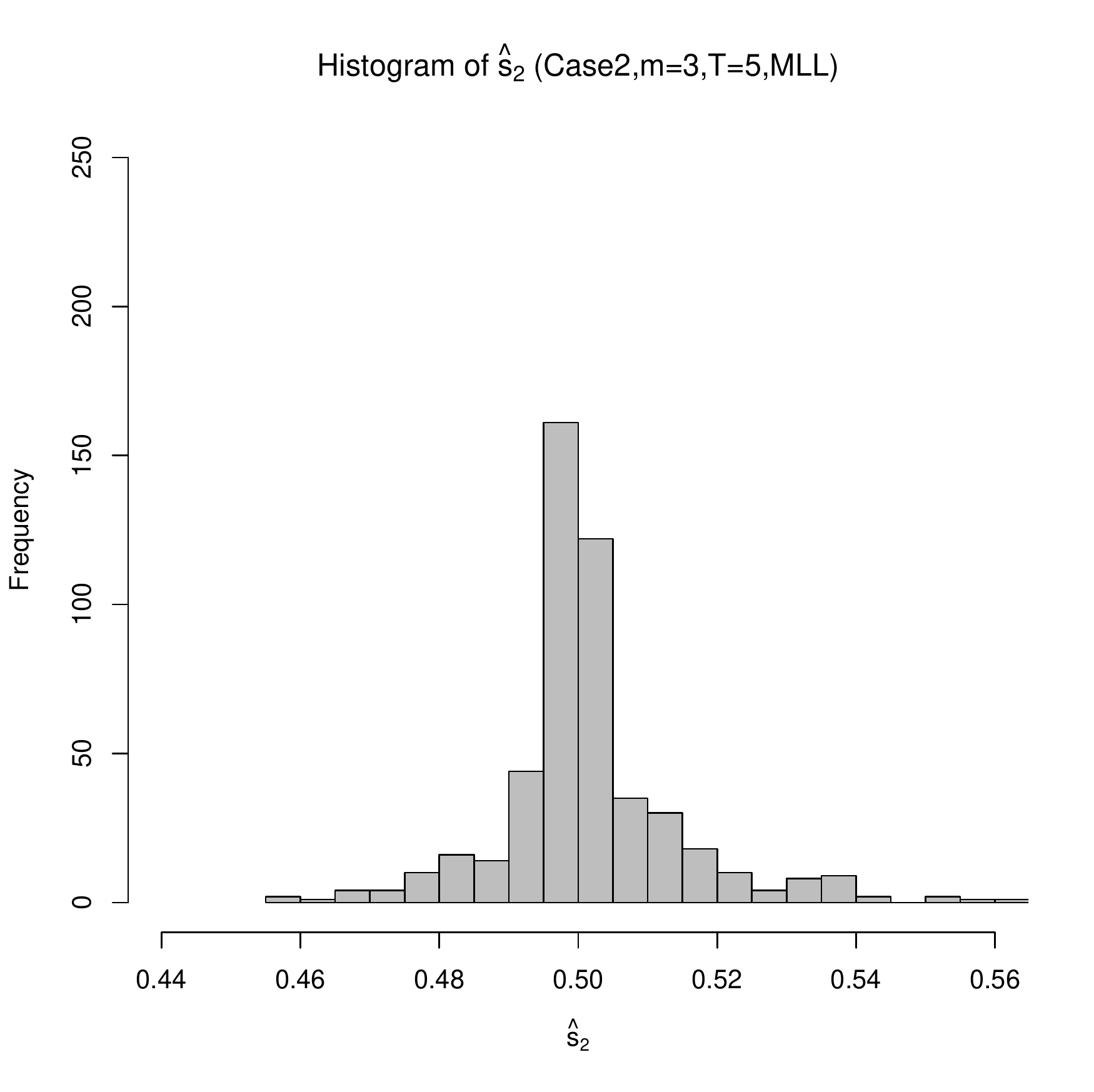}
\includegraphics[height=1.85in,width=2.04in]{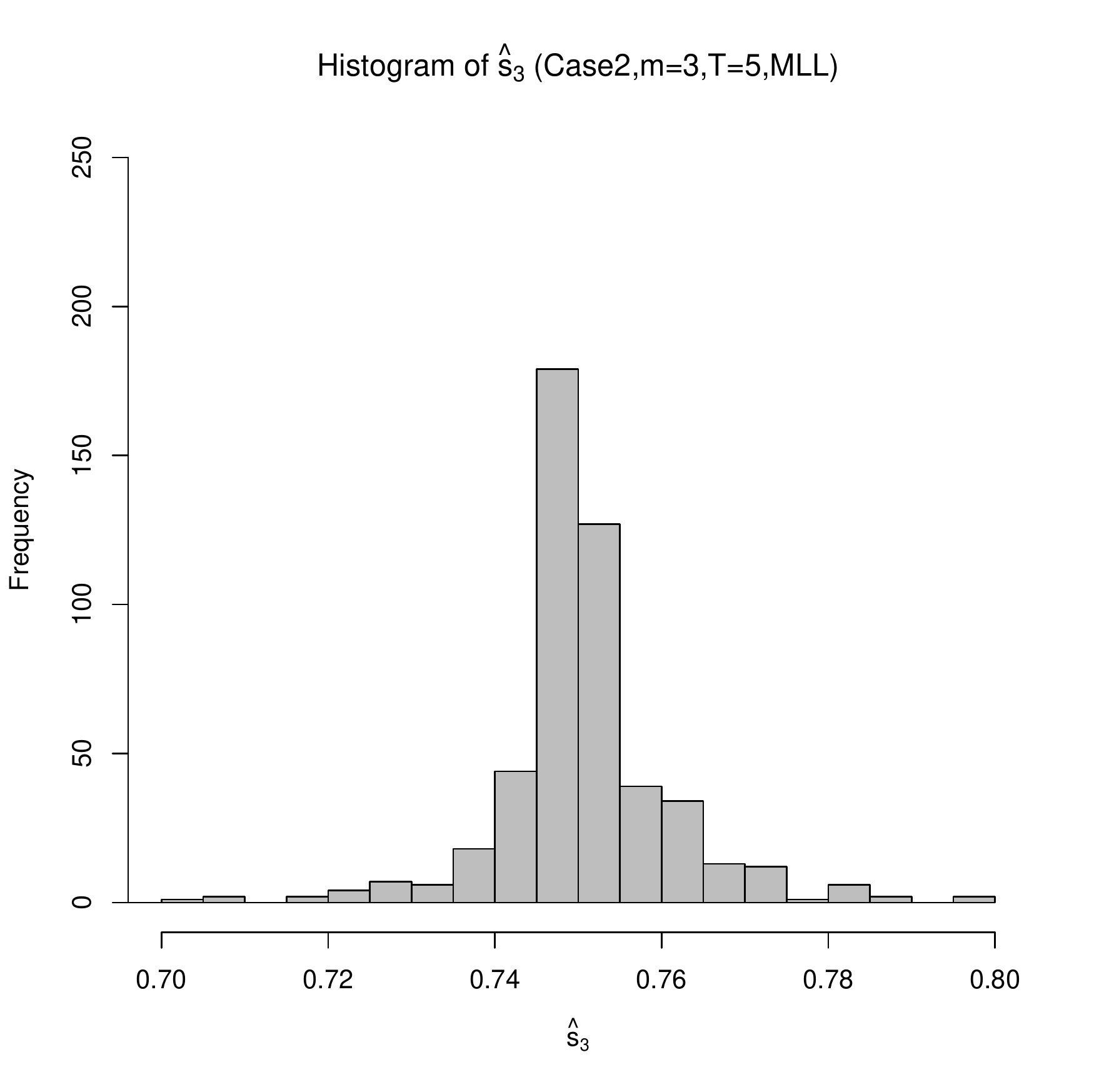}
\caption{\small Histogram of $\hat{s}$ for $m^0=3$ and $T=5$ with exact value $s^0=(0.25, 0.50, 0.75)$}
\label{figure4}
\end{figure}

\begin{figure}[htbp]
\includegraphics[height=2.35in,width=2.04in]{histLSSECase1m3t10s1.pdf}
\includegraphics[height=2.35in,width=2.04in]{histLSSECase1m3t10s2.pdf}
\includegraphics[height=2.35in,width=2.04in]{histLSSECase1m3t10s3.pdf}\\
\includegraphics[height=2.35in,width=2.04in]{histMLLCase1m3t10s1.pdf}
\includegraphics[height=2.35in,width=2.04in]{histMLLCase1m3t10s2.pdf}
\includegraphics[height=2.35in,width=2.04in]{histMLLCase1m3t10s3.pdf}\\
\includegraphics[height=2.35in,width=2.04in]{histLSSECase2m3t10s1.pdf}
\includegraphics[height=2.35in,width=2.04in]{histLSSECase2m3t10s2.pdf}
\includegraphics[height=2.35in,width=2.04in]{histLSSECase2m3t10s3.pdf}\\
\includegraphics[height=2.35in,width=2.04in]{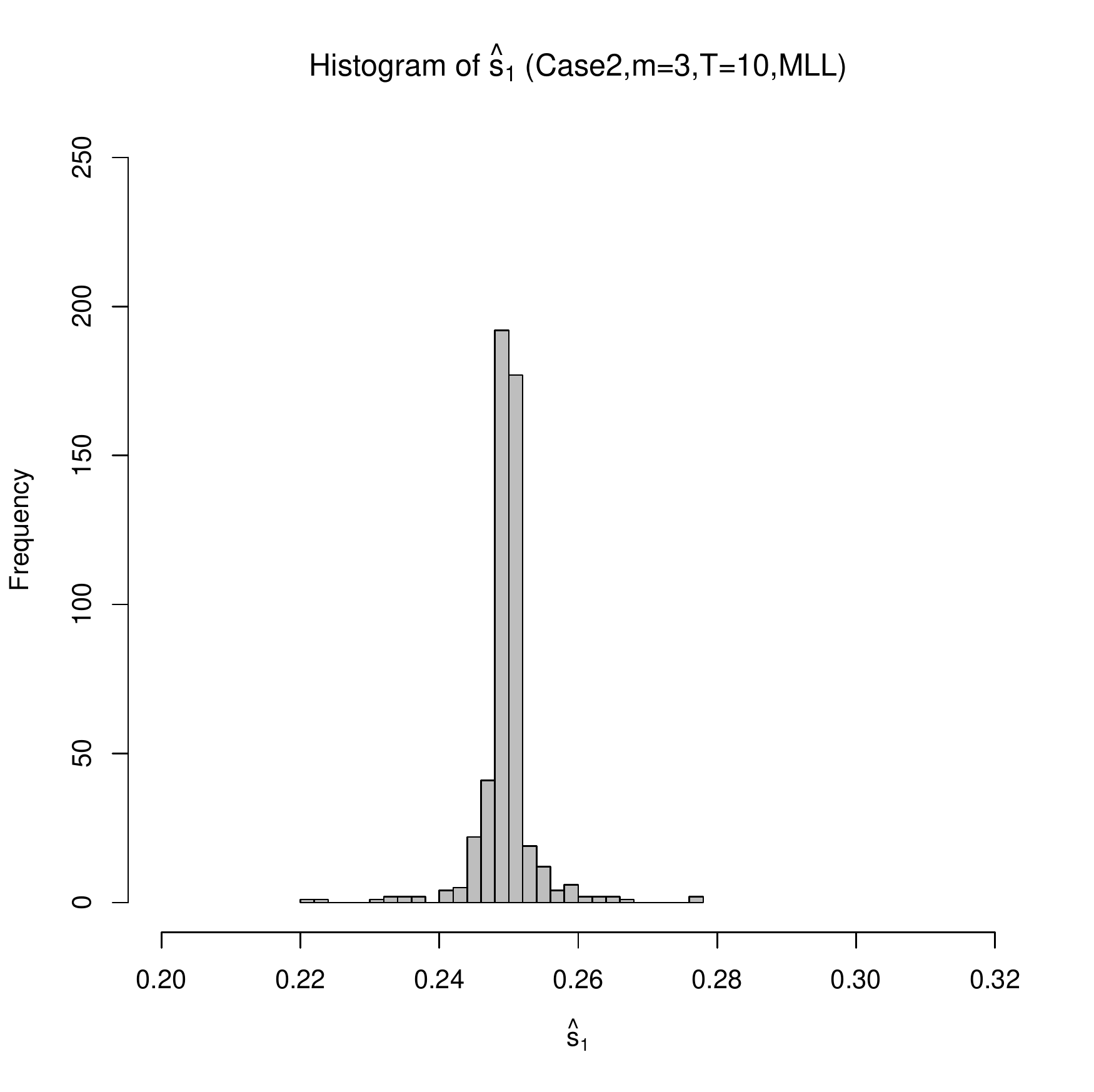}
\includegraphics[height=2.35in,width=2.04in]{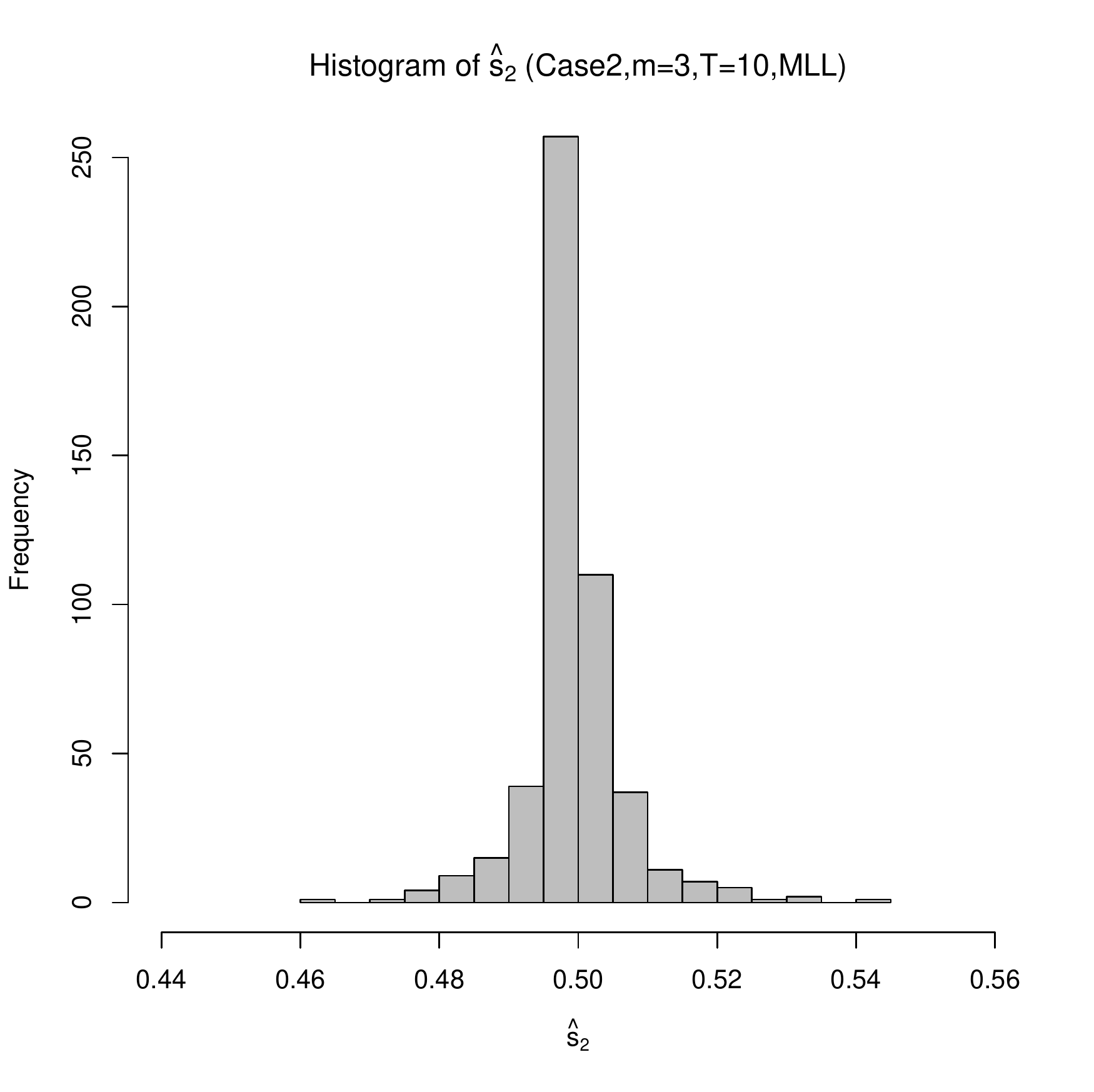}
\includegraphics[height=2.35in,width=2.04in]{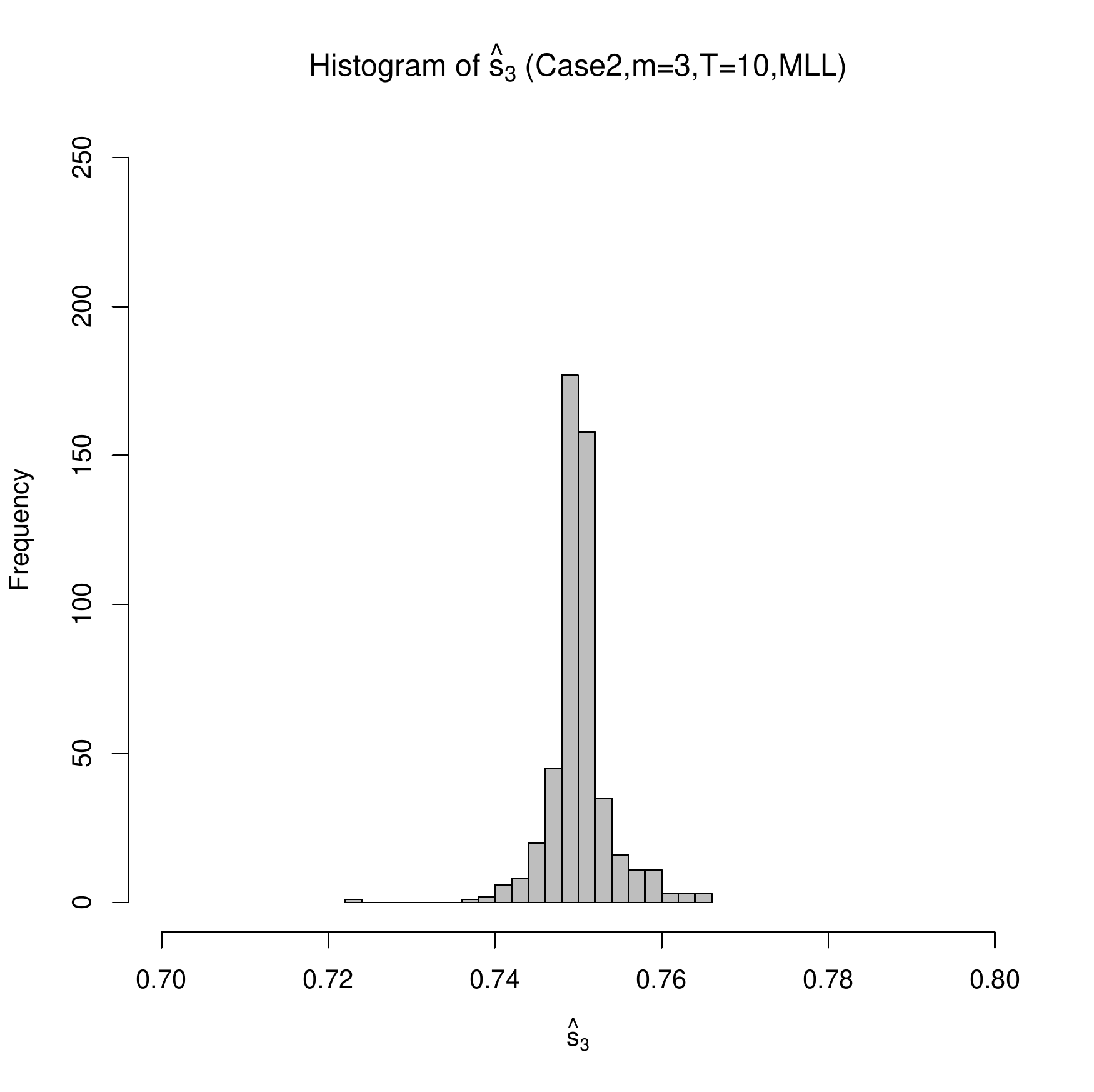}
\caption{\small Histogram of $\hat{s}$ for $m^0=3$ and $T=10$ with exact value $s^0=(0.25, 0.50, 0.75)$}
\label{figure5}
\end{figure}
\begin{figure}[htbp]
\includegraphics[height=2.35in,width=2.04in]{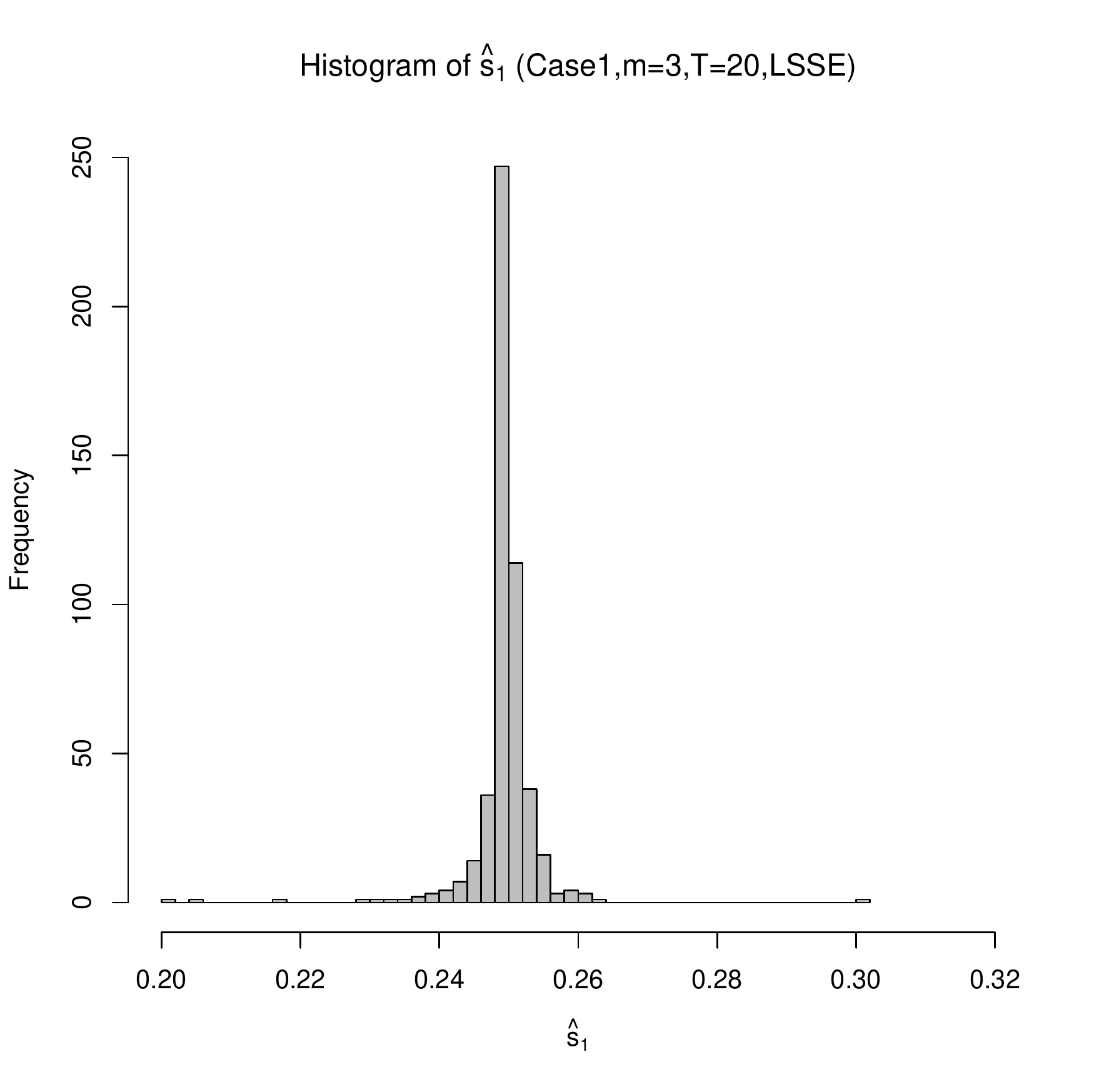}
\includegraphics[height=2.35in,width=2.04in]{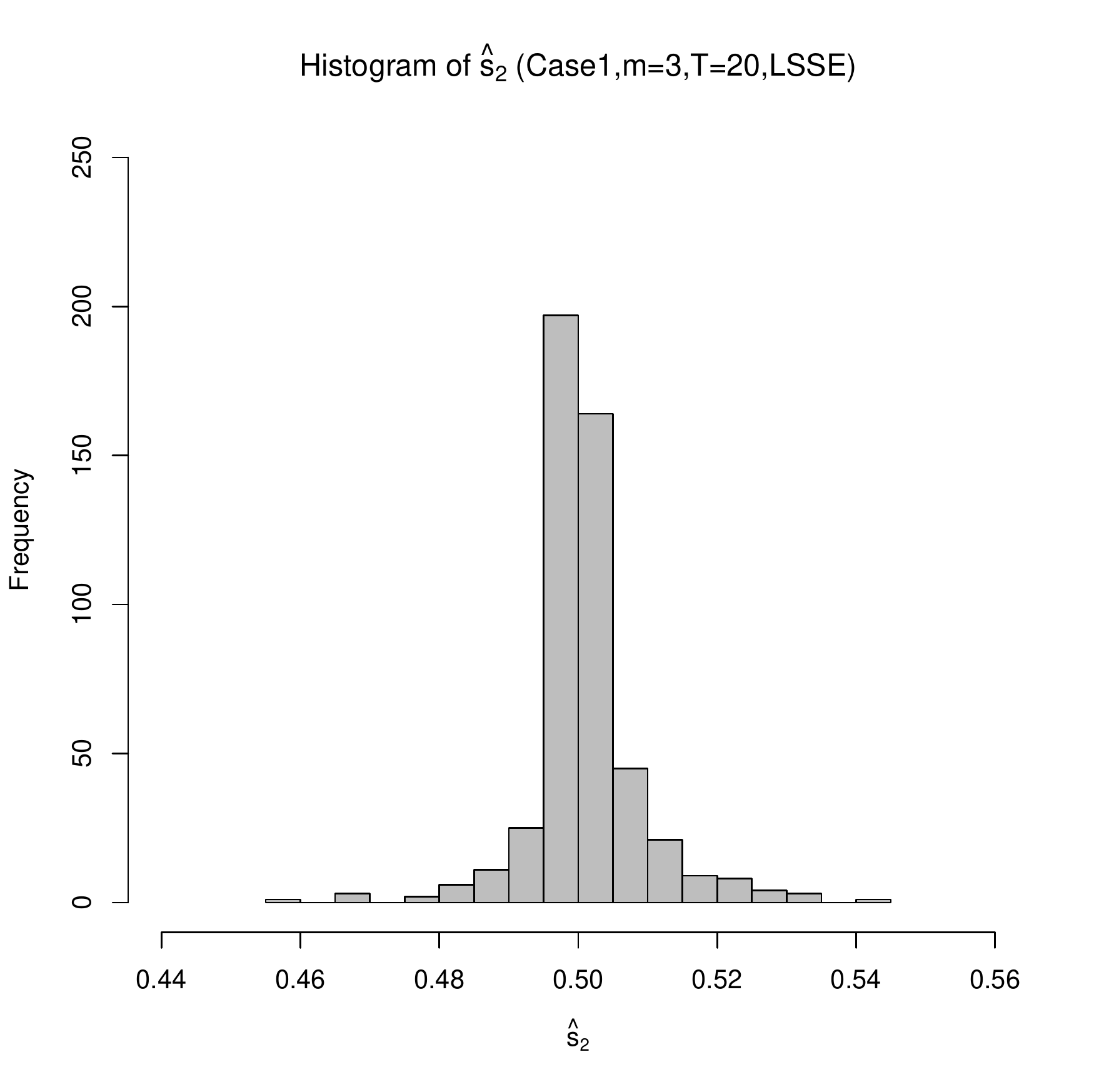}
\includegraphics[height=2.35in,width=2.04in]{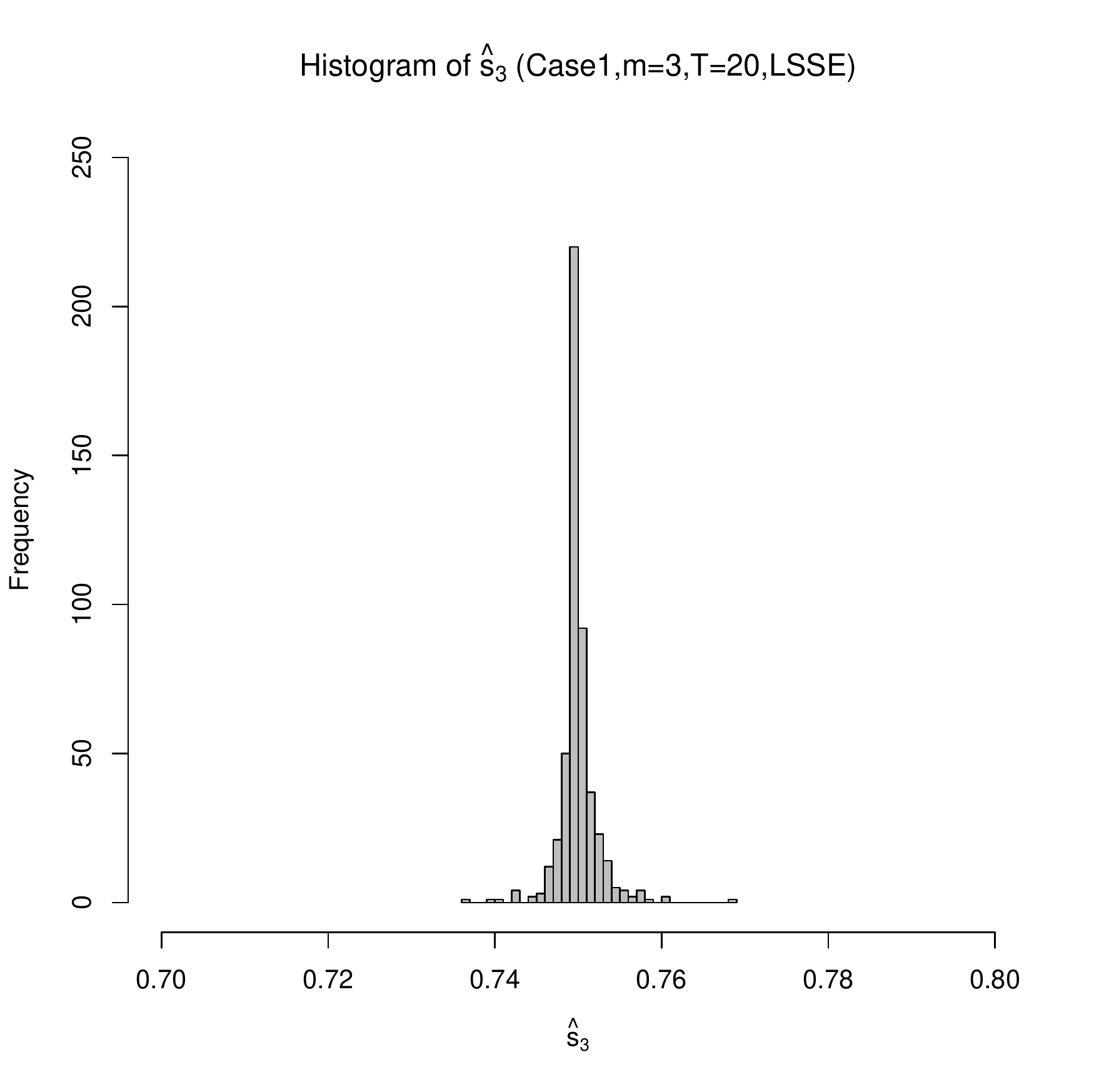}\\
\includegraphics[height=2.35in,width=2.04in]{histMLLCase1m3t20s1.pdf}
\includegraphics[height=2.35in,width=2.04in]{histMLLCase1m3t20s2.pdf}
\includegraphics[height=2.35in,width=2.04in]{histMLLCase1m3t20s3.pdf}\\
\includegraphics[height=2.35in,width=2.04in]{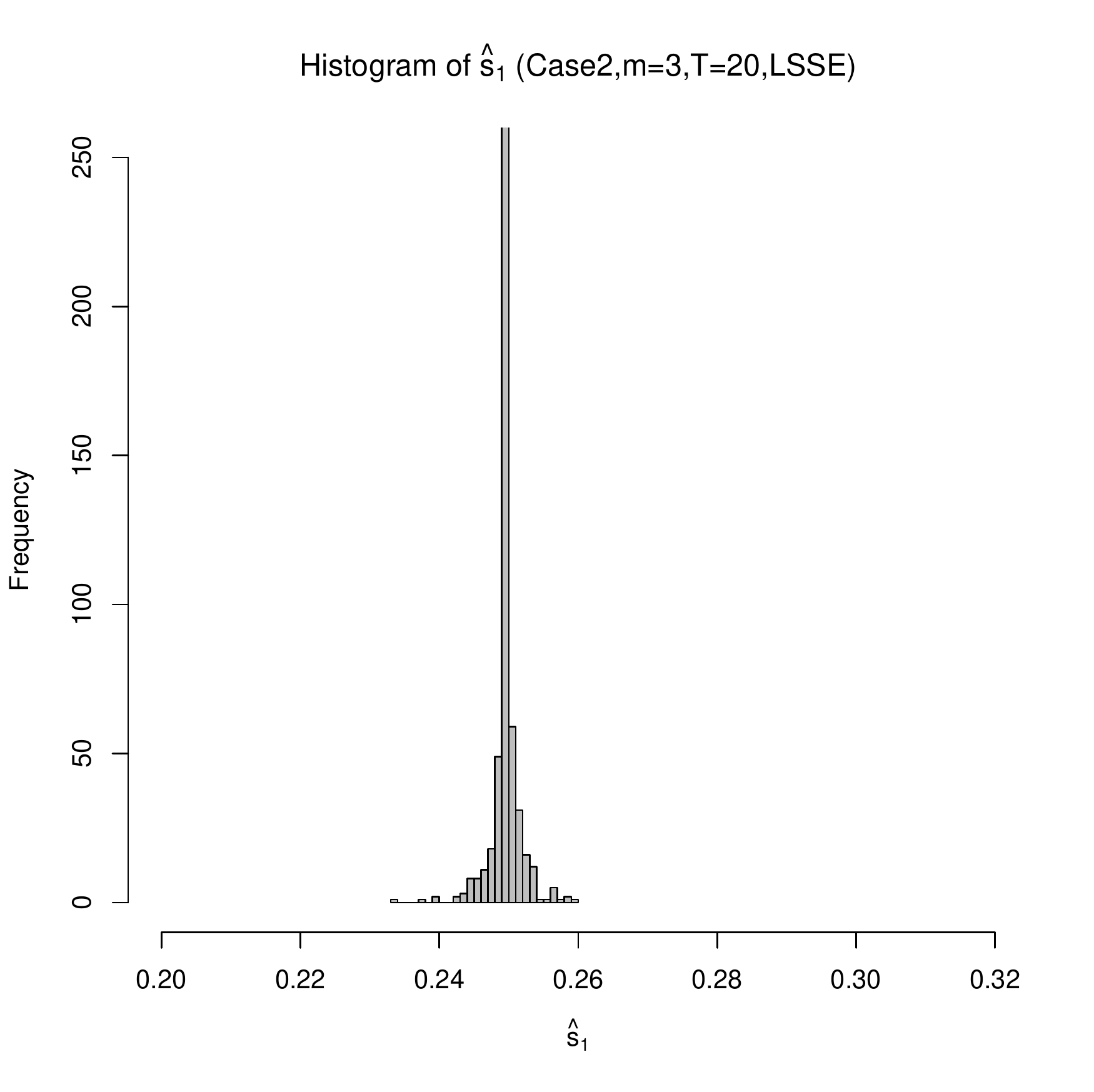}
\includegraphics[height=2.35in,width=2.04in]{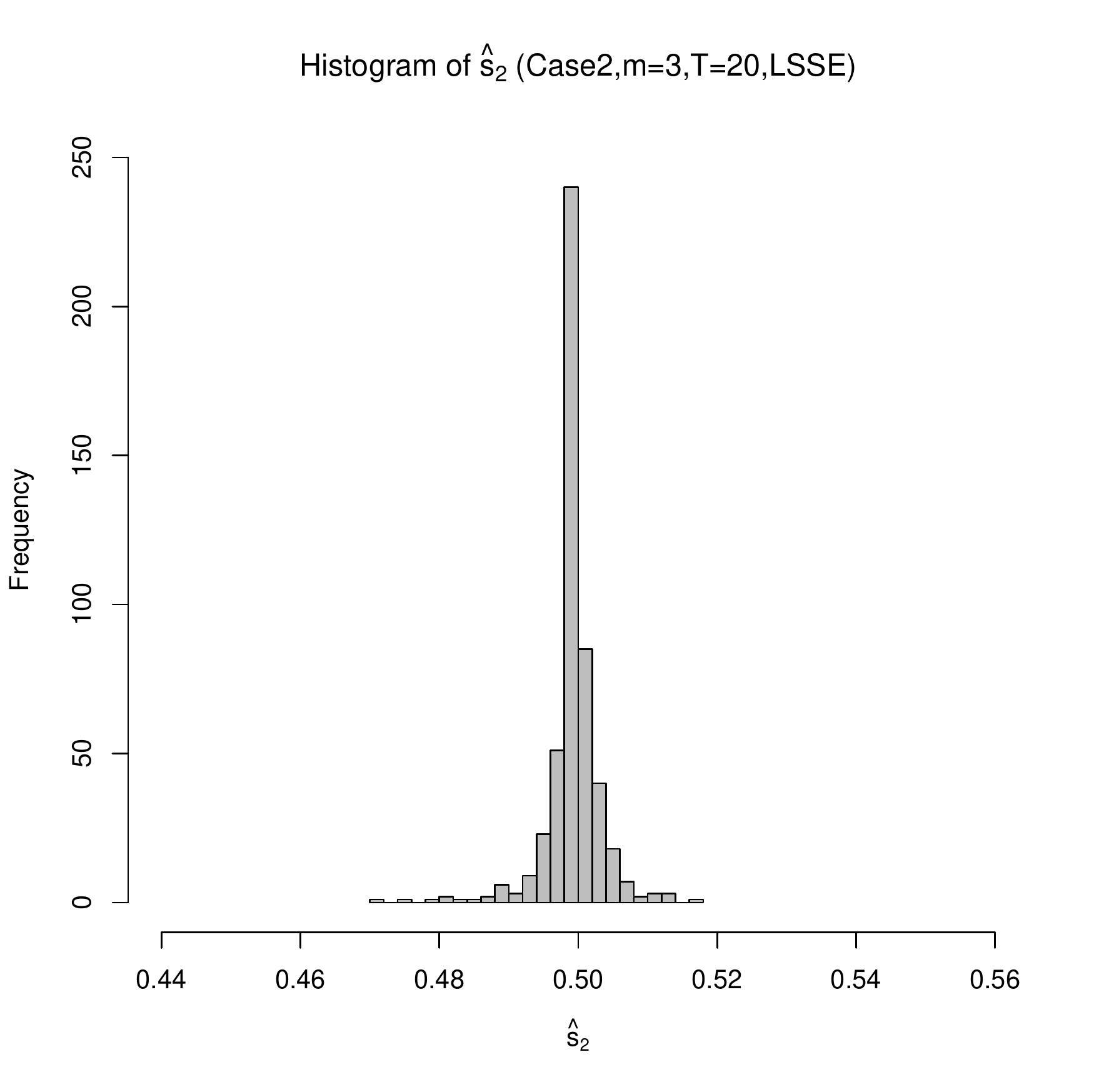}
\includegraphics[height=2.35in,width=2.04in]{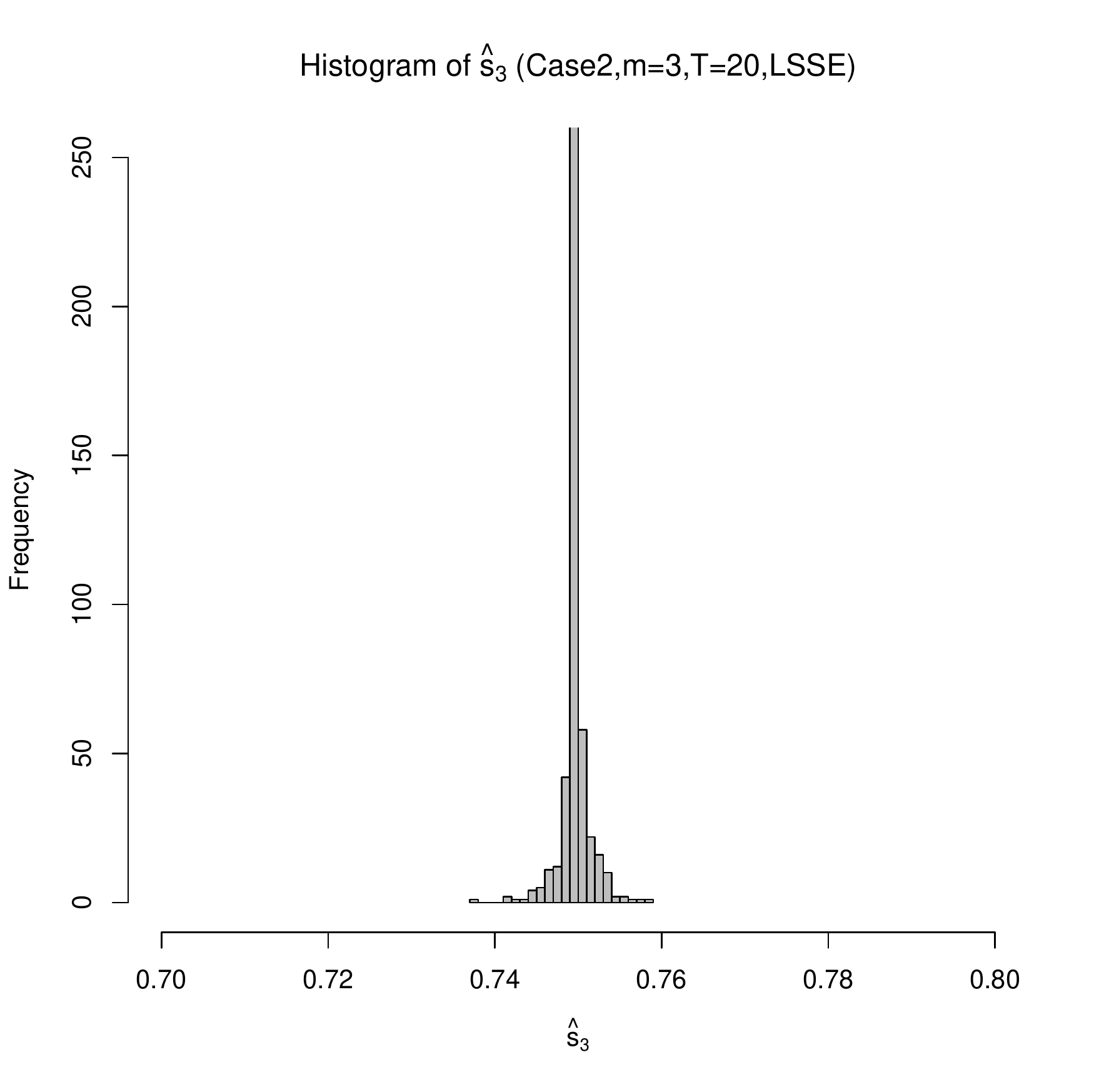}\\
\includegraphics[height=2.35in,width=2.04in]{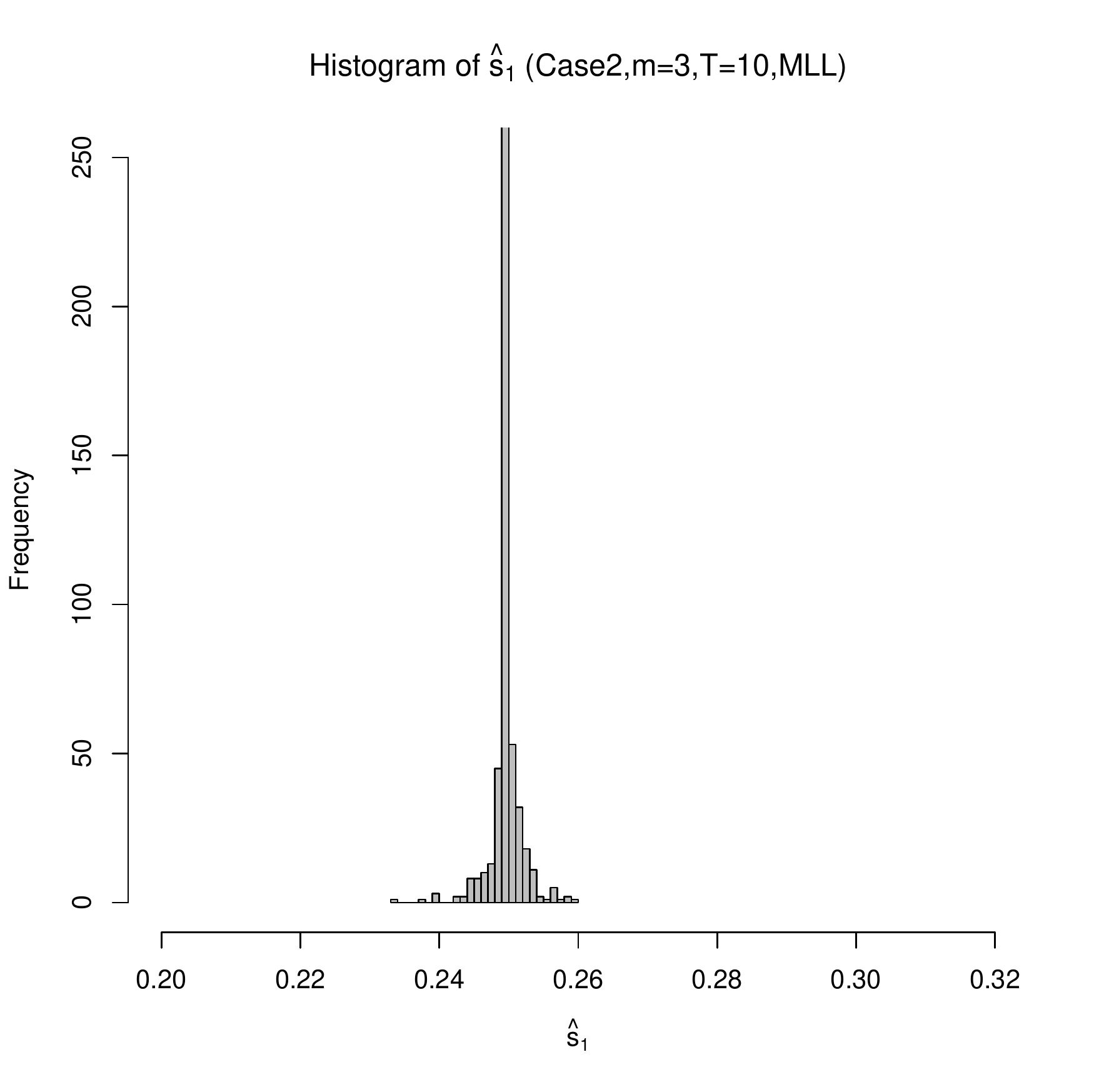}
\includegraphics[height=2.35in,width=2.04in]{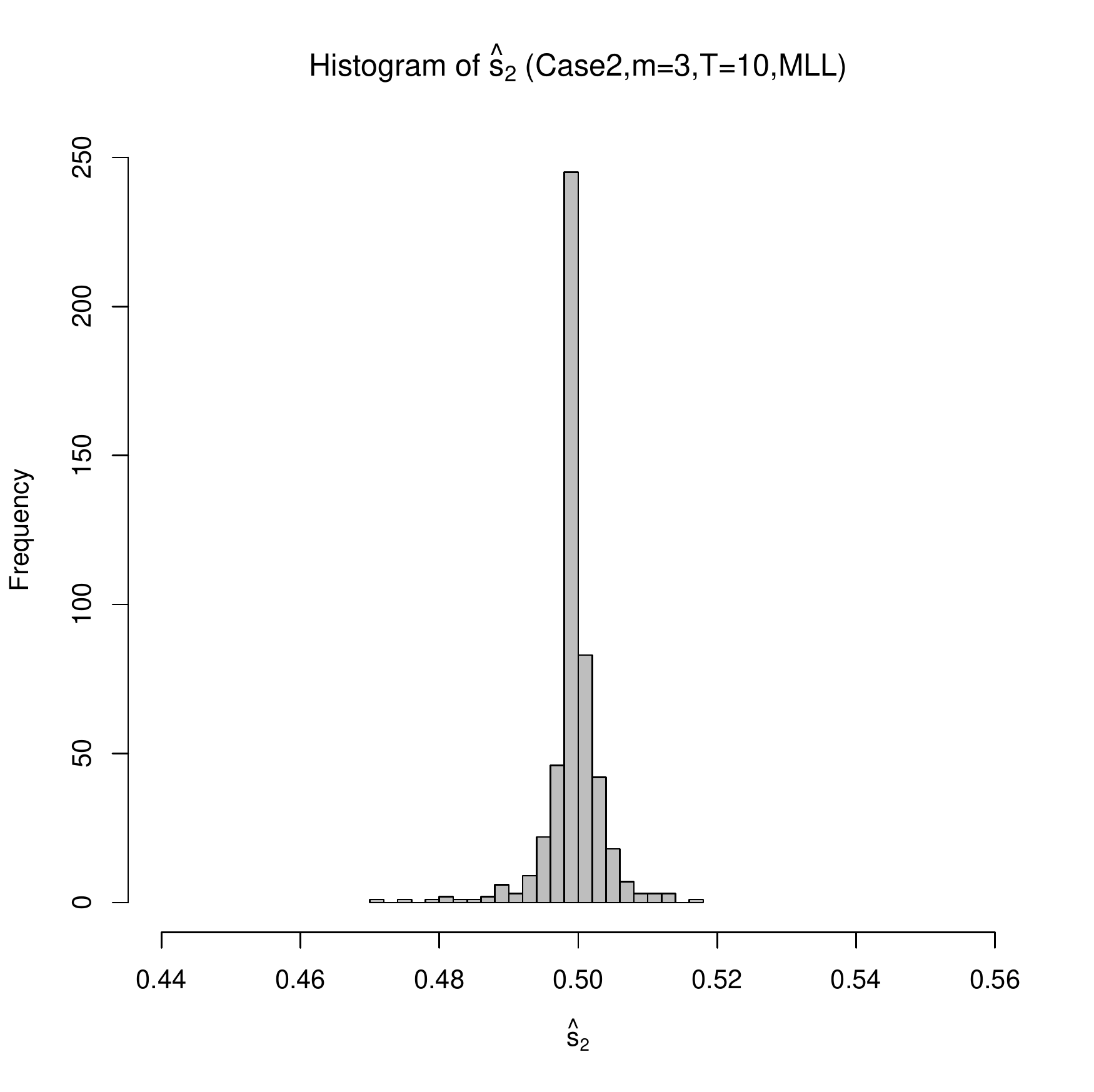}
\includegraphics[height=2.35in,width=2.04in]{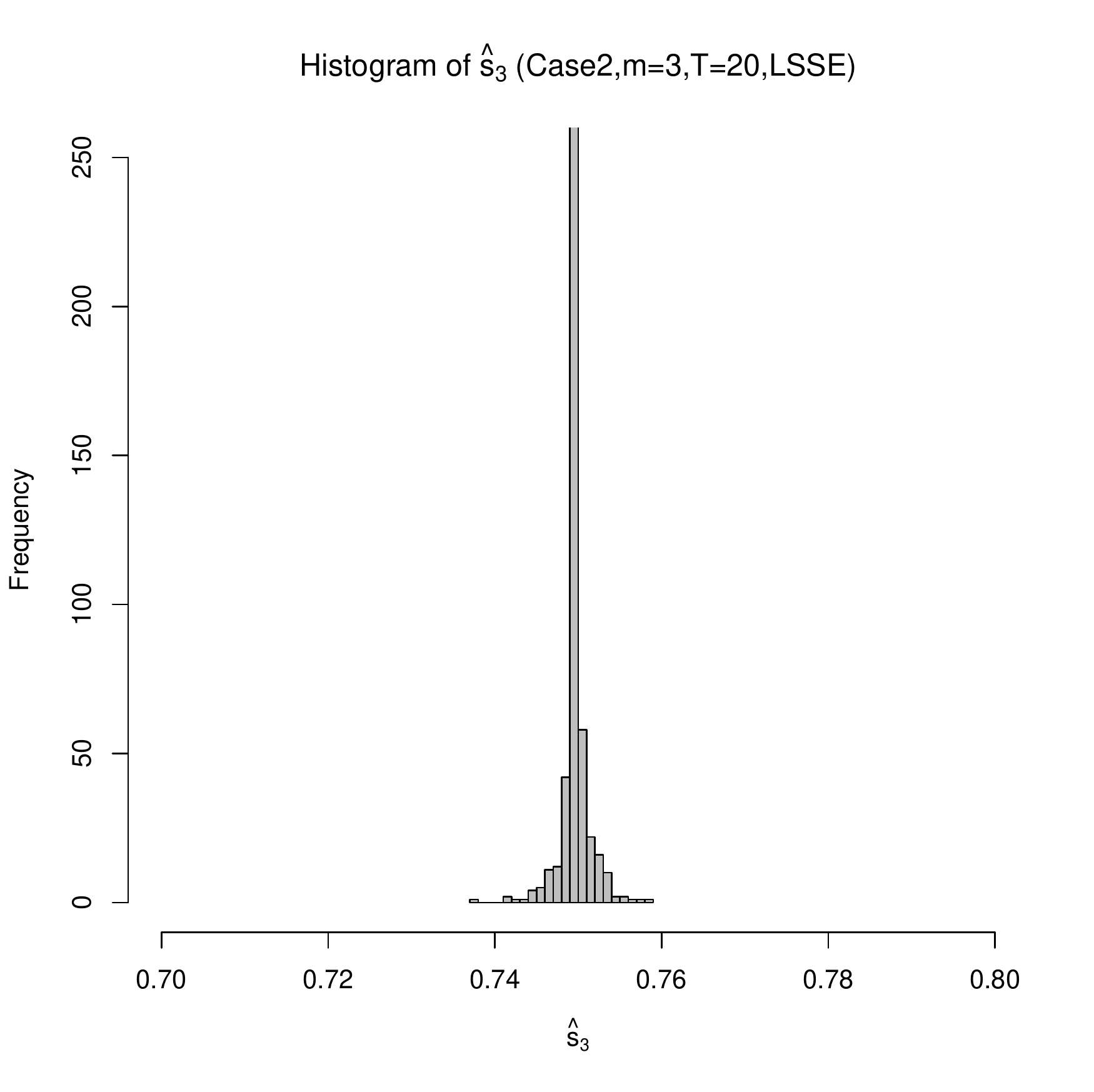}
\caption{\small Histogram of $\hat{s}$ for $m^0=3$ and $T=20$ with exact value $s^0=(0.25, 0.50, 0.75)$}
\label{figure6}
\end{figure}

\newpage
\noindent {\bf Acknowledgements:} \textit{F. Chen and R. Mamon wish to thank the hospitality
and financial support of the Fields Institute for Research in Mathematical
Sciences, Toronto, Ontario, Canada, where this research was initially
conceived and partially conducted. Generous support
on this research collaboration from the Dean of Faculty of Science
is gratefully acknowledged.}\\
\ \\

\end{document}